\documentclass[a4paper,11pt]{article}
\usepackage{jheppub}

\usepackage{bbm}
\usepackage{amsthm}
\newtheorem{theorem}{Theorem}
\theoremstyle{definition}
\newtheorem{definition}{Definition}[section]

\usepackage[utf8]{inputenc}

\usepackage{amssymb,amsmath,amsbsy,mathtools}
\usepackage{braket}
\usepackage{graphicx}
\usepackage{xcolor}
\usepackage{braket}
\usepackage{multirow}
\usepackage{physics}
\usepackage{dcolumn}
\usepackage{hyperref}

\usepackage{natbib}

\numberwithin{equation}{section}

\newtheorem{lemma}{Lemma}[section]

\newtheorem{proposition}{Proposition}[section]

\newtheorem{corollary}[lemma]{Corollary}
\newtheorem{theorem/definition}{Theorem/Definition}
\newenvironment{claim}[1]{\par\noindent\underline{Claim:}\space#1}{}

\newcommand{\mE}{\mathcal{E}}
\newcommand{\mF}{\mathcal{F}}

\newcommand{\mH}{\mathcal{H}}

\newcommand{\mM}{\mathcal{M}}
\newcommand{\mN}{\mathcal{N}}

\newcommand{\mO}{\mathcal{O}}

\newcommand{\lb}{\left(}

\newcommand{\rb}{\right)}

\newcommand{\p}{\partial}

\newcommand{\nn}{\nonumber}
\newcommand{\ee}{\end{equation}}

\usepackage{graphicx}
\usepackage{float}
\usepackage{longtable}

\usepackage{graphicx}
\usepackage{dcolumn}
\usepackage{bm}

\usepackage{physics}
\title{Operator Algebras and Third Quantization}
\begin{document}
	\author[1,a]{Yidong Chen}
	\author[2,a]{Marius Junge}
	\author[3,b]{Nima Lashkari}
	\abstract{
In quantum gravity, the gravitational path integral involves a sum over topologies, representing the joining and splitting of multiple universes. To account for topology change, one is led to allow the creation and annihilation of closed and open universes in a framework often called third quantization or universe field theory. We argue that since topology change in gravity is a rare event, its contribution to exponentially late-time physics is universally described by a Poisson process. This universal Poisson process is responsible for the plateau in the multi-boundary generalization of the spectral form factor at late times. In the Fock space of closed baby universes, this universality implies that the statistics of the total number of baby universes are captured by a coherent state.

Allowing for the creation of asymptotic open universes calls for a noncommutative generalization of a Poisson process. We propose such an operator algebraic framework, called Poissonization, which takes as input the observable algebra and a (unnormalized) state of a quantum system and outputs a von Neumann algebra of a many-body theory represented on its symmetric Fock space.  Physically, Poissonization is a generalization of the coherent state vacua of bipartite quantum systems or matrix quantum mechanics. 

The multi-boundary correlators of the Marolf-Maxfield toy model of baby universes, and the sum over bordisms of closed-open 2D topological quantum field theory, are entirely captured by Poissonization. In the Jackiew-Teitelboim gravity, the universal Poisson process reproduces the late-time correlators in the $\tau$-scaling limit. In Poissonization, asymptotic open boundaries are treated as distinguishable. We argue that the indistinguishability of the bulk dynamical probes, such as with the end-of-the-world branes, is tied to the eigenstate thermalization hypothesis and a projection (conditional expectation) to the center of the one-boundary algebra.}
 
\affiliation[a]{Department of Mathematics, University of Illinois at Urbana-Champaign, Illinois, IL 61801, USA}
\affiliation[b]{Department of Physics and Astronomy, Purdue University, West Lafayette, IN 47907, USA}
 
\emailAdd{$^1$yidongc2@illinois.edu}
\emailAdd{$^2$mjunge@illinois.edu}
\emailAdd{$^3$nima@purdue.edu}

\maketitle

\section{Introduction}


In quantum mechanics and quantum field theory (QFT), the transition amplitudes between states are given by Lorentzian path integrals that sum over all interpolating configurations weighted by the exponential of the action. In the semi-classical regime of small $\hbar$ (large action), the path integral can be written as a sum labeled by the classical solutions (saddle points of the path integral), with each term in the summand capturing the perturbation theory around that classical solution:
\begin{eqnarray}
    \bra{\varphi_f;t_f}U(t_f;t_i)\ket{\varphi_i;t_i}&&=\int_{\phi(t_i)=\varphi_i}^{\phi(t_f)=\varphi_f}\mathcal{D}\phi\: e^{i S[\phi]/\hbar}\simeq \sum_{\phi^{cl}} \int_{\delta\phi(t_i)=0}^{\delta\phi(t_f)=0} \mathcal{D}(\delta \phi) e^{i S[\phi^{cl}+\delta \phi]/\hbar}
\end{eqnarray}
where $\phi^{cl}$ are classical solutions to the equations of motion.
There is an exponentially small probability for the system to tunnel between each pair of classical solutions $\phi^{cl}_1$ and $\phi^{cl}_2$. In the semiclassical approximation, the probability rate for such transitions per unit volume is $\Gamma/V\sim e^{-2S_E/\hbar}$ where $S_E$ is the one-shell action on the instanton that interpolates between the two saddles $\phi^{cl}_1$ and $\phi^{cl}_2$  \cite{coleman1977fate,callan1977fate} \footnote{An instanton is a finite action classical solution to the Euclidean equations of motion. The factor of $2$ has to do with the fact that the transition probability is given by the action of a ``bounce" solution. A bounce is a classical solution of the Euclidean equations of motion that describes a system starting in a metastable vacuum (a local minimum of the potential), ``bouncing" over the potential barrier, and returning to the metastable vacuum. The action of the bounce is twice the action of the instanton that traverses the potential only once.}.


Multi-instanton solutions describe multi-tunneling events. In a general theory with non-linear equations of motion, the instantons can interact, which means that tunneling events at different times are correlated. However, for instantons that are well-separated in time, it is reasonable to use a dilute gas approximation where the total action of a configuration of multiple instantons is the sum of the individual instanton actions\footnote{The assumption is that the interaction between the instantons is small and can be treated perturbatively.}. In this approximation, the probability for $k$ tunneling events is proportional to 
$\frac{(e^{-2S_E/\hbar})^k}{k!}$ which is a Poisson distribution \cite{coleman1988aspects}. The power $k$ arises from treating the tunneling events as independent, and the factor $1/k!$ stems from treating instantons as indistinguishable, as we are only concerned with the total number of events.


A well-known example of tunneling in quantum mechanics is the alpha decay of radioactive atoms. Euclidean methods compute the probability rate for multiple tunneling exactly by summing over the multi-instanton solutions. However, to understand the statistics of many such events over a long time, one might not need the microscopic details of the action. In the dilute gas approximation, one uses the probability rate for a single transition and models the decay as a rare stochastic process. Assuming that decays at different times are independent, the Poisson limit theorems tell us that for exponentially long times the statistics of such events are independent of the detailed microscopics of individual events, and universally modeled by a Poisson distribution\footnote{See Section \ref{subsec:rareevents} for a review of the Poisson limit theorem and a precise statement of its universality.}. In the study of quantum tunneling of charge through a tunneling junction (a thin insulator separating two conductors) this Poisson distribution is called shot noise, or Poisson noise \cite{blanter2000shot}. The Poisson limit theorem applies to the statistics of the total number of occurrences of any set of rare independent events, enjoying such great experimental success that it is often referred to as the law of rare events. 


The Poisson limit theorem should be compared to the central limit theorem. Gaussian and Poisson random variables take values on continuous and discrete sets, respectively. The central limit theorem explains a universal Gaussian kinetic term controlling thermal or quantum fluctuations\footnote{In the thermodynamic limit, the fluctuations around each saddle are $\hbar$ suppressed. Assuming independence, by the central limit theorem, the statistics of these fluctuations are controlled by a Gaussian distribution.}. Here, we argue that the Poisson limit theorem controls the statistics of multi-tunneling events such as topology fluctuations in quantum gravity. 
Tunneling between saddles is a rare event with a probability that is exponentially suppressed $e^{-O(1/\hbar)}$. By the Poisson limit theorem, the statistics of multi-tunneling events at late times are universally controlled by Poisson distributions. It is important to emphasize that, in this work, by late time, we mean $t\sim 1/\Gamma$ where $\Gamma$ is the probability rate for the rare events. The Poisson approximation is violated for both parameterically longer and shorter time scales.


In gauge theories, topologically non-trivial gauge configurations can give rise to distinct classical solutions that contribute to the path integral. 
Similarly, in gravity, one is instructed to include in the gravitational path-integral a sum over all non-trivial topologies that interpolate between the fixed boundaries\footnote{See \cite{hawking1980quantum,hawking1982unpredictability,hartle1983wave,strominger1984vacuum,hawking1987quantum,hawking1988wormholes} for an incomplete list of references.}. The transition amplitude to go from a boundary manifold $\Sigma_i$ with metric $h_i$ to a boundary manifold $\Sigma_f$ with induced metrics $h_f$ is given by 
\begin{eqnarray}
    \braket{h_f}{h_i}=\sum_{\mathcal{M}}\int_{g|_{\Sigma_i}=h_i}^{g|_{\Sigma_f}=h_i} \mathcal{D}g\:e^{i S[g]/G_N}
\end{eqnarray}
where $\mM$ are topologically inequivalent manifolds (cobordisms) that interpolate between $\Sigma_i$ and $\Sigma_f$ with metrics $h_i$ and $h_f$, respectively. In quantum gravity, topology fluctuations is interpreted as processes that create and annihilate open and closed boundaries; see Figure \ref{fig1}. Euclidean methods establish that in the semiclassical regime of small $G_N$, the probability rate for topology fluctuation in quantum gravity suffers an exponential suppression $e^{-2S_E/G_N}$ where $S_E$ is the Euclidean action of the instanton \cite{gross1982instability}. In the remainder of this work, we often work in units where $\hbar$ and $G_N$ are set to one so that the semi-classical limit corresponds to a large Euclidean action. We will see in Section \ref{sec:sumovertopology} that the sum over topology fluctuations in the Euclidean gravitational path-integral, in the dilute gas approximation, results in a Poisson-type distribution. Intuitively, this is expected because in the dilute gas approximation, summing multi-instanton configurations counts the number of set partitions, which is intimately related to the moments of a Poisson distribution through the Bell polynomials and the Bell numbers.

Topology fluctuations in quantum gravity are rare events with a probability rate $\Gamma\sim e^{-2S_E}$ that remains constant in time. In holography, the contribution of topology-changing events to boundary observables (e.g., boundary correlation functions) becomes significant at exponentially long times $T\sim 1/\Gamma$. 
The first result of this work is that, by the Poisson limit theorem, analogous to alpha decay, there is a universal Poisson term in the contribution of topology change to late-time physics that is independent of the microscopics of the quantum gravity model. 

According to the Eigenstate Thermalization Hypothesis (ETH), our discussion can be generalized from gravity to any general quantum chaotic theory. By the ETH, an energy transition due to the insertion of a ``simple" operator is a rare event\footnote{The off-diagonal elements of a simple operator in energy eigenbasis have the interpretation of the amplitude for such transitions. According to the ETH, these off-diagonal elements are exponentially small in entropy (see equation \ref{ETH}).}. We will see that the universal Poisson process can be understood as approximating ``simple" operators and their powers in a microcanonical ensemble of energy $E$ by $\mO\simeq \text{tr}(P_E\mO) P_E$, where $P_E$ is the projection to the microcanonical ensemble. This approximation is relevant for late-time physics and the plateau in the multi-point generalization of the spectral form.

Much of the earlier work on topology fluctuations focused on the probe limit of one asymptotic open boundary (parent universe), interpreting topology fluctuations as the creation and annihilation of closed baby universes \cite{giddings1988axion,giddings1988loss,coleman1988black}. In this limit,  the Hilbert space is $\mathcal{H}_{open}\otimes \mF$, where $\mF$ is the Fock space of closed baby universes. Integrating out this Fock space results in second-order processes that correspond to the emission and reabsorption of baby universes. From the point of view of observers on the open boundary, integrating out this Fock space results in topology fluctuations. When the state of the Fock space is a coherent state $\ket{\alpha}$, the coefficient $\alpha$ becomes a constant of nature in the parent universe. In a general state of the Fock space, integrating out the Fock space results in an ensemble for these couplings. Note that the second-order process of the emission and reabsorption of a baby universe in the Fock space is the number operator $a^\dagger a$ up to a constant shift. We will see that since the statistics of the number operator in a coherent state is given by the Poisson distribution, the Poisson limit theorem suggests that coherent states capture the universal late-time physics of topology fluctuations.


More generally, going beyond the probe limit, one has to introduce operators that create and annihilate closed and open universes, in a framework that is often referred to as the {\it third quantization} or {\it universe field theory} \cite{coleman1991quantum}. In third quantization and universe field theory, we consider the space of boundary conditions $J$ (superspace), and introduce the field operator $\phi(J)\sim a(J)+a^\dagger(J)$ with operators $a(J)$ that annihilate boundaries with boundary condition $J$. The Klein-Gordon equation for $\phi(J)$ is the Wheeler-DeWitt equation, and its vacuum is the Hartle-Hawking no-boundary wave-function \cite{hartle1983wave}. In the universe field theory, the process of emission/absorption of universes is described via a vertex interaction that couples boundaries (e.g., see Section ``Baby Universes" in \cite{coleman1991quantum}). The inclusion of this interaction term results in a renormalization of the Hartle-Hawking vacuum $\ket{HH}$ and the propagator of a single universe. The connected $p$-point correlators $\bra{HH}\phi(J_1)\cdots \phi(J_p)\ket{HH}_{conn}$ $\varphi(J)$ are interpreted as the $p$-boundary wormholes with boundary conditions $J_1, \cdots, J_p$\footnote{In probability theory and statistics, connected correlation functions are called cumulants.}.
The universal Poisson term we are concerned with comes from fixing this connected correlator to be a particular correlation function in the Hilbert space of a single boundary, and the count of configurations with multiple disjoint universes with $p$ boundaries. Note that the set partitions are intimately related to the Poisson distribution (see Appendix \ref{app:setpartitions} for a review).

\begin{figure}[t]
    \centering
    \includegraphics[width=1\linewidth]{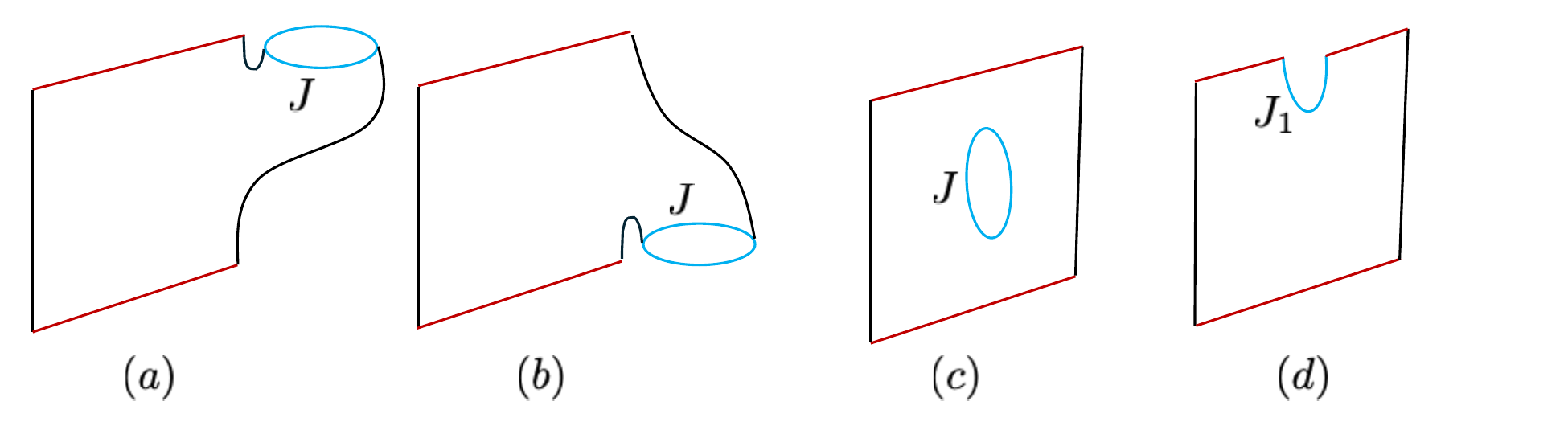}
    \caption{\small{The (a) emission and (b) absorption of a closed baby universe are topologically equivalent to (c). Here, $J$ is the choice of boundary conditions on the baby universe. (d) The process of emitting an open universe with boundary condition $J_1$ on the half-circle.}}
    \label{fig1}
\end{figure}

The Euclidean manifold that describes the emission of a single baby universe can be cut differently according to the figure \ref{fig2}. In this picture, the process of emission of a closed baby universe is a second-order effect describing the creation and annihilation of an open universe\footnote{The equivalence of these two pictures in two-dimensional quantum gravity is reminiscent of the closed-open string duality.}. The Hilbert space of closed baby universes is one-dimensional. If the closed baby universes come in different types, e.g., energy or length in JT gravity, we can associate a commutative algebra to baby universes\footnote{For a concrete example, see the commutative algebra of operators generated by (\ref{Ob}) in pure JT gravity.}. We will see that the Poisson limit theorem suggests that multi-mode coherent states in the Fock space of closed baby universes capture a universal contribution to the late-time physics of the open boundary. In the case of JT, the Poisson distribution manifests itself in the so-called $\tau$-scaling limit and the plateau in the multi-point generalization of the spectral form factors \cite{okuyama2020multi,blommaert2023integrable,saad2024convergent}. 


Recent progress in lower-dimensional models of quantum gravity revived interest in Euclidean wormholes \cite{saad2019jt}. In \cite{MM}, Marolf and Maxfield revisited the Fock space of closed and open boundaries in topological 2D gravity, elaborating the role of topology fluctuations in quantum gravity in asymptotically Anti-de Sitter spacetimes. 
Their construction was later generalized to arbitrary 2D open/closed topological quantum field theories in \cite{gardiner20212D,@Banerjee-Moore}, and the Jackiew-Teitelboim (JT) gravity in \cite{penington2023algebras}. For an explicit construction of a universe field theory for JT, see \cite{post2022universe}.

More generally, one can allow for the creation and annihilation of both closed and open universes, including asymptotic boundaries \cite{penington2023algebras}. Asymptotic open boundaries, in general, have noncommutative algebras. We introduce creation and annihilation operators for asymptotic open universes, and using a noncommutative Poisson limit theorem, argue that late-time physics ($t\sim 1/\Gamma$) is universally controlled by a noncommutative generalization of a Poisson process. A key result of this work is a mathematically rigorous framework called Poissonization that we argue captures this universal aspect of late-time fluctuations and summing over topologies.

\begin{figure}[t]
    \centering
    \includegraphics[width=0.6\linewidth]{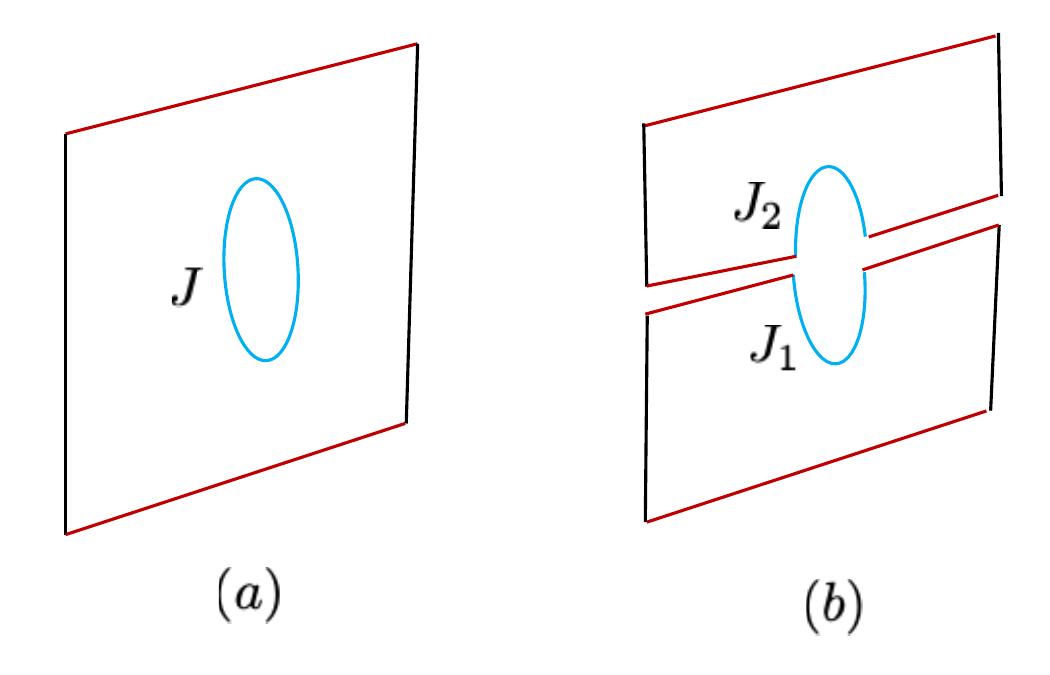}
    \caption{\small{Cutting a closed baby universe in half results in a one-dimensional Hilbert space that we refer to as an open baby universe. (b) We can choose different boundary conditions on the upper and lower half-circles. Then, the emission of a baby universe in (a) can be interpreted as a second-order process of emitting and reabsorbing an open baby universe.}}
    \label{fig2}
\end{figure}

In general relativity, one is expected to treat excitations on asymptotic open boundaries as distinguishable. This distinguishability naturally leads to the Poissonization framework. Poissonization takes as input the (noncommutative) observable algebra of a quantum system in an unnormalized state (weight) and lifts it to a (noncommutative) von Neumann algebra in a (noncommutative) coherent state represented on the Fock space.
We argue that indistinguishable bulk dynamical degrees of freedom can be accounted for by symmetrizing the correlators of Poissonization. We will see that the bulk dynamical fields are the trace of the boundary random operators. Physically, one can use the ETH to justify random operators, and the trace can be viewed as a conditional expectation from the noncommutative algebra of the open boundaries to the commutative algebra of closed boundaries.

This paper is structured as follows. In Section \ref{sec:sumovertopology}, we argue that the late-time effects of topology fluctuation as a rare discrete event can be modeled as a Poisson process. We motivate a noncommutative generalization of the Poisson process (Poissonization) using creation and annihilation operators for open universes, including the asymptotic ones. In Section \ref{sec:bipartitecoherent}, we present Poissonization as coherent states of bipartite quantum systems, and compare it to second quantization. In Section \ref{sec:babyuniversescommutative} and \ref{sec:babyuniversesnonAbel}, we show that the Marolf-Maxfield toy model of baby universes, the sum over bordisms of closed-open Topological Quantum Field Theory (TQFT), and the simplified Jackiew-Teitelboim (JT) gravity with the End of the World (EOW) branes are all examples of Poissonization. To keep the notation easier to follow, in Appendix \ref{app:symbols}, we have included a list of symbols and their definitions. 
 
\section{Sum over topologies and Poisson distribution}\label{sec:sumovertopology}

In subsection \ref{subsec:rareevents}, we start by reviewing the Poisson limit theorem, explaining why the statistics of rare and independent events are universally controlled by a Poisson distribution. The example familiar to physicists is the emission of quanta from a radioactive source. The key assumption that underlies this universality is that the emissions of particles at different times are rare and statistically independent events. The example of $\alpha$-decay most resembles our discussion of topology change, because both processes are discrete and rare due to the exponential suppression of the probability of quantum tunneling. However, the universality of the Poisson limit theorem also applies to any source that, for any other physical reason (e.g. very weak coupling), emits quanta with vanishingly small probability\footnote{In the Poisson limit theorem, all that matters is that the events are independent and rare, and we wait long enough so that the expectation value of the total number of emissions is order one (exponentially long times).}.

In Subsection \ref{subsec:coherentrareevent}, we model the emission and absorption of photons by radioactive atoms as quantum events, and argue that the universal Poisson distribution is reflected in the statistics of number operators in a coherent state of photons. In the remainder of this section, we argue that, in quantum gravity, the emission and absorption of closed and open universes are rare events, with an amplitude exponentially suppressed in entropy. In analogy with the emission of photons from a radioactive source, it is best modeled using Poisson statistics and coherent states. 

In the case of closed baby universes, the algebra of observables is commutative\footnote{For instance, as we see in Section \ref{sec:babyuniversescommutative}, in 2D closed topological TQFTs like Dijkgraaf-Witten, each boundary is labeled by an irreducible representation of a gauge group. The commutative algebra is generated by $\pi_r$, the projections to the irreducible representation $r$.}. The Poisson limit theorem points to the universality of multi-mode coherent states. In 2D dimensions, we can use the open/closed duality to view the process of emission of a closed baby universe as a second-order process of emission and reabsorption of an open baby universe. This allows us to contemplate the further generalization of allowing for the creation and annihilation of asymptotic open universes, as required by third quantization or universe field theory. The algebra of such boundaries is noncommutative. Subsection \ref{subsec:asymptoticopen} uses physics intuition to motivate the key ingredients of our main construction (Poissonization) as a universal model for third quantization of a general asymptotic boundary theory; e.g., the algebra of a single copy of $N=4$ super Yang-Mills.

\subsection{Rare events and Poisson distribution}\label{subsec:rareevents}


Let us start by modeling the emission of photons by decaying atoms as classical events. Define a random variable $n_t$ (number of emitted photons) associated with a time interval $(t,t+\Delta t)$, which is equal to one if a decay occurs in this time interval or zero if it does not. We say we have a rare event if the probability $p(n_t=1)\sim \epsilon\ll 1$ and we can neglect the probability of emitting multiple photons in a single time interval.

For a rare event, each $n_t$ is a binary random variable. Assume for now that the probability of decay in the time interval $(t, t+\delta t)$ is independent of $t$ (time-translation invariant process), and the random variable $n_t$ at different times is independent: $p(n_t, n_{t'}) = p(n_t)p(n_{t'})$. In other words, $n_t$'s are independent and identically distributed (i.i.d.) random variables. Then the probability of emitting at least $1$ photon after time $T = M\Delta T$ is given by\footnote{We warn the reader that, in this work, the variable $M$ will be used to denote the total number of events to avoid potential confusion with the total number of emissions that we denote by $N$.}:
\begin{equation}
    1 - (1 - \epsilon)^{M} \approx M\epsilon
\end{equation}
This probability is of $O(1)$ if $M = O(1/\epsilon)$ (i.e. $T = O(\Delta t / \epsilon)$). Typically, this is an exponentially late time (c.f. Subsection \ref{subsec:coherentrareevent}). At this late time, the probability of emitting a total of $k$ photons is given by the binomial distribution:
\begin{equation}
    N(T) = \sum_{t = 1}^{M} n_t
\end{equation}
\begin{equation}
    \text{Prob}(N(T) = k) = \binom{M}{k}\epsilon^k(1 - \epsilon)^{M - k}
\end{equation}
In the limit $\epsilon \rightarrow 0$, the probability density has a well-defined limit:
\begin{equation}
    \lim_{\epsilon\rightarrow 0}\text{Prob}(N(T) = k) = \frac{e^{-\lambda}\lambda^k}{k!}
\end{equation}
where $\lambda = M\epsilon = O(1)$. This is the Poisson distribution with intensity $\lambda$. For finite but infinitesimally small $\epsilon$, the binomial distribution can be effectively approximated by this Poisson distribution.

Alternatively, for a fixed time interval $[0, T]$, assume the \textit{rate} of emission is constant and is given by $\frac{\lambda}{T}$. We divide the time interval into $M$ equal-length subintervals: $0 = t_0 < \frac{T}{M} = t_1 < \frac{2T}{M} = t_2 < \cdots < \frac{(M - 1)T}{M} = t_{M - 1} < T = t_M$. Since the rate is constant, the probability of an emission within each subinterval is given by $\frac{\lambda}{T}\frac{T}{M} = \frac{\lambda}{M}$. This probability is small when $M \gg 1$. In this case, we can assume that the number of emissions with each subinterval is a Bernoulli random variable with success probability $\frac{\lambda}{M}$ and each Bernoulli random variable is mutually independent. The total number of emissions also follows the binomial distribution:
\begin{equation}
    \text{Prob}(N(T) = k) = \binom{M}{k}\lb \frac{\lambda}{M}\rb^{k}\lb 1 - \frac{\lambda}{M}\rb^{M - k}
\end{equation}
When $M\rightarrow\infty$, the Poisson limit arises:
\begin{equation}
    \lim_{M\rightarrow\infty}\text{Prob}(N(T) = k)= \frac{e^{-\lambda}\lambda^k}{k!}
\end{equation}
The result above is called the Poisson limit theorem. The universality of this result lies in the fact that it continues to hold even if the events have weak dependence on time and there are weak violations of independence. This means that the result above is robust against correlations that fall off fast enough at late times.

 More generally, consider a collection of atoms interacting with photons in a cavity. Assume that initially there are $k_i$ photons present in the cavity. At each time interval, the atoms can emit a photon or absorb one of the photons, and assume that emission and absorption are independent and equally rare events. Atoms can emit photons and reabsorb them later in an event that is doubly suppressed.
 Repeating the argument of the Poisson limit theorem, we find that the probability for the atoms to absorb a total of $k_i-m$ atoms and emit $k_f-m$ atoms is universally given by the multiplication of Poisson distributions
 \begin{eqnarray}\label{rareeventscoleman}
     p_\lambda(k_i,k_f,k;m)={k_i \choose m} {k_f \choose m} p_\lambda(k_i-m) p_\lambda(k_f-m)
 \end{eqnarray}
 where the binomials ${k_i \choose m}$ and ${k_f \choose m}$ count the number of ways one can emit $k_i-m$ and $k_f-m$ photons from $k_i$ and $k_f$ present, respectively. In subsection \ref{subsec:topol}, we will see that the argument above can be repeated if we replace photons with baby universes. Then, the Euclidean path-integral provides a pictorial representation of the Poisson process above; see Figure \ref{fig3}.

\begin{figure}[t]
    \centering
    \includegraphics[width=\linewidth]{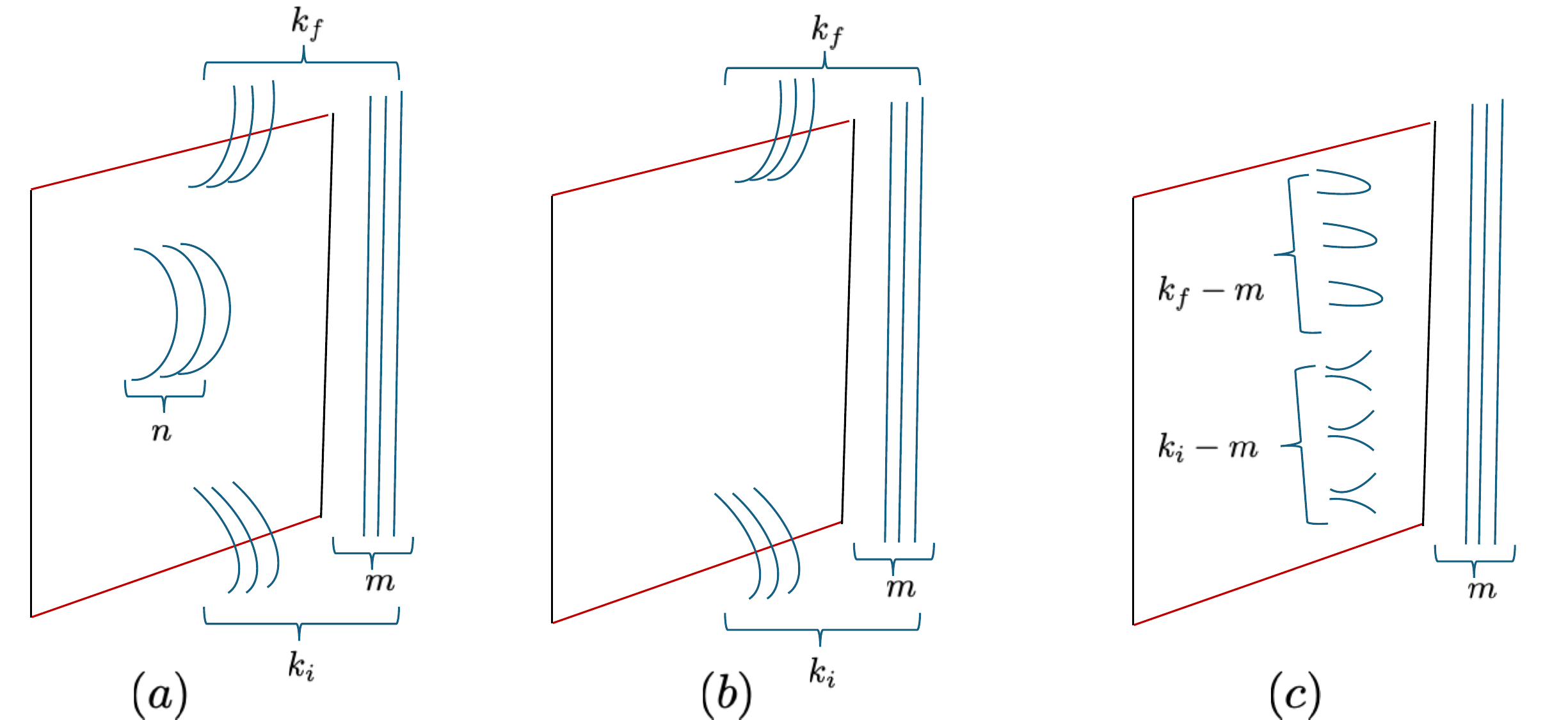}
    \caption{\small{(a) The Euclidean path-integral that represents the absorption of $k_i-m$ baby universes and emission of $k_f-m$ baby universes, while emitting and reabsorbing $n$ (b) Setting $n=0$, we can view this as an amplitude whose probability is described by (a) with $n\to k_f$ and $k_f\to k_i$. (c) By rearranging the resulting diagram, we can express it as the multiplication of probabilities (independent events).}}
    \label{fig3}
\end{figure}

Note that the Poisson limit theorem is a statement about the probability distribution of the total number of events\footnote{If the events are indistinguishable, e.g., photons in this example, then the total number of events is the only meaningful notion.}. Now, consider the atoms can emit photons in $k$ different wavelengths $b_i$, each as an independent rare event. Photons of the same wavelengths are indistinguishable; however, photons of different wavelengths are distinguishable. Repeating the Poisson limit argument, we find that the limiting distribution is the product of $l$ independent Poisson distributions.
\begin{eqnarray}
    p(N(b_1)=k(b_1), \cdots N(b_l)=k(b_l))=\prod_{i=1}^l p_{\lambda_i}(N(b_i)=k(b_i))\ .
\end{eqnarray}
where we have assumed that the weights $\lambda_i$ are different for different wavelengths.

\subsection{Rare quantum events, coherent states, and Poisson distribution}\label{subsec:coherentrareevent}
The discussion above treated emission and absorption of photons as classical events. Here, we generalize our discussion of the Poisson limit theorem, where the emission and absorption of atoms are treated as rare quantum events. We provide a microscopic explanation of how the Poisson limit theorem and coherent states arise as effective descriptions of quantum rare events for late-time observers.

Consider a parent quantum system $P$ and its Hilbert space $\mathcal{H}_P$. The parent system may undergo \textit{rare} events of emitting/absorbing subsystems. We assume all emitted subsystems are identical and all events are mutually independent. For the current discussion, we do not need to know any other details. Let $\mathcal{H}_C$ denote the subsystem's Hilbert space \footnote{The subscript $C$ stands for children.}. Because all emitted subsystems are identical, the entire composite system has a Hilbert space $\mathcal{H}_P\otimes\mathcal{F}_\text{sym}(\mathcal{H}_C)$.

For simplicity, we assume that the initial state of the system is $\ket{\Omega}$ \footnote{We use the familiar notation for vacuum to indicate that the initial state is assumed to be free of emissions.}. Each emission is described by a creation operator $a^\dagger_i$ where the emitted system is in the state $\ket{i}$. The entire system may have a nontrivial Hamiltonian-driven time evolution. However, this is inessential for our discussion here. Again, for simplicity, we assume the emitted subsystem is always in the stationary state $\ket{f}$, and the system only goes through the emission/absorption process with no other time evolution. The creation operator for the state $\ket{f}$ is denoted as $a^\dagger_f$, and the annihilation operator is denoted as $a_f$. 

Since we are concerned with late-time observers, we focus on quantum events in a fixed time interval $[0,T]$. As before, we consider dividing the interval into $M$ subintervals (c.f. Subsection \ref{subsec:rareevents}). Each subinterval has length $\frac{T}{M}$. When $M$ is sufficiently large, we assume there is at most one emission in each subinterval. By mutual independence, the number of emitted subsystems in each subinterval can be modeled by a Bernoulli (0-1) random variable $M_i$ ($1\leq i\leq M$) and the total number of emissions is simply the sum $\sum_{1\leq i\leq M}M_i$. The success probability of each Bernoulli random variable is proportional to the length of the interval $p_M := \lambda\frac{T}{M}$ \footnote{Here the subscript $M$ emphasizes the fact that the probability is proportional to the length of the subinterval and depends on $M$.}. The key feature is that the probability is proportional to the length of each subinterval \footnote{This is the key to arrive at the Poisson limit.}. In each subinterval, emission happens stochastically. The entire emission process can be modeled by:
\begin{equation}\label{equation:emissionevolution}
    \ket{\Omega}\mapsto\prod_{1\leq i\leq N}\lb a^\dagger_f\rb^{M_i}\ket{\Omega} = \lb a_f^\dagger\rb^{\sum_i M_i}\ket{\Omega}
\end{equation}
Notice the operator $\lb a_f^\dagger\rb^{M_i}$ is a random operator. It is identity $\mathbbm{1}$ when $M_i = 0$ and $a^\dagger_f$ when $M_i = 1$. This is one way to mathematically model the physical process of stochastic emission. The resulting state $\prod_{1\leq i\leq M}\lb a^\dagger_f\rb^{M_i}\ket{\Omega}$ is a random state. In terms of the number basis, it can be simply written as $\ket{\sum_i M_i}$. This state is random because the total number $\sum_i M_i$ is a random variable. Using the binomial distribution, the probability of $\sum_i M_i = k$ is given by the multinomial distribution:
\begin{equation}
    \text{Prob}(\sum_i M_i = k) = \binom{M}{k}p_M^k(1 - p_M)^{M - k}
\end{equation}
Therefore, we can write the density matrix of the random state in the more familiar form:
\begin{equation}
    \sum_{0\leq k\leq M}\text{Prob}(\sum_i M_i = k)\ket{k}\bra{k}
\end{equation}
where $\ket{k} := \frac{1}{\sqrt{k!}}\lb a_f^\dagger\rb^k\ket{\Omega}$. The factorial is to ensure the number state is properly normalized.

When $M$ approaches infinity, the Poisson limit kicks in. In this case, the density matrix effectively becomes:
\begin{align}
    \begin{split}
        &\lim_{M\rightarrow\infty}\sum_{0\leq k\leq M}\text{Prob}\lb\sum_i M_i = k\rb\ket{k}\bra{k} \nn\\
        &= \lim_{M\rightarrow\infty}\sum_{0\leq k\leq M}\binom{M}{k}\lb\lambda\frac{T}{M}\rb^k(1 - \lb\lambda\frac{T}{M}\rb)^{M - k}\ket{k}\bra{k}
        \\&=\sum_{k\geq 0}\frac{e^{-\lambda T}\lb\lambda T\rb^k}{k!}\ket{k}\bra{k}
    \end{split}
\end{align}
This is the classical-quantum state corresponding to the Poisson distribution. This state is intimately related to the coherent state $\ket{W((\lambda T)^{1/2})}$ (c.f. Subsection \ref{subsection:coherentexamples}). First, we calculate the expectation value of the total number operator $\widehat{N}$:
\begin{align}
    \begin{split}
        \bra{W((\lambda T)^{1/2})}\widehat{N}\ket{W((\lambda T)^{1/2})} &= \sum_{k,l\geq 0}\frac{e^{-\frac{\lambda T}{2}}(\lambda T)^{k/2}}{\sqrt{k!}}\bra{k}\widehat{N}\frac{e^{-\frac{\lambda T}{2}}(\lambda T)^{l/2}}{\sqrt{l!}}\ket{l}
        \\&=\sum_{k\geq 0}\frac{e^{-\lambda T}(\lambda T)^k}{k!}\bra{k}\widehat{N}\ket{k}
    \end{split}
\end{align}
This is precisely the expectation value of $\widehat{N}$ under the Poisson classical-quantum state. More generally, for any operator that \textit{preserves} the number of emissions (i.e., commutes with $\widehat{N}$) \footnote{It is well-known that the number-preserving operators form a von Neumann algebra. This algebra is generated by operators of the form: $\Gamma(x) = \sum_{k\geq 0}\mathcal{O}^{\otimes k}$ \cite{quantStoch}.}, we have:
\begin{align}
    \begin{split}
        \bra{W((\lambda T)^{1/2})}\Gamma(x)\ket{W((\lambda T)^{1/2})} &= \sum_{k,l\geq 0}\frac{e^{-\frac{\lambda T}{2}}(\lambda T)^{k/2}}{\sqrt{k!}}\bra{k}\Gamma(x)\frac{e^{-\frac{\lambda T}{2}}(\lambda T)^{l/2}}{\sqrt{l!}}\ket{l}
        \\
        &=\sum_{k\geq 0}\frac{e^{-\lambda T}(\lambda T)^k}{k!}\bra{k}\mathcal{O}^{\otimes k}\ket{k}
    \end{split}
\end{align}
Again, this is precisely the expectation value of a generic number-preserving operator $\Gamma(x)$ under the Poisson classical-quantum state. 

The Poisson classical-quantum state is an example of what we call a Poisson state in Poissonization (c.f. Section \ref{sec:bipartitecoherent} \footnote{Here the Poisson state is a classical-quantum state. In general, this need not be the case. For more examples of Poisson state in Poissonization, refer to Section \ref{sec:bipartitecoherent}}.  The discussion above shows that \textit{if the observables are restricted to number-preserving operators}, the Poisson state coincides with the pure density matrix formed by the coherent state. 

\paragraph{Quantum Markovian Evolution} An alternative description of the emission process discussed above is given by a quantum Markov semigroup. It is well-known that under mild continuity conditions, a quantum Markov semigroup is driven by a Lindbladian \cite{lindbladian, davies, davies2, davies3, cipriani}. And for spontaneous decay with a single decay mode (i.e. $\ket{f}$), the emission Lindbladian is given by \footnote{The Lindbladian here does not satisfy detailed balance because we only considered the emission process.}:
\begin{equation}
    \mathcal{L}_*\rho = a_f\rho a_f^\dagger - \frac{1}{2}\{a^\dagger_fa_f, \rho\}
\end{equation}
And the master evolution equation is given by:
\begin{equation}
    \frac{d}{dt}\rho(t) = \mathcal{L}_*\rho
\end{equation}
For an observer interested in the evolution of individual quantum trajectories, one can apply the so-called \textit{quantum jump method} to decompose the evolution of individual quantum trajectories into two steps:
\begin{enumerate}
    \item \textbf{Continuous (non-)Hermitian evolution:} During this step within an infinitesimal time interval $\delta t$, the system evolves with non-Hermitian Hamiltonian. In our case, this non-Hermitian Hamiltonian is given by:
    \begin{equation}
        H_\text{eff} = \frac{i\hbar}{2}a^\dagger_fa_f
    \end{equation}
    This is proportional to the number operator. This non-Hermitian Hamiltonian drives dissipation. For each number state $\ket{k}$, the state decays by $e^{-\frac{\hbar k}{2}}\ket{k}$ eventually converging to $0$. In other words, in the long-time limit, this non-Hermitian evolution projects onto the $0$-emission state. 
    \item \textbf{Quantum jumps:} Within an infinitesimal time interval $\delta t$, a jump (i.e. an emission) occurs with probability proportional to $\tr(a_f\rho a_f^\dagger)\delta t$. For the Poisson state $\rho = \sum_{k\geq 0}\frac{e^{-\lambda}\lambda^k}{k!}\ket{k}\bra{k}$, the jump probability is given by:
    \begin{equation}
        \tr(a_f\rho a^\dagger_f)\delta t = \lambda \delta t
    \end{equation}
    This shows that the pure jump component of the quantum trajectory is indeed modeled by a Poisson process with intensity $\lambda$.
\end{enumerate}

\subsection{Topology change as a rare event}\label{subsec:topol}

 In quantum gravity, we can repeat the analysis above, replacing atoms with an asymptotic open universe with a Hilbert $\mH_{open}$ and photons with closed baby universes to mimic the probe limit of topology fluctuations. Going beyond the probe limit, in an analogy, we think of a cavity that emits entangled pairs of atoms (asymptotic open universes) in some state of $\mH_{open}$, as well. In this section, we mostly focus on two-dimensional quantum gravity.
 
 In gravity, we will consider three different types of universes:
 \begin{itemize}
     \item {\bf Asymptotic open universes:} Asymptotic open universes have large-dimensional Hilbert spaces ($e^S\gg 1$) and a non-commutative algebra of observables. We can view the open universe as the gravity dual of an entangled state of two holographic boundary conformal field theories. It is often convenient to consider the thermofield double state so that the state of the open universe is prepared using a gravitational Euclidean path integral on a half-disk in 2D. 
     There are various equivalent ways to describe the state of an open boundary (see Figure \ref{fig4}): 1) a (entangled) bipartite state of two boundary theories  2) an operator in the algebra of one boundary CFT that creates this state by acting on the unnormalized maximally entangled state: $\ket{\mO}_{11'}=\sum_l(\mO\otimes 1)\ket{ll}_{11'}$, 3) the gravitational wave-function on the co-dimension one Cauchy surface of the bulk dual of the state.

\item {\bf Closed baby universes:} Closed baby universes have one-dimensional Hilbert spaces and a trivial Abelian algebra if they come in different flavors $J$. In a gravitational theory, closed baby universes carry labels set by the boundary conditions of the metric $J$. In a topological gauge theory, such as BF theory, they can be labeled by an irreducible representation of a gauge group. In 2D pure JT gravity, closed baby universes are labeled by a positive parameter $b>0$, which corresponds to the circumference (volume) of the closed baby universe.

\item{\bf Open baby universes:} We can slice a gravitational path-integral such that the closed baby universe is cut in half, resulting in what we call open baby universes; see Figure \ref{fig2}. Each open baby universe is labeled by the boundary condition $J$ of the metric, the disk in 2D (half-sphere in higher dimensions). This allows for introducing different boundary conditions $J$ and $J'$ on the two halves of a closed baby universe; see Figure \ref{fig2}. 

 \end{itemize}
 In gravity, it is natural to treat excitations on asymptotic boundaries as distinguishable and baby universes as indistinguishable. This means that a configuration with two asymptotic boundaries described by the pair of operators $(\mO_1,\mO_2)$ is distinct from that of $(\mO_2,\mO_1)$; i.e. order matters \footnote{Equivalently, in describing asymptotic open universes in vectors in bipartite systems the set $(\ket{\Psi_1},\ket{\Psi_2})$ is an ordered set.}. Whereas for baby universes, there is no meaningful way of defining an order; see Figure \ref{fig5}. This is reflected in the fact that, as opposed to open baby universes that have a commutative algebra, the algebra of observables of asymptotic open boundaries is, in general, noncommutative. The location of the baby universe in the gravitational path integral is decided dynamically in the bulk, which means that in the quantum theory, we will have to sum over all insertions. In comparison, the states of the asymptotic boundary universes are fixed and not summed over.

\begin{figure}[t]
    \centering
    \includegraphics[width=1\linewidth]{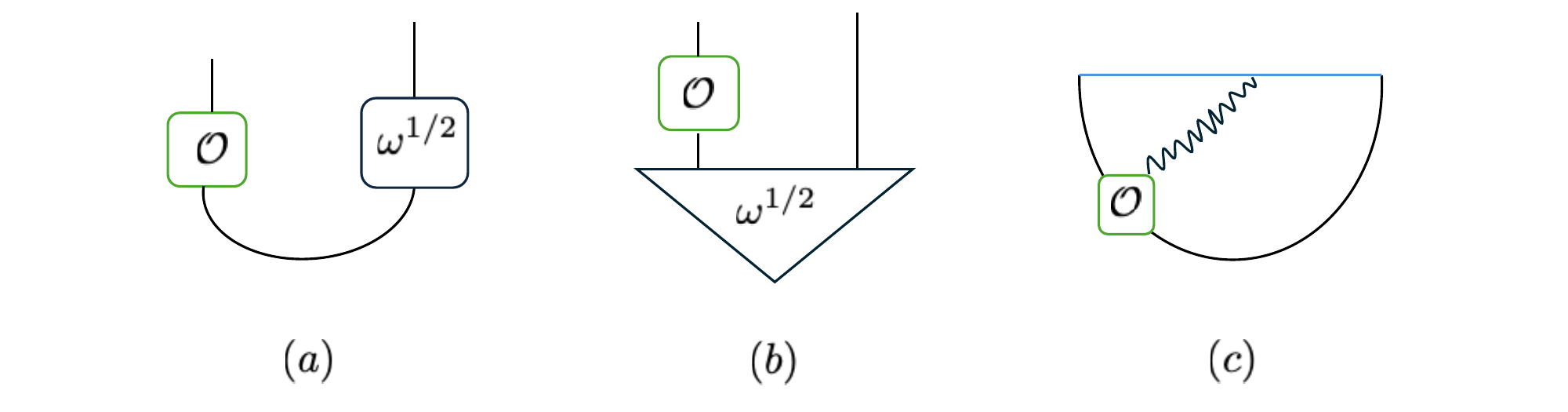}
    \caption{\small{(a) and (b) are equivalent representations of the state $(\mO\otimes 1)\ket{\omega^{1/2}}$ using tensor diagrams. (c) In a gravitational setup, the insertion of some operator $\mO$ is dual to a deformation of the bulk geometry.}}
    \label{fig4}
\end{figure}

\paragraph{Euclidean gravitational path-integral:}

Earlier discussions of topology change in gravity focused on the emission and absorption of closed baby universes by the probe limit of one asymptotic open universe (parent universe). In \cite{giddings1988loss}, using Euclidean gravitational path integrals and the dilute gas approximation of instantons, the authors computed the amplitude for a parent universe with $k_i$ initial baby universes to evolve to $k_f$ final baby universes after Euclidean evolution for $T$, emitting and absorbing a total of $k_i-m$ and $k_f-m$ baby universes, respectively, and emitting and reabsorbing $n$ baby universes:
\begin{eqnarray}\label{instantoncount}
    \bra{k_f}e^{-H T}\ket{k_i}_{m}&&\equiv \sum_n  \bra{k_f}e^{-H T}\ket{k_i}_{m,n}\nn\\
    &&=\sum_n \frac{\sqrt{k_i!k_f!}}{2^nm!n!}\frac{z^{2n+k_i+k_f-2m}}{(k_i-m)!(k_f-m)!} \nn\\
    &&=e^{\frac{1}{2}z^2}\sqrt{k_i!}\sqrt{k_f!}\sum_{m=0}^{\min(k_i,k_f)}\frac{z^{k_i+k_f-2m}}{m!(k_i-m)!(k_f-m)!}
\end{eqnarray}
for $z=e^{-S_E}VT\equiv \sqrt{\lambda}$ where $V$ is the volume of the parent universe, and $H=e^{-S_E}V (a+a^\dagger)$. Here, $S_E$ is the on-shell action of the instanton that describes the tunneling amplitude corresponding to the emission/absorption of a baby universe; see Figure \ref{fig3} \cite{giddings1988axion}. 
The $n=0$ term in the summand in (\ref{instantoncount}) can be interpreted as the amplitude of tunneling from a state with $k_i$ initial baby universes to a state with $k_f$ final baby universes, in a process that involves only the absorption of $k_i-m$ and the emission of $k_f-m$ baby universes. The probability for this process is
\begin{eqnarray}
&&|\bra{k_f}e^{-HT}\ket{k_i}_{m,0}|^2={k_i \choose m}{k_f \choose m}p_\lambda(k_f-m)p_\lambda(k_i-m)
\end{eqnarray}
where $p_\lambda(n)$ is a Poisson distribution of weight $\lambda=z^2$. This matches precisely the answer we found in (\ref{rareeventscoleman}).

As we argued earlier, the appearance of Poisson distributions in a count of instantons in the dilute gas approximation is not surprising. The result (\ref{instantoncount}) follows from a counting argument whose combinatorics are controlled by the number of set partitions. The moments of Poisson distributions precisely capture such combinatorics; see Appendix \ref{app:setpartitions}. Figure \ref{fig3} shows pictorially why we should expect the instanton count to match the counting in (\ref{rareeventscoleman}). 

In the two-dimensional quantum gravity models we study explicitly, the small parameter $e^{-S_E}$ is replaced by a small parameter $e^{-S_0}$ that controls the genus expansion. We will see that, in the random matrix theory description, the small parameter $e^{-S_0}$ is proportional to the averaged density of states $\langle\rho(E)\rangle$ that we sometimes denote by $\rho_0(E)$. For this reason, in the remainder of this section, we treat $S_0$ as a measure of the entropy of the system in some microcanonical or canonical ensemble, and sometimes denote it by $S$.

\paragraph{Lorentzian gravitational path-integral:} The connection to the Poisson limit theorem becomes more direct once we consider the contribution of baby universes to the real-time path integrals. Consider the path integral that prepares a thermofield double state with some inverse temperature $\beta$. 
This state evolves in real time with some complicated Hamiltonian $H_L+H_R$. We could alternatively use the bulk Hamiltonian to evolve the Lorentzian geometry forward in time. We denote the eigenstates of the bulk Hamiltonian by $\ket{E}$. There is a natural timescale associated with this time evolution that we set to one in our normalization of the Hamiltonian. In a real-time interval $(0, \Delta t)$ with $\Delta t=O(1)$, the emission or absorption of a baby universe is a rare quantum tunneling event with probability suppressed by $e^{-2S_E}$, and we can ignore the possibility of multiple events.

We split real time into many intervals $I_n=(n\Delta t,(n+1)\Delta t)$. In terms of the bulk path-integral, the emission/absorption process is described by inserting a cylinder inside the Lorentzian path-integral at some time interval $I_n$; see Figure \ref{fig5}. 
The insertions are local in the time of the parent universe; however, there is no sense in chronologically ordering them in the bulk. This implies that the operator insertions that describe the emission/absorption of baby universes should commute. 
The emissions/absorptions of baby universes in different time intervals are independent and rare events. The Poisson limit theorem applies, and we conclude that there exists a universal Poisson term in the statistics of the number of closed baby universes for times $T\sim e^{2S_0}$ \footnote{Note that the Poisson limit theorem can tolerate correlations as long as they fall off fast enough. Strictly speaking, the emission/absorptions during different time-intervals are not independent; however, as we saw in Section \ref{subsec:rareevents} (see also Appendix \ref{app:universalityPoisson}), if the correlations between them fall off fast enough, we can still apply the Poisson limit theorem to argue that after time $T\sim e^{2S_0}$, there is a universal Poisson contribution to the probability of observing $k$ closed universes.}.

It is instructive to elaborate more on the origin of the Poisson limit theorem in this case. To each time interval $I_n$, we associate a two-dimensional Hilbert space $\mH_{I_n}$, spanned by the vectors $\ket{0_n}$ (no baby universe) and $\ket{1_n}$ (a baby universe)\footnote{There are no kets $\ket{2}$ at this point, because the event is so rare that we ignore the multiple occurances in a single time interval.}. The total Hilbert space is 
\begin{eqnarray}
    \mH_{parent}\otimes \lb\otimes_n \mH_{I_n}\rb\ .
\end{eqnarray}
To allow for the emission/absorption of a baby universe we couple $\mH_{parent}$ to each $\mH_{I_n}$ via the interaction Hamiltonian $H^{(n)}_{int}=g(\mO_{parent}\otimes X_n)$, where $X_n=\ket{1_n}\bra{0_n}+\ket{0_n}\bra{1_n}$. Then, the amplitude for the emission/absorption of a baby universe in a single interval in the expansion $g\Delta t \sim e^{-S_0}\ll 1$ is  
\begin{eqnarray}\label{transitionamp}
    \bra{E',1_n}e^{-i(H_{parent}+H_{int})(\Delta t)}\ket{E,0_{n-1}}&&=(g\Delta t)e^{-iH_E(\Delta t)}\bra{E'}\mO\ket{E}+O(g\Delta t)^2
\end{eqnarray}
where $\ket{E}$ are the eigenkets of $H_{parent}$. We will be working in the order of limits $1\ll \Delta t\ll e^{S_0}$ \footnote{More specifically, we need to keep $\Delta t$ large but much smaller than the scrambling time \cite{chandrasekaran2023large}.}. Note that in the description above, we are representing an event that is rare due to quantum tunneling by a spontaneous emission that is rare due to a small coupling $g$. 

To make our discussion of the universal Poisson term most pedagogical, we focus on a ``topological" idealization that assumes $\mH_{parent}$ is one-dimensional. For example, this can occur if the parent universe is an open baby universe (a closed baby universe, circle cut in half). In this simplified model, the only dynamical process allowed is the emission/absorption of baby universes. This idealization has the advantage that, as we see in topological examples of Section \ref{sec:babyuniversescommutative}, the universal Poisson term is the only term that controls the statistics of baby universes\footnote{In systems with non-topological fluctuations, one needs to take a late-time limit to isolate this Poisson term. We will see that in JT gravity, this late-time limit is the $\tau$-scaling limit studied in \cite{okuyama2020multi,blommaert2023integrable,saad2024convergent}.}.

 Since the baby universes are indistinguishable, we should symmetrize our tensor product Hilbert space in the symmetric Fock by direct summing sectors with different numbers of baby universes: 
  \begin{eqnarray}
      \mathcal{F}_{closed} =\overline{\oplus_n \otimes_\text{sym}^n\mH_{closed}}\ .
  \end{eqnarray}
By the universality of the Poisson limit theorem, as long as $e^{-S_0}$  is exponentially small, the operator $1\otimes X_n$ can be replaced by any other evolution that acts trivially in the parent universe. In the Fock space of closed baby universes, a convenient choice of interaction is $1\otimes g(a+a^\dagger)$, resulting in the time evolution operator that is a Weyl unitary operator $1_{parent}\otimes W(\delta z)$ with $\delta z=(g \Delta t)\sim e^{-S_0}$. 
As we saw in Section \ref{subsec:coherentrareevent}, by the Poisson limit theorem, the statistics of $\hat{N}$ after long time $T\sim e^{2S_0}$ is well-approximated by the coherent state $W(Te^{-S_0})\ket{0}$ in $\mathcal{F}_{closed}$:
  \begin{eqnarray}
      |\bra{n}W(\sqrt{T}e^{-S_0})\ket{0}|^2=p_{Te^{-2S_0}}(n)
  \end{eqnarray}
  where the Poisson weight $Te^{-2S_0}$ grows linearly in time, and ket $\ket{n}$ is the eigenket of $\hat{N}$ (the total baby universe number operator).

If closed baby universes come in $d$ different types. Then, we have a rare event with $d+1$ possibilities: no emission, or emission of a baby universe of type $i$. Repeating the argument above, from the Poisson limit theorem, we obtain that the state is well-described by a $d$-mode coherent state with a pair of creation/annihilation operators for each type: $[a_i,a^\dagger_j]=\delta_{ij}$. In Section \ref{sec:babyuniversescommutative}, we will find that multi-mode coherent states control the statistics of baby universes in a general closed 2D topological quantum field theory\footnote{In a general 2D quantum gravity, we can label a closed baby universe boundary condition $J$ on the boundary circle as a parameter. For example, in JT gravity, closed baby universes can be labeled by positive $b>0$, which is the geodesic length corresponding to a fixed length boundary condition.}.

In a general 2D topological model, the algebra of the parent universe is the commutative algebra $\oplus_r x_r \mathbb{I}_r$ where $\mathbb{I}_r$ are $d_r\times d_r$ dimensional identity operators; see Section \ref{sec:babyuniversescommutative}. We still have $H_{parent}=0$, but the interaction can take the form
\begin{eqnarray}
    H=\sum_r g_r \mathbb{I}_r\otimes X_r^{(n)}\ . 
\end{eqnarray}
For example, this can occur if the parent universe is an open baby universe that comes in different types $r=1,\cdots, d$. We see in Section \ref{sec:babyuniversescommutative}, as the Poisson limit theorem suggests, the sum over topologies in this model matches the statistics of the number operator in a multi-mode coherent state.

\subsection{Open baby universes}\label{subsec:openbaby}

In the Euclidean path-integral of 2D quantum gravity, e.g. JT gravity, the amplitude for the emission of a closed baby universe with label $J$ is described by the overlap $\braket{\Psi_2,J}{\Psi_1}$ where the incoming state is $\ket{\Psi_1}$ is in the  Hilbert space are states of the asymptotic open boundary and the outgoing state is the tensor product $\ket{\Psi_2,J}\equiv\ket{\Psi_2}\otimes \ket{J}$ and we have emitted an open baby universe in state $\ket{J}$. The complex conjugate amplitude $\braket{\Psi_1}{\Psi_2,J}$ has the interpretation of absorption of a baby universe.

A more symmetric way of cutting this path integral is to split the closed baby universe in half. In this way of cutting the manifold, we interpret the path integral as the norm of a vector $\ket{\Psi,J_i}$, where $\ket{J_i}$ belongs to a different Hilbert space labeled by the boundary condition on a half-disk, and we have the inner product
\begin{eqnarray}
    \braket{J_i}{J_j}\sim\delta_{ij}\ .
\end{eqnarray}
Suppose there are $k$ different types of open (and hence closed) baby universes. Denoting $\ket{J_i}$ by $\ket{i}$ we have an commutative algebra of the projections $e_{ii}=\ket{i}\bra{i}$ with $i=1,\cdots, k$. The Poisson statistics naturally arise in this setting. The parent universe stochastically emits and then annihilates open baby universes. The key observation is that once we introduce annihilation/creation operators $a_i, a_i^\dagger$ for an open universe of $i$-th type, closed baby universes are automatically included as a second-order effect of a creation followed by annihilation\footnote{Mathematically it is well-known that classical Poisson distribution or Poisson random process (or more generally infinitely divisible process) can be realized by introducing creation/annihilation operators on a symmetric Fock space. This is the basis of quantum stochastic calculus \cite{quantStoch}.}.

One way to study the stochastic emission process is to focus on the evolution of the no-open-baby-universe initial state $\ket{\Omega}$ (c.f. Equation \ref{equation:emissionevolution}):
\begin{equation}
    \ket{\Omega}\mapsto \prod_{1\leq i\leq k}\lb a^\dagger_i\rb^{N_i(T)}\ket{\Omega}
\end{equation}
where $N_i(T)$ is the total number of emissions of type $i$ open universe after time $T$. This equation is a simple mathematical model where within time $T$ there could be $N_i(T)$ emissions of open universes of type $i$. Because the emissions are stochastic, the number $N_i(T)$ is a random variable. Hence, this state is a random state with a diagonal density matrix:
\begin{equation}
    \sum_{\ell_1,\cdots,\ell_k\geq 0}\text{Prob}(N_1(T) = \ell_1,\cdots,N_k(T) = \ell_k)\ket{\ell_1,\cdots,\ell_k}\bra{\ell_1,\cdots,\ell_k}
\end{equation}
where $\ket{\ell_1,\cdots,\ell_k}$ is the number eigenstate with $\ell_i$ many type $i$ open universes. It turns out that the probability is a product of Poisson distributions  \footnote{If we coarse-grain the interval to $M$ equal-length subintervals, the random variable $N_i(T)$ can be approximated by the sum of $M$ mutually independent identically distributed Bernoulli random variables. As $M\rightarrow \infty$ (i.e. as the division becomes finer), the limiting distribution of $N_i(T)$ approaches the Poisson distribution. This is the well-known Poisson limit theorem. And this is the mathematical reason why the probability is Poisson.}:
\begin{equation}
    \text{Prob}(N_1(T) = \ell_1,\cdots,N_k(T) = \ell_k) = \prod_{1\leq i\leq k}\text{Prob}(N_i(T) = \ell_i) = \prod_{1\leq i\leq k}\frac{e^{-\lambda_i T}(\lambda_i T)^{\ell_i}}{\ell_i!}
\end{equation}
Here $\lambda_i$ is the \textit{rate} of emission of the open universe of type $i$. Notice this rate of emission is exponentially suppressed in $G_N$. The diagonal density matrix can be written as a tensor product:
\begin{eqnarray}
        &&\sum_{\ell_1,\cdots,\ell_k\geq 0}\text{Prob}(N_1(T) = \ell_1,\cdots,N_k(T) = \ell_k)\ket{\ell_1,\cdots,\ell_k}\bra{\ell_1,\cdots,\ell_k}\nn\\
        &&= \otimes_{1\leq i\leq k}\sum_{\ell_i\geq 0}\frac{e^{-\lambda_i T}(\lambda_i T)^{\ell_i}}{\ell_i!}\ket{\ell_i}\bra{\ell_i}
\end{eqnarray}
From a mathematical point of view, the density matrix above describes a classical probability distribution over pure states (a random state), which we may also denote by
\begin{equation}
    \ket{N_1(T),\cdots,N_k(T)}
\end{equation} 
where $N_i(T)$ are classical random variables.

This random state corresponds to a particular way of cutting the path integral that splits the closed baby universes in half. A general operator insertion in the path integral results in a formula of the generic form $\bra{\Psi, i}\mathcal{O}\ket{\Psi, j}$ where $\mathcal{O}$ is a closed-baby-universe observable. In our concrete model, the generic formula takes the concrete form:
\begin{equation}
\mathbb{E}\bra{N_1(T),\cdots,N_k(T)}\mathcal{O}\ket{N_1(T),\cdots,N_k(T)}
\end{equation}
where the expectation is on the random variables $N_i(T)$'s.

The Hilbert space of closed baby universes is trivial (one-dimensional), and they come in $k$ different flavors; therefore, the only relevant observables in the Fock space are $N_i(T)$, which generate a commutative algebra. 
From the point of view of an observer in the parent universe, the state of the system is effectively a multi-modal coherent state\footnote{The statistics of the number operator in a coherent state is given by a Poisson distribution. For a more detailed discussion of coherent state and Poisson distribution, see  \cite{QuantumOptics}.}:
\begin{equation}
    \ket{W(z_1,\cdots,z_k)} = e^{\sum_iz_ia_i - \bar{z_i}a^\dagger_i}\ket{\Omega}
\end{equation}
And a generic correlator has the form:
\begin{equation}
    \bra{W(z_1,\cdots,z_k)}N_{i_1}N_{i_2}\cdots N_{i_p}\ket{W(z_1,\cdots,z_k)}
\end{equation}

The generating function for the connected correlations is:
\begin{equation}
    \log\bra{W(z_1,\cdots,z_k)}e^{\sum_i u_i N_i}\ket{W(z_1,\cdots,z_k)} = \sum_i |z_i|^2(e^{u_i} - 1)
\end{equation}
As a result, a generic connected correlator is given by:
\begin{equation}
    \bra{W(z_1,\cdots,z_k)}N_{i_1}\cdots N_{i_p}\ket{W(z_1,\cdots,z_k)}_\text{conn} = |z_{i_1}\cdots z_{i_p}|^2\delta_{i_1\cdots i_p}
\end{equation}
Here, $\delta_{i_1\cdots i_p}=1$ if $i_1=i_2=\cdots =i_p$ and zero otherwise.
It follows that 
\begin{eqnarray}
&&\bra{W(z_1,\cdots,z_k)}\lb N_i-|z_i|^2\rb\ket{W(z_1,\cdots,z_k)}=0\nn\\
       && \bra{W(z_1,\cdots,z_k)}\lb N_i-|z_i|^2\rb\lb N_j-|z_j|^2\rb\ket{W(z_1,\cdots,z_k)}=|z_i|^2 \delta_{ij}\ .
\end{eqnarray}
We can interpret the equations above as the one-point removed number operators
\begin{eqnarray}
 N_i-\bra{W(z_1,\cdots,z_k)}N_i\ket{W(z_1,\cdots, z_k)}=N_i-|z_i|^2   
\end{eqnarray}
creating type $i$ excitation on top of the coherent state. The Hilbert space spanned by $\{(N_i-|z_i|^2)\ket{W(z_1,\cdots,z_k)}\}_{1\leq i\leq k}$ is isomorphic to the Hilbert space of \textit{single} closed baby universe states. 

From this perspective, the coherent state $\ket{W(z_1,\cdots, z_k)}$ is a Hartle-Hawking state, and the action of the one-point removed number operator creates a closed open universe excitation. The form of the connected correlation function shows that only closed baby universes of the same type can have correlations. A simple combinatorial calculation will show that the full correlation function between $N$ closed baby universes of the same type counts the number of ways these universes connect. This discussion will come up again in the Marolf-Maxfield toy model of 2D topological gravity in subsection \ref{MMmodel}.

\begin{figure}[t]
    \centering
    \includegraphics[width=.8\linewidth]{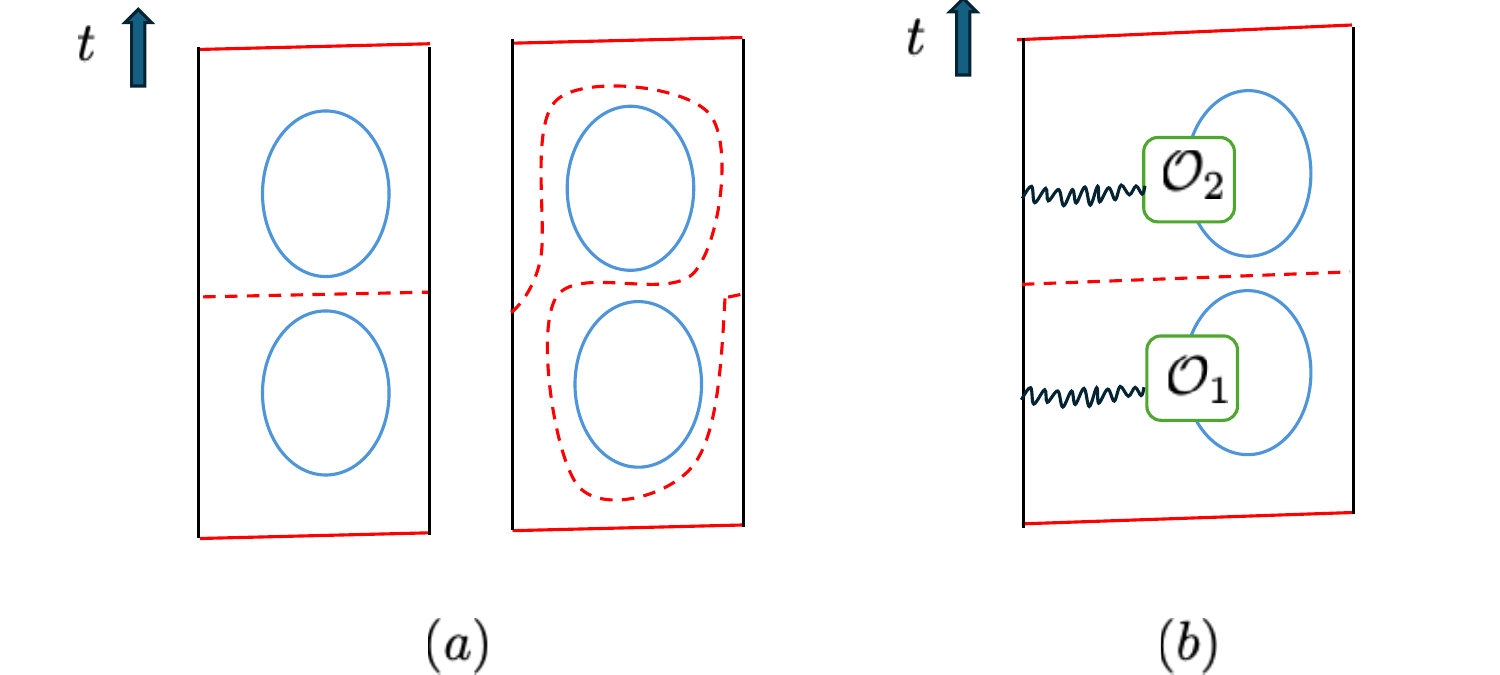}
    \caption{\small{(a) The Lorentzian path integral that corresponds to the emission of two baby universes can be cut, resulting in a change in the order of operator insertions. The equivalence of the results implies that the algebra of closed baby universes is commutative. (b) For an asymptotic open universe with a deformed state corresponding to the insertion of operator $\mO$, the order matters unless $\mO_1=\mO_2$.}}
    \label{fig5}
\end{figure}

\subsection{Asymptotic open universes and noncommutative Poisson}\label{subsec:asymptoticopen}

We can now generalize our construction, allowing for the creation of asymptotic open universes with noncommutative algebras. These are no longer baby universes because we can have a copy of CFT living on each asymptotic boundary. An asymptotic 2D open universe can be viewed as a Euclidean path integral over a half-disk with two endpoints on the boundary semi-circle. Each endpoint carries a label which represents a basis vector in the Hilbert space $\mathcal{K}$. The open-universe Hilbert space is given by:
\begin{equation}
    \mathcal{H}_\text{open} = \mathcal{K}\otimes\mathcal{K}
\end{equation}
The creation of an asymptotic open universe is a two-step process: 
\begin{enumerate}
    \item Create a thermofield double state from the vacuum
    \item Insert an operator that modifies the boundary condition (i.e., the boundary label)
\end{enumerate}
Note that open universes in the same state are indistinguishable, whereas those in different states are distinguishable \footnote{More generally, in gravity we should treat excitations on asymptotic boundaries as distinguishable and dynamic excitations in the bulk as indistinguishable.  These bulk excitations are summed over when we sum over all possible bulk geometries in the gravitational path integral. We will come back to this in Section \ref{sec:babyuniversescommutative}.}. 

To model these non-commutative excitations, we introduce creation/annihilation operators:
\begin{equation}\label{eqn:openunivCCR}
    [a_{ik}, a^\dagger_{jl}] = \delta_{ij}\delta_{kl}
\end{equation}
The creation operator $a^\dagger_{jl}$ acting on the Fock space ground state $\ket{\Omega}$ creates a pair of asymptotic open universes in the state $\ket{jl}$. 
Consider the thermofield double state corresponding to a Gibbs state of inverse temperature $\beta$
\begin{eqnarray}
    \ket{TFD}= \sum_j \frac{e^{-\beta E_j/2}}{\sqrt{Z}}\ket{jj}
\end{eqnarray}
where $E_j$ is the eigen-energy of the $j$-th state. The thermofield double state is the canonical purification of the Gibbs density:
\begin{equation}
    \rho_\beta:= \sum_j \frac{e^{-\beta E_j}}{Z}\ket{j}\bra{j}
\end{equation}
Suppose within the time interval $[0, T]$, $N(T)$ thermofield double states have been emitted \footnote{We reiterate that because the emission is stochastic, the variable $N(T)$ is a random variable.}. After emission, the state is 
\begin{eqnarray}
    \ket{TFD}^{\otimes N(T)}
\end{eqnarray}
which is a probability distribution over states with different numbers of copies of $\ket{TFD}$.
Then, the following operator creates an open universe in the state $\ket{ij}$:
\begin{align}\label{equation:statecreation}
    \begin{split}
        \sum_k& a_{ik}^\dagger a_{jk}\ket{TFD}^{\otimes N(T)}= 
        \sum_k a^\dagger_{ik}\sqrt{N(T)}\bra{jk}\ket{TFD}\ket{TFD}^{\otimes(N(T) - 1)}
        \\&
        =\sum_k \delta_{jk}\frac{e^{-\beta E_j/2}}{\sqrt{Z}}\sqrt{N(T)}a^\dagger_{ik}\ket{TFD}^{\otimes (N(T) - 1)}
        \\&
        =\frac{e^{-\beta E_j/2}}{\sqrt{Z}}\sqrt{N(T)}\sum_{1\leq s\leq N(T) - 1}\frac{1}{\sqrt{N(T)}}\ket{TFD}^{\otimes s}\otimes\ket{ij}\otimes\ket{TFD}^{\otimes (N(T)-1-s)}
        \\&
        =\frac{e^{-\beta E_j/2}}{\sqrt{Z}}\ket{ij}\otimes_\text{sym}\ket{TFD}^{\otimes(N(T) - 1)}
    \end{split}
\end{align}
where $\otimes_\text{sym}$ is the symmetric tensor product. The state $\frac{e^{-\beta E_j/2}}{\sqrt{Z}}\ket{ij}\otimes_\text{sym}\ket{TFD}^{\otimes(N(T) - 1)}$ represents $N(T) -1$ thermofield double states and an open universe in the state $\frac{e^{-\beta E_j/2}}{\sqrt{Z}}\ket{ij}$. Therefore, the operator $\sum_ka^\dagger_{ik}a_{jk}$ can be understood as converting one of the TFD states to an open universe in state $\ket{ij}$.  

Because the emission of the thermofield double state is stochastic and the number of emissions $N(T)$ follows the Poisson distribution, insertion of a single operator $\sum_k a^\dagger_{ik}a_{jk}$ leads to the following expected amplitude:
\begin{align}
    \begin{split}
        \sum_{N\geq 0}&\frac{e^{-\lambda T}(\lambda T)^N}{N!}\bra{TFD}^{\otimes N}\sum_k a^\dagger_{ik}a_{jk}\ket{TFD}^{\otimes N}\\&=\frac{e^{-\beta E_j/2}}{\sqrt{Z}}\sum_{N\geq 0}\frac{e^{-\lambda T}(\lambda T)^N}{N!}N\bra{TFD}\ket{ij}\bra{TFD}\ket{TFD}^{N - 1}
        \\
        &=\frac{e^{-\beta (E_i + E_j)/2}\delta_{ij}}{Z}e^{-\lambda T + \lambda T\bra{TFD}\ket{TFD}} = \frac{e^{-\beta(E_i + E_j)/2}\delta_{ij}}{Z}
    \end{split}
\end{align}
where $\lambda > 0$ is the intensity of emission. Note, we have properly normalized the thermofield double state, thus $\bra{TFD}\ket{TFD} = 1$.

Using the connection between the Poisson distribution and the coherent state, we can introduce the coherent thermofield double state:
\begin{equation}
    \ket{W(\rho_\beta^{1/2})} = \sum_{N\geq 0}\frac{e^{-\frac{\lambda T}{2}}(\lambda T)^{\frac{N}{2}}}{\sqrt{N!}}\ket{TFD}^{\otimes N}
\end{equation}
where $\rho_\beta$ is the Gibbs state of inverse temperature $\beta$.
The same amplitude can be written as:
\begin{align}
    \begin{split}
        \bra{W(\rho_\beta^{1/2})}&\sum_ka^\dagger_{ik}a_{jk}\ket{W(\rho_\beta^{1/2})} =\sum_k\bra{TFD}\ket{ik}\bra{jk}\ket{TFD}
        \\&=\sum_k\frac{e^{-\beta(E_i + E_j)/2}}{Z}\delta_{ik}\delta_{jk}
        =\frac{e^{-\beta(E_i + E_j)/2}}{Z}\delta_{ij}
    \end{split}
\end{align}
Hence, the observer effectively observes a coherent state of thermofield double excitations, and we refer to it as the Hartle-Hawking (HH) state. Acting on this HH-state, the one-point remove operator $\sum_k a^\dagger_{ik}a_{jk}-\delta_{ij}e^{-\beta E_i}/Z$ creates an open boundary in the state $\frac{e^{-\beta E_j/2}}{\sqrt{Z}}\ket{ij}$. In terms of the boundary condition changes, this state can be represented by:
\begin{equation}
    (\ket{i}\bra{j}\otimes\mathbbm{1})\ket{TFD} = \frac{e^{-\beta E_j/2}}{\sqrt{Z}}\ket{ij}
\end{equation}
Here, the matrix $\ket{i}\bra{j}\otimes\mathbbm{1}$ represents the boundary operator that changes the boundary label from $j$ to $i$. Note that this boundary operator is an operator on the \textit{single} open-universe Hilbert space. On the multi-open-universe Hilbert space (a symmetric Fock space) $\mathcal{F}_\text{sym}(\mathcal{H}_\text{open})$, this operator is lifted to:
\begin{equation}
    \lambda(\ket{i}\bra{j}):=\sum_ka^\dagger_{ik}a_{jk}
\end{equation}
Mathematically, this lift is canonical in the sense that we can consider a more general lifting:
\begin{equation}
    \lambda(\mathcal{O}) = \sum_{ijk} \mathcal{O}_{ij}a^\dagger_{ik}a_{jk}
\end{equation}
This represents the insertion of a boundary operator $\mathcal{O}$ that changes the boundary label in one of the thermofield double states. The one-point removed ``third quantized" operator 
\begin{eqnarray}
    \lb \lambda(\mO)-\bra{W(\rho_\beta^{1/2})}\lambda(\mO)\ket{W(\rho_\beta^{1/2})}\rb \ket{W(\rho_\beta^{1/2})}
\end{eqnarray}
creates an open universe in the superposition: $\sum_{ij}\mathcal{O}_{ij} \frac{e^{-\beta E_j/2}}
{\sqrt{Z}}\ket{ij}$ on top of the coherent state; i.e. Hartle-Hawking state.

It is clear from definition of $\lambda(\mathcal{O})$ that the $\lambda$ map is linear and preserves adjoint:
\begin{equation}
    \lambda(\mathcal{O}_1 + \mathcal{O}_2) = \lambda(\mathcal{O}_1) + \lambda(\mathcal{O}_2)
\end{equation}
\begin{equation}
    \lambda(\mathcal{O}^\dagger) = \sum_{ijk}\overline{\mathcal{O}}_{ji}a^\dagger_{ik}a_{jk} = \lambda(\mathcal{O})^\dagger
\end{equation}

As we argued before, asymptotic open boundaries with different excitations $\mathcal{O}_1$ and $\mathcal{O}_2$ are distinguishable. 
It follows from the commutation relations in (\ref{eqn:openunivCCR}) that on the Fock space, we have the commutation relations
\begin{eqnarray}\label{commutators}
    [\lambda(\mathcal{O}_1),\lambda(\mathcal{O}_2)]=\lambda([\mathcal{O}_1,\mathcal{O}_2])\ .
\end{eqnarray}
The commutation relation (\ref{commutators}) is required if we are to interpret $\lambda(\mathcal{O})$ as the lift of $\mathcal{O}$ to the Fock space so that the restriction of the $C^*$ algebra of operators $\lambda(\mathcal{O})$ to a single boundary sector is the same as the $C^*$ algebra of operators $\mathcal{O}$.

\begin{figure}[t]
    \centering
    \includegraphics[width=1\linewidth]{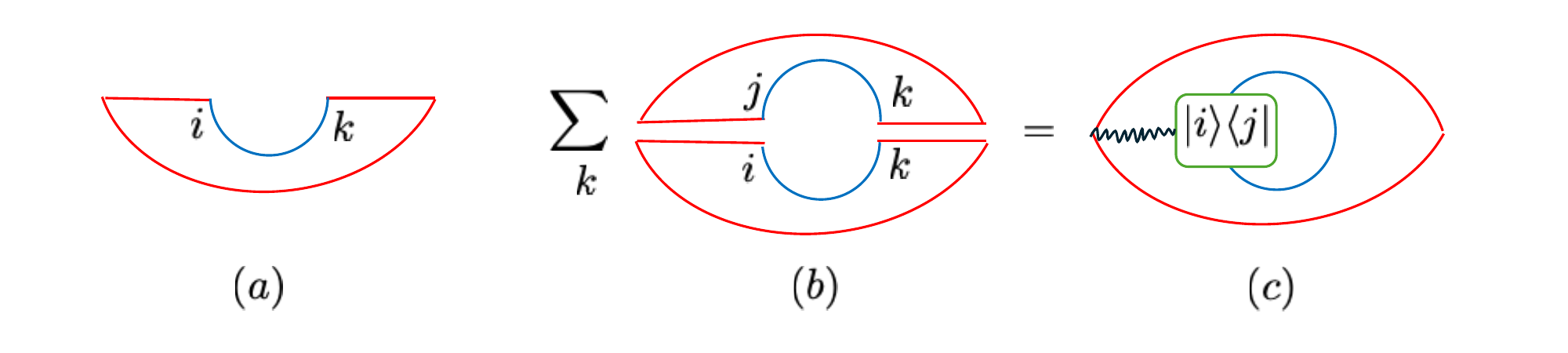}
    \caption{\small{(a) The operator $a_{ik}^\dagger$ creates an open boundary with in the state $\ket{ik}$. (b) The Poissonization map $\lambda(\ket{i}\bra{j})=\sum_k a_{ik}^\dagger a_{jk}$ can be understood as the insertion of an asymptotic boundary with a particle insertion as in (b).}}
    \label{fig6}
\end{figure}

A single-mode coherent state shows Poisson statistics for the commutative algebra of number operators. Similarly, a coherent thermofield double state in the Fock space exhibits a noncommutative generalization of Poisson statistics with respect to the noncommutative algebra of $\lambda(\mathcal{O})$. To understand what this means, note that if we pick any fixed self-adjoint $\lambda(\mathcal{O})$, then the state shows Poisson statistics with respect to the commutative algebra generated by $\lambda(\mathcal{O})$: 
\begin{eqnarray}
\bra{W(\rho_\beta^{1/2})}\lambda(\mathcal{O})^p\ket{W(\rho_\beta^{1/2})}_\text{conn}=\sum_i \frac{e^{-\beta E_i}}{Z} \mathcal{O}_{ii}^p=\bra{TFD}(\mathcal{O}^p\otimes \mathbb{I})\ket{TFD}\nn\\
\end{eqnarray}
In fact, requiring the relation above to hold for the commutative algebra generated by any single $\lambda(\mathcal{O})$ uniquely fixes the vector $\ket{W(\rho_\beta^{1/2})}$.

In particular, for $p = 1$, the full correlation function is the connected correlation function. Hence, we have:
\begin{equation}
    \bra{W(\rho_\beta^{1/2})}\lambda(\mathcal{O})\ket{W(\rho_\beta^{1/2})} = \bra{TFD}\mathcal{O}\otimes\mathbbm{1}\ket{TFD} = \sum_i \frac{e^{-\beta E_i}}{Z}\mathcal{O}_{ii}
\end{equation}
In particular, for off-diagonal matrix units, the expectation value vanishes:
\begin{equation}
    \bra{W(\rho_\beta^{1/2})}\lambda(\ket{i}\bra{j})\ket{W(\rho_\beta^{1/2})} = \bra{TFD}\lb\ket{i}\bra{j}\otimes\mathbbm{1}\rb\ket{TFD} = 0
\end{equation}
where $i\neq j$. To simplify our discussion, in the following, we often implicitly assume that the operator $\mathcal{O}$ has a vanishing one-point function.

The algebra generated by the $\lambda(\mathcal{O})$ operators is generally noncommutative. To see this, we note that for each matrix unit $\ket{i}\bra{j}$, the corresponding operator $\lambda(\ket{i}\bra{j})$ creates a nontrivial open universe state $\ket{i}\bra{j}\otimes\mathbbm{1}\ket{TFD} = \frac{e^{-\beta E_j/2}}{\sqrt{Z}}\ket{ij}$ from the coherent thermofield double state (c.f. Equation \ref{equation:statecreation}). In terms of Figure \ref{fig6}(c), the annulus with operator insertion of $\ket{i}\bra{j}$ has a the wiggly line in the path-integral to denote the insertion of an operator in the bulk, and the figure represents only the created state of $\frac{e^{-\beta E_j/2}}{\sqrt{Z}}\ket{ij}$. The state $\lambda(\ket{i}\bra{j})\ket{W(\rho_\beta^{1/2})}$ is a symmetric tensor product of Figure \ref{fig6}(c) and $\ket{W(\rho_\beta^{1/2})}$ (c.f. Equation \ref{equation:statecreation}).
Acting a second operator $\lambda(\ket{i'}\bra{j'})$ on this state creates two scenarios:
\begin{enumerate}
    \item Either a second operator $\ket{i'}\bra{j'}$ is inserted in Figure \ref{fig6}(c); 
    \item Or the operator $\ket{i'}\bra{j'}$ is inserted in a brand new annulus representing $\ket{TFD}$.
\end{enumerate}
Pictorially, the first scenario can be represented by gluing two annuli (i.e., two Figure \ref{fig6}(c)) along the cut. To be able to glue, the boundary conditions must match. In terms of formula, this is simply the observation that the state $\lb\ket{i'}\bra{j'}\ket{i}\bra{j}\otimes\mathbbm{1}\rb\ket{TFD}$ is only nonzero if $j' = i$. The resulting figure is an annulus with one branch cut, but the boundary condition is changed to $(i',j)$ representing the state $\frac{e^{-\beta E_j/2}}{\sqrt{Z}}\ket{i'j}$. The second scenario is simply two separate annuli (i.e, two separate Figure \ref{fig6}(c)), one with boundary condition $(i,j)$ and the other with $(i',j')$.

The second scenario is commutative because the two annuli are disjoint. However, the first scenario is \textit{not} commutative because the order of operator insertion matters (c.f. Figure \ref{fig5}(b)). Therefore, the algebra generated by $\lambda(\mathcal{O})$ is generally noncommutative. 

The infinite temperature thermofield double state corresponds to the maximally mixed state of a single boundary
\begin{eqnarray}
    \ket{TFD_0} = \frac{1}{d}\sum_i\ket{ii}
\end{eqnarray}
This is the purification of the normalized trace on the matrix algebra. We can also replace the thermofield double state with the purification of an unnormalized trace $|\mu|^2\sum_i\ket{ii}$. We denote the corresponding coherent state as
\begin{eqnarray}
    \ket{W(\mu\mathbbm{1})}=e^{\mu\sum_i (a_{ii}^\dagger+a_{ii})}\ket{\Omega}\ .
\end{eqnarray}
where $|\mu|^2\mathbbm{1}$ is the density matrix corresponding to the bipartite state $\mu\sum_i\ket{ii}$. The coherent state $\ket{W(\mu\mathbbm{1})}$ and the coherent thermofield double state $\ket{W(\rho_\beta^{1/2})}$ are basic examples of bipartite coherent states, which are the key components of Poissonization. We will study these states and Poisson in depth in Section \ref{sec:bipartitecoherent}.

\section{Chaotic quantum systems and late-time plateau}\label{sec:chaoticsystems}
In this Section, we take a digression from our discussion of topology change and Poisson processes to review some basic facts about chaotic quantum systems. The goal is to argue that the height of the late-time plateau expected in chaotic quantum systems can be explained using the Poisson limit theorem.

Generalizing beyond the topological model, we can consider a parent universe with a finite-dimensional Hilbert space, a non-commutative algebra of observables, and a general Hamiltonian $H_{parent}$ with non-degenerate eigenkets $\ket{E}$. Consider the two-point function of general operators in the Hilbert space of the parent universe
\begin{eqnarray}
    \tr(\mO_1(0)\mO_2(t))=\sum_E (\mO_1)_{EE}(\mO_2)_{EE}+\sum_{E_1\neq E_2}(\mO_1)_{E_1E_2}(\mO_2)_{E_2E_1}e^{i(E_1-E_2)t}\ .
\end{eqnarray}
In chaotic quantum systems, one expects that the correlator above decays initially, but of course, this decay cannot continue forever, because the first term is time-independent and equal to the long-time average of the correlator:
\begin{eqnarray}
    \lim_{t\to \infty}\frac{1}{(\Delta t)}\int_0^{\Delta t}dt \tr(\mO_1(0)\mO_2(t))=\sum_E (\mO_1)_{EE}(\mO_2)_{EE}\ .
\end{eqnarray}
At very late times, one expects that the two-point function oscillates erratically around its average value above. The precise nature of these oscillations depends on the choice of operator $\mO$; however, one can make statistical statements by taking an average over various choices of $\mO$. For instance, choose the operator to be unitary $u$ and averaging over all unitary operators with a Haar random measure, we find
\begin{eqnarray}
    \int du\: \tr(u(0) u(t))=\sum_{E_1,E_2}e^{i(E_1-E_2)t}=|Z(it)|^2
\end{eqnarray}
where $Z(it)=\tr(e^{i t H})$ is the partition function continued to real time. Choosing $\mO=u e^{-\beta H}$ we find that
\begin{eqnarray}\label{SFF}
    \int du \: \tr(u e^{-\beta H/2} \lb u(t) e^{-\beta H/2}\rb^\dagger)=\sum_{E_1,E_2}e^{i(E_1-E_2)t}e^{-\beta(E_1+E_2)}=|Z(\beta+it)|^2
\end{eqnarray}
which is called the spectral form factor. In a quantum chaotic theory, we consider an ensemble of Hamiltonians \cite{mehta2004random}. More generally, one can define a generalization of the spectral form factor with $t_1+\cdots+t_p=0$:
\begin{eqnarray}\label{highptSFF}
    Z(\beta+it_1)\cdots Z(\beta+it_p)=\sum_{E_1\cdots E_p}e^{i(E_1 t_1+\cdots +E_pt_p)}e^{-\beta(E_1+\cdots E_p)}\ .
\end{eqnarray}

\subsection{Diagonal terms and conditional expectation}

The expectation is that the ensemble-averaged spectral form factor initially decays. At times of order $e^{S/2}$, that is often referred to as the Thouless time or the onset of the random matrix theory, it reaches a minimum, turns around, and grows linearly in time. This linear growth is referred to as the ramp. Finally, at times of order $e^S$, it reaches a plateau. Splitting the right-hand side of (\ref{SFF}) into a diagonal and off-diagonal terms, the height of the plateau comes from the diagonal term \cite{prange1997spectral,saad2019late}
\begin{eqnarray}
    \sum_{E_1=E_2}e^{-\beta (E_1+E_2)}=Z(2\beta)\ .
\end{eqnarray}
The intuition is that at late times, the off-diagonal term averages to zero. This expectation extends to the generalization of the spectral form factor in (\ref{highptSFF}). In a quantum chaotic theory at late times, $t_1,\cdots t_p>e^S$, the spectral form factor averages to a plateau with a height $Z(p\beta)$ that comes from the diagonal term $E_1=E_2=\cdots=E_p$.

Similarly, for general operators $\mO$ inserted at times $t_n\in I_n$ we have the correlator
\begin{eqnarray}
&&\tr(\mO_1(t_1)\cdots \mO_p(t_p))=\nn\\
&&\sum_{E_1,\cdots, E_p}(\mO_1)_{E_pE_1}(\mO_2)_{E_1E_2}\cdots (\mO_p)_{E_{p-1}E_p}e^{i(E_1-E_p)t_1}e^{i(E_2-E_1)}\cdots e^{i(E_{p-1}-E_p)t_p} \ .
\end{eqnarray}
There is a term in the sum above that corresponds to $E_1=E_2=\cdots=E_p$ and matches the long-time average
\begin{eqnarray}
&&\sum_E (\mO_1)_{EE}\cdots (\mO_p)_{EE}=\lim_{\Delta \to \infty}\frac{1}{(\Delta t)^p}\int_{t_1\in I_1,\cdots, t_p\in I_p}\mO_1(t_1)\cdots \mO_2(t_p)
\end{eqnarray}
Once again, in a quantum chaotic theory, separating subsequent correlators by a large $\Delta t$, one expects the off-diagonal terms to decay away \footnote{In noncommutative ergodic theory, this is often referred to as the $k$-mixing property.}. In this limit, we expect the correlator to oscillate around this average value, which we also refer to as the plateau.

The time-scale at which the physics of the plateau (the diagonal term) kicks in depends on how large $\Delta t$ in the average should be so that the diagonal approximation is valid. If there are very small energy gaps in the spectrum, one needs a larger $\Delta t$ for the approximation to work. In a quantum chaotic system, the probability distribution of energy gaps is conjectured to be given by the Wigner surmise \cite{mehta2004random}. The eigenvalue repulsion of random matrix theory makes it very unlikely for the energy gaps to be arbitrarily small. However,  note that no matter what eigenvalue spectrum we have, in the absence of degeneracies for large enough $\Delta t$, the off-diagonal terms go away\footnote{If there are exact degeneracies $\ket{E,\alpha}$ with $\alpha$ keeping track of eigenkets of same energy $E$, the matrix elements $\bra{E;\alpha}\mO\ket{E,\beta}$ also remain time-independent.}.

In chaotic quantum systems, the timescale of the plateau is $e^S$. In integrable models with no eigenvalue repulsion, there is no ramp, and the decay ends in a plateau\footnote{Note that the absence of eigenvalue repulsion is the assumption of the independence of the matrix eigenvalues, which is often referred to as a Poisson spectrum. We warn the reader that this is not related to the Poisson process in time we are describing in this work.}. Therefore, the transition to the plateau occurs before $e^S$. Our simple argument based on the Poisson limit theorem requires a much longer timescale $e^{2S}$. The Poisson limit argument is not complex enough to chaotic properties of the system, such as the matrix degrees of freedom, and the eigenvalue repulsion. The  Poisson limit theorem provides a universal explanation for the plateau at very late times; however, it is not refined enough to predict the timescale of the onset of the plateau. 

Consider a quantum system with a Hilbert space of large dimension $d\gg 1$. Its operator algebra is the algebra of $d\times d$ matrices with the trivial center $z\mathbb{I}$. For simplicity, we will think of this Hilbert space as a narrow microcanonical ensemble of average energy $\bar{E}$, i.e. $d\sim e^{-S(\bar{E})}$ where $S_E$ is the microcanonical entropy. Consider an arbitrary pure state 
\begin{eqnarray}
    \ket{\Psi}=\sum_i \psi_i \ket{E_i}
\end{eqnarray}
The time-average of observables over time-scales much larger than $d$ decoheres the observable in the energy eigenbasis
\begin{eqnarray}
    \lim_{t\to \infty}\frac{1}{t}\int_0^t dt'\braket{\Psi}{\mO(t'
    )\Psi}=\sum_i |\psi_i|^2 \mO_{ii}\ .
\end{eqnarray}
This suggests that the system never forgets the diagonal matrix elements of observables, which is in tension with the intuition that in a chaotic system, the system forgets about the initial state. The resolution proposed by the ETH is that the diagonal matrix elements of {\it simple} probe observables are independent of the microscopic label $i$ and only dependent on the energy $E_i$. Therefore, in a chaotic theory, any diagonal matrix element of an observable is the same as the microcanonical average, up to exponentially small corrections. More generally, the expectation is that the correlation functions of {\it simple} observables separated by large times are independent of the choice of the state. This motivates the Eigenstate Thermalization Hypothesis (ETH) that postulates the matrix elements of simple observables in the energy eigenbasis are given by
\begin{eqnarray}\label{ETH}
    \braket{E_i}{\mO E_j}=\mO(\bar{E})\delta_{ij}+e^{-S(\bar{E})/2}\mO(E_i,E_j)\psi_{ij}
\end{eqnarray}
where $\bar{E}=(E_i+E_j)/2$ is a smooth function of $\bar{E}$ and $\psi_{ij}$ are random matrices that are traditionally taken to be Gaussian random matrices \footnote{It is known that non-Gaussianities are needed to reproduce the expected behavior of the out-of-time ordered correlators in a chaotic system \cite{foini2019eigenstate,jafferis2023matrix}. This has motivated Generalized ETH \cite{foini2019eigenstate,pappalardi2022eigenstate}.}.

At large entropy (or short times), each energy sum can be approximated by an integral over $E$ with a measure $\rho_0(E)$ and an integrand that is a smooth function of energies. In the continuum of energies, we denote energy eigenkets by $\ket{E;\alpha}$, where $\alpha$ keeps track of the degeneracies. If we assume that $H_{parent}$ is chaotic and $\mO$ in (\ref{transitionamp}) are simple  observables, then invoking the ETH, we conclude that to the lowest order
\begin{eqnarray}
    \bra{E;\alpha}\mO\ket{E';\beta}=\delta_{EE'}\mathbb{I}_E\mO(E)+O(e^{-S_E/2})\ .
\end{eqnarray}
In other words, to the lowest order, simple operators are invariant under the conditional expectation below:
\begin{eqnarray}\label{conditionalexpec}
    &&\mathcal{E}(\mO)=\sum_E e^{-S_E}\tr(\mathbb{I}_E\mO)\mathbb{I}_E
\end{eqnarray}
assuming that there are no large degeneracies.
This conditional expectation projects the noncommutative algebra of the parent universe to a commutative subalgebra that only keeps track of the microcanonical traces of the operators. This commutative subalgebra is the same as our example of topological model we discussed above, when we identify $r$ with the energy label $E$. The universal Poisson term captures the physics of the diagonal operator $\mE(\mO)$.

\subsection{Off-diagonal terms and random operators}\label{subsec:ETH}
To go beyond the physics of the diagonal term, one needs to keep track of the random and exponentially suppressed off-diagonal matrix elements in (\ref{ETH}). A transition in energy due to the insertion of simple operators is an exponentially rare event. To focus on the off-diagonal fluctuation, we redefine simple operators by subtracting their diagonal terms (one-point functions):
\begin{eqnarray}
    \mO\to \mO-\mathcal{E}(\mO)
\end{eqnarray}
so that our simple operators are (c.f. Subsection \ref{subsection:coherentexamples})
\begin{eqnarray}
    \mO^{(\psi)}=\psi \mO \psi^\dagger\ .
\end{eqnarray}
Even though we subtracted the diagonal piece, higher powers of the simple operator $\mO^{(\psi)}$ will contain diagonal terms coming from the second-order effect of a transition from some initial energy $\ket{E}$ to $\ket{E'}$ and back. These diagonal terms contribute to the time averages:
\begin{eqnarray}
     \lim_{T\to \infty}\frac{1}{T}\int_0^T dt\braket{\Psi}{(\mO(t')^{(\psi)})^p\Psi}&&=\mathbb{E}_\Psi\text{tr}\lb (\mO^{(\psi)})^p\rb\ .
\end{eqnarray}
Assuming that there are no large degeneracies, we compute these diagonal terms in Subsection \ref{subsection:coherentexamples} (c.f. Equation \ref{equation:simplecombinatorics})\footnote{Note that in the notation of Subsection \ref{subsection:coherentexamples} and Appendix \ref{app:openclosed}, the role of the $a$ index is played by $E$ and the role of the microstate $i$ is replaced by the microstate $E_i$ and 
\begin{eqnarray}
\psi_{ij}=\bra{E_i}\mathbb{I}^{(\psi)}\ket{E_j}=\sum_a\bra{E_i}\psi e_{aa}\psi^\dagger\ket{E_j}
\end{eqnarray}
providing a physical justification for the Poissonization of random weights in Subsection \ref{subsection:coherentexamples}.
} in the powers of simple operators to find\footnote{See Equation 4.4 of \cite{srednicki1999approach}.}
\begin{eqnarray}
\mathbb{E}_\Psi\text{tr}\lb (\mO^{(\psi)})^p\rb&&= \sum_{\substack{l_1,l_2,\cdots, l_p=0\\
    l_1+2l_2+\cdots pl_p=p}}^p    \frac{x^p p!}{\prod_{k=1}^p k^{l_k}l_k!}  d^{\sum_k l_k} \text{tr}(\mO)^{l_1}\text{tr}(\mO^2)^{l_2}\cdots \text{tr}(\mO^p)^{l_p}\ .
\end{eqnarray}

\section{Poissonization and bipartite coherent states}
\label{sec:bipartitecoherent}
In this section, we introduce Poissonization as a mathematical framework for third quantization. In essence, what we desire to achieve in third quantization is a framework that takes as input a general quantum system (a single copy of boundary theory) and outputs an algebra of operators that create and annihilate boundary copies in various states, represented on a symmetric Fock space.
Formally speaking, the input is a von Neumann algebra represented on a Hilbert space $\mH$ and an unnormalized state (weight)\footnote{We say unnormalized because, as we will see in some examples, we promote the partition function (norm of an unnormalized thermofield double state) to an operator.} 
\begin{eqnarray}
\ket{\omega^{1/2}}=\sum_i \omega_{i}^{1/2}\ket{ii}\in \mH \ .
\end{eqnarray}
Here, by unnormalized we mean that the sum of $\omega_{i}$ need not be one, or even finite. Denote by $\mathcal{O}\in B(\mathcal{H})$  operators in the one copy theory. The output should be an algebra of operators (lifts of $\mathcal{O}$) generated by $\lambda(\mathcal{O})$ acting on the symmetric Fock space of $\mH$: $\mathcal{F}_\text{sym}(\mathcal{H}) = \overline{\oplus_{n\geq 0}\otimes_\text{sym}^n\mathcal{H}}$ and a coherenet vector $\ket{W(\omega^{1/2})}$ (lift of the weight $\omega$). The action of $\lambda(\mathcal{O})$ on $\ket{W(\omega^{1/2})}$ in this Fock space has the interpretation of creating an open universe in a state $\sum_{ij}\mathcal{O}_{ij}\omega_j^{1/2}\ket{ij}$ \footnote{Once again, we are implicitly assuming that we are removing the one-point functions of $\mathcal{O}$.}.

In the previous section, we motivated Poissonization from gravitational path integrals; however, our discussion only required a set of rare independent events. Hence, in this section, we discuss Poissonization with no direct reference to any gravitational context. We formulate Poissonization using the intuitive physics picture of coherent states of bipartite quantum systems (or, equivalently, coherent states of matrix quantum systems). The connection to gravity will be provided in the next sections. We postpone the full mathematical theory of Poissonization to Section \ref{app:math}.

\subsection{Multi-mode coherent states and the second quantization}\label{subsec:multimodecoherent}


First, we briefly review the second quantization and the physics of coherent states. The purpose is mainly to fix some notations. The second quantization describes a many-body system using one-particle information. In terms of the Hilbert space, the second quantization lifts a one-particle Hilbert space $\mathcal{H}$ to the symmetric Fock space $\mathcal{F}_\text{sym}(\mathcal{H}) = \overline{\oplus_{n\geq 0}\otimes_\text{sym}^n\mathcal{H}}$. In terms of operators, the second quantization constructs the creation/annihilation operators $a_f$ and $a_f^\dagger$ from a vector $\ket{f}\in\mathcal{H}$. The creation/annihilation operators satisfy the canonical commutation relation:
\begin{equation}
    [a_f, a^\dagger_g] = \bra{f}\ket{g}
\end{equation}
In addition, the creation/annihilation operators generate the Weyl unitary operators:
\begin{equation}
    W(f) := \exp(i p(f))
\end{equation}
where $p(f):=i(a_f - a^\dagger_f)$. From the canonical commutation relation, we can see that the generators satisfy:
\begin{equation}
    [p(f), p(g)] = 2i\Im\bra{f}\ket{g}
\end{equation}
And the Weyl operators satisfy the integrated version of the canonical commutation relation:
\begin{equation}
    W(f)W(g) = \exp(i\Im\bra{f}\ket{g})W(f + g)
\end{equation}
where $W(f + g)$ is the Weyl operator lifted from the vector $\ket{f} + \ket{g}\in\mathcal{H}$.

In addition to the Weyl operators, one can also lift one-particle operators to operators on $\mathcal{F}_\text{sym}(\mathcal{H})$. Suppose the one-particle system evolves under the Hamiltonian $H$, then allowing simultaneous independent evolution on the many-body system gives the following time evolution map:
\begin{equation}\label{eqution:distribution}
    \ket{f_1}\otimes_\text{sym}\cdots\otimes_\text{sym}\ket{f_N} \mapsto e^{iHt}\ket{f_1}\otimes_\text{sym}\cdots\otimes_\text{sym} e^{iHt}\ket{f_N}
\end{equation}
where $\otimes_\text{sym}$ is the symmetric tensor product. The generator of this simultaneous evolution is given by:
\begin{equation}
    \lambda_N(H) := \sum_{1\leq i\leq N}1\otimes\cdots\underbrace{\otimes H\otimes}_{\text{i-th position}}\cdots\otimes 1
\end{equation}
The generator $\lambda_N(H)$ acts only the $N$-particle subspace of $\mathcal{F}_\text{sym}(\mathcal{H})$. Collecting all such generators, we can define the operator:
\begin{equation}
    \lambda(H) := \sum_{N\geq 0}\lambda_N(H)
\end{equation}
where by convention $\lambda_0(H) = 0$. The one-parameter unitary group generated by $\lambda(H)$ is given by:
\begin{equation}
    \Gamma(e^{iH}) := \sum_{N\geq 0}e^{i\lambda_N(H)} = \sum_{N\geq 0}\otimes_\text{sym}^N e^{iH}
\end{equation}
We can generalize the definition of $\lambda(H)$ and $\Gamma(e^{iH})$ beyond the Hamiltonian $H$. In fact for any self-adjoint operator $\mathcal{O}$ on $\mathcal{H}$, we can define $\lambda(\mathcal{O})$ and $\Gamma(e^{i\mathcal{O}})$. From the definition, it is clear that $\lambda$ is linear. Hence we can define $\lambda(\mathcal{O})$ for any operator $\mathcal{O}$ on $\mathcal{H}$. In addition, from definition, the following commutation relation holds\footnote{Note that this is different from what we do in symmetric orbifold theories, where the lift is normalized by an extra factor of $1/\sqrt{n}$: $\mathcal{O}\to \frac{1}{\sqrt{n}}\sum_i\mathcal{O}^{(i)}$. This extra factor of $n^{-1/2}$ is responsible for the large $n$  factorization because $[n^{-p/2}\lambda(\mathcal{O}_1),n^{-p/2}\lambda(\mathcal{O}_2)]=n^{-p}\lambda([\mathcal{O}_1,\mathcal{O}_2])$.}:
\begin{equation}
    [\lambda(\mathcal{O}_1), \lambda(\mathcal{O}_2)] = \lambda([\mathcal{O}_1, \mathcal{O}_2])
\end{equation}
A particularly interesting example of $\lambda(\mathcal{O})$ operators is the number operator:
\begin{equation}
    \widehat{N}  = \lambda(\mathbbm{1})
\end{equation}
This can be seen from the definition:
\begin{equation}
    \lambda(\mathbbm{1}) = \sum_{N\geq 0}\lambda_N(\mathbbm{1}) = \sum_{N\geq 0}Np_{N}
\end{equation}
where $p_N$ is the projection from $\mathcal{F}_\text{sym}(\mathcal{H})$ to the $N$-particle subspace. From this observation, it is clear that all $\lambda(\mathcal{O})$ operators commute with the number operator:
\begin{equation}
    [\lambda(\mathcal{O}),\lambda(\mathbbm{1})] = \lambda([\mathcal{O},\mathbbm{1}]) = 0
\end{equation}

If we fix a basis $\{\ket{\xi_i}\}$ for the one-particle Hilbert space $\mathcal{H}$, we can rewrite the generator of the Weyl operator as:
\begin{equation}
    p(f) = i\sum_jf_j(a_j -a_j^\dagger)
\end{equation}
where $\ket{f} = \sum_j f_j\ket{\xi_j}$ and $a_j  = a_{\xi_j}\text{ , }a_j^\dagger = a_{\xi_j}^\dagger$. In addition, because the $\lambda(\mathcal{O})$ operators preserve the particle number, we can rewrite $\lambda(\mathcal{O})$ in terms of a bilinear of creation/annihilation operators:
\begin{equation}
    \lambda(\mathcal{O}) = \sum_{jk}\mathcal{O}_{jk}a^\dagger_ja_k
\end{equation}

A typical choice of vacuum in the symmetric Fock space is a vector $\ket{\Omega}$ that is killed by all annihilation operators. Using this vacuum vector, we can calculate the moment generating function for $p(f)$:
\begin{equation}\label{equation: moment}
    \bra{\Omega}W(tf)\ket{\Omega} = \bra{\Omega}e^{itp(f)}\ket{\Omega} = e^{-\frac{t^2}{2}\bra{f}\ket{f}}
\end{equation}
where $t\in\mathbb{R}$ is a real parameter. This is analogous to the classical Gaussian moment generating function. The Weyl operator acting on the vacuum state creates the coherent state:
\begin{equation}
    \ket{W(f)} := e^{-\frac{\bra{f}\ket{f}}{2}}\sum_{k\geq 0}\frac{1}{\sqrt{k!}}\ket{f}^{\otimes_\text{sym} k}
\end{equation}
where $\ket{f}^{\otimes_\text{sym} k}$ is the $k$-fold symmetric tensor product of $\ket{f}$. Notice that $\ket{W(f)}$ is properly normalized. The coherent states form a dense subset of $\mathcal{F}_\text{sym}(\mathcal{H})$. 

In general in addition to $\lambda(\mathcal{O})$, we can always consider bilinears with terms $a_ka_j^\dagger$, however since $[a_j,a^\dagger_k]=\delta_{jk}$ we can always bring the operator into the normal-ordered form by subtracting the {\it one-point function} to make sure $\lambda(\mathcal{O})=:\lambda(\mathcal{O}):$. Notice that for the operator $\sum_{jk}\mathcal{O}_{jk}a_ja_k^\dagger$ not in the normal form, its one-point function under the vacuum state is the trace:
\begin{equation}
    \bra{\Omega}\sum_{jk}\mathcal{O}_{jk}a_ja_k^\dagger\ket{\Omega} = \sum_{jk}\mathcal{O}_{jk}\delta_{jk} = \tr(\mathcal{O})
\end{equation}
Thus, subtracting the one-point function simply means keeping the traceless part of $\mathcal{O}$. For the normal-ordered $\lambda(\mathcal{O})$ operators, their one-point functions always vanish under the vacuum state $\ket{\Omega}$. In addition, we can calculate the moment generating function:
\begin{equation}
    \bra{\Omega}\Gamma(e^{it\mathcal{O}})\ket{\Omega} = \sum_{k\geq 0}\frac{(it)^k}{k!}\bra{\Omega}\lambda(\mathcal{O})^k\ket{\Omega} = \bra{\Omega}\ket{\Omega} = 1
\end{equation}
where all terms $\bra{\Omega}\lambda(\mathcal{O})^k\ket{\Omega}$ with $k \geq 1$ vanish because $\lambda(\mathcal{O})$ is normal-ordered.

In terms of the coherent state and the normal ordered operator $\lambda(\mathcal{O})$, the one-point function is preserved in the following sense:
\begin{equation}
    \bra{W(f)}\lambda(\mathcal{O})\ket{W(f)} = \sum_{jk}\mathcal{O}_{jk}\bra{f}\ket{\xi_j}\bra{\xi_k}\ket{f} = \bra{f}\mathcal{O}\ket{f}
\end{equation}
where $\{\ket{\xi_j}\}$ is an orthonormal basis for the one-particle Hilbert space $\mathcal{H}$. 
Connected correlation functions are defined as:
\begin{align}
    \begin{split}
        \bra{W(f)}&\lambda(\mathcal{O}_1)\cdots\lambda(\mathcal{O}_k)\ket{W(f)}_\text{conn} \\&:= (-i)^k\frac{\partial^k}{\partial s_1\cdots\partial s_k}|_{s_1,\cdots,s_k = 0}\log \bra{W(f)}e^{is_1\lambda(\mathcal{O}_1)}\cdots e^{is_k\lambda(\mathcal{O}_k)}\ket{W(f)}\ .
    \end{split}
\end{align}
More generally, we observe that the coherent states have the following property:
\begin{equation}
    \bra{W(f)}\lambda(\mathcal{O}_1)\cdots\lambda(\mathcal{O}_k)\ket{W(f)}_\text{conn} = \bra{f}\mathcal{O}_1\cdots \mathcal{O}_k\ket{f}\ .
\end{equation}
We can calculate the logarithm directly. Notice that we have
\begin{align}\label{equation:action}
    \begin{split}
        e^{is\lambda(\mathcal{O})}\ket{W(f)} &= \sum_{k\geq 0}\frac{e^{-\frac{\bra{f}\ket{f}}{2}}}{\sqrt{k!}}\lb e^{is\mathcal{O}}\ket{f}\rb^{\otimes_\text{sym} k}
         = \ket{W(e^{is\mathcal{O}}f)}.
    \end{split}
\end{align}
Therefore, we have:
\begin{align}
    \begin{split}
        \log&\bra{W(f)}e^{is_1\lambda(\mathcal{O}_1)}\cdots e^{is_k\lambda(\mathcal{O}_k)}\ket{W(f)} = \log\bra{W(f)}\ket{W(e^{is_1\mathcal{O}_1}\cdots e^{is_k\mathcal{O}_k}f)} \\&= \log \exp(-\frac{\bra{f}\ket{f} +\bra{e^{is_1\mathcal{O}_1}\cdots e^{is_k\mathcal{O}_k}f}\ket{e^{is_1\mathcal{O}_1}\cdots e^{is_k\mathcal{O}_k}f}}{2} + \bra{f}\ket{e^{is_1\mathcal{O}_1}\cdots e^{is_k\mathcal{O}_k}f})
        \\&
        =\bra{f}e^{is_1\mathcal{O}_1}\cdots e^{is_k\mathcal{O}_k}\ket{f} - 1
    \end{split}
\end{align}
Hence, we have the following preservation property:
\begin{align}
    \begin{split}
        \bra{W(f)}\lambda(\mathcal{O}_1)\cdots\lambda(\mathcal{O}_k)\ket{W(f)}_\text{conn} &= (-i)^k\frac{\partial^k}{\partial s_1\cdots\partial s_k}|_{s_1,\cdots,s_k = 0}\lb\bra{f}e^{is_1\mathcal{O}_1}\cdots e^{is_k\mathcal{O}_k}\ket{f} -  1\rb \\&= \bra{f}\mathcal{O}_1\cdots\mathcal{O}_k\ket{f}    
    \end{split} 
\end{align}
The existence of connected correlators implies that coherent states are not quasi-free (Gaussian) because the connected correlators of normal-ordered operators under the quasi-free states are zero. To better understand the statistical nature of the coherent states $\ket{W(f)}$ consider the simple case of the number operator $\lambda(1)$:
\begin{eqnarray}
    \bra{W(f)}\lambda(1)^k\ket{W(f)}_\text{conn} = \bra{f}\ket{f}
\end{eqnarray}
The $k$-point connected correlator is independent of $k$, which reveals that the coherent state has Poisson statistics, as is clear from the moment generating function:
\begin{eqnarray}
    \bra{W(f)}e^{is\lambda(1)}\ket{W(f)} = \exp\lb|f|^2(e^{is} - 1)\rb
\end{eqnarray}
This is the classical moment generating function of a Poisson random variable with intensity $|f|^2$.

We emphasize that the key data of the calculation above is contained in Equation \ref{equation: moment} and Equation \ref{equation:action}. All other equations can be derived from these two key equations.
\subsection{Coherent vacua of bipartite systems}\label{subsec:coherentbipartite}
In the second quantization, each vector in the one-particle Hilbert space (i.e., pure states) is lifted to a Weyl operator and a coherent state on the symmetric Fock space. More generally, we can consider mixed states of the one-particle Hilbert space and their canonical purification in the double-copy Hilbert space:
\begin{equation}
    \ket{\omega^{1/2}} := \sum_j\omega_j^{1/2}\ket{j\bar{j}}\in\mathcal{H}\otimes\overline{\mathcal{H}}
\end{equation}
For simplicity, we consider a finite-dimensional Hilbert space $\mathcal{H}$ with $\dim\mathcal{H} = d$. The double-copy Hilbert space is the space of $d\times d$-Hilbert-Schmidt matrices. Similar to the second quantization, we consider the symmetric Fock space: $\mathcal{F}_\text{sym}(\mathcal{H}\otimes\overline{\mathcal{H}})$ and define the creation/annihilation operators $a_{i\bar{k}}, a_{j\bar{l}}^\dagger$. These creation/annihilation operators satisfy the following canonical commutation relations:
\begin{equation}
    [a_{i\bar{k}},a^\dagger_{j\bar{l}}] = \delta_{ij}\delta_{kl}
\end{equation}
One can think of these operators simply as regular creation/annihilation operators. Alternatively, one can also organize these operators into a matrix of operators:
\begin{equation}
    \textbf{a} :=\begin{bmatrix}
        a_{1\bar{1}} & \cdots & a_{1\bar{d}}\\a_{2\bar{1}} &\cdots &a_{2\bar{d}}\\\cdots\\a_{d\bar{1}}&\cdots&a_{d\bar{d}}
    \end{bmatrix}
\end{equation}
\begin{equation}
    \textbf{a}^\dagger = \begin{bmatrix}
        a^\dagger_{1\bar{1}} & a^\dagger_{2\bar{1}} &\cdots & a^\dagger_{d\bar{1}}\\a^\dagger_{1\bar{2}}& a^\dagger_{2\bar{2}} & \cdots & a^\dagger_{d\bar{2}}\\\cdots\\a^\dagger_{1\bar{d}}&a^\dagger_{2\bar{d}}&\cdots&a^\dagger_{d\bar{d}}
    \end{bmatrix}
\end{equation}
Notice that this setup is different from what we do in finite-temperature field theory. 
In thermal QFT, we consider the Fock space $\mathcal{F}_\text{sym}(\mathcal{H})\otimes\mathcal{F}_\text{sym}(\overline{\mathcal{H}}) = \mathcal{F}_\text{sym}(\mathcal{H}\oplus\overline{\mathcal{H}})$, and a linear Bogoliubov transformation 
\begin{eqnarray}
    &&c_k=b_+ a_k-b_- a_{\bar{k}}^\dagger,\qquad b_+^2-b_-^2=1\nn\\
    &&c_{\bar{k}}=-b_-a^\dagger_k-b_+a_{\bar{k}}\nn\\
    &&[c_k,c_{k'}^\dagger]=\delta_{kk'}, \qquad [c_{\bar{k}},c_{\bar{k'}}^\dagger]=\delta_{\bar{k}\bar{k'}}
\end{eqnarray}
This linear Bogoliubov transformation disentangles the thermofield double state on the duplicated Fock space to $\ket{\Omega}\otimes \ket{\bar{\Omega}}$ where $\ket{\Omega}$ and $\ket{\bar{\Omega}}$ are respectively defined as vectors in the Fock space killed by all $c_k$ and $c_{\bar{k}}$. There are no operators that can naturally be interpreted as analogs of $a_{i\bar{k}}$ \footnote{Mathematically, the distinction between our construction and finite temperature field theory is evident. Our operators $a_{i\bar{k}}$ acts on the symmetric Fock space $\mathcal{F}(\mathcal{H}\otimes \mathcal{H}^\dagger)$, whereas $a_k$ and $a_{\bar{k}}$ acts on a different space $\mathcal{F}_\text{sym}(\mathcal{H})\otimes \mathcal{F}_\text{sym}(\overline{\mathcal{H}})$. In second quantization and Poissonization, the direct sum of two single-particle Hilbert spaces is lifted to a tensor product in the Fock space: $\mathcal{H}_1\oplus \mathcal{H}_2\to \mathcal{F}_\text{sym}(\mathcal{H}_1)\otimes \mathcal{F}_\text{sym}(\mathcal{H}_2)$. This is used to define finite temperature field theory. Whereas the lift of the tensor product of single-particle Hilbert spaces is the space $\mathcal{F}_\text{sym}(\mathcal{H}_1\otimes \mathcal{H}_2)$, which we use in Poissonization.}. It might be tempting to think $a_{i\bar{k}}$ can be replaced by bilinears $a_i a_{\bar{k}}$, an operator that annihilates a state $\ket{i}$ in $\mathcal{F}_\text{sym}(\mathcal{H})$ and an operator that annihilates a state $\ket{\bar{k}}$ in $\mathcal{F}_\text{sym}(\overline{\mathcal{H}})$ but this is incorrect as one can check explicitly that 
\begin{eqnarray}\label{TFDannihilatoon}
    [a_ia_{\bar{k}},a^\dagger_j a^\dagger_{\bar{l}}]\neq [a_{i\bar{k}},a^\dagger_{j\bar{l}}] \ .
\end{eqnarray}

On the symmetric Fock space $\mathcal{F}_\text{sym}(\mathcal{H}\otimes\overline{\mathcal{H}})$, we consider the operators:
\begin{equation}\label{equation:poisslambda}
    \lambda(\mathcal{O}):=\lambda(\mathcal{O}\otimes \mathbbm{1}) = \sum_{ijk}\mathcal{O}_{ij}a^{\dagger}_{i\bar{k}}a_{j\bar{k}}
\end{equation}
\begin{equation}\label{equation:poisslambda'}
    \lambda'(\widetilde{\mathcal{O}}) := \lambda(\mathbbm{1}\otimes \widetilde{\mathcal{O}}) = \sum_{ikl}\widetilde{\mathcal{O}}_{kl}a^\dagger_{i\bar{k}}a_{i\bar{l}}
\end{equation}
For simplicity, we have short-handed $\lambda(\mathcal{O}\otimes\mathbbm{1})$ simply as $\lambda(\mathcal{O})$. From definition, it is clear that $[\lambda(\mathcal{O}), \lambda'(\widetilde{\mathcal{O}})] = 0$. We also consider the coherent state:
\begin{equation}
    \ket{W(\omega^{1/2})} := e^{\sum_j \omega_j^{1/2}(a^\dagger_{j\bar{j}} - a_{j\bar{j}})}\ket{\Omega}
\end{equation}
Using the same calculation as the previous section, one can show that the connected correlator of $\lambda(\mathcal{O})$ operators is given by:
\begin{align}
    \begin{split}
        \bra{W(\omega^{1/2})}\lambda(\mathcal{O}_1)\cdots\lambda(\mathcal{O}_p)\ket{W(\omega^{1/2})}_\text{conn} = \bra{\omega^{1/2}}\lb \mathcal{O}_1\cdots\mathcal{O}_p\rb\otimes\mathbbm{1}\ket{\omega^{1/2}} = \omega(\mathcal{O}_1\cdots\mathcal{O}_p)
    \end{split}
\end{align}
Using this formula, it is easy to see that the full $p$-point correlation function is given by:
\begin{equation}
    \bra{W(\omega^{1/2})}\lambda(\mathcal{O}_1)\cdots\lambda(\mathcal{O}_p)\ket{W(\omega^{1/2})} = \sum_{\sigma\in\mathcal{P}_p}\prod_{A\in\sigma}\omega(\overrightarrow{\prod}_{j\in A}\mathcal{O}_j)
\end{equation}
where $\mathcal{P}_p$ is the set of partitions on the set $[p]:=\{1,\cdots,p\}$, $\sigma$ is a particular partition and $A\in\sigma$ is a subset in the partition $\sigma$. To be explicit, the full one-point function is given by:
\begin{equation}
    \bra{W(\omega^{1/2})}\lambda(\mathcal{O})\ket{W(\omega^{1/2})} = \omega(\mathcal{O})
\end{equation}
The full two-point function is given by:
\begin{equation}
    \bra{W(\omega^{1/2})}\lambda(\mathcal{O}_1)\lambda(\mathcal{O}_2)\ket{W(\omega^{1/2})} = \omega(\mathcal{O}_1\mathcal{O}_2) + \omega(\mathcal{O}_1)\omega(\mathcal{O}_2)
\end{equation}
And the full three-point function is given by:
\begin{align}
    \begin{split}
        \bra{W(\omega)^{1/2}}&\lambda(\mathcal{O}_1)\lambda(\mathcal{O}_2)\lambda(\mathcal{O}_3)\ket{W(\omega^{1/2})} =\omega(\mathcal{O}_1\mathcal{O}_2\mathcal{O}_3) \\&+\omega(\mathcal{O}_1\mathcal{O}_2)\omega(\mathcal{O}_3) +  \omega(\mathcal{O}_1\mathcal{O}_3)\omega(\mathcal{O}_2) +\omega(\mathcal{O}_2\mathcal{O}_3)\omega(\mathcal{O}_1) \\&+ \omega(\mathcal{O}_1)\omega(X\mathcal{O}_2)\omega(\mathcal{O}_3)
    \end{split}
\end{align}

The $\lambda$ map preserves the commutation relation:
\begin{align}
    \begin{split}
        [\lambda(\mathcal{O}), \lambda(\widetilde{\mathcal{O}})] &= \sum_{ijk,lmn}\mathcal{O}_{ij}\widetilde{\mathcal{O}}_{lm}[a^\dagger_{i\bar{k}}a_{j\bar{k}}, a^\dagger_{l\bar{n}}a_{m\bar{n}}] = \sum_{ijk,lmn}\mathcal{O}_{ij}\widetilde{\mathcal{O}}_{lm}\lb a^\dagger_{i\bar{k}}a_{m\bar{n}}\delta_{jl}\delta_{kn} - a^\dagger_{l\bar{n}}a_{j\bar{k}}\delta_{im}\delta_{kn}\rb
        \\&
        =\sum_{ijkm}\mathcal{O}_{ij}\widetilde{\mathcal{O}}_{jm}a^\dagger_{i\bar{k}}a_{l\bar{k}} - \widetilde{\mathcal{O}}_{li}\mathcal{O}_{ij}a^\dagger_{l\bar{k}}a_{j\bar{k}} = \lambda([\mathcal{O}, \widetilde{\mathcal{O}}])
    \end{split}
\end{align}
We have now arrived at our first construction of Poissonization:
\begin{definition}
    Given a matrix algebra $M_d(\mathbb{C})$ and a mixed state $\omega$, \textbf{\textit{the Poisson algebra $\mathbb{P}_\omega M_d(\mathbb{C})$}} is the von Neumann algebra generated by $\lambda(\mathcal{O})$ operators under the coherent state $\ket{W(\omega^{1/2})}$. These operators act on the symmetric Fock space $\mathcal{F}_\text{sym}(\mathcal{H}\otimes\overline{\mathcal{H}})$. The pure density matrix $\ket{W(\omega^{1/2})}\bra{W(\omega^{1/2})}$ defines a state $\varphi_\omega$ on the Poisson algebra. This state is \textbf{\textit{the Poisson state}}. The linear map $\lambda$ is \textbf{\textit{the Poisson quantization map}}. The general $p$-point correlation function is called \textbf{\textit{the Poisson moment formula}}.
\end{definition}
Much of this construction can be stated in the language of matrix quantum mechanics using $\textbf{a}$ and $\textbf{a}^\dagger$. Define the matrix ``position" and ``momentum" operators:
\begin{equation}
    \mathbf{M} = \frac{1}{\sqrt{2}}\lb\mathbf{a} + \mathbf{a}^\dagger\rb
\end{equation}
\begin{equation}
    \mathbf{\Pi} = \frac{1}{\sqrt{2}i}\lb \mathbf{a} - \mathbf{a}^\dagger\rb
\end{equation}
\begin{equation}
    [\mathbf{M}_{ik}, \mathbf{\Pi}_{jl}] = \delta_{ij}\delta_{kl}
\end{equation}
Using normalized Killing-Cartan basis $\{\mathbf{T}^A\}$ of the Lie algebra $\mathfrak{su}(d)$ \footnote{Here the normalization ensures $\tr(\mathbf{T}^A\mathbf{T}^B) = \delta_{AB}$}, we can decompose the matrix ``position" and ``momentum" operators: $\mathbf{M} = \sum_A M_A\mathbf{T}^A$ and $\mathbf{\Pi} = \sum_B\Pi_B\mathbf{T}^B$ such that $[M_A, \Pi_B] = i\delta_{AB}$. 

Consider the free matrix Hamiltonian:
\begin{equation}
    \mathbf{H} := \frac{1}{2}\tr\lb\mathbf{M}^\dagger\mathbf{M} + \mathbf{\Pi}^\dagger\mathbf{\Pi}\rb = \frac{1}{2}\tr\lb\mathbf{a}^\dagger\mathbf{a} + \mathbf{a}\mathbf{a}^\dagger\rb
\end{equation}
We remind the reader that in the equation above $\tr$ represents the trace in $d\times d$-matrices. This Hamiltonian admits two $U(d)$-symmetries:
\begin{equation}
    \mathbf{a}\mapsto \mathbf{U}\mathbf{a}\text{ , }\mathbf{a}\mapsto\mathbf{a}\mathbf{U}
\end{equation}
where $\mathbf{U}$ is a $d\times d$ unitary matrix. Suppose $\mathbf{U}$ is generated by a self-adjoint matrix $\mathcal{O}$. Then observe:
\begin{equation}
    [\mathbf{a}_{mn}, \sum_{ijk}\lb\mathbf{a}^\dagger\rb_{ki} \mathcal{O}_{ij}\mathbf{a}_{jk}] = \sum_{ijk}[\mathbf{a}_{mn},\mathbf{a}^\dagger_{ki}]\mathcal{O}_{ij}\mathbf{a}_{jk} = \sum_j\mathcal{O}_{mj}\mathbf{a}_{jn} 
\end{equation}
This is the infinitesimal unitary transformation:
\begin{equation}
    (-i)\frac{d}{dt}|_{t = 0}e^{it\mathcal{O}}\mathbf{a} = \mathcal{O}\mathbf{a}
\end{equation}
In other words, the left $U(d)$-symmetry is implemented by the \textit{adjoint} action:
\begin{equation}
    e^{-it\sum_{ijk}\lb\mathbf{a}^\dagger\rb_{ki}\mathcal{O}_{ij}\mathbf{a}_{jk}} \mathbf{a} e^{it\sum_{ijk}\lb\mathbf{a}^\dagger\rb_{ki}\mathcal{O}_{ij}\mathbf{a}_{jk}} = e^{it\mathcal{O}}\mathbf{a}
\end{equation}
Similarly one can check that the right $U(d)$-symmetry is generated by the \textit{adjoint} action by $\sum_{ikl}\mathbf{a}_{il}\widetilde{\mathcal{O}}^T_{lk}\lb\mathbf{a}^\dagger\rb_{ki}$. These generators are precisely the $\lambda(\mathcal{O})$ and $\lambda'(\widetilde{\mathcal{O}})$ operators. The operator $\lambda(\mathcal{O})$ can be rewritten as:
\begin{equation}
    \tr(\textbf{a}^\dagger\mathcal{O}\textbf{a}) = \sum_{ijk}\lb \textbf{a}^\dagger\rb_{ki}\mathcal{O}_{ij}\textbf{a}_{jk} = \sum_{ijk}\mathcal{O}_{ij}a^\dagger_{i\bar{k}}a_{j\bar{k}} = \lambda(\mathcal{O})
\end{equation}
Similarly, $\lambda'(\widetilde{\mathcal{O}})$ can be rewritten as:
\begin{equation}
    \tr(\lb\textbf{a}\rb \widetilde{\mathcal{O}}^T \textbf{a}^\dagger) = \sum_{ikl}\textbf{a}_{il}\widetilde{\mathcal{O}}^T_{lk}\lb\textbf{a}^\dagger\rb_{ki} = \sum_{ikl}\widetilde{\mathcal{O}}_{kl}a^\dagger_{i\bar{k}}a_{i\bar{l}} = \lambda'(\widetilde{\mathcal{O}})
\end{equation}
Therefore, from the perspective of matrix quantum mechanics, the Poisson algebra is generated by the generators of the left $U(d)$-symmetries of the free matrix model. We will not be pursuing the perspective of matrix quantum mechanics further in this work, reserving it for future work.

We motivated the construction above using the second quantization. However, mathematically, the prescription above is a well-defined framework that does not require the introduction of creation and annihilation operators. It takes as input a quantum system (a von Neumann algebra with a state, or more generally a weight) and outputs another von Neumann algebra represented on the Fock space, where both the operators and the states are lifted in a canonical structure-preserving (functorial) way. We call this construction {\it Poissonization}. In general, in Poissonization, there are no operators similar to the coherent operators, and the algebra is solely generated by the lifted operators $\lambda(\mathcal{O})$. As opposed to the conventional second quantization, where the choice of a noncommutative Gaussian (quasi-free) state is to define a free quantum field theory, in Poissonization, the natural choice is the Poisson state, resulting in connected correlators that match the correlators of the state of the input algebra. We postpone a more mathematically rigorous discussion of Poissonization until Section \ref{app:math}.
\subsection{Examples of coherent vacuua of bipartite systems}\label{subsection:coherentexamples}
Here we collect a few important examples of the construction above. These examples will be used throughout this work. 
\begin{enumerate}
    \item \textit{Restricted to diagonal matrices} Given the canonical purification of the input weight $\ket{\omega^{1/2}} = \sum_j \omega_j^{1/2}\ket{j\bar{j}}$, we restrict the possible types of operators to diagonal matrices only:
    \begin{equation}
        \lambda(\mathcal{O}) = \sum_{ik}\mathcal{O}_{ii}a^\dagger_{i\bar{k}}a_{i\bar{k}} 
    \end{equation}
    where $\mathcal{O}$ is a diagonal matrix. In this case, the connected correlation function is given by:
    \begin{align}\label{equation:diagonalconn}
        \begin{split}
            \bra{W(\omega^{1/2})}\lambda(e_{i_1i_1})\cdots\lambda(e_{i_pi_p})\ket{W(\omega^{1/2})}_\text{conn} &= \omega(e_{i_1i_1}\cdots e_{i_pi_p}) = \sum_j\omega_j\delta_{ji_1\cdots i_p}
        \end{split}
    \end{align}
    where $e_{ii} = \ket{i}\bra{i}$ is a diagonal projiection. A general diagonal matrix is a linear combination of these projections. In this case, the construction is equivalent to the following. We map $a^\dagger_{i\bar{k}}\mapsto a^\dagger_i$ and $a_{i\bar{k}}\mapsto a_i$. Then the bipartite coherent state $\ket{W(\omega^{1/2})}$ is mapped to a multimodal coherent state $e^{\sum_j \omega_j^{1/2}(a^\dagger_j -a_j)}\ket{\Omega}$ and the operator $\lambda(\mathcal{O})$ is mapped to a weighted number operator $\sum_i\mathcal{O}_{ii}a^\dagger_ia_i$. In this form, it is easy to see that the algebra generated by $\lambda(\mathcal{O})$ operators is commutative. Such a commutative Poisson algebra will be studied in depth in Appendix \ref{app:closed} for its connection with closed 2D TQFTs \cite{tqftreview, frobeniusTQFT2}.
    \item \textit{Thermofield double state} Given a thermofield double state $\ket{TFD^{1/2}} = \sum_j \frac{e^{-\beta E_j / 2}}{\sqrt{Z}}\ket{j\bar{j}}$, the connected correlation function is given by:
    \begin{align}\label{equation:diagonal}
        \begin{split}
            \bra{W(\rho_\beta^{1/2})}\lambda(e_{i_1j_1})\cdots\lambda(e_{i_pj_p})\ket{W(\rho_\beta^{1/2})}_\text{conn} &= \sum_j \frac{e^{-\beta E_j}}{Z}\delta_{ji_1}\delta_{j_1i_2}\cdots\delta_{j_{p-1}i_p}\delta_{j_pj}
        \end{split}
    \end{align}
    This is a quintessential example of Poissonization of matrix algebra.
    \item \textit{Wishart random matrices} Consider a $d\times k$ rectangular Gaussian random matrices $\psi$ where each entry $\psi_i^a := \bra{i}\psi\ket{a}$ is a standard complex Gaussian random variable \footnote{Recall a standard complex Gaussian random variable follows the complex normal distribution with density $\frac{1}{2\pi}e^{-\frac{|z|^2}{2}}d\Re{z}d\Im{z}$}. Here $1\leq i\leq d$ and $1\leq a\leq k$. Then for each $k\times k$ matrix $\mathcal{O}$, we can consider a \textit{random embedding}:
    \begin{equation}
        \mathcal{O}\mapsto\mathcal{O}^{(\psi)} := \psi\mathcal{O}\psi^\dagger = \sum_{ab}\mathcal{O}^{ab}\psi^{ab}
    \end{equation}
    where $\psi\mathcal{O}\psi^\dagger $ is a matrix multiplication, and the second equation results from inserting the resolution of identity $\sum_a\ket{a}\bra{a}$. The $d\times d$ matrix $\psi^{ab} := \sum_{ij}\psi^a_i\bar{\psi}^b_j\ket{i}\bra{j}$ and $\bar{\psi}^b_i$ is the standard complex Gaussian random variable conjugate to $\psi^b_i$. We will study the mathematical properties of this map in depth in Appendix \ref{subsection:simplechannel}. Notice that $\mathcal{O}^{(\psi)}$ is a $d\times d$ \textit{random matrix}. The matrix $\psi^{ab}$ is a \textit{standard complex Wishart matrix}. The following expectation value is useful (c.f. Appendix \ref{app:openclosed}):
    \begin{align}
        \begin{split}
            \mathbb{E}\tr\lb \psi^{a_1b_1}\cdots\psi^{a_pb_p} \rb &= \sum_{i_1j_1\cdots i_pj_p} \mathbb{E}\lb\psi^{a_1}_{i_1}\bar{\psi}^{b_1}_{j_1}\cdots\psi^{a_p}_{i_p}\bar{\psi}^{b_p}_{j_p}\rb\delta_{j_pi_1}\delta_{j_1i_2}\cdots\delta_{j_{p-1}i_p}
            \\& = \sum_{i_1j_1\cdots i_pj_p}\sum_{\pi\in\mathcal{S}_p}\prod_{1\leq k\leq p}\mathbb{E}\lb\psi^{a_k}_{i_k}\bar{\psi}^{b_{\pi(k)}}_{j_{\pi(k)}}\rb\delta_{j_pi_1}\delta_{j_1i_2}\cdots\delta_{j_{p-1}i_p}
            \\&
            =\sum_{i_1j_1\cdots i_pj_p}\sum_{\pi\in\mathcal{S}_p}\prod_{1\leq k\leq p}\lb\delta^{a_k b_{\pi(k)}}\delta_{i_kj_{\pi(k)}}\rb\delta_{j_pi_1}\delta_{j_1i_2}\cdots\delta_{j_{p-1}i_p}
            \\&
            =\sum_{\pi\in\mathcal{S}_p}d^{\text{cycles}(\pi^0\pi)}\delta^{a_1b_{\pi(1)}}\cdots\delta^{a_pb_{\pi(p)}}
        \end{split}
    \end{align}
    where $\mathcal{S}_p$ is the permutation group on $p$ indices, $\pi^0 = (12\cdots p)$ is the cyclic permutation, and $\text{cycles}(\pi^0\pi)$ counts the number of irreducible cycles in the permutation $\pi^0\pi$ \footnote{For a more detailed derivation and various extensions of this formula, see Appendix \ref{app:openclosed}.}. Using this formula, we can calculate the connected correlation function of the $p$ operators $\lambda(\mathcal{O}_1^{(\psi)})\cdots\lambda(\mathcal{O}_p^{(\psi)})$. This is given by:
    \begin{align}
        \begin{split}
            \mathbb{E}\tr(\mathcal{O}_1^{(\psi)}\cdots\mathcal{O}_p^{(\psi)}) &= \sum_{a_1b_1\cdots a_pb_p}\lb\mathcal{O}_1\rb_{a_1b_1}\cdots\lb\mathcal{O}_p\rb_{a_pb_p}\mathbb{E}\tr(\psi^{a_1b_1}\cdots\psi^{a_pb_p})
            \\&
            =\sum_{a_1b_1\cdots a_pb_p}\sum_{\pi\in\mathcal{S}_p}d^{\text{cycles}(\pi^0\pi)}\delta^{a_1b_{\pi(1)}}\cdots\delta^{a_pb_{\pi(p)}}\lb\mathcal{O}_1\rb_{a_1b_1}\cdots\lb\mathcal{O}_p\rb_{a_pb_p}
            \\&
            =\sum_{\pi\in\mathcal{S}_p}d^{\text{cycles}(\pi^0\pi)}\prod_{\gamma\in\pi}\tr(\overrightarrow{\prod}_{k\in\gamma}\mathcal{O}_k)
        \end{split}
    \end{align}
    where $\gamma$ is an irreducible cycle in $\pi$. For a more detailed discussion and various extensions of this formula, we refer the readers to Appendix \ref{app:openclosed}. Notice that because $\mathcal{O}^{(\psi)}$ is a random matrix, the algebra generated by $\lambda(\mathcal{O}^{(\psi)})$ is noncommutative.

    Instead of embedding a $k\times k$ matrix to a random $d\times d$ matrix, we can also embed a $d\times d$ matrix to a random $k\times k$ matrix:
    \begin{equation}
        \mathcal{O}\mapsto\mathcal{O}_{(\psi)} :=\psi^\dagger\mathcal{O}\psi = \sum_{ij,ab} \bra{b}\psi^\dagger\ket{j}\bra{j}\mathcal{O}\ket{i}\bra{i}\psi\ket{a}\ket{b}\bra{a} = \sum_{ij,ab}\mathcal{O}_{ji}\psi^a_i\bar{\psi}^b_j\ket{b}\bra{a}
    \end{equation}
    Denote $\psi_{ij} :=\sum_{ab}\psi^a_i\bar{\psi}^b_j\ket{b}\bra{a}$. This is a $k\times k$ random matrix. Using the same calculation \footnote{The only difference is to exchange the indices.} as before, we have the following formula:
    \begin{align}
        \begin{split}
            \mathbb{E}\tr\lb\psi_{j_1i_1}\cdots\psi_{j_pi_p}\rb &= \sum_{a_1b_1\cdots a_pb_p}\mathbb{E}\lb\psi^{a_1}_{i_1}\bar{\psi}^{b_1}_{j_1}\cdots\psi^{a_p}_{i_p}\bar{\psi}^{b_p}_{j_p}\rb\delta_{a_1b_2}\delta_{a_2b_3}\cdots\delta_{a_pb_1}
            \\&=\sum_{a_1b_1\cdots a_pb_p}\sum_{\pi\in\mathcal{S}_p}\prod_{1\leq k\leq p}\mathbb{E}\lb\psi^{a_\pi(k)}_{i_{\pi(k)}}\bar{\psi}^{b_k}_{j_k}\rb\delta_{a_1b_2}\delta_{a_2b_3}\cdots\delta_{a_pb_1}
            \\&=\sum_{\pi\in\mathcal{S}_p}k^{\text{cycles}(\pi^0\pi)}\delta_{j_1i_{\pi(1)}}\cdots\delta_{j_p i_{\pi(p)}}
        \end{split}
    \end{align}
    Using this formula, we can calculate the connected correlation function of the $p$ operators $\lambda(\mathcal{O}_{1,(\psi)})\cdots\lambda(\mathcal{O}_{p,(\psi)})$:
    \begin{align}\label{equation:lowerpsi}
        \begin{split}
            \mathbb{E}\tr\lb\mathcal{O}_{1,(\psi)}\cdots\mathcal{O}_{p,(\psi)}\rb &= \sum_{i_1j_1\cdots i_pj_p}\lb\mathcal{O}_1\rb_{j_1i_1}\cdots\lb\mathcal{O}_p\rb_{j_pi_p}\mathbb{E}\tr\lb\psi_{j_1i_1}\cdots\psi_{j_pi_p}\rb
            \\&
            =\sum_{\pi\in\mathcal{S}_p}k^{\text{cycles}(\pi^0\pi)}\prod_{\gamma\in\pi}\tr\lb\overrightarrow{\prod}_{k\in \gamma}\mathcal{O}_k\rb
        \end{split}
    \end{align}
    \item \textit{Restrict to the tracial component of Wishart random matrices} Every matrix can be decomposed into a tracial component $\frac{1}{d}\tr\lb\mathcal{O}\rb\mathbbm{1}$ and a traceless remainder term. Following the previous example, we can project (apply a conditional expectation) the Wishart random matrix $\mathcal{O}^{(\psi)}$ to its tracial component: 
    \begin{equation}
        \overline{\mathcal{O}^{(\psi)}} := \frac{1}{d}\sum_{i,ab}\mathcal{O}^{ab}\psi^a_i\bar{\psi}^b_i\mathbb{I}
    \end{equation}
    In this case, we can also calculate the connected correlation function of $\lambda(\overline{\mathcal{O}_1^{(\psi)}})\cdots\lambda(\overline{\mathcal{O}_p^{(\psi)}})$. 
    \begin{align}
        \begin{split}
            \mathbb{E}\tr\lb\overline{\mathcal{O}_1^{(\psi)}}\cdots\overline{\mathcal{O}_p^{(\psi)}}\rb &= \frac{1}{d^p}\sum_{i_1\cdots i_p}\sum_{a_1b_1\cdots a_pb_p}\mathcal{O}^{a_1b_1}\cdots\mathcal{O}^{a_pb_p}\mathbb{E}(\psi^{a_1}_{i_1}\bar{\psi}^{b_1}_{i_1}\cdots\psi^{a_p}_{i_p}\bar{\psi}^{b_p}_{i_p})
            \\&=\frac{1}{d^p}\sum_{i_1\cdots i_p}\sum_{a_1b_1\cdots a_pb_p}\mathcal{O}^{a_1b_1}\cdots\mathcal{O}^{a_pb_p}\sum_{\pi\in\mathcal{S}_p}\prod_{1\leq k\leq p}\lb\delta^{a_kb_{\pi(k)}}\delta_{i_ki_{\pi(k)}}\rb
            \\&
            =\frac{1}{d^p}\sum_{\pi\in\mathcal{S}_p}d^{\text{cycles}(\pi)}\prod_{\gamma\in\pi}\tr(\overrightarrow{\prod}_{k\in \gamma}\mathcal{O}_k)
            \\&= \frac{1}{d^p}\sum_{\pi\in\mathcal{S}_p}\prod_{\gamma\in\pi}d\tr(\overrightarrow{\prod}_{k\in\gamma}\mathcal{O}_k)
        \end{split}
    \end{align}
    Notice that the algebra generated by the $\lambda(\overline{\mathcal{O}^{(\psi)}})$ operators is commutative because they are all proportional to the $d\times d$ identity matrix. When all operators $\mathcal{O}_i$ are the same, the connected correlation function is given by:
    \begin{equation}
        \mathbb{E}\tr\lb\lb\overline{\mathcal{O}^{(\psi)}}\rb^p\rb = \frac{1}{d^p}\sum_{\pi\in\mathcal{S}_p}\prod_{\gamma\in\pi}d\tr(\mathcal{O}^{|\gamma|})
    \end{equation}
    where $|\gamma|$ is the size of the irreducible cycle $\gamma$. We can reorganize the sum above as follows: for each permutation $\pi\in\mathcal{S}_p$, we are distributing $p$ copies of $\mathcal{O}$ into indistinguishable bins such that there are $\ell_1$ bins with one single $\mathcal{O}$, $\ell_2$ bins with two $\mathcal{O}$'s, and $\ell_k$ bins with $k$ many $\mathcal{O}$'s. Since we have $p$ many $\mathcal{O}$'s in total, we have:
    \begin{equation}
        \sum_{1 \leq k\leq p}k\ell_k = p
    \end{equation}
    The number of ways of having $\ell_k$ bins with $k$ many $\mathcal{O}$'s is given by $\frac{p!}{\prod_{1\leq k\leq p}k^{\ell_k}\ell_k!}$. To each bin with $k$ many $\mathcal{O}$'s we associate it with $d\tr(\mathcal{O}^k)$. Therefore to each permutation $\pi$ we associate:
    \begin{eqnarray}
        d^{\sum_k\ell_k}\tr(\mathcal{O})^{\ell_1}\tr\lb\lb\mathcal{O}\rb^2\rb^{\ell_2}\cdots\tr\lb\lb\mathcal{O}\rb^p\rb^{\ell_p}
    \end{eqnarray}
    Hence, we can rewrite:
    \begin{equation}\label{equation:simplecombinatorics}
        \mathbb{E}\tr\lb\lb\overline{\mathcal{O}^{(\psi)}}\rb^p\rb = \frac{1}{d^p}\sum_{\substack{0\leq\ell_1,\cdots,\ell_p\leq p\\ \ell_1 + 2\ell_2 + \cdots p\ell_p = p}}\frac{p!}{\prod_{1\leq k\leq p}k^{\ell_k}\ell_k!}d^{\sum_k\ell_k}\tr(\mathcal{O})^{\ell_1}\cdots\tr\lb\lb\mathcal{O}\rb^p\rb^{\ell_p}
    \end{equation}
    Using the formula, we can calculate the moment generating function:
    \begin{align}
        \begin{split}
            \mathbb{E}\tr\lb e^{\mu\overline{\mathcal{O}^{(\psi)}}}\rb &= \sum_{p\geq 0}\frac{\mu^p}{p!}\mathbb{E}\tr\lb\lb\overline{\mathcal{O}^{(\psi)}}\rb^p\rb\\&
            =\exp(-d\tr\log\lb \mathbbm{1} - \frac{\mu}{d}\mathcal{O}\rb) = \frac{1}{\det(\mathbbm{1} - \frac{\mu}{d}\mathcal{O})^d}
        \end{split}
    \end{align}
    This is the moment generating function of the Wishart distribution. Using this formula, we can directly calculate the moment generating function of the operators $\lambda(\overline{\mathcal{O}^{(\psi)}})$. The detailed derivation of this formula and its relation to Marolf-Maxfield's topological gravity with end-of-the-world branes \cite{MM} are discussed in depth in Appendix \ref{app:openclosed}.
    \end{enumerate}
    
\section{Baby universes and Poissonization of commutative algebras}\label{sec:babyuniversescommutative}

In this section, we establish that the simplest examples of our construction, namely the Poissonization of a commutative algebra, correspond to models of 2D closed topological QFTs whose path-integral sums over all Riemann surfaces interpolating between boundary circles. As the first example, in subsection \ref{MMmodel}, we show that the Marolf-Maxfield model of topological gravity in two-dimensions \cite{MM} is the Poissonization of the commutative algebra of complex numbers. Then, in subsection \ref{closedTQFFT}, we establish that more general theories of baby universe obtained in \cite{gardiner20212D} and \cite{banerjee2022comments} by summing over bordisms of circles in a general closed 2D topological field theory correspond to the Poissonization of the commutative algebra of diagonal matrices. In subsection \ref{subsec:otherexamples}, we comment on the cases of 2D Yang-Mills theory and the commutative algebra of pure JT gravity. For completeness, in Appendix \ref{app:closed}, we have included a review of closed and open/closed 2D TQFT and more mathematical details of the calculations of this Section.

\subsection{Marolf-Maxfield closed baby universes}\label{MMmodel}

In this section, we discuss the Marolf-Maxfield toy model of closed baby universes. First, we present a concise review of this simple theory. Then, we reformulate the theory in terms of Poissonization. In \cite{MM}, Marolf and Maxfield (MM) considered a topological two-dimensional theory that associates to a genus $g$ Riemann surface $\mM_{g,m}$ with $m$ closed boundaries a topological action that is
\begin{eqnarray}
    S[\mathcal{M}_{g,m}]=-S_0\chi(\mathcal{M}_{g,m})-S_\p \:m=-S_0(2-2g)-(S_\p-S_0) m
\end{eqnarray}
where $\chi$ is the Euler character, $g$ is the genus, and $S_0$ and $S_\p$ are parameters of the theory \footnote{Later, we will specialize to the case $S_0=S_\partial$.}.
This is a topological QFT, and in the holographic picture, the boundary quantum mechanics is trivial: the Hamiltonian is zero, the partition function is independent of $\beta$ and given by the dimension of the Hilbert space (a positive integer $d$). 

In the MM model, the gravity path-integral associated with the manifold $\mathcal{M}_{g,m}$ defines the weight
\begin{eqnarray}
    Z[\mathcal{M}_{g,m}]=e^{-S[\mathcal{M}_{g,m}]}
\end{eqnarray}
Summing over the weights that correspond to all two-dimensional manifolds (bordisms), connected or not, that interpolate between $m$-boundaries gives 
\begin{eqnarray}\label{ZnMM}
    Z[m]=\sum_{\substack{\text{$\mathcal{M}$ s.t. $\p\mathcal{M}$ has}\\
     \text{$m$ components.}}}\mu(\mathcal{M})e^{-S[\mathcal{M}]}
\end{eqnarray}
where $\mu(\mathcal{M})=\frac{1}{\prod_g n_g!}$ and $n_g$ is the number of connected components of genus $g$ with no boundaries in $\mathcal{M}$. From the expression above, it is clear that we are assuming the partition function of disconnected spacetimes multiply 
\begin{eqnarray}
    Z[\mathcal{M}_1\sqcup \mathcal{M}_2]=Z[\mathcal{M}_1]Z[\mathcal{M}_2]
\end{eqnarray}
and the sum in (\ref{ZnMM}) involves summing over all ways of partitioning $m$ boundaries into $k$ indistinguishable bins (connected spacetimes). The combinatorics of such partitions are controlled by the Touchard polynomials $B_m(\lambda)$, which are intimately tied to the Poisson distribution (see Appendix \ref{app:setpartitions}):
\begin{eqnarray}
    &&Z[m]=e^{Z[0]} B_m(Z[0]) e^{m(S_\p-S_0)}\nn\\
 &&Z[0]=\sum_ge^{S_0(2-2g)}=\frac{e^{2S_0}}{1-e^{-2S_0}}\ .
\end{eqnarray}
The quantity $Z[0]$ is the (re)-normalization of the Hartle-Hawking vacuum. MM interpreted the ratio below as the correlation functions of the ``partition function" operator $\hat{Z}$ that creates closed baby universes in a no-boundary state (Hartle-Hawking state) $\ket{HH}$.
\begin{eqnarray}
    \langle\hat{Z}^m\rangle\equiv \bra{HH}\hat{Z}^m\ket{HH}=\frac{Z[m]}{Z[0]}\ .
\end{eqnarray}
The equation in (\ref{ZnMM})  expressed in terms of the connected correlators becomes
\begin{eqnarray}\label{ZconMM}
  && \langle\hat{Z}^m\rangle_\text{conn}=Z[0] \sum_g e^{-S[\mathcal{M}_{g,m}]}=Z[0]x^m\ .
\end{eqnarray}
The generating function for these connected correlators is
\begin{eqnarray}
  &&\log\langle e^{u \hat{Z}}\rangle=Z[0]\lb e^{u x}-1 \rb\nn\\
   && x= e^{(S_\p-S_0)}\ .
\end{eqnarray}
Using these correlators, MM constructed the GNS Hilbert space of the commutative algebra of $\hat{Z}$ by interpreting the correlators as an overlap between excited states $\ket{Z^m}$:
\begin{eqnarray}
    &&\langle \hat{Z}^{m_1+m_2}\rangle=\braket{Z^{m_1}}{Z^{m_2}}\nn\\
    &&\ket{Z^m}=\hat{Z}^m\ket{HH}
\end{eqnarray}
The vacuum $\ket{HH}$ is called the Hartle-Hawking state:
\begin{eqnarray}\label{pnHH}
    \langle \hat{Z}^m\rangle=\bra{HH}\hat{Z}^m\ket{HH}=\sum_{n=0}^\infty p_{Z[0]}(n) (nx)^m\ .
\end{eqnarray}
where $p_{Z[0]}(n)$ is a Poisson distribution with intensity variable $\lambda=Z[0]$. The problem with this expression is that in a topological field theory, the Hamiltonian is zero and the partition function is a positive integer, namely the dimension of the Hilbert space. Promoting the partition function to a commutative operator means that $\hat{Z}$ is a random variable that takes values in $\mathbb{N}$. This contradicts the equation (\ref{pnHH}), which suggests that the eigenvalues of $\hat{Z}$ are non-integer values $nx$. To solve this problem, Marolf-Maxfield specializes to the case $S_\p=S_0$ so that $x=1$ and the eigenvalues $\hat{Z}$ are integers. Alternatively, one can allow for non-zero $x$ by promoting $xZ$ to an operator.

It is important to point out that in this model, the connected correlators of $m$ copies of the partition function are independent of $m$: 
\begin{eqnarray}\label{connectedZp}
    \langle Z^p\rangle_\text{conn}=\sum_{g>0}e^{-S_0(2g-2)}=\frac{e^{2S_0}}{1-e^{-2S_0}}\equiv \lambda\ .
\end{eqnarray}
This is the unique characteristic of a Poisson distribution
\begin{eqnarray}
    p_\lambda(n)=\frac{e^{-\lambda}\lambda^n}{n!}\ .
\end{eqnarray}

In Poissonization, the above model is the Poissonization of the commutative algebra of complex numbers (see also Appendix \ref{subsection:MMDW}). Each boundary circle is associated with a $1$-dimensional Hilbert space with trivial Hamiltonian $H = 0$. On this trivial system, we introduce a weight $\ket{\mu 1}$ such that $\bra{\mu 1}z\ket{\mu 1} = |\mu|^2z$ where $z\in \mathbb{C}$. Performing Poissonization on the data $(\mathbb{C}, \ket{\mu 1})$, we model the Hartle-Hawking state by the coherent vector $\ket{W(\mu 1)}$ where $\mu$ is fixed to satisfy $|\mu|^2 = Z[0]$ and the \textit{partition function operator} $\widehat{Z}$ by the Poisson quantized operator $\lambda(1)$ (i.e. the number operator). Notice that in Poissonization $\widehat{Z}$ is a second-order process involving emission and absorption of open baby universes (see also Section \ref{sec:sumovertopology}). As a brief check, the connected correlation function is given by:
\begin{equation}
    \bra{W(\mu 1)}\lambda(1)^m\ket{W(\mu 1)}_\text{conn} = \bra{\mu 1}1^m\ket{\mu 1} = |\mu|^2 = Z[0]
\end{equation}
This is exactly Equation \ref{ZnMM}. For more detailed discussion and comparison, please refer to Appendix \ref{subsection:MMDW}.

Notice that because $\widehat{Z}$ is modeled as a second-order process, and the excited states are found by acting the Poisson quantized operator (i.e., the number operator in this case) on the Hartle-Hawking state:
\begin{equation}
    \lambda(1)\ket{W(\mu 1)} = \mu a^\dagger\ket{W(\mu 1)}
\end{equation}
In addition, the following \textit{factorization} property\footnote{In the conventional approach to the universe field theory, the states that satisfy factorization are called $\alpha$-states.} holds:
\begin{equation}
    \bra{n}\lambda(z_1)\cdots\lambda(z_p)\ket{n} = z_1\cdots z_p\bra{n}\lambda(1)^p\ket{n} = z_1\cdots z_pn^p = \bra{n}\lambda(z_1)\ket{n}\cdots\bra{n}\lambda(z_p)\ket{n}
\end{equation}
where $z_i\in\mathbb{C}$ and $\lambda(z) = z\lambda(1) = z\widehat{N}$. The state $\ket{n}$ is the number eigenstate with $n$ closed baby universes. 

\subsection{General closed 2D TQFT}\label{closedTQFFT}

In \cite{gardiner20212D,banerjee2022comments}, the authors generalized the MM model by considering the action that is the sum of $S_0\chi$ and the Dijkgraaf-Witten action for topological gauge fields with finite gauge group $G$. It is convenient to label the matter boundary conditions on each circle by an irreducible representation $r$ of $G$. The result of the path-integral over a Riemann surface $\mathcal{M}_{g,m}$ with boundary labels $r_i$ with $i=1,\cdots, m$ associated to the boundary circle $i$ is \cite{gardiner20212D}
\begin{eqnarray}\label{dijWitt}
    Z[\mathcal{M}_{g,m},r_1,\cdots, r_m]= \lb\frac{e^{S_0}d_{r_1}}{|G|}\rb^{2-2g-m}\delta_{r_1\cdots r_m}\ .
\end{eqnarray}
where $d_r$ is the dimension of the irreducible representation $r$.
As before, to manifolds with disjoint pieces, we associate the multiplication of the partition functions of their connected components.
Summing over all bordisms with $m$ boundaries with the boundary conditions $r_1$ to $r_m$ to define the total partition function
\begin{eqnarray}\label{corr2DTQFT}
    Z[r_1,\cdots, r_m]=\sum_{\substack{\text{$\mathcal{M}$ s.t. $\p\mathcal{M}$ has}\\
     \text{$m$ components.}}}\mu(\mathcal{M})Z[\mathcal{M},r_1,\cdots, r_m]
\end{eqnarray}
Following MM \cite{gardiner20212D} and \cite{banerjee2022comments} interpreted the sum over bordisms in (\ref{corr2DTQFT}) as the correlators of a {\it partition function operator}
\begin{eqnarray}
    \langle \hat{Z}(r_1)\cdots \hat{Z}(r_m)\rangle=\frac{Z[r_1,\cdots,r_m]}{Z[0]}
\end{eqnarray}
where
\begin{eqnarray}\label{map2DTQFT}
    &&Z[0]=\sum_g\sum_rx_r^{2g-2}=\sum_rZ_r[0]\\\label{map2DTQFT2}
    &&x_r= \frac{|G|}{d_r} e^{-S_0},\qquad Z_r[0]= \frac{x_r^{-2}}{1-x_r^{2}}
\end{eqnarray}
The connected correlators and their generating function associated to (\ref{corr2DTQFT}) are
\begin{eqnarray}\label{cumulantsclosedTQFT}
    &&\langle \hat{Z}(r_1)\cdots \hat{Z}(r_m)\rangle_\text{conn}=Z_{r_1}[0]\: x_{r_1}^{m}\delta_{r_1\cdots r_m}\nn\\
    &&\log\langle e^{\sum_r u_r\hat{Z}(r)}\rangle=\sum_r Z_r[0]\lb e^{u_rx_r}-1\rb\ .
\end{eqnarray}
 More generally, one can consider the sum over bordisms in a general 2D semi-simple closed TQFT where the label $r$ and parameter $x_r$ corresponding to an irreducible representation of the group are replaced by the label $r$ for the eigenkets of the handle-creation operator with eigenvalue $x_r^2$. The formulae in (\ref{corr2DTQFT}),(\ref{map2DTQFT}) and (\ref{cumulantsclosedTQFT}) continue to hold \cite{gardiner20212D,banerjee2022comments}.

Following MM, \cite{gardiner20212D,banerjee2022comments}  construct the baby universe Hilbert space by interpreting the correlator as an inner product 
\begin{eqnarray}
    \langle \hat{Z}(r_1)\cdots \hat{Z}(r_m)\rangle=\braket{Z(r_1)\cdots Z(r_k)}{Z(r_{k+1})\cdots Z(r_m)}
\end{eqnarray}
in a GNS Hilbert space built on top of the Hartle-Hawking vacuum
\begin{eqnarray}\label{closed2DTQFTHH}
    \langle \hat{Z}(r_1)\cdots \hat{Z}(r_m)\rangle=\braket{\hat{Z}(r_1)\cdots \hat{Z}(r_m)}{HH}\ .
\end{eqnarray}
The excitations are
\begin{eqnarray}
    \ket{Z(r_1)\cdots Z(r_k)}=\hat{Z}(r_1)\cdots \hat{Z}(r_k)\ket{HH}\ .
\end{eqnarray}
It is clear from (\ref{cumulantsclosedTQFT}) and (\ref{closed2DTQFTHH}) that the Hartle-Hawking vacuum factors as $\ket{HH}=\otimes_r \ket{HH_r}$ and
\begin{eqnarray}
    &&\bra{HH}e^{\sum_r u_r\hat{Z}(r)}\ket{HH}=\prod_r\bra{HH_r}e^{u_r\hat{Z}(r)}\ket{HH}_r\nn\\
    &&\bra{HH_r}e^{u_r\hat{Z}(r)}\ket{HH_r}=e^{Z_r[0](e^{u_rx_r}-1)}\ .
\end{eqnarray}
As before, if we interpret the operator $\hat{Z}(r)$ as independent Poisson random variables, we find that it has Poisson probability $p_r(Z_r[0])$ associated with non-integer eigenvalues $n_r x_r$. In other words, we can write the expression above as
\begin{eqnarray}
   &&\bra{HH_r}e^{u_r\hat{Z}(r)}\ket{HH_r}=\int d\alpha \: p_{r,\alpha} e^{u_r\alpha_r}\nn\\
   &&p_{r,\alpha}=p_{n_r}(Z_r[0])\delta(\alpha_r-x n_r)
\end{eqnarray}
where $\alpha$ is a continuous variable, and $p_{n_r}(Z_r)$ is a Poisson random variable. 

We now show that the above construction is the Poissonization of a commutative algebra generated by the operator $X=\oplus_r x_r \mathbb{I}_r$ where $\mathbb{I}_r$ is the identity operator of dimension $d_r$. Note that the matrices $\mathbb{I}_{r}$ are introduced for convenience, and the same analysis applies for a set of rank one projections $e_{rr}$ with the parameter $x_r$ correctly identified in Poissonization. Consider the ``Hamiltonian" defined by the equation $X=e^{-\beta H}$ in the input algebra of diagonal matrices, where $\beta$ is an inverse temperature, and $x_r=e^{-E_r}$ with $E_r$ the eigenvalues of the ``one-boundary" Hamiltonian. Comparing to (\ref{map2DTQFT2}) suggests the identification 
\begin{eqnarray}
    &&E_r=E_0+\log d_r\\
    &&E_0=S_0-\log |G|
\end{eqnarray}
where $E_0$ is the {\it vacuum energy} corresponding to the trivial representation. The operator $\lambda(x_r\mathbb{I}_r)$ creates a closed boundary with boundary conditions $r$. We choose the weight $\oplus_r|\mu_r|^2\mathbb{I}_r$ and the Hartle-Hawking coherent state becomes
\begin{eqnarray}
    \ket{W(\mu \mathbbm{1})}=\otimes_r \ket{W(\mu_r \mathbbm{1}_r)}
\end{eqnarray}
In Poissonization, the connected correlation function is given by (c.f. Equation \ref{equation:diagonalconn})
\begin{eqnarray}
    \bra{W(\mu \mathbbm{1})}\lambda(x_{r_1}\mathbb{I}_{r_1})\cdots \lambda(x_{r_m}\mathbb{I}_{r_m})\ket{W(\mu \mathbbm{1})}_\text{conn}=(|\mu_{r_1} |^2D_{r_1})x_{r_1}^m\delta_{r_1\cdots r_m}\ .
\end{eqnarray}
Compared to equation (\ref{cumulantsclosedTQFT}), we find that summing over bordisms of a general 2D closed TQFT is the same as the Poissonization of the commutative algebra of diagonal matrices $\oplus x_r e_{rr}$ with the following identification
\begin{eqnarray}
    Z_r[0]=|\mu_r|^2D_r,\qquad \ket{HH}=\ket{W(\mu 1)}\ .
\end{eqnarray}
The factorized states correspond to the eigenstates of the number operator $\lambda(x_r\mathbb{I}_r)=x_r\hat{N}_r$.

For a more detailed discussion on Poissonization and 2D closed TQFTs, please refer to Appendix \ref{app:closed}. For a more detailed discussion on Poissonization and Dijkgraaf-Witten theory, please refer to Subsection \ref{subsection:MMDW}.

\subsection{2D Yang-Mills and JT}\label{subsec:otherexamples}

\paragraph{2D Yang-Mills:} To take a step away from topological theories, we can consider the example of  Yang-Mills theory in two dimensions. The action is 
\begin{eqnarray}
    I=-\frac{1}{4e^2}\int d^2x \sqrt{g} \tr(F_{\mu\nu}F^{\mu\nu})
\end{eqnarray}
which is invariant under the group of area-preserving diffeomorphisms. This group is so large that this theory has a type of topological QFT, as it depends only on $e^2A$, where $A=\int d^2x\sqrt{g}$ is the total area of the Riemann surface. Similar to a TQFT, we can label each boundary circle with an irreducible representation of the gauge group $r$. Then, the partition function of 2D Yang-Mills on a Riemann surface $\mM_{g,m}$ of genus $g$ with $m$ boundary circles is \cite{witten1991quantum}
\begin{eqnarray}\label{YangMills}
    Z[e^2A, \mM_{g,m},r_1,\cdots, r_m]=e^{-e^2A c_2(r)/2}d_r^{2-2g-m}\delta_{r_1\cdots r_m}
\end{eqnarray}
which is the same as the topological answer in (\ref{dijWitt}) with an extra exponential $e^{-e^2A\: c_2(r)/2}$ where $c_2(r)$ is the quadratic Casimir of $G$. 

Since the partition function of 2D Yang-Mills depends on $e^2A$, in addition to summing over topologies, we will have to integrate over areas as well. For a discussion of a measure for such integrals and the divergences that can arise, see \cite{Banerjee:2022szm}. Here, we just point out that the expression in (\ref{YangMills}) is compatible with a model of Poissonization for a commutative algebra where each close boundary with representation $r$ corresponds to the insertion of a projection $e_r$ and the $\delta_{r_1\cdots r_m}$ can be attributed to $\tr(e_{r_1}\cdots e_{r_m})$ in Poissonization.

\paragraph{Pure JT gravity:}
JT gravity is a simple 2D theory of gravity with the action
\begin{eqnarray}\label{JTaction}
    I_{JT}=S_0\chi-\left(\frac{1}{2}\int_\mathcal{M}\sqrt{g}\phi(R+2)+\int_{\partial\mathcal{M}}\phi K\right)
\end{eqnarray}
where $\chi=2-2g-n$ is the Euler characteristic of a Riemann surface of genus $g$ with $n$ disks removed.
In \cite{saad2019jt}, it was shown that JT gravity is dual to the double-scaling limit of a random matrix theory with potential $V(H)$
\begin{eqnarray}
    \mathcal{Z}=\int dH\:e^{-L\text{tr}(V(H))}
\end{eqnarray}
where $H$ is an $L\times L$ Hermitian matrix, and the partition function $Z(\beta)=\text{tr}(e^{-\beta H})$ may be viewed as an observable. In the double scaling limit $L\to \infty$ and the potential is tuned to zoom in over the edge of the spectrum where the density of eigenvalues takes the finite value $e^{S_0}$. In this correspondence, $S_0\sim 1/G_N$ in JT gravity, and the contribution of higher topologies to connected correlators of $Z(\beta)$ suffers a non-perturbative suppression  $e^{\chi S_0}$ where $\chi=2-2g-n$ is the Euler character of the surface
\begin{eqnarray}\label{Conneccor}
    \mathbb{E}(Z(\beta_1)\cdots Z(\beta_n))_{conn}=\sum_{g>0}e^{-S_0(2g+n-2)}Z_{g,n}(\beta_1,\cdots ,\beta_n)\ .
\end{eqnarray}
The higher genus terms represent the processes that involve the joining and splitting of closed asymptotic universes (the third quantization).

In the lowest order in this expansion, the one-point function of $Z(\beta)$ is given by the disk partition function in JT gravity \cite{cotler2017black,kitaev2019statistical,yang2019quantum}
\begin{eqnarray}
\mathbb{E}(Z(\beta))=\frac{e^{S_0}e^{\pi^2/\beta}}{\sqrt{2\pi}(2\beta)^{3/2}}
\end{eqnarray}
 whose inverse Laplace transform gives the density of states $\rho_0(E)$ with
\begin{eqnarray}\label{equation:spectraldensity}
    \mathbb{E}(\rho_0(E))=\frac{e^{S_0}}{4\pi^2}\sinh(2\pi\sqrt{E})\ .
\end{eqnarray}
Similarly, the dominant (genus zero) contribution to the connected correlator is given by the double trumpet
\begin{eqnarray}
    \mathbb{E}(Z(\beta_1)Z(\beta_2))_{conn}=\frac{\sqrt{\beta_1\beta_2}}{2\pi(\beta_1+\beta_2)}
\end{eqnarray}
whose inverse Laplace transform gives
\begin{eqnarray}
    \mathbb{E}(\rho_0(E_1)\rho_0(E_2))=\frac{-1}{2\pi^2}\frac{E_1+E_2}{\sqrt{E_1E_2}(E_1-E_2)^2}\ .
\end{eqnarray}
The higher-order correlators are computed using Mirzakhani recursion relations \cite{saad2019jt}.
 
At the level of the action, using the first-order formalism, one can rewrite the pure JT action in (\ref{JTaction}) as a topological BF gauge theory with gauge group $SL(2,\mathbb{R})$. However, $SL(2,\mathbb{R})$ is non-compact and its irreducible representations are infinite-dimensional. This goes beyond the conventional definition of topological TQFT a la Atiyah's axioms. 
However, similar to the topological theories, the algebra of operators of JT gravity is commutative \cite{penington2023algebras}. Recall that for a commutative algebra, the Poisson formula implies
\begin{eqnarray}\label{Poissoncommutative}
    \langle \hat{Z}(\beta_1)\cdots \hat{Z}(\beta_p)\rangle_{conn}=Z(\beta_1+\cdots +\beta_p)\ .
\end{eqnarray}
As we argued earlier, the Poisson limit theorem implies that at late times $T\sim e^{2S_0}$ the multi-point generalization of the spectral form factor with  $t_1+\cdots t_p=0$ is given by the Poisson answer above. Since JT is a chaotic theory with matrix degrees of freedom and eigenvalue repulsion, the onset of the plateau is earlier $T\sim e^{S_0}$\footnote{Of course, the plateau persists for later times as well.}. To isolate the late-time physics associated with the Poisson limit theorem (the plateau), following \cite{okuyama2020multi,blommaert2023integrable,saad2024convergent}, one can define the $\tau$-scaling limit $e^{S_0}\to \infty$ and $t\to \infty$ with $\tau=te^{-S}$ kept fixed\footnote{In \cite{saad2019late}, the author studied the contribution of closed baby universes to the ramp physics at time scale $T\sim e^{S_0/2}$. Here, we are concerned with later times and the physics of the plateau.}. Then, for $\rho_0 (E)<t_1,\cdots, t_p$ one obtains \footnote{We thank Luca Iliseiu for pointing out the recent work on the $\tau$-scaling limit of JT to us.}
\begin{eqnarray}
    \langle Z(\beta+t_1)\cdots Z(\beta+t_p)\rangle_{conn}=\int_0^\infty dE e^{-(\beta_1+\cdots +\beta_p)E}\rho_0(E) =Z(\beta_1+\cdots +\beta_p)
\end{eqnarray}
which matches the Poisson expression in (\ref{Poissoncommutative}).
For the two-point and three-point functions, the formula above was derived by resuming the sum over topologies in the $\tau$-scaling limit of the Mirzakhani recursion relations \cite{okuyama2020multi,blommaert2021eigenbranes,blommaert2023integrable,saad2024convergent}. Note that for short times, the sum over genus is divergent; however, at late enough times, the recursion relations simplify, and we can resum the sum over genus explicitly.

In JT gravity, the probe limit of a single open universe and a Fock space of closed baby universes in JT was studied in \cite{penington2023algebras}. For completeness, we summarize some of the key aspects of this Fock space and the universe field theory of JT gravity studied in \cite{post2022universe} in Appendix \ref{app:JT}.

\section{Open universes and Poissonization of matrix algebras}\label{sec:babyuniversesnonAbel}

In this section, we consider three examples of 2D theories with sum over topologies that include dynamical degrees of freedom in the bulk: the MM models with the EOW branes in subsection \ref{subsec:MMEOW}, the open-closed 2D TQFT in subsection \ref{subsection:openclosed}, and the simplified JT gravity with the EOW branes in subsection \ref{subsec:JTEOW}. We realize the sum over topologies in these examples as Poissonization of matrix algebras with random operators (see also Appendix \ref{app:openclosed}). Our discussion of simple operators of chaotic systems in Section \ref{sec:chaoticsystems} provided a physical justification of the randomness required in matching Poissonization with the sum over topologies in these examples \footnote{The study of lower-dimensional models of quantum gravity, and their connection to random matrices and chaos, has hightligthed the physics importance of randomness in simple operators.}. For more details on the calculations of this Section, see Appendix \ref{app:openclosed},

However, before giving explicit examples, in subsection \ref{subsec:indistinguish}, we point out that even in the absence of randomness, if we symmetrize the correlators of Poissonization of matrix algebras, we can match models with dynamical degrees of freedom in the bulk. The origin of this is simply the fact that Poissonization treats excitations as distinguishable, but bulk dynamical degrees of freedom should be treated as indistinguishable.
In subsection \ref{subsec:ETH}, we argued that this symmetrization can be explained using random operators, which is physically rooted in the ETH for chaotic quantum systems. 

\subsection{Indistinguishability and symmetrization}\label{subsec:indistinguish}

In Poissonization, we obtain a noncommutative algebra on the symmetric Fock space of the boundary Hilbert space. The Poisson statistics follow from counting the number of partitions of a set of $n$ distinguishable objects into $m$ indistinguishable bins. In comparison to a sum over topologies, each bin corresponds to a connected universe, and disconnected universes are indistinguishable. Connected correlators correspond to connected universes. In Poissonization, the contribution of a single connected universe to the correlators of $\lambda(\mathcal{O}_1)\cdots \lambda(\mathcal{O}_n)$ is proportional to $\text{tr}(\mathcal{O}_1\cdots \mathcal{O}_n)$. The noncommutativity of the algebra of Poissonization is related to the fact that $\lambda(\mathcal{O}_i)$ are treated as {\it distinguishable} operators, which is reflected in the fact that the order of operators in $\text{tr}(\mathcal{O}_1\cdots \mathcal{O}_n)$ matters. Insisting on a description that treats $\lambda(\mathcal{O}_i)$ as indistinguishable results in a set of correlation functions that do not care about order. We will see that the combinatorics of indistinguishable operators $\lambda(\mathcal{O})$ can be derived from Poissonization simply by summing over all orders:
\begin{eqnarray}\label{transtocomm}
    \langle \lambda(\mathcal{O}_1)\cdots \lambda(\mathcal{O}_n)\rangle_\text{Indis, conn}=\sum_{\pi\in c_n}\langle \lambda(\mathcal{O}_{\pi(1)})\cdots \lambda(\mathcal{O}_{\pi(n)})\rangle_\text{Poiss,conn} 
\end{eqnarray}
where $c_n$ are all $n$-cycles. Note that we need to restrict the sum to $n$-cycles to ensure that the correlated (universe) remains connected. Of course, after symmetrization, commuting the operators $\lambda(\mathcal{O}_i)$ does not change the correlation function.

In the discussion above, one might get the impression that the distinguishability of operators $\lambda(\mathcal{O}_i)$ is a choice. If we choose to be distinguishable, we obtain Poissonization, and if we choose them to be indistinguishable, we have to perform the sum in (\ref{transtocomm}), and we obtain the result that matches the examples with dynamical degrees of freedom. However, as we will see, the randomness in simple operators results in the statistics required by indistinguishable bulk probes.

\subsection{Marolf-Maxfield model with EOW branes}\label{subsec:MMEOW}

In this subsection, we revisit Marolf-Maxfield's toy model of topological quantum gravity. To come closer to realistic models of gravity, Marolf and Maxfield generalized their model in \cite{MM} by introducing a new type of boundary that is an interval ending on a pair of the {\it end-of-the-world} (EOW) branes. This generalization is analogous to cutting a closed boundary in the middle, as we described in Section \ref{sec:sumovertopology}, and setting different boundary conditions on the lower and upper half-circles; see Figure \ref{fig2}. They allow $k$ EOW brane types, labeling boundary conditions by a pair $(a,b)$ with $a,b=1,\cdots, k$. The EOW branes are dynamical, and the gravitational path-integral sums over all EOW configurations allowed. This means that different EOW configurations in the sum are indistinguishable.

To build a Hilbert space of baby universes with these new labels, Marolf and Maxfield introduced new operators that we denote by $\psi^{ab}$ \footnote{In the original work, Marolf and Maxfield used indices $1\leq i,j\leq k$ to label the types of EOW branes. They denote these new operators as $\widehat{(\psi_j,\psi_i)} = \sum_{1\leq a\leq d}\psi^a_i\bar{\psi}^a_j$. In our notation $\psi^{ab} := \frac{1}{d}\sum_{i}\psi^a_i\bar{\psi}^b_i$. Notice that we have added the normalization factor $\frac{1}{d}$ and use a different set of indices.} next to their closed-boundary operators $\hat{Z}$. Then, using the Euclidean gravitational path integral, they argued that these operators have to commute in the symmetric Fock space
\begin{eqnarray}\label{commuteboundary}
    [\psi^{ab}, \psi^{cd}]=0\ .
\end{eqnarray}
A similar conclusion was reached in earlier work on baby universes in \cite{coleman1988black,giddings1988loss}. 
The commutation relation in (\ref{commuteboundary}) implies that these operators can be simultaneously diagonalized\footnote{In their analysis, the commutation relation plays an important role in the ensemble interpretation of gravity and the definition of the so-called {\it $\alpha$-states} \cite{coleman1988black,MM}.}.

The modern approach to summing over topologies in the gravitational path-integral is called the {\it universe field theory}, which draws an analogy with the worldsheet string theory. However, the commutativity assumption in (\ref{commuteboundary}) is in tension with this modern picture. From the perspective of the universe field theory, the operators that create boundaries with different boundary conditions need not commute\footnote{For concrete examples, see the discussion in Section 8 of \cite{giddings1989baby} or more recently in \cite{anous2020density}. For more recent attempts to resolve this tension by defining a noncommutative algebra of observables, see \cite{casali2022baby}.}. 

Later studies of the EOW branes in simplified models of two-dimensional gravity, such as simplified JT gravity in \cite{Penington:2019kki} and more detailed work in \cite{jafferis2023jackiw,jafferis2023matrix}, have highlighted the importance of the ETH and chaos in gravity. In another line of development, the MM model with the EOW branes was generalized to arbitrary 2D open-closed TQFTs in \cite{gardiner20212D,banerjee2022comments}. We discuss these models in the next subsection.

Since EOW branes are bulk dynamical objects, Marolf and Maxfield treat them as indistinguishable, and sum over all two-dimensional topologies with the same measure as before, now with the boundary conditions $\widehat{(\psi_{j_\alpha}, \psi_{i_\alpha})}$ \footnote{Here we conform to Marolf-Maxfield's notation.} on the $\alpha$-th open boundary. Summing all possible configurations, the following generating function is found (c.f. Equation 3.34 \cite{MM} and Equation \ref{equation:MMkey}):
\begin{equation}
    \langle\exp\lb uZ + \sum_{1\leq i,j\leq k}t_{ij}\widehat{(\psi_j,\psi_i)}\rb\rangle = \exp(\lambda\frac{e^u}{\det(1 - t)})
\end{equation}
The generating function is calculated in the Hartle-Hawking state.

To model this construction using Poissonization, we use the central simple operators \footnote{For a detailed explanation of where this terminology and definition come from, please refer to Appendix \ref{app:openclosed}.} of the form (c.f. Subsection \ref{subsection:centralsimple}) $\overline{\mathcal{O}^{(\psi)}} = \frac{1}{d^2}\sum_{i,ab}\mathcal{O}^{ab}\psi^a_i\bar{\psi}^b_i\mathbbm{1}$. These operators are central because they are all proportional to the identity matrix. In particular, when $d = 1$, we have: $\overline{\mathcal{O}^{(\psi)}} = \sum_{i,ab}\mathcal{O}^{ab}\psi^a_i\bar{\psi}^b_i$ where we have dropped the notation for identity matrix because for $d = 1$ the identity matrix is simply $1$. Notice that for $d = 1$, we can rewrite our definition of $\overline{\mathcal{O}^{(\psi)}}$ in terms of Marolf-Maxfield's original notation \cite{MM}:
\begin{equation}
    \overline{\mathcal{O}^{(\psi)}} = \sum_{i,ab}\mathcal{O}^{ab}\psi^a_i\bar{\psi}^b_i = \sum_{ab}\mathcal{O}^{ab}\psi^{ab} = \sum_{ab}\mathcal{O}^{ab}\widehat{(\psi_b, \psi_a)}
\end{equation}
This is exactly the boundary operator considered in the generating function (c.f. Equation 3.34 \cite{MM}). In terms of Poissonization, the same generating function is written as \footnote{Here we used slightly different notation for various constants in the generating function (e.g. $\lambda\mapsto\mu$). This is to avoid notational confusion in this work.}:
\begin{equation}
    e^\mu\varphi_{\mu\mathbb{E}}\lb e^{u\lambda(1) + \lambda(\overline{\mathcal{O}^{(\psi)}})} \rb = \exp(\mu\frac{e^u}{\det(1 - \mathcal{O})})
\end{equation}
For a detailed derivation and discussion, please refer to Subsection \ref{subsection:MMEOWdetail}.
\subsection{Summing over bordisms in open/closed 2D TQFT}\label{subsection:openclosed}

As the second example, following \cite{gardiner20212D,banerjee2022comments}, we consider a general open/closed 2D TQFT. Compared with closed 2D TQFT, the key new ingredient of an open/closed TQFT is the so-called open-to-closed bordism \cite{banerjee2022comments, Lauda-Pfeiffer}. Using this new ingredient, the n-point correlation function, which corresponds to a bordism that maps $n$ open boundaries to $\mathbb{C}$, can be written as \cite{banerjee2022comments}:
\begin{equation}
    \overline{\mathcal{A}}^\text{open}(n,0):\mathcal{O}_1\cdots\mathcal{O}_n\mapsto\sum_{\pi\in\mathcal{S}_n}B_{|\pi|}(\mu)e^\mu\prod_{\gamma\in\pi}\mu^{-1}\tr(\overrightarrow{\prod}_{k\in\gamma}\mathcal{O}_k)
\end{equation}
Due to the presence of the Touchard polynomial, it is almost immediately clear that one can model a general open/closed 2D TQFT using Poissonization. This is carried out in detail in Subsection \ref{subsection:openclosedappendix}. 
\subsection{Simplified JT gravity with EOW branes}\label{subsec:JTEOW}

Finally, we consider the example of simplified JT gravity with the EOW branes as in \cite{Penington:2019kki}. 
One can include the EOW branes in JT gravity, by adding to the action $\mu\: l$ where $l$ is the proper length of the brane and $\mu$ is the brane's tension:
\begin{eqnarray}
    I=I_{JT}+\mu \int_{brane} ds\ .
\end{eqnarray}
The closed boundary sector of JT is a random matrix theory, and topology fluctuations (adding genus) are suppressed by $e^{-S_0}\sim 1/d$. In the main text of the original work of \cite{Penington:2019kki}, the authors focused on the planar limit where $d\rightarrow\infty$ (c.f. Section 2 \cite{Penington:2019kki}). As we establish here using Poissonization of $k\times k$ random matrices embedded in $d$-dimensional Hilbert spaces, we can capture:
\begin{enumerate}
    \item The ensemble boundary dual of simplified JT with EOW brane (c.f. Appendix D \cite{Penington:2019kki})
    \item And the planar limit when $d\rightarrow\infty$
\end{enumerate}

Besides the gravitational operator $Z(\beta)$ corresponding to a closed boundary of renormalized circumference $\beta$, JT gravity with EOW brane also admits open boundaries which correspond to $k\times k$ matrix degrees of freedom representing matter fields. The global Hilbert space is $\mathcal{H}_\text{grav}\otimes\mathcal{H}_\text{matter}$ and the total system is in the state (c.f. Equations 2.5, D.2, D.3 \cite{Penington:2019kki}):
\begin{equation}
    \ket{\Psi(\beta)} = \frac{1}{\sqrt{k}}\sum_a\ket{\psi_a(\beta)}\ket{a}
\end{equation}
\begin{equation}
   \ket{\psi_a(\beta)} = \sum_i 2^{1/2 - \mu}e^{-\beta E_i}\Gamma(\mu - \frac{1}{2} + i\sqrt{2E_i})\psi_i^a\ket{i}
\end{equation}
where $\psi_i^a$ is a standard complex Gaussian random variable \footnote{Here the notation suggests no energy degeneracy. However, the exact same calculation can be used to include degeneracy with slightly more complicated notations.}. The partial density matrix $\rho(\beta) := \tr_\text{grav}\ket{\Psi(\beta)}\bra{\Psi(\beta)}$ of the entire system \footnote{Here the gravitational subsystem is traced out.} is given by Equation \ref{equation:JTdensity}. This is a random embedding of $d\times d$ matrices in a $k$-dimensional Hilbert space. Poissonization of such densities/operators is studied in depth in Appendix \ref{app:openclosed}. Applying the general formula of connected correlation function (c.f. Appendix \ref{app:openclosed}), we can derive a formula for the connected correlation between $n$-boundaries $\mathbb{E}\tr\lb\rho(\beta_1)\cdots\rho(\beta_n)\rb$ (c.f. Equation \ref{equation:corrJTEOW}).

Notice that the algebra generated by \textit{all} $d\times d$ matrices randomly embedded in a $k$-dimensional Hilbert space is noncommutative. However, the subalgebra generated by the density matrices $\rho(\beta)$ (where $\beta$ is a free parameter) is commutative. Hence, we can introduce a spectral density to rewrite the connected correlation function. In the limit when $k\rightarrow\infty$, the Poissonization calculation gives a well-defined limit for the connected correlation function. It turns out that the Poissonization answer matches the calculation in \cite{Penington:2019kki} (c.f. Equation 2.32) in the double scaling limit (both $k$ and $d$ go to infinity). For detailed discussion and derivations, we refer the reader to Appendix \ref{subsection:mathJTEOW}.

\section{Discussion}

It is well-known that tunneling effects in quantum mechanics are rare events with a probability rate that is exponentially small $\Gamma\sim e^{-O(S_E/\hbar)}$ in the semi-classical regime. It follows from the Poisson limit theorem that, at late times, $T\sim 1/\Gamma$, the statistics of multi-tunneling events are universally described by a Poisson process\footnote{In tunnel junctions, this Poisson process is called shot noise \cite{blanter2000shot}.}. In this work, first, we argued that since topology fluctuations in gravity are also exponentially suppressed, the contribution of baby universes to exponentially late-time physics of quantum gravity correlators can be universally described by a Poisson process. 

In third quantization or universe field theory, it is natural to allow for the creation and annihilation of all universes, including asymptotic open universes with noncommutative algebras. This suggests that the late-time and coarse-grained statistics of such rare events are universally captured by a noncommutative generalization of the Poisson process. We presented such a framework called Poissonization. We worked out simple 2D models of quantum gravity as examples of Poissonization. We observed that the indistinguishability of the bulk dynamical degrees of freedom is tied to the randomness of simple operators in ETH, quantum chaos, and gravity.

Similarly, in a quantum chaotic system, by the ETH, the amplitude for a simple operator to change the macroscopic energy of the system is also exponentially small and random. We argued that by the Poisson limit theorem, the contribution of these rare processes to powers of the simple operators at late time is captured by Poissonization of Gaussian random operators. 

Coleman had argued that, to avoid causality problems, the algebra of operators that create and annihilate baby universes $a^\dagger(J)/a(J)$ has to be commutative. Later, Marolf and Maxfield used the Euclidean path-integral methods to reach the same conclusion \cite{MM}. In simple models such as 2D Topological QFTs (TQFT) and pure JT gravity, this algebra is indeed commutative. This is in tension with the universe field theory paradigm, where the algebra of creation/annihilation operators is that of a quantum field theory, and hence noncommutative. The explicit construction of the universe field theory in simple  String theory motivated models and JT gravity resulted in a noncommutative algebra \cite{giddings1988axion,post2022universe}. See \cite{anous2020density,casali2022baby} for attempts to reconcile the two pictures. 
More generally, once we allow creation/annihilation operators for asymptotic open boundaries as well. It is no longer reasonable to expect a commutative algebra. Our operator algebraic approach to third quantization (Poissonization) provides a resolution to the puzzle above.

As we saw in this work, for the Poisson limit theorem, one needs to wait for times of order $e^{2S}$, which is significantly longer than the timescale of the onset of the plateau in the spectral form factor of quantum chaotic systems and Gaussian random matrices. We attribute this to the fact that our simplified model of the Poisson process does not incorporate the finite dimensionality of the total Hilbert space and the fact that we have matrix degrees of freedom. 

One possibility to go beyond our simplified description is as follows: Instead of the order limits ($e^S\to \infty$ first and then $T\to\infty$ we took in Section \ref{sec:chaoticsystems} that led to independent noncommutative events, we can separate events by the scrambling time or Thouless time. In this limit, successive events are freely independent \cite{chandrasekaran2023large}. The analog of the Poisson limit theorem in this case is a free Poisson limit theorem with an asymptotic universal free Poisson distribution, otherwise known as the Pastur-Marchenko distribution. It is known that free Poisson and free Gaussian distributions (GUE ensemble) are related by a change of variables.  We postpone an exploration of free-Poissonization and more generally $q$-deformed Poissonization and its connections to double-scaled SYK model and JT gravity to future work \cite{berkooz2019towards,lin2022bulk,lin2023symmetry}.

In this work, we argued that our universal Poisson process, in essence, captures the count of the number of ways to partition $p$ into $k$ connected universes. From the point of view of quantum chaos, these are non-perturbative corrections that capture deviation from random matrix theory, and non-perturbative corrections to the spectral density. More explicitly, in the universe field theory of JT (Kodaira-Spencer theory), there are two types of branes: non-compact ones and compact ones \cite{post2022universe}. The random matrix dual of JT arises from integrating out the non-compact branes. The physics of non-compact branes results in non-perturbative corrections that are not captured in the random matrix model \cite{saad2019jt,post2022universe,altland2023quantum}. It is interesting to explore potential connections between our universal Poisson distribution and the semi-classical theory of periodic orbits, which have been argued to play a role in the physics of the plateau\footnote{For instance, see the recent work \cite{saad2024convergent}.}.

Our discussion of Poisson distributions resonates with the recent work in \cite{de2024principle} that applies the maximum ignorance principle to semiclassical gravity. In that work, the authors argue that wormhole contributions capture the ignorance of the microstate of quantum gravity, given a fixed semiclassical bulk description. It is worth noting that given a discrete random variable that takes values on positive integers (e.g., number of events) with a fixed one-point expectation value, a Poisson distribution maximizes the entropy (ignorance)\footnote{Recall that Gaussians are the maximum entropy distributions for a real random variable for fixed one-point and two-point expectation values. Note that if, in addition to a one-point function, for a discrete random variable, we also fix the second moment, the maximum entropy distribution is an exponential distribution.}. We postpone the exploration of the connections between Poissonization, wormhole physics, and the maximum uncertainty principle to upcoming work.

Finally, we would like to point out that in the probe limit approach to third quantization (one parent universe and integrating out the Fock space of baby universes), the mathematics of Poissonization resembles the construction of continuous matrix product states in \cite{ganahl2017continuous,haegeman2013calculus}. We postpone the exploration of this connection to future work.

We saw in Section \ref{subsec:coherentbipartite} that the output algebra of Poissonization always has a center generated by functions of the number operator $N$. As can be seen from the commutation relation $[\lambda(\mathcal{O}),N]=\lambda([\mathcal{O},1])=0$.
In a holographic context, it is desirable to take the algebra of a single boundary to be infinite-dimensional $I_\infty$ and the input weight to be the tracial weight. The construction of Section \ref{subsec:coherentbipartite} using creation/annihilation operators no longer applies because the trace is not finite, and one has to use the more general definition of Poissonization we discussed in Section \ref{app:math}. Quite surprisingly, in this case, the output algebra of Poissonization is a factor! It would be interesting to explore the implications of this fact for the ensemble interpretation of gravity.

We would like to end by pointing out a key question that needs further investigation. We saw that, in the framework of third quantization, the rare event is a topology fluctuation. In our discussion of ETH, it is not clear to the authors what the rare event is that is responsible for the late-time Poisson behavior of the generalized spectral form-factor.

As a guide to the appendices, Appendices \ref{app:closed}, \ref{app:openclosed}, \ref{app:JT}, and \ref{app:plateau} are physics appendices complementing and extending various discussions in the main text. Appendices \ref{app:setpartitions}, \ref{app:universalityPoisson}, and \ref{app:math} are math appendices laying the necessary foundations for Poisson statistics and noncommutative Poissonization.
    
\section*{Acknolwedgements}

NL would like to thank Scott Collier, Shawn Cui, Luca Iliesiu, Juan Maldacena, Herny Maxfield, Gregory Moore, Moshe Rozali, Julian Sonner, Jeremy van der Heijden, Eric Verlinde and Qi Zhou for discussions. YC would like to thank Roy Araiza for the discussions. NL's work is supported by the DOE grants DE-SC0007884 and DE-SC0025547. MJ's work is supported the NSF grant DMS-2247114.

\appendix
\section{List of Symbols}\label{app:symbols}
\begin{longtable}[H]{c|c}
        $\mathcal{N}$ or $\mathcal{M}$ & a von Neumann algebra of single-particle/single-boundary observables  \\
        $\omega$ & a normal faithful (semi-)finite weight on the algebra $\mathcal{N}$ \\
        $\sigma_t^\omega$ & the modular automorphism of $\omega$
        \\
        $J_\omega$ & the modular conjugation of $\omega$
        \\
        $\mathbb{P}_\omega \mathcal{N}$ & a Poisson algebra generated by the pair of input date $(\mathcal{N},\omega)$ \\
        $\varphi_\omega$ & a Poisson state associated with the input weight $\omega$ \\
        $\sigma_t^{\varphi_\omega}$ & the modular automorphism  of the Poisson state $\varphi_\omega$\\
        $J_{\varphi_\omega}$ & the modular conjugation of the Poisson state $\varphi_\omega$
        \\
        $\mathcal{H}$ & a Hilbert space \\
        $\overline{\mathcal{H}}$ & the conjugate Hilbert space
        \\
        $M_n(\mathbb{C})$ & the algebra of $n\times n$-matrices
        \\
        $\mathbb{B}(\mathcal{H})$ & an algebra of bounded operators on the Hilbert space $\mathcal{H}$ \\
        $\tr$ & unnormalized trace on $\mathbb{B}(\mathcal{H})$
        \\
        $\tr_n$ & unnormalized trace on $M_n(\mathbb{C})$
        \\
        $\mathbb{P}_{\tr}\mathbb{B}(\mathcal{H})$ & the Poisson algebra generated by $(\mathbb{B}(\mathcal{H}), \tr)$
        \\
        $\mathcal{R}$ & the hyperfinite type $II_1$ factor
        \\
        $\tau$ & the canonical trace on $\mathcal{R}$
        \\
        $\lambda$ & the Poisson quantization map
        \\
        $\lambda(\mathcal{O})$ & a Poisson-quantized operator
        \\
        $\lambda_N(\mathcal{O})$ & an $N$-th degree Poisson-quantized operator
        \\
        $\Gamma(x)$ & an integrated Poisson-quantized operators
        \\
        $\Gamma(e^{ix})$ & an integrated Poisson-quantized operators associated with unitary $e^{ix}$
        \\
        $L_2(\mathcal{N},\omega)$ or $\mathcal{H}_\omega$ & standard Hilbert space of $\mathcal{N}$ under the weight $\omega$
        \\
        $L_2(\mathbb{P}_\omega \mathcal{N}, \varphi_\omega)$ or $\mathcal{H}_{\varphi_\omega}$ & standard Hilbert space of the Poisson algebra $\mathbb{P}_\omega\mathcal{N}$ under the state $\varphi_\omega$
        \\
        $\mathcal{F}_\text{sym}(\mathcal{H})$ & the symmetric Fock space of a single-particle Hilbert space $\mathcal{H}$
        \\
        $\mathcal{F}_\text{sym}(L_2(\mathcal{N},\omega))$ & the symmetric Fock space of the standard Hilbert space $L_2(\mathcal{N},\omega)$
        \\
        $\mathcal{N}^{\otimes_\text{sym} N}$ & $N$-fold symmetric tensor product of $\mathcal{N}$
        \\
        $\mathcal{T}_\text{sym}(\mathcal{N})$ & the tensor algebra of $\mathcal{N}$
        \\
        $\omega^{\otimes N}$ & $N$-fold tensor product weight of $\omega$
        \\
        $\mathbb{N}$ & the set of natural numbers
        \\
        $\ell_2(\mathbb{N})$ & the Hilbert space of square-integrable sequences
        \\
        $[n]$ & the set of natural numbers $\{1,2,\cdots,n\}$
        \\
        $\mathcal{P}_n$ & the set of partitions on $[n]$
        \\
        $\mathcal{P}_{\text{pair}, 2k}$ & the set of pair partitions on $[2k]$
        \\
        $\sigma$ & a partition of some discrete set
        \\
        $S(n,m)$ & the Stirling number of second kind
        \\
        $B_n$ & the Bell number
        \\
        $\overrightarrow{\prod}$ & ordered product
        \\
        $\mathcal{S}_n$ & the permutation group on $[n]$
        \\
        $\mathcal{S}_A$ & the permutation group on a finite set $A$
        \\
        $\pi$ & a permutation 
        \\
        $\pi_A$ & a permutation on a finite set $A$
        \\
        $\gamma$ & an irreducible cycle in a permutation
        \\
        $\gamma_A$ & an irreducible cycle in a permutation on a finite set $A$
        \\
        $\vec{i}:[n]\rightarrow[m]$ & a set map from $[n]$ to $[m]$
        \\
        $\mathcal{A}$ & a unital $*$-associative algebra
        \\
        $\phi$ & a general linear functional
        \\
        $V$ & a general vector space
        \\
        $\ket{W(\mu^{1/2})}$ & coherent state with parameter $\mu$
        \\
        $\ket{W(\omega^{1/2})}$ &  coherent state associated to the canonical purification of the density matrix $\ket{\omega^{1/2}}$
        \\
        $\psi_i^a$ & standard complex Gaussian random variable
        \\
        $\Psi_d$ & quantum channel that creates $d\times d$ simple operators from finite matrices
        \\
        $\Psi_d^\dagger$ & dual channel (unital CP map) of $\Psi_d$
        \\
        $1\leq i,j,\cdots\leq d$ & indices for the gravitational sector
        \\
        $1\leq a,b,\cdots\leq k$ & indices for the matter sector
        \\
        $\langle\cdot\rangle_\text{conn}$ & connected correlation function
        \\
        $\langle\cdot\rangle$ & correlation function
        
\end{longtable}

\section{Closed 2D TQFTs and Poissonization - Mathematical Theory}\label{app:closed}

This is the first half of a two-part Appendix. The purpose of these two Appendices \ref{app:closed} and \ref{app:openclosed} is to provide a unified and rigorous connection between Poissonization and 2D TQFT, and how summing over 2D topologies is naturally included in the Poisson moment formula. In particular, we provide more mathematical details for the examples discussed in the main text in Section \ref{sec:babyuniversescommutative} and \ref{sec:babyuniversesnonAbel}.

In this section, we focus on summation over topologies and baby universes in closed 2D TQFTs. To prepare for the discussion of open/closed TQFTs in the next Appendix, here, we discuss a general relation between Frobenius algebras and von Neumann algebras.

Mathematically, a 2D TQFT (without summation over bordisms) is a monoidal functor from an appropriate category of 2D cobordisms to the linear category of complex vector spaces \cite{tqftreview, @Moore-Segal, abrams1996two, frobeniusTQFT2, frobeniusTQFT3}. In addition, it can be equivalently characterized as an appropriate Frobenius algebra. For example, it is well-known that any closed 2D TQFT is a symmetric monoidal functor from the category of oriented 2D cobordisms to the category of complex vector spaces, and such a theory is equivalently described by a commutative Frobenius algebra \cite{frobeniusTQFT2}. Note that by definition, complex Frobenius algebras are finite-dimensional. 

We propose Poissonization as a natural and general framework to study summation over bordisms in 2D TQFTs. First, we observe that in finite dimensions, the input data to Poissonization - namely, a von Neumann algebra $N$ and a normal faithful (finite) weight $\omega$\footnote{In this section, we only consider the case where $\omega$ is finite.} - can naturally be identified with a Frobenius algebra. From the TQFT's perspective, one can understand Poissonization as a canonical procedure that takes a Frobenius algebra (or, equivalently, a 2D TQFT) as input and canonically adds in the effect of summation over all bordisms \footnote{In mathematical terminology, the term ``canonical procedure" by the term ``functor", meaning that the map is structure preserving.}. One can understand the output algebra as a gravitational 2D TQFT, where the word ``gravitational" solely means summing over all bulk topologies (bordisms). To make the discussion less abstract, we discuss the concrete examples of closed 2D TQFTs, including Marolf-Maxfield's 2D topological gravity with baby universes and 2D Dijkgraaf-Witten theory with finite gauge groups. 

In this section, we will adopt notations familiar to the TQFT literature.

\subsection{Von Neumann algebra with weight as Frobenius algebra}

We start by reviewing the relevant definitions of Frobenius algebras \cite{abrams1996two, frobeniusTQFT2, frobeniusTQFT3}, emphasizing that a finite-dimensional von Neumann algebra with a normal faithful finite weight is a Frobenius algebra \footnote{In general, one can define a Frobenius algebra object in any monoidal (i.e., tensor) category. However, for us, we only consider Frobenius algebras in the monoidal category of complex vector spaces.}. 
\begin{definition}\label{def1Frobeinus}
    A {\it complex Frobenius algebra} is a unital associative algebra $A$ over $\mathbb{C}$ along with nondegenerate bilinear form $\sigma$ such that:
    \begin{enumerate}
        \item $\sigma(xy, z) = \sigma(x, yz)$
        \item $\varphi:A\rightarrow A^\dagger:x\mapsto \varphi_x(\cdot):=\sigma(x\cdot)$ is an isomorphism
    \end{enumerate}

    A {\it symmetric Frobenius algebra} is a Frobenius algebra $(A,\sigma)$ such that the ``tracial" property is satisfied $\sigma(xy) = \sigma(yx)$. A {\it commutative Frobenius algebra} is a Frobenius algebra where $xy = yx$.
\end{definition}
Using the isomorphism between $A$ and the dual $A^\dagger$, one can define a coalgebra homomorphism:
\begin{equation}
    \Delta: A\rightarrow A\otimes A
\end{equation}
and a co-unit (i.e. a linear functional):
\begin{equation}
    \epsilon:A\rightarrow \mathbb{C}
\end{equation}
In fact, an equivalent definition of a complex Frobenius algebra is the following:
\begin{definition}\label{def2Frobenius}
    A complex Frobenius algebra is a unital associative algebra $A$ that is also a coalgebra. In addition, the comultiplication and multiplication satisfy the following compatibility condition:
    \begin{equation}
        \Delta(xy) = m\otimes id_A(x\otimes\Delta(y)) = id_A\otimes m(\Delta(x)\otimes y)
    \end{equation}
    where $m:A\otimes A\rightarrow A:x\otimes y\mapsto xy$ is the multiplication and $\Delta:A\rightarrow A\otimes A$ is the comultiplication.
\end{definition}
Both definitions imply that $A$ is finite-dimensional. Then, it follows from the Artin-Wedderburn theorem that $A$ is a direct sum of matrix algebra $M_d(\mathbb{C})$ ($d\geq1$) \cite{lang02}. The counit $\epsilon$ defines a faithful linear functional on the matrix algebra. In the language of von Neumann algebras, $A$ is a type $I$ von Neumann algebra. In finite dimensions, the faithful linear functional $\epsilon$ is normal and the nondegenerate bilinear form is simply given by:
\begin{equation}
    \sigma(x,y) := \epsilon(xy)
\end{equation}

The equivalence of the two definitions relies on the fact that in finite dimensions, the existence of a nondegenerate bilinear form on a vector space $V$ implies that $V$ is isomorphic to its dual. This is not true in infinite dimensions. For example, if a Banach space is isomorphic to its dual, then the Banach space must be finite-dimensional. Here we choose to adopt the following generalization of Definition \ref{def1Frobeinus}:
\begin{definition}
    A complex Frobenius algebra (potentially infinite dimensional) is a unital associative algebra $N$ with a nondegenerate bilinear form $\sigma$ such that:
    \begin{equation}
        \sigma(x,yz) = \sigma(xy, z)
    \end{equation}
\end{definition}
This definition is almost identical to the original definition, but we omit the requirement that $N$ and $N^\dagger$ are isomorphic. Using this extended definition, any von Neumann algebra with a fixed normal faithful finite weight is a Frobenius algebra.

\begin{claim}
    A von Neumann algebra $N$ with a fixed normal faithful finite weight $\omega$ is a Frobenius algebra (in this extended sense) where the nondegenerate bilinear form is given by:
    \begin{equation}
        \sigma(x,y) := \omega(xy)
    \end{equation}
    It is clear that $\sigma(xy,z) = \omega(xyz) = \sigma(x,yz)$. In addition, this Frobenius algebra is symmetric if and only if $\omega$ is tracial:
    \begin{equation}
        \omega(xy) = \sigma(x,y) = \sigma(y,x) = \omega(yx)
    \end{equation}
\end{claim}
It is well-known that any closed 2D TQFT is equivalent to a commutative Frobenius algebra \cite{abrams1996two, frobeniusTQFT2}. The generators of (the category of) closed 2D cobordims are given by Figure \ref{figure:closedtqft}. For the relations these generators must satisfy, please refer to \cite{abrams1996two, frobeniusTQFT2}.

Now we restrict our attention on a particular type of Frobenius algebra. It is known that an open/closed 2D TQFT is equivalent to a complex knowledgable Frobenius algebra \cite{frobeniusTQFT3}.
\begin{definition}
    A {\it complex knowledgable Frobenius algebra} is a tuple of the following data:
    \begin{enumerate}
        \item A symmetric complex Frobenius algebra $(A, \omega)$ - This algebra contains both the closed and the open sectors of the TQFT. For the cobordism generators of the open sector, please refer to Figure \ref{figure:opentaft}. For the relations these generators must satisfy, please refer to \cite{frobeniusTQFT3}
        \item A commutative complex Frobenius algebra $(C, \omega_C)$ - This is the closed sector of the knowledgable Frobenius algebra
        \item An algebra homomorphism (c.f. Figure \ref{figure:openclosed}):
        \begin{equation}
            \iota: C\rightarrow A
        \end{equation}
        \item A complex linear map (c.f. Figure \ref{figure:openclosed}):
        \begin{equation}
            \iota^\dagger: A\rightarrow C
        \end{equation}
    \end{enumerate}
    such that the following three conditions hold:
    \begin{enumerate}
        \item The algebra homomorphism $\iota$ maps $C$ to a subalgebra of the center of $A$:
        \begin{equation}
            \iota(C)\subset Z(A)
        \end{equation}
        \item The complex linear map $\iota^\dagger$ is dual to $\iota$ in the sense that:
        \begin{equation}
            \omega(\iota(x)y) = \omega_C(x\iota^\dagger(y))
        \end{equation}
        \item The so-called Cardy condition holds:
        \begin{equation}
            m\circ\tau_{A,A}\circ\Delta(x) = \iota\circ\iota^\dagger(x)
        \end{equation}
        where $\tau_{A,A}:A\otimes A\rightarrow A\otimes A:x\otimes y\mapsto y\otimes x$ is the twist map, $m$ is the multiplication map of the Frobenius algebra $A$, and $\Delta$ is the comultiplication map of $A$.
    \end{enumerate}
\end{definition}
\begin{figure}
    \centering
    \includegraphics[width=0.2\linewidth]{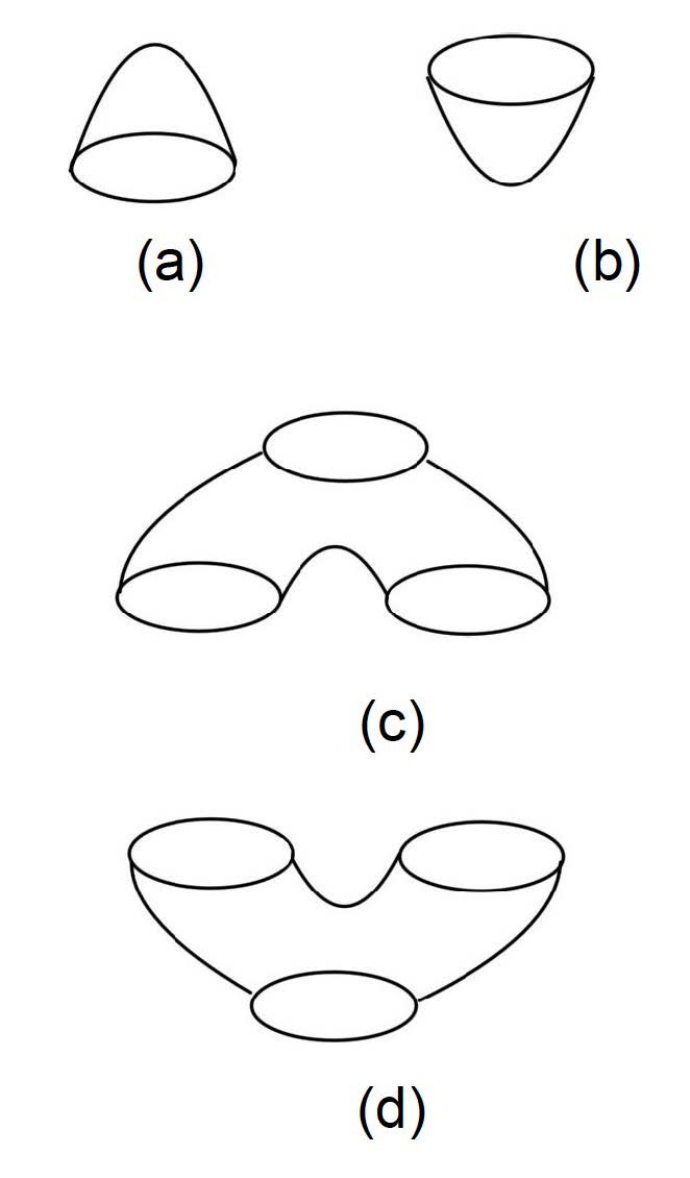}
    \caption{\small{(a) This cobordism connects the empty set (at the top of the figure) to a closed circle. This represents the unit in the commutative Frobenius algebra $C$. (b) This cobordism connects a closed circle to the empty set. This represents the counit (i.e. a nondegenerate linear functional) of the commutative Frobenius algebra $C$. (c) This cobordism connects one closed circle to a disjoint union of two closed circles. This represents the comultiplication of the commutative Frobenius algebra $C$. (d) This cobordism connects a disjoint union of two closed circles to one closed circle. This represents the multiplication of the commutative Frobenius algebra $C$.}}
    \label{figure:closedtqft}
\end{figure}
\begin{figure}
    \centering
    \includegraphics[width=0.2\linewidth]{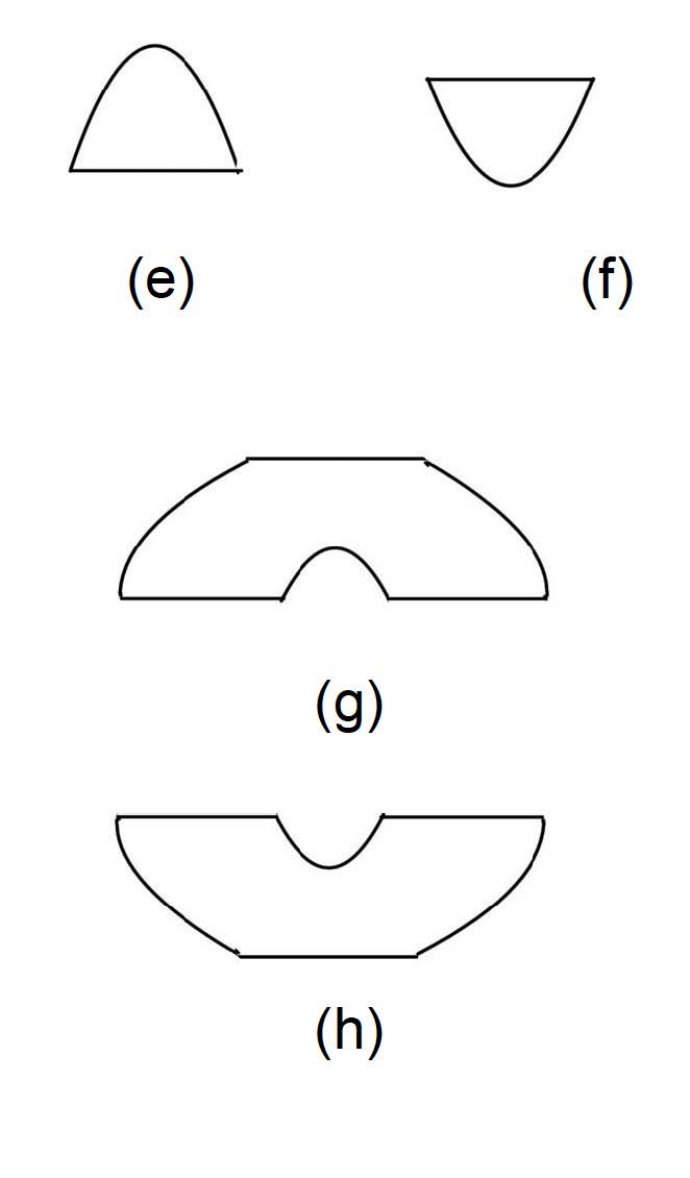}
    \caption{\small{(e) This cobordism connects the empty set (at the top of the figure) to a closed interval. This represents the unit in the symmetric (noncommutative) Frobenius algebra $A$. (f) This cobordism connects a closed interval to the empty set. This represents the counit (i.e. a nondegenerate linear functional) of the symmetric Frobenius algebra $A$. (g) This cobordism connects one closed interval to a disjoint union of two closed intervals. This represents the comultiplication of the symmetric Frobenius algebra $A$. (h) This cobordism connects a disjoint union of two closed intervals to one closed interval. This represents the multiplication of the symmetric Frobenius algebra $A$.}}
    \label{figure:opentaft}
\end{figure}
\begin{figure}
    \centering
    \includegraphics[width=0.2\linewidth]{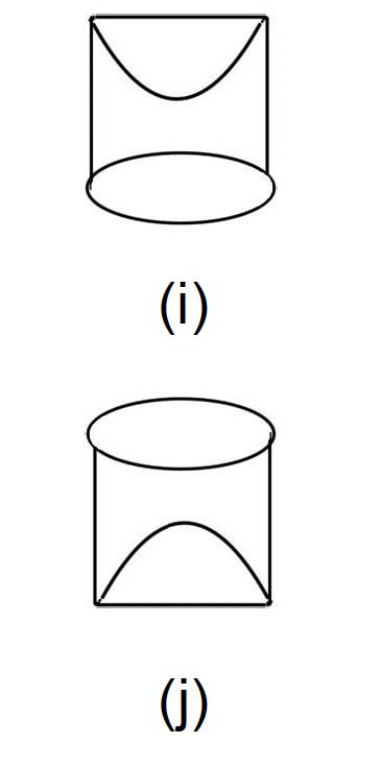}
    \caption{\small{(i) This cobordism connects a closed interval (at the top of the figure) to a closed circle. This is the so-called open-to-closed map. Algebraically, this represents the algebra homomorphism $\iota:C\rightarrow A$. (j) This cobordism connects a closed circle to a closed interval. This is the so-called closed-to-open map. This map is dual to the open-to-closed map. Algebraically, this represents the dual linear map $\iota^\dagger$.}}
    \label{figure:openclosed}
\end{figure}
To see that the definition is not vacuous, we recall that a matrix factor $M_d(\mathbb{C})$ can be used to construct a knowledgeable Frobenius algebra.

\begin{claim}
    A matrix factor $M_d(\mathbb{C})$ along with a finite weight $\tr$ is a symmetric Frobenius algebra. The nondegenerate bilinear form is given by $\tr(xy)$. The comultiplication is defined by:
    \begin{equation}
        \Delta(x) := \sum_{ijk}x_{ij}e_{ik}\otimes e_{kj}
    \end{equation}
    where $x = \sum_{ij}x_{ij}e_{ij}$ and $e_{ij}$ is the matrix unit in $M_d(\mathbb{C})$. The normalization is chosen so that the following duality holds:
    \begin{equation}
        \tr(x^\dagger yz) = \tr\otimes\tr(\Delta(x)^\dagger(y\otimes z))
    \end{equation}
    This follows from a simple calculation:
    \begin{align}
        \begin{split}
            \tr\otimes\tr(\Delta(x)^\dagger(y\otimes z)) &= \sum_{ijk}\bar{x}_{ij}\tr\otimes\tr((e^\dagger_{ik}\otimes e^\dagger_{kj})(y\otimes z)) = \sum_{ijk}\bar{x}_{ij}\tr(e_{ki}y)\tr(e_{jk}z)\\&=\sum_{ijk}\bar{x}_{ij}y_{ik}z_{kj} = \tr(x^\dagger yz)
        \end{split}
    \end{align}
    This duality relies on the conjugation. In general, a Frobenius algebra is not required to have a well-defined notion of conjugation. The compatibility conditions between multiplication and comultiplication can be checked explicitly as well:
    \begin{align}
        \begin{split}
            m\otimes id(x\otimes \Delta(y)) &= \sum_{ijklm}x_{lm}y_{ij}m\otimes id(e_{lm}\otimes e_{ik}\otimes e_{kj})
            =\sum_{ijklm}x_{lm}y_{ij}(e_{lm}e_{ik})\otimes e_{kj}
            \\
            &=\sum_{ijkl}x_{li}y_{ij}e_{lk}\otimes e_{kj} = \Delta(xy)
        \end{split}
    \end{align}
    \begin{align}
        \begin{split}
            id\otimes m(\Delta(x)\otimes y) &= \sum_{ijklm}x_{ij}y_{lm}id\otimes m(e_{ik}\otimes e_{kj}\otimes e_{lm})
            =\sum_{ijklm}x_{ij}y_{lm}e_{ik}\otimes(e_{kj}e_{lm}
            )
            \\
            &=\sum_{ijkm}x_{ij}y_{jm}e_{ik}\otimes e_{km} = \Delta(xy)
        \end{split}
    \end{align}
    Because the weight is tracial, this Frobenius algebra is symmetric.

    The center of $M_d(\mathbb{C})$ is $\mathbb{C}$. Trivially, it is a commutative Frobenius algebra where the nondegenerate bilinear form is simply the multiplication. There exists an algebraic embedding:
    \begin{equation}
        \iota:\mathbb{C}\rightarrow M_d(\mathbb{C}):z\mapsto z \mathbbm{1}
    \end{equation}
    and its dual normal conditional expectation (by Takesaki's theorem or explicitly by the tracial weight):
    \begin{equation}
        \iota^\dagger: M_d(\mathbb{C})\rightarrow \mathbb{C}: x\mapsto \tr(x)
    \end{equation}
    The compatibility condition between $\iota$ and $\iota^\dagger$ is simply the definition for conditional expectation:
    \begin{equation}
        \mu\tr(\iota(z)x) = \mu\tr((z\mathbbm{1})x) = \mu (z\tr(x))
    \end{equation}
    Finally, the Cardy condition can be checked by an explicit calculation:
    \begin{align}
        \begin{split}
            m\circ\tau\circ\Delta(x)&=\sum_{ijk}x_{ij}m\circ\tau(e_{ik}\otimes e_{kj}) = \sum_{ijk}x_{ij}m(e_{kj}\otimes e_{ik}) \\&= \sum_{ijk}x_{ij}\delta_{ij}e_{kk} = \tr(x)\mathbbm{1} = \iota\circ\iota^\dagger(x)
        \end{split}
    \end{align}
    Thus, the matrix factor with the tracial weight together with its center defines a knowledgeable Frobenius algebra. 
\end{claim}

Next, we show how the data of a knowledgeable Frobenius algebra defines an open/closed 2D TQFT. The basic idea is that the symmetric Frobenius algebra corresponds to the algebra of observables on the open boundary intervals. The commutative Frobenius algebra corresponds to the algebra of observables on the closed boundary circles. The map $\iota^\dagger$ is the so-called open-to-closed bordism that ties the two ends of an open boundary interval and forms a closed circle \cite{frobeniusTQFT3}. 

If we focus on the closed sector of an open/closed 2D TQFT, then the data of a knowledgeable Frobenius algebra reduces to the commutative Frobenius algebra. This is the well-known equivalence between closed 2D TQFTs and commutative Frobenius algebras \cite{frobeniusTQFT2}. Using the language of von Neumann algebras, a knowledgeable Frobenius algebra is characterized by the following data:
\begin{enumerate}
    \item A tracial von Neumann algebra $\mathcal{N}$ with a tracial weight $\tau$.
    \item A central subalgebra $\mathcal{C}\subset Z(\mathcal{N})$ with $\tau$ restricted to $\mathcal{C}$.
    \item A normal conditional expectation $\mathcal{E}:\mathcal{N}\rightarrow \mathcal{C}$\footnote{If $\mathcal{C}$ is the center of $\mathcal{N}$, then such a conditional expectation always exists.}.
\end{enumerate}
Crucially, the Cardy condition is not easy to characterize using the language of von Neumann algebra and requires the existence of a comultiplication. This is not an issue in finite dimensions. The extension to infinite dimensions would be beyond the scope of this work.
\subsection{Summing over bordisms in closed 2D TQFTs and Poissonization of commutative von Neumann Algebras}
In this section, we discuss how Poissonization sums over bordisms given the input data of a 2D TQFT. In this subsection and the next section, we demonstrate this by considering two classes of 2D TQFTs: 
\begin{enumerate}
    \item Closed 2D TQFT. Such a TQFT is described by a (finite-dimensional) commutative Frobenius algebra. Equivalently, it is described by a (finite-dimensional) commutative von Neumann algebra with a fixed finite weight
    \item Open/Closed 2D TQFT. Such a TQFT is described by a (finite-dimensional) knowledgeable Frobenius algebra. For simplicity, we consider the case of a single matrix algebra with the canonical tracial weight. As we have seen in the main text, because the boundaries are indistinguishable, we need to modify the input data by coupling each matrix with a random Wishart matrix. The input algebra is effectively the algebra of Gaussian random matrices. This modification ensures that the Poisson moment formula is permutation invariant and correctly models the indistinguishability of the boundaries in an open/closed 2D TQFT.
\end{enumerate}
Using a mutually-orthogonal family of idempotents (i.e., projections), any commutative Frobenius algebra can be decomposed as:
\begin{equation}
    A := \oplus_{1\leq i \leq n} \mathbb{C}\epsilon_i
\end{equation}
where $\epsilon_i\epsilon_j = \delta_{ij}\epsilon_i$. The multiplicative unit is given by:
\begin{equation}
    1 = \sum_i\epsilon_i
\end{equation}The Frobenius bilinear form is given by:
\begin{equation}
    \sigma(\sum_i\alpha_i\epsilon_i, \sum_j\beta_j\epsilon_j) = \sum_i\alpha_i\beta_i\sigma(\epsilon_i,\epsilon_i)
\end{equation}
Equivalently, $A$ is a finite-dimensional commutative von Neumann algebra isomorphic to $\ell_\infty^n$. When $\sigma(\epsilon_i, \epsilon_i) > 0$, the Frobenius form is equivalent to a normal faithful finite weight:
\begin{equation}
    \omega(\sum_i\alpha_i\epsilon_i) := \sigma(1, \sum_j\alpha_j\epsilon_j) = \sigma(\sum_i\epsilon_i, \sum_j\alpha_j\epsilon_j) = \sum_i\alpha_i\mu_i
\end{equation}
where $\mu_i:=\sigma(\epsilon_i,\epsilon_i)$. In addition, the coproduct is given by:
\begin{equation}
    \Delta(\sum_i\alpha_i\epsilon_i) = \sum_i\frac{\alpha_i}{\mu_i}\epsilon_i\otimes\epsilon_i
\end{equation}
This is the input data of a 2D TQFT. To sum over bordisms, we need to perform two summations. For a closed 2D gravitational TQFTs, given a boundary manifold (i.e. a number of closed circles), the topological type of the bordism connecting the boundary is characterized by 
\begin{enumerate}
    \item The number of connected components.
    \item The genus of each connected component.
\end{enumerate}
Hence, we need to sum over genus and sum over connected components. The summation over genus corresponds to a renormalization of the Hartle-Hawking state by a multiplicative modification of the weight $\omega$. Poissonization, in essence, follows from counting the number of ways a fixed number of boundaries can be partitioned into several connected components. Note that the set partitions are intimately related to the Poisson distribution; see Appendix \ref{app:setpartitions}.

\paragraph{Summing over genera}

Using Morse theory, any bordism connecting boundary circles can be decomposed into a normal form. Concretely, suppose there are $n_\text{in}$-many incoming boundary circles and $n_\text{out}$-many outgoing boundary circles, then the normal form of a connected genus-g bordism connecting these boundary circles is given by
\begin{enumerate}
    \item First using pair-of-pants, all $n_\text{in}$-many incoming circles are sequentially merged to a single circle $S^1_\text{in}$
    \item Then connect $S^1_\text{in}$ to another circle $S^1_\text{out}$ by a genus $g$ surface
    \item Finally using reversed pair-of-pants, $S^1_\text{out}$ sequentially splits into $n_\text{out}$-many outgoing circles
\end{enumerate}
Equivalently, instead of morphisms we can consider topological amplitudes from $n$ boundary circles to the empty set. Then the normal form of topological amplitude is given by:
\begin{enumerate}
    \item First using pair-of-pants, all $n$ circles are sequentially merged to a single circle $S^1$
    \item The circle $S^1$ is connected to the empty set via a capped genus-g surface
\end{enumerate}
From the normal form of topological amplitude, the summation over genus is equivalent to a modification of the counit (i.e. the weight). Since each genus is represented by a concatenation of comultiplication and multiplication, we can calculate:
\begin{equation}
    m\circ\Delta(\sum_i\alpha_i\epsilon_i) = \sum_i\frac{\alpha_i}{\mu_i}\epsilon_i^2 = \sum_i\frac{\alpha_i}{\mu_i}\epsilon_i
\end{equation}
where $\mu_i$ is the suppression factor of genus creation. In a physically meaningful theory, this suppression factor is given by $e^{-2S_0^{(i)}}$ (i.e. $\mu_i = e^{2S_0^{(i)}}$) where $S_0^{(i)}$ is the topological action associated with the $i$-th topological sector and $S_0^{(i)} > 0$. Hence, summing up all genera, the co-unit (i.e. the weight) is modified to:
\begin{equation}
    \omega_\text{all-genus}(\sum_i\alpha_i\epsilon_i) = \sigma(\sum_{g\geq 0}(m\circ\Delta)^g(\sum_i\epsilon_i), \sum_i\alpha_i\epsilon_i) = \sum_i\alpha_i\mu_i\sum_{g\geq 0}\mu_i^{-g} = \sum_i\frac{\alpha_ie^{2S_0^{(i)}}}{1 - e^{-2S_0^{(i)}}}
\end{equation}
In terms of the original weight, we have:
\begin{equation}
    \omega_\text{all-genus}(\cdot) = \omega(\sum_i\frac{1}{1 - e^{-2S_0^{(i)}}}\epsilon_i\cdot)
\end{equation}
Hence, the input data to Poissonization is the commutative algebra $A$ and the weight $\omega_\text{all-genus}$.

\paragraph{Summing over connected components}

The summation over connected components requires summing over all possible ways to connect the set of $n$ circles. This is equivalent to partitioning the set of $n$ circles and connecting each subset in the partition. This is exactly what Poissonization achieves.

Concretely, given a commutative algebra with the genus-modified weight $(A,\omega_\text{all-genus})$, the Poisson algebra along with the Poisson state $(\mathbb{P}_{\omega_\text{all-genus}}A, \varphi_{\omega_\text{all-genus}})$ gives a canonical way to calculate topological amplitude that includes bulk contributions (connected and disconnected):
\begin{equation}
    \varphi_{\omega_\text{all-genus}}(\lambda(\sum_i\alpha_i^{(1)}\epsilon_i)\cdots\lambda(\sum_i\alpha_i^{(n)}\epsilon_i)) = \sum_{\sigma\in\mathcal{P}_n}\prod_{A\in\sigma}\omega_\text{all-genus}(\prod_{k\in A}(\sum_i\alpha_i^{(k)}\epsilon_i))
\end{equation}
The innermost amplitude $\omega_\text{all-genus}(\prod_{k\in A}(\sum_i\alpha_i^{(k)}\epsilon_i))$ is itself a summation over all topological types of connected bulk manifolds with $|A|$-many boundaries. Each boundary circle is associated with a Poisson quantized operator $\lambda(\sum_i\alpha_i\epsilon_i)$. The outermost summation sums over all possible ways to connect $n$ boundaries.

\subsubsection*{Marolf-Maxfield closed baby universes and 2D Dijkgraaf-Witten}\label{subsection:MMDW}

In this Subsection, we apply the general theory of Poissonization of commutative von Neumann algebras to two particular examples of closed 2D TQFTs. We emphasize the mathematical structures relevant to these two theories. All discussions related to the underlying physics are in the main text. 

\paragraph{Marolf-Maxfield closed baby universes}

This is the most trivial example of Poissonization and the simplest example of a closed 2D TQFT. As discussed in the main text, the algebra of boundary observables in Marolf-Maxfield closed baby universe theory is the Poisson algebra of complex numbers $\mathbb{P}_{|\mu|^2}\mathbb{C}$. The weight is simply given by multiplication:
\begin{equation}
    \mathbb{C}\rightarrow\mathbb{C}:z\mapsto|\mu|^2z
\end{equation}
As discussed in the previous section, this weight (i.e. the constant $|\mu|^2$) should be understood as including a summation over all genera. In terms of the original notation of Marolf and Maxfield, we have\footnote{In the original work \cite{MM}, this constant is denoted $\lambda > 0$. Coincidentally, this is the notation of the Poisson quantization map. To avoid notational confusion, we opt to denote this constant $|\mu|^2$.}:
\begin{equation}
    |\mu|^2 = Z[0] = \sum_{g\geq 0} e^{S_0(2-2g)} = \frac{e^{2S_0}}{1 - e^{-2S_0}}
\end{equation}
The Poisson state $\ket{W(\mu 1)}$ (i.e. the purifying vector of the Poisson state $\varphi_{|\mu|^2}$) is the coherent state. In the context of the Marolf-Maxfield closed baby universe theory, this is the normalized Hartle-Hawking state:
\begin{equation}
    \ket{W(\mu1)} = \ket{HH}
\end{equation}
Equivalently, we have:
\begin{equation}
    \varphi_{|\mu|^2}(\cdot) = \bra{W(\mu 1)}\cdot\ket{W(\mu 1} = \frac{1}{\mathcal{Z}}\bra{HH}\cdot\ket{HH}
\end{equation}
where $\mathcal{Z} = ||\ket{HH}||^2$ is the norm of the Hartle-Hawking state.

The algebra of boundary observables in the Marolf-Maxfield closed baby universe theory is generated by the ``partition function operator" $\widehat{Z}$. In terms of Poissonization, this is the Poisson quantized operator:
\begin{equation}
    \lambda(1) = \widehat{Z}
 \end{equation}
To check these claims, we calculate the m-point function explicitly. This serves as a direct comparison with the calculation done in \cite{MM}:
\begin{equation}
    \varphi_{|\mu|^2}(\lambda(1)^m) = \sum_{\sigma\in\mathcal{P}_m}\prod_{A\in\sigma}|\mu|^21^{|A|} = \sum_{\sigma\in\mathcal{P}_m}|\mu|^{2|\sigma|} = B_m(|\mu|^2)
\end{equation}
where $B_m(|\mu|^2)$ is the m-th Touchard polynomial. Compared with the original calculation of \cite{MM}, we have an exact match:
\begin{equation}
    \varphi_{|\mu|^2}(\lambda(1)^m) = \bra{W(\mu 1)}\lambda(1)^m\ket{W(\mu 1)} = \frac{1}{\mathcal{Z}}\bra{HH}\widehat{Z}^m\ket{HH}
\end{equation}
\paragraph{2D Dijkgraaf-Witten theory}
Mathematically, 2D Dijkgraaf-Witten theory with finite gauge group $G$ is characterized by the commutative Frobenius algebra of class functions on $G$ \cite{DW}:
\begin{equation}
    \mathbb{C}[G]^G :=\{f:G\rightarrow \mathbb{C}:f(hgh^{-1}) = f(g)\}
\end{equation}
Here, the notation emphasizes that this algebra is the fixed-point algebra of $\mathbb{C}[G]$ under the conjugation of $G$. Since $G$ is finite, there are finitely many conjugacy classes and hence finitely many irreducible representations. Therefore, we have the decomposition
\begin{equation}
    \mathbb{C}[G]^G := \oplus_q \mathbb{C}\chi_q
\end{equation}
where the direct sum is over all irreducible representations $\alpha_q:G\rightarrow End(V_q)$ and $\chi_q(g) := \tr(\alpha_q(g))$ is the character of the representation $q$. To study the summation over bordisms in 2D Dijkgraaf-Witten theory, we consider the following Frobenius algebra structure on $\mathbb{C}[G]^G$:
\begin{enumerate}
    \item The multiplication is simply element-wise multiplication:
    \begin{equation}
        m:\mathbb{C}[G]^G\otimes\mathbb{C}[G]^G\rightarrow\mathbb{C}[G]^G:\sum_q\alpha_q\chi_q\otimes\sum_{q'}\beta_{q'}\chi_{q'} \mapsto \sum_q\alpha_q\beta_q\chi_q
    \end{equation}
    where $d_q \in\mathbb{N}$ is the dimension of the irreducible representation $q$. It is finite because $G$ is a finite group.
    \item The unit is the identity function: $\sum_q\chi_q$
    \item The comultiplication is diagonal:
    \begin{equation}
        \Delta:\mathbb{C}[G]^G\rightarrow\mathbb{C}[G]^G\otimes\mathbb{C}[G]^G:\sum_q\alpha_q\chi_q\rightarrow\sum_q\alpha_q\frac{|G|^2}{d_q^2e^{2S_0}}\chi_q\otimes\chi_q
    \end{equation}
    \item The count is given by:
    \begin{equation}
        \omega(\sum_q\alpha_q\chi_q) = \sum_q\frac{d_q^2e^{2S_0}}{|G|^2}\alpha_q
    \end{equation}
\end{enumerate}
As discussed in the previous section, this Frobenius algebra is naturally a commutative von Neumann algebra $\mathbb{C}[G]^G$ with a weight $\omega$. The effect of including all genera amounts to a modification of the weight $\omega$, and the corresponding all-genus correction to $\omega$ is given by:
\begin{equation}
    \omega_\text{all-genus}(\sum_q\alpha_q\chi_q) = \sum_q\sum_{g\geq 0}\lb\frac{d_q}{|G|}e^{S_0}\rb^{2-2g}\alpha_q
\end{equation}
Our claim is that the algebra of boundary observables in a gravitational 2D Dijkgraaf-Witten theory
is the Poisson algebra:
\begin{equation}
    \mathbb{P}_{\omega_\text{all-genera}}\bigg(\mathbb{C}[G]^G\bigg)\ .
\end{equation}
This can be checked directly by calculating the Poisson moment formula:
\begin{align}
    \begin{split}
        \varphi_{\omega_\text{all-genera}}(\lambda(\frac{|G|}{d_{q_1}e^{S_0}}\chi_{q_1})\cdots\lambda(\frac{|G|}{d_{q_n}e^{S_0}}\chi_{q_n})) &= \sum_{\sigma\in\mathcal{P}_n}\prod_{A\in\sigma}\omega_\text{all-genera}(\prod_{k\in A}\frac{|G|}{d_{q_k}e^{S_0}}\chi_{q_k})
    \end{split}
\end{align}
In particular the connected correlation function is given by (c.f. Equation \ref{equation:diagonalconn}):
\begin{equation}
    \omega_\text{all-genera}\bigg((\frac{|G|}{d_{q_1}e^{S_0}}\chi_{q_1})\cdots(\frac{|G|}{d_{q_n}e^{S_0}}\chi_{q_n})\bigg) = \sum_q\delta_{q_1\cdots Q_n}(\frac{d_q}{|G|}e^{S_0})^{2-2g-n}
\end{equation}
This reproduces the Mednykh formula \cite{mednykh} and the results from \cite{gardiner20212D}.

\section{Open/closed 2D TQFTs, EOW branes, and Poissonization -  Mathematical Theory}\label{app:openclosed}

In the second half of the two-part Appendix, we continue to discuss Poissonization as the proper mathematical framework to study summation over topologies and baby universes. We focus on 2D gravitational theories with end-of-the-world branes and open/closed 2D TQFTs. As opposed to the previous section, the asymptotic boundaries in theories considered in this section are generally allowed to be open intervals. This leads to Poisson algebras that are generally noncommutative. We apply Poissonization to study the following three theories:
\begin{enumerate}
    \item Simplified JT gravity with EOW branes \cite{Penington:2019kki}.
    \item Open/closed 2D TQFTs as defined by Moore and Segal \cite{@Moore-Segal,@Banerjee-Moore}.
    \item Marolf-Maxfield 2D topological gravity with EOW branes \cite{MM}.
\end{enumerate}
The mathematical descriptions of these theories are intimately related. JT gravity with EOW branes can be described by the Poissonization of random matrices. Through a conditional expectation to the diagonal subalgebra, open/closed 2D TQFTs can be described using Poissonization of random diagonal matrices. Finally, through another conditional expectation to the center (i.e., the trace), we obtain the mathematical description of Marolf-Maxfield 2D topological gravity with EOW branes. 


In this Appendix, we start by discussing the general setup, specifying the input data of Poissonization and its physical interpretation. Then, we calculate the moment formula (i.e., the correlation function) of Poisson quantized random matrices. This calculation is central to this Appendix. Then, we discuss the three examples of Section \ref{sec:babyuniversesnonAbel} in more detail. 

\subsection{Functorial framework for single-boundary data}\label{subsection:simplechannel}

When applying Poissonization to model quantum gravitational systems, the general philosophy dictates that the input data to Poissonization should represent physical data of a single asymptotic boundary. The theories considered in the previous Appendix contained only closed boundaries. In 2D theories, the closed boundaries are simply circles. Open boundaries require additional data. This is the case for regular (i.e. non-gravitational) open/closed 2D TQFTs. While a closed 2D TQFT can be fully described by a commutative Frobenius algebra, an open/closed 2D TQFT requires the data of a knowledgeable Frobenius algebra, whose center is a commutative Frobenius subalgebra representing the closed sector of the theory. The general theory of extended 2D TQFTs in \cite{Lauda-Pfeiffer} specifies the boundary data as follows:
\begin{enumerate}
    \item The objects are the endpoints of boundary intervals. Each endpoint is associated with a $d$-dimensional vector space.
    \item The 1-morphisms are the boundary intervals. Each boundary interval is associated with the linear maps (i.e. $d\times d$-matrices) mapping the vector space of one endpoint to another. 
\end{enumerate}
Here we are studying toy models of 2D quantum gravitational systems. 
As we argued in Section \ref{subsec:ETH}, the observables in these topological toy models can only probe very coarse-grained information (time-averages) or simple operators. Specifically in terms of the energy eigenstates, one cannot expect these observables to distinguish states whose energies differ by a negligible amount. Such considerations are derived directly from ETH, and hence suggest the following modification to the boundary data:
\begin{enumerate}
    \item The objects (end points of boundary intervals) are associated with $d$-dimensional random vectors $\psi := \sum_{1\leq i\leq d}\psi_i \ket{i}$, where $\ket{i}$ is the basis vector in $\mathbb{C}^d$ and $\psi_i$ is a standard complex Gaussian random variable on $\mathbb{C}^k$. Expand $\psi_i$ in terms of the basis of $\mathbb{C}^k$, we have $\psi_i := \sum_{1\leq a\leq k}\psi^a_i\ket{a}$ where $\ket{a}$ is the basis vector in $\mathbb{C}^k$ and $\psi^a_i$'s are mutually independent standard complex Gaussian random variables on $\mathbb{C}$.
    \item The 1-morphisms (boundary intervals) are associated with $M_k(\mathbb{C})$. Each 1-morphism is dressed by its two end-points:
    \begin{equation}
        \sum_{1\leq a,b\leq k}\mathcal{O}^{ab}\ket{a}\bra{b} \mapsto \mathcal{O}^{(\psi)}:=\frac{1}{d}\sum_{1\leq a,b\leq k}\mathcal{O}^{ab}\psi^{ab}
    \end{equation}
    where $\psi^{ab} := \sum_{1\leq i,j\leq d}\psi^{ab}_{ij}\ket{i}\bra{j} = \sum_{1\leq i,j\leq d}\psi^a_i\bar{\psi^b_j}\ket{i}\bra{j}$ \footnote{From the perspective of ETH, this map $\mO\mapsto \mO^{(\psi)}$ outputs a $d \times d$-simple operator from a $k \times k$-matrix.}.
\end{enumerate}
The fact that we choose Gaussian random vectors is a modeling choice rather than a derived requirement. This choice is to ensure the Poisson moment formulae match with known calculations of correlation functions.

Note that the map $\mathcal{O}\mapsto \mathcal{O}^{(\psi)}$ sends a $k\times k$-matrix to a $d\times d$-random matrix. From the perspective of ETH, one can understand $\mathcal{O}^{(\psi)}$ as a simple operator, and the coarse-graining map $\mathcal{O}\mapsto \mathcal{O}^{(\psi)}$ outputs simple-operators.
Before we move on to the main calculations, we show that the simple-operator coarse-graining map is a CPTP map by constructing its Stinespring dilation. Mathematically, this map connects the following two spaces:
\begin{equation}
    \Psi_d: L_k^1(\mathbb{C})\rightarrow L^1(\mathbb{C}^k, L_d^1(\mathbb{C})): \mathcal{O}\mapsto \mathcal{O}^{(\psi)}
\end{equation}
where $L_k^1(\mathbb{C})$ is the space of $k\times k$-trace class matrices and $L^1(\mathbb{C}^k, L_d^1(\mathbb{C}))$ is the integrable functions over $\mathbb{C}^k$ with value in $d\times d$-trace class matrices. Notice that we formulate everything in terms of $L^1$-spaces instead of $L^\infty$-spaces. This is because the Gaussian random variables $\psi^a_i$ is $L^1$-integrable on $\mathbb{C}^k$ but not $L^\infty$ \footnote{These random variables are simply coordinate functions on $\mathbb{C}^k$, and hence not bounded.}. The tracial weight on $L^1(\mathbb{C}^k, L_d^1(\mathbb{C}))$ is given by $\mathbb{E}_\text{Gaussian}\tr$ - i.e. the Gaussian expectation of the unnormalized trace. The tracial weight on $L^1_k(\mathbb{C})$ is simply the canonical unnormalized trace.

The map $\Psi_d$ is trace-preserving:
\begin{align}
    \begin{split}
        \mathbb{E}_\text{Gaussian}\tr(\Psi_d(\mathcal{O})) = \frac{1}{d}\sum_{1\leq a,b\leq k}\mathcal{O}^{ab}\sum_{1\leq i,j\leq d}\mathbb{E}(\psi^a_i\bar{\psi}^b_j)\tr(\ket{i}\bra{j}) = \frac{1}{d}\sum_{ab}\mathcal{O}^{ab}\sum_{ij}\delta^{ab}\delta_{ij} = \tr(\mathcal{O})\ .
    \end{split}
\end{align}
Before proving complete positivity, we first observe the dual map of $\Psi_d$ is given by:
\begin{equation}
    \Psi_d^\dagger: L^\infty(\mathbb{C}^k, M_d(\mathbb{C}))\rightarrow M_k(\mathbb{C}): \sum_{1\leq i,j\leq d}g_{ij}\ket{i}\bra{j}\mapsto\frac{1}{d}\sum_{1\leq i,j\leq d}\sum_{1\leq a,b\leq k}\mathbb{E}_\text{Gaussian}(g_{ij}\bar{\psi}^a_i\psi^b_j)\ket{a}\bra{b}
\end{equation}
where $g_{ij} \in L^\infty(\mathbb{C}^k)$ is a bounded function on $\mathbb{C}^k$. The duality can be checked explicitly:
\begin{align}
    \begin{split}
        \mathbb{E}_\text{Gaussian}\tr(\Psi_d(\mathcal{O})^\dagger\sum_{ij}g_{ij}\ket{i}\bra{j}) &= \frac{1}{d}\sum_{ab}\bar{\mathcal{O}}^{ab}\sum_{ijkl}\mathbb{E}_\text{Gaussian}(\bar{\psi}^a_k\psi^b_lg_{ij})\tr(\ket{l}\bra{k}\ket{i}\bra{j})
        \\
        &=\frac{1}{d}\sum_{ab}\bar{\mathcal{O}}^{ab}\sum_{ij}\mathbb{E}_\text{Gaussian}(g_{ij}\bar{\psi}^a_i\psi^b_j)
        \\
        &=\sum_{ab}\bar{\mathcal{O}}^{ab}\tr(\ket{b}\bra{a}\Psi^\dagger_d(\sum_{ij}g_{ij}\ket{i}\bra{j}))
        \\
        &=\tr(\mathcal{O}^\dagger\Psi^\dagger_d(\sum_{ij}g_{ij}\ket{i}\bra{j}))
    \end{split}
\end{align}
In addition, $\Psi^\dagger_d$ is unital:
\begin{equation}
    \Psi^\dagger_d(1) = \frac{1}{d}\sum_{ij,ab}\mathbb{E}_\text{Gaussian}(\delta_{ij}\bar{\psi}^a_i\psi^b_j)\ket{a}\bra{b} = \sum_{ab}\delta^{ba}\ket{a}\bra{b} = \mathbbm{1}
\end{equation}
where $1$ denotes the unit function on $\mathbb{C}^k$ taking constant value of the identity matrix.

Now, we directly construct the Stinespring dilation of $\Psi^\dagger_d$. Consider the isometry:
\begin{equation}
    V_d:\mathbb{C}^k\rightarrow L^2(\mathbb{C}^k, \mathbb{C}^d)\otimes\mathbb{C}^k:\ket{a}\mapsto\frac{1}{\sqrt{d}}\sum_{1\leq i\leq d}{\psi}^a_i\ket{i}\otimes\ket{a}
\end{equation}
where $L^2(\mathbb{C}^k,\mathbb{C}^d)$ is $\mathbb{C}^d$-vector-valued square-integrable functions on $\mathbb{C}^k$. Its inner product is defined using the standard Gaussian measure. The dual map is given by:
\begin{equation}
    V^\dagger_d:L^2(\mathbb{C}^k,\mathbb{C}^d)\otimes\mathbb{C}^k\rightarrow\mathbb{C}^k: g^a_i\ket{i}\otimes\ket{a}\mapsto\frac{1}{\sqrt{d}}\sum_{b}\mathbb{E}_\text{Gaussian}(g^a_i\bar{\psi}^b_i)\ket{b}
\end{equation}
where $g^a_i$ is a square-integrable function on $\mathbb{C}^k$. To check this is an isometry, we have:
\begin{equation}
    V^\dagger_dV_d(\ket{a}) = \frac{1}{\sqrt{d}}V^\dagger_d(\sum_{1\leq i\leq d}{\psi}^a_i\ket{i}\otimes\ket{a}) = \frac{1}{d}\sum_{i,b}\mathbb{E}_\text{Gaussian}({\psi}^a_i\bar{\psi}^b_i)\ket{b} = \ket{a}
\end{equation}
In addition, consider the $*$-homomorphism:
\begin{equation}
    \pi_d: L^\infty(\mathbb{C}^k,M_d(\mathbb{C}))\rightarrow L^\infty(\mathbb{C}^k, M_d(\mathbb{C}))\otimes M_k(\mathbb{C}):\sum_{ij}g_{ij}\ket{i}\bra{j}\mapsto\sum_{ij}g_{ij}\ket{i}\bra{j}\otimes\mathbbm{1}
\end{equation}
where $g_{ij}$ is a bounded function on $\mathbb{C}^k$. Then the dual channel $\Psi^\dagger_d$ can be decomposed as:
\begin{align}
    \begin{split}
        V^\dagger_d\pi_d(\sum_{ij}g_{ij}\ket{i}\bra{j})V_d\ket{a} &= \frac{1}{\sqrt{d}}V^\dagger_d\pi_d(\sum_{ij}g_{ij}\ket{i}\bra{j})(\sum_k{\psi}^a_k\ket{k}\otimes\ket{a}) \\&= \frac{1}{\sqrt{d}}V^\dagger_d(\sum_{ijk}g_{ij}{\psi}^a_k\ket{i}\bra{j}\ket{k}\otimes\ket{a})
        \\&=\frac{1}{d}\sum_{ij,b}\mathbb{E}_\text{Gaussian}(g_{ij}{\psi}^a_j\bar{\psi}^b_i)\ket{b}
        \\&
        = \Psi^\dagger_d(\sum_{ij}g_{ij}\ket{i}\bra{j})\ket{a}
    \end{split}
\end{align}
Therefore, $\Psi^\dagger_d$ and consequently its dual $\Psi_d$ are both completely positive.

Regarding physics applications, we occasionally would require embedding $d\times d$-matrices to $k\times k$-random matrices \footnote{This is the case for applications to simplified JT gravity with EOW branes \cite{Penington:2019kki}.}. The mathematical construction above can be copied verbatim to construct a CPTP map:
\begin{equation}
    \Psi_k: L_d^1\rightarrow L^1(\mathbb{C}^d, L_k^1(\mathbb{C})):\mathcal{O}\mapsto\mathcal{O}_{(\psi)}:=\frac{1}{k}\sum_{ij,ab}\mathcal{O}^{ij}\psi^a_i\bar{\psi}^b_j\ket{a}\bra{b}
\end{equation}
All mathematical properties of this map is completely analogous to $\Psi_d$. We use a different notation $\mathcal{O}_{(\psi)}$ primarily for the benefit of physical applications.

\subsection{Poissonization of ``simple" operators}

Following the general philosophy of Poissonization, because the simple operators $\mathcal{O}^{(\psi)}$'s are associated with single asymptotic boundaries, we need to consider Poisson quantized operators $\lambda(\mathcal{O}^{(\psi)})$ in order to study summation of bordisms. In terms of the algebra, $\mathcal{O}^{(\psi)}$'s are random matrices and they generate the von Neumann algebra $L^\infty(\mathbb{C}^k, M_d(\mathbb{C})) = L^\infty(\mathbb{C}^k)\otimes M_d(\mathbb{C})$. We take the expectation of trace $\mu\mathbb{E}\otimes\tr$ as the canonical weight on this algebra. For generality, we add a positive constant $\mu > 0$ to the weight. Therefore, we are led to consider the Poissonization of the pair of data $(L^\infty(\mathbb{C}^k, M_d(\mathbb{C})), \mu\mathbb{E}\tr)$. The resulting Poisson algebra is denoted as: $\mathbb{P}_{\mu\mathbb{E}\tr}L^\infty(\mathbb{C}^k, M_d(\mathbb{C}))$. The following calculation is central to our discussion in this section. 

Before stating the result, we fix a piece of notation. Given a subset $A\subset\{1,2,\cdots,p\}$, there exists a natural ordering of elements in $A$ which is derived from the natural ordering of $\{1,2,\cdots,p\}$:
\begin{equation}
    A := \{\alpha_1,\alpha_2,\cdots,\alpha_{|A|}\}
    \text{ where } \alpha_i < \alpha_{i+1}
\end{equation}
where $|A|$ is the cardinality of $A$. Given this ordering, there exists an isomorphism between the following two sets:
\begin{equation}
    \{1,2,\cdots,|A|\}\rightarrow A
\end{equation}
Using this isomorphism, we can transfer a permutation on the set $\{1,2,\cdots,|A|\}$ to a permutation on $A$:
\begin{equation}
    \pi(i) = j \rightarrow \pi_A(\alpha_i) = \alpha_j
\end{equation}
where $i,j\in \{1,2,\cdots,|A|\}$ and $\alpha_i, \alpha_j\in A$. It is easy to see that this map is a group isomorphism between permutation groups
\begin{equation}
    f:\mathcal{S}_{|A|}\rightarrow\mathcal{S}_A
\end{equation}
where $\mathcal{S}_{|A|}$ is the group of permutations on the set $\{1,2,\cdots,|A|\}$ and $\mathcal{S}_A$ is the group of permutations on the set $A$.

For our calculation, we need a particular cyclic permutation. Let $\pi^0:=(12\cdots|A|)\in\mathcal{S}_{|A|}$ be the cyclic permutation that sends $i$ to $i + 1$ (mod $|A|$). Then via the isomorphism $f$, we consider $\pi_A^0 := f(\pi^0)$. This is the cyclic permutation of $A$ that sends $\alpha_i$ to $\alpha_{(i + 1)\mod |A|}$.
\begin{proposition}\label{proposition:workhorse}
    Consider the Poisson algebra $\mathbb{P}_{\mu\mathbb{E}\tr}L^\infty(\mathbb{C}^k, M_d(\mathbb{C}))$ where $\mu > 0$ is a constant parameter. By Wick contraction and the general theory of Poissonization, the Poisson moment formula (i.e., the correlation function) can be written as:
    \begin{equation}
        \varphi_{\mu\mathbb{E}\tr}(\lambda(\mathcal{O}_1^{(\psi)})\cdots\lambda(\mathcal{O}^{(\psi)}_p)) = \frac{1}{d^p}\sum_{\sigma\in\mathcal{P}_p}\prod_{A\in\sigma}\sum_{\pi_A\in S_A}\mu \prod_{\gamma_A\in\pi_A}\tr(\overrightarrow{\prod}_{k\in\gamma_A}\mathcal{O}_k)d^{\text{cycles}(\pi_A^0\pi_A)}
    \end{equation}
    where $\mathcal{P}_p$ is the set of partitions of $\{1,2,\cdots,p\}$, $\sigma\in\mathcal{P}_p$ is a particular partition, $A\in\sigma$ is a subset of $\{1,2,\cdots,p\}$ that is part of the partition $\sigma$, $\mathcal{S}_A$ is the permutation group of the subset $A$, $\pi_A\in \mathcal{S}_A$ is a particular permutation, and $\gamma_A\in\pi_A$ is a cycle that is part of the permutation $\pi_A$. $\text{cycles}(\cdot)$ counts the number of cycles in a permutation. And finally, the product inside the trace is an ordered product over matrices.
    
    The connected correlation function is given by:
    \begin{equation}
        \frac{\mu}{d^p}\sum_{\pi\in S_p}d^{\text{cycles}(\pi^0\pi)}\prod_{\gamma\in\pi}\tr(\overrightarrow{\prod}_{k\in\gamma}\mathcal{O}_k)
    \end{equation}
    where $\pi^0 = (12\cdots p)$.
\end{proposition}
We can also consider the Poisson algebra $\mathbb{P}_{\mathbb{E}\tr}L^\infty(\mathbb{C}^d, M_k(\mathbb{C}))$. This is the algebra generated by $\lambda(\mathcal{O}_{(\psi)})$'s (c.f. Equation \ref{equation:lowerpsi}). Using the same calculation (only exchanging $d\leftrightarrow k$), the Poisson moment formula of this algebra is given by:
\begin{equation}\label{equation:momentforJT}
    \varphi_{\mathbb{E}\tr}\lb\lambda(\mathcal{O}_{1,(\psi)})\cdots\lambda(\mathcal{O}_{p,(\psi)})\rb = \frac{1}{k^p}\sum_{\sigma\in\mathcal{P}_p}\prod_{A\in\sigma}\sum_{\pi_A\in\mathcal{S}_A}\prod_{\gamma\in\pi_A}\tr\lb\overrightarrow{\prod}_{l\in \gamma_A}\mathcal{O}_l\rb k^{\text{cycles}(\pi_A^0\pi_A)}
\end{equation}
The corresponding connected correlation function is given by (c.f. Equation \ref{equation:lowerpsi}):
\begin{equation}\label{equation:connectedmomentforJT}
    \varphi_{\mathbb{E}\tr}\lb\lambda(\mathcal{O}_{1,(\psi)})\cdots\lambda(\mathcal{O}_{p,(\psi)})\rb_\text{conn} = \frac{1}{k^p}\sum_{\pi\in\mathcal{S}_p}\prod_{\gamma\in\pi}\tr(\overrightarrow{\prod}_{l\in \gamma}\mathcal{O}_l)k^{\text{cycles}(\pi^0\pi)}
\end{equation}
\begin{corollary}\label{corollary:planar}
    Using the same notation as the Proposition, in the limit where $d\gg 1$, the leading contribution to the connected correlation function is given by:
    \begin{equation}
        \tr(\mathcal{O}_p\mathcal{O}_{p-1}\cdots\mathcal{O}_1) = \tr(\overleftarrow{\prod}_{k\in \{1,2,\cdots,p\}}\mathcal{O}_k)
    \end{equation}
    where the matrix product reverses the canonical ordering of the set $\{1,2,\cdots,p\}$. Similarly, the leading contribution to the correlation function is given by:
    \begin{equation}
        \sum_{\sigma\in\mathcal{P}_p}\prod_{A\in\sigma}\tr(\overleftarrow{\prod}_{k\in A}\mathcal{O}_k)
    \end{equation}
    where the matrix product again reverses the canonical ordering of the set $\{1,2,\cdots,p\}$. 
\end{corollary}
The proof follows from simple calculations and the following combinatorial observation:
\begin{lemma}
    \begin{equation}
        \sum_{\sigma\in\mathcal{P}_n}\prod_{A\in\sigma}\sum_{\pi_A\in \mathcal{S}_{|A|}}\mu = \sum_{\pi\in\mathcal{S}_n}B_{|\pi|}(\mu)
    \end{equation}
    where $\mathcal{P}_n$ is the set of partitions of $n$-elements, $\sigma$ is a partition, and $\mathcal{S}_{|A|}$ is the permutation group on $|A|$ elements.
\end{lemma}
\begin{proof}
    For any partition $\sigma =\sqcup A\in\mathcal{P}_n$, we have a set map:
    \begin{equation}
        f_\sigma:\sqcup_A\mathcal{S}_{|A|}\rightarrow \mathcal{S}_n: (\pi_A)_A\mapsto \prod_A\pi_A
    \end{equation}
    Consider the combined set map:
    \begin{equation}
        f = \sqcup_\sigma f_\sigma: \sqcup_\sigma(\sqcup_A\mathcal{S}_{|A|})\rightarrow \mathcal{S}_n
    \end{equation}
    This map is many-to-one. For each fixed permutation $\pi \in\mathcal{S}_n$, consider its cycle decomposition. Let $|\pi|$ be the number of cycles in $\pi$. The cycle decomposition defines a set map:
    \begin{equation}
        \pi:[n]\rightarrow[|\pi|]
    \end{equation}
    where $\pi(i) = \pi(j)$ if and only if $i,j$ are in the same cycle. Then any partition $\eta$ of $[m]$ uniquely induces a partition $\pi^{-1}(\eta)$ of $[n]$ via pullback of $\pi$. In addition, for each subset $A$ in the induced partition $\pi^{-1}(\eta)$ of $[n]$, the cycles of the permutation $\pi$ uniquely determines a permutation $\pi_A\in\mathcal{S}_{|A|}$. Moreover, the product:
    \begin{equation}
        \prod_{A\in \pi^{-1}(\eta)}\pi_A = \pi
    \end{equation}
    Therefore, the preimage under the map $f$ is given by:
    \begin{equation}
        f^{-1}(\pi) = \{(\pi_A)_{A\in \pi^{-1}(\eta)}:\eta\in\mathcal{P}_{|\pi|}\text{ , and }\pi_A\text{ is product of cycles in }\pi\}
    \end{equation}
    We emphasize that once the partition $\eta\in\mathcal{P}_{|\pi|}$ is fixed, the permutation in the preimage $f^{-1}(\pi)$ is uniquely defined by the cycle decomposition of $\pi$.
    Then we have:
    \begin{equation}\label{equation:resummation}
        \sum_{\sigma\in\mathcal{P}_n}\prod_{A\in\sigma}\sum_{\pi_A\in\mathcal{S}_{|A|}}\mu = \sum_{\sigma\in\mathcal{P}_n}\sum_{(\pi_A)_A\in\sqcup_A\mathcal{S}_{|A|}}\mu^{|\sigma|} = \sum_{\pi\in\mathcal{S}_n}\sum_{\eta\in \mathcal{P}_{|\pi|}}\mu^{|\eta|} = \sum_{\pi\in\mathcal{S}_n}B_{|\pi|}(\mu)
    \end{equation}
    In the last equation, we have used the following formula of Touchard polynomial:
    \begin{equation}
        B_n(\mu) = \sum_{\sigma\in\mathcal{P}_n}\mu^{|\sigma|}\qedhere
    \end{equation}
\end{proof}
For the next corollary, we consider the following class of multiplicative functions on the permutation group $S_n$:
\begin{definition}
    $f:S_n\rightarrow\mathbb{C}$ is a multiplicative function on $S_n$ if for any two commuting permutations $\pi_1$ and $\pi_2$, $f(\pi_1\pi_2) = f(\pi_1)f(\pi_2)$.
\end{definition}
A multiplicative function is defined by fixing its values on irreducible cycles. This is because any permutation can be uniquely decomposed into a product of irreducible cycles (i.e., decomposition using the orbits of the permutation). Notice that a multiplicative function is \textit{not} a group homomorphism. For two generic permutations $\pi_1,\pi_2$ that do not commute with each other, in general $f(\pi_1\pi_2) \neq f(\pi_1)f(\pi_2)$. The multiplicative property holds only when the two permutations commute with each other. For this reason, we have coined the name multiplicative function to distinguish the notion from a group homomorphism.
\begin{corollary}
    Let $f:S_n\rightarrow\mathbb{C}$ be a multiplicative function on the permutation group $S_n$ then we have:
    \begin{equation}
        \sum_{\sigma\in\mathcal{P}_n}\prod_{A\in\sigma}\sum_{\pi_A\in S_{A}}\mu \prod_{\gamma_A\in S_A}f(\gamma_A) = \sum_{\pi\in S_n}B_{|\pi|}(\mu)\prod_{\gamma\in\pi}f(\gamma)
    \end{equation}
    where $\gamma\in\pi$ are irreducible cycles of $\pi\in S_n$ and $\gamma_A\in\pi_A$ are irreducible cycles of $\pi_A\in S_A$.
\end{corollary}
\begin{proof}
    The proof is a simple modification of the proof of the Lemma. Using the same notation, we recalculate Equation \ref{equation:resummation}. 
    \begin{align}
        \begin{split}
            \sum_{\sigma\in\mathcal{P}_n}\prod_{A\in\sigma}\sum_{\pi_A\in S_{A}}\mu\prod_{\gamma_A\in\pi_A}f(\gamma_A) &=\sum_{\sigma\in\mathcal{P}_n}\sum_{(\pi_A)_A\in\sqcup_A S_{A}}\mu^{|\sigma|}\prod_{\pi_A}f(\pi_A)
            \\&
            =\sum_{\pi\in S_n}\sum_{\eta\in\mathcal{P}_{|\pi|}}\mu^{|\eta|}f(\pi)
            \\&
            =\sum_{\pi\in S_n}B_{|\pi|}(\mu)\prod_{\gamma\in\pi}f(\gamma)\qedhere
        \end{split}
    \end{align}
    where in the last equation we used the definition of a multiplicative function to decompose $f(\pi)$ into a product over its cycles. Also note the $\gamma_A$'s and $\pi_A$'s are pairwise commuting because the subsets $A$ are pairwise disjoint.
\end{proof}
There are two natural types of multiplicative functions that we need in the calculation of the moment formula. 
\begin{enumerate}
    \item Fix a set of matrices $\{\mathcal{O}_1,\cdots,\mathcal{O}_p\}$. For any irreducible cycle $\gamma\in S_p$, we define:
    \begin{equation}
        f_{\tr}(\gamma) := \tr(\overrightarrow{\prod_{k\in\gamma}}\mathcal{O}_k)
    \end{equation}
    where the product of matrices is ordered. For a given cycle $\gamma = (i_1,i_2,\cdots,i_k)$, there is cyclic ambiguity in representing the cycle in the sense that:
    \begin{equation}
        \gamma = (i_2,i_3,\cdots i_k,i_1) = \cdots = (i_k,i_1,\cdots,i_{k-1})
    \end{equation}
    This corresponds to the cyclic invariance of the tracial weight:
    \begin{equation}
        \tr(\mathcal{O}_{i_1}\cdots\mathcal{O}_{i_k}) = \tr(\mathcal{O}_{i_2}\cdots\mathcal{O}_{i_k}\mathcal{O}_{i_1}) = \cdots =\tr(\mathcal{O}_{i_k}\mathcal{O}_{i_1}\cdots\mathcal{O}_{i_{k-1}})
    \end{equation}
    Therefore, the function $f_{\tr}$ is well-defined on cycles. Then for general permutations $\pi\in S_n$, the multiplicative generalization of $f$ is given by:
    \begin{equation}
        f_{\tr}(\pi) = \prod_{\gamma\in\pi}f_{\tr}(\gamma)
    \end{equation}
    By construction, $f_{\tr}$ is multiplicative;
    \item Given a constant $d > 0$, for any permutation $\pi \in S_p$ the counting function $d^{\text{cycles}(\pi)}$ where $\text{cycles}(\pi)$ counts the number of irreducible cycles in $\pi$. This function is multiplicative. This multiplicative counting function arises as the trace of a particular operator, which we explain below.

    For any given index set $A = \{\alpha_1 < \alpha_2 < \cdots < \alpha_{|A|}: \alpha_k\in \{1,2,3,\cdots,p\}\}$\footnote{$A\subset \{1,2,\cdots,p\}$}, the permutation group $S_{|A|}$ acts naturally $A$:
    \begin{equation}
        \pi_A(\alpha_k):=\alpha_{\pi(k)}
    \end{equation}
    where $\pi\in S_{|A|}$ and $\pi_A\in S_A$. $S_A$ is the permutation group of the set $A$ (not its labels). The map $\pi\mapsto\pi_A$ induces an isomorphism $S_{|A|}\cong S_A$. We want to calculate:
    \begin{align}
        \begin{split}
            \sum_{1\leq i_k,j_k\leq d}\prod_{k\in A}\delta_{i_kj_{\pi_A(k)}} &= \sum_{1\leq i_k,j_k\leq d}\tr(\ket{i_1\cdots i_{|A|}}\bra{i_{\pi_A(1)}\cdots i_{\pi_A(|A|)}})
        \end{split}
    \end{align}
    The summands are non-zero only if all labels in the same cycle of $\pi_A$ have the same index value. More precisely:
    \begin{equation}
        \sum_{1\leq i_k,j_k\leq d}\tr(\ket{i_1\cdots i_{|A|}}\bra{i_{\pi_A(1)}\cdot i_{\pi_A(|A|)}}) = \sum_{1\leq i_\gamma\leq d}\prod_{\gamma\in\pi_A}(\prod_{k\in\gamma}\delta_{i_k,i_\gamma}) = d^{\text{cycles($\pi_A$)}}
    \end{equation}
    where $\gamma\in\pi_A$ is a cycle. For simplicity of notation, we introduce the notation:
    \begin{equation}
        P^A_\pi := \sum_{1\leq i_k,j_k\leq d}\ket{i_1\cdots i_{|A|}}\bra{i_{\pi_A(1)}\cdots i_{\pi_A(|A|)}}
    \end{equation}
    where $\pi_A\in S_{A}$. Then $\tr(P^A_\pi) = d^{\text{cycles($\pi_A$)}}$.
\end{enumerate}
\begin{proof}\textbf{Proof of Proposition}
    Now we present the calculation of the moment formula:
    \begin{align}
        \begin{split}
            \varphi_{\mu\mathbb{E}\tr}&(\lambda(\mathcal{O}^{(\psi)}_1)\cdots\lambda(\mathcal{O}^{(\psi)}_p)) = \frac{1}{d^p}\sum_{\sigma\in\mathcal{P}_p}\prod_{A\in\sigma}\mu\mathbb{E}\tr\bigg(\overrightarrow{\prod}_{k\in A}(\sum_{a_kb_k,i_kj_k}\mathcal{O}_k^{a_kb_k}\psi^{a_k}_{i_k}\bar{\psi}^{b_k}_{j_k}\ket{i_k}\bra{j_k})\bigg)
            \\& = \frac{1}{d^p}\sum_{\sigma\in\mathcal{P}_p}\prod_{A\in\sigma}\sum_{a_kb_k, i_kj_k}(\prod_{k\in A}\mathcal{O}_k^{a_kb_k})\mu\mathbb{E}(\prod_{k\in A}\psi^{a_k}_{i_k}\bar{\psi}^{b_k}_{j_k})\tr(\overrightarrow{\prod}_{k\in A}\ket{i_k}\bra{j_k})
            \\
            &=\frac{1}{d^p}\sum_{\sigma\in\mathcal{P}_p}\prod_{A\in\sigma}\sum_{a_kb_k,i_kj_k}(\prod_{k\in A} \mathcal{O}_k^{a_kb_k})\mu\bigg(\sum_{\pi_A\in \mathcal{S}_{A}}\prod_{k\in A}\delta_{a_kb_{\pi_A(k)}}\delta_{i_kj_{\pi_A(k)}}\bigg)\bigg(\prod_{k\in A}\delta_{i_{\pi_A^0(k)}j_k}\bigg)
            \\&
            =\frac{1}{d^p}\sum_{\sigma\in\mathcal{P}_p}\prod_{A\in\sigma}\sum_{\pi_A\in S_{A}}\mu(\sum_{a_kb_k}\prod_{k\in A}\mathcal{O}_k^{a_kb_k}\delta_{a_kb_{\pi_A(k)}})(\sum_{i_kj_k}\prod_{k\in A}\delta_{i_kj_{\pi_A(k)}}\delta_{i_{\pi_A^0(k)}j_{k}})
            \\&
            =\frac{1}{d^p}\sum_{\sigma\in\mathcal{P}_p}\prod_{A\in\sigma}\sum_{\pi_A\in\mathcal{S}_{A}}\mu\prod_{\gamma\in\pi_A}\tr(\overrightarrow{\prod}_{k\in\gamma}\mathcal{O}_k)(\sum_{j_k}\prod_{k\in A}\delta_{j_{\pi_A^0\pi_A(k)}j_{k}})
            \\
            &=\frac{1}{d^p}\sum_{\sigma\in\mathcal{P}_p}\prod_{A\in\sigma}\sum_{\pi_A\in\mathcal{S}_{A}}\mu\prod_{\gamma\in\pi_A}\tr(\overrightarrow{\prod}_{k\in\gamma}\mathcal{O}_k)\tr(P^A_{\pi_A^0\pi_A})
        \end{split}
    \end{align}
    where the product $\prod_{\gamma\in\pi_A}$ is over the irreducible cycles in permutation $\pi_A$ and $\pi_A^0$ acts on $A := \{\alpha_1 < \alpha_2 < \cdots\alpha_{|A|}:1\leq \alpha_k\leq p\}$ by sending $\alpha_k$ to $\alpha_{k+1}$.

    By the general discussion of Poissonization, the connected component of the Poisson moment formula corresponds to the partition $\sigma = \{1,2,\cdots,p\}$. Hence, the connected correlation function reduces to \footnote{The constant parameter $\mu$ is omitted for simplicity.}:
    \begin{equation}
        \frac{1}{d^p}\sum_{\pi\in S_p}\prod_{\gamma\in\pi}\tr(\overrightarrow{\prod}_{k\in\gamma}\mathcal{O}_k)d^{\text{cycles}(\pi^0\pi)}\qedhere
    \end{equation}
\end{proof}
\begin{proof}
    \textbf{Proof of Corollary}
    When $d \gg 1$, the leading contribution to the connected correlation function comes from the maximum of the combinatorial factor $d^{\text{cycles}(\pi^0\pi)}$. It is easy to see that the maximum is unique and it is achieved by the permutation $(\pi^0)^{-1} = (p\cdots21)$ which is the inverse of the permutation $\pi^0$. For this permutation, the counting of cycles yields $d^{\text{cycles}(\pi^0(\pi^0)^{-1})} = d^p$. Hence, the leading term of the connected correlation function is determined by \footnote{Again, the constant parameter $\mu$ is omitted for simplicity.}:
    \begin{equation}
        \frac{1}{d^p}\tr(\overrightarrow{\prod}_{k\in (p\cdots21)}\mathcal{O}_k)d^p = \tr(\mathcal{O}_p\cdots\mathcal{O}_2\mathcal{O}_1)
    \end{equation}
    For the general correlation function, the same argument as above shows that the leading contribution is given by:
    \begin{align}
        \begin{split}
            \frac{1}{d^p}\sum_{\sigma\in\mathcal{P}_p}\prod_{A\in\sigma}\mu\tr(\overleftarrow{\prod}_{k\in A}\mathcal{O}_k)d^{\text{cycles}(\pi_A^0(\pi_A^0)^{-1})} &= \frac{1}{d^p}\sum_{\sigma\in\mathcal{P}_p}\mu^{|\sigma|}d^{\sum_{A\in\sigma}|A|}\prod_{A\in\sigma}\tr(\overleftarrow{\prod}_{k\in A}\mathcal{O}_k)
            \\& = \sum_{\sigma\in\mathcal{P}_p}\mu^{|\sigma|}\prod_{A\in\sigma}\tr(\overleftarrow{\prod}_{k\in A}\mathcal{O}_k)
        \end{split}
    \end{align}
    where $(\pi_A^0)^{-1}$ is the inverse of $\pi_A^0$ and $|\sigma|$ is the cardinality of the partition $\sigma$. Because $(\pi_A^0)^{-1}$ has only one cycle and it reverses the natural order of the set $A$, the matrix product is order-reversing $\overleftarrow{\prod}_{k\in A}$.
\end{proof}
This large-$d$ correlation function is the Poisson moment formula of the Poisson algebra $\mathbb{P}_{\tr}M^\text{op}_k(\mathbb{C})$ \footnote{Ignoring the constant $\mu$.}, where $M^\text{op}_k(\mathbb{C})$ is the \textit{opposite} von Neumann algebra of $M_k(\mathbb{C})$. In terms of the standard form of $M_k(\mathbb{C})$ with respect to the canonical trace, the opposite algebra $M^\text{op}_k(\mathbb{C})$ can be realized as the commutant of the left representation of $M_k(\mathbb{C})$ on the Hilbert-Schmidt space $L^2_k(\mathbb{C})$. The opposite algebra is the representation of $M_k(\mathbb{C})$ on $L^2_k(\mathbb{C})$ by \textit{right} multiplication. 
\begin{corollary}
    Consider the Poisson algebra $\mathbb{P}_{\tr}M^\text{op}_k(\mathbb{C})$, the Poisson moment formula is given by:
    \begin{equation}
        \varphi_{\tr}(\lambda(\mathcal{O}_1^\text{op})\cdots\lambda(\mathcal{O}^\text{op}_p)) = \sum_{\sigma\in\mathcal{P}_p}\prod_{A\in\sigma}\tr(\overrightarrow{\prod}_{j\in A}\mathcal{O}^\text{op}_k) = \sum_{\sigma\in\mathcal{P}_p}\prod_{A\in\sigma}\tr(\overleftarrow{\prod}_{j\in A}\mathcal{O}_k)
    \end{equation}
\end{corollary}
The proof is a straightforward calculation and is omitted.

\subsection{Projection to diagonal simple operators}

Poissonization of simple operators will be used to model simplified JT gravity with end-of-the-world branes. To study open/closed 2D TQFTs, we need an additional projection to the diagonal matrices and consider diagonal simple operators. Recall that we have considered a unital CP map that outputs simple operators:
\begin{equation}
\Psi_d: L_k^1(\mathbb{C})\rightarrow L^1(\mathbb{C}^k,L^1_d(\mathbb{C}))
\end{equation}
We now compose this channel with the projection to diagonal matrices:
\begin{equation}
    \mathcal{E}_d:=(id\otimes P_\text{diag})\circ\Psi_d: L_k^1(\mathbb{C})\rightarrow L^1(\mathbb{C}^k, L_d^1(\mathbb{C}))\rightarrow L^1(\mathbb{C}^k, \ell_d^1(\mathbb{C}))
\end{equation}
where $P_\text{diag}: L^1_d(\mathbb{C})\rightarrow\ell^1_d(\mathbb{C})$ is the projection to diagonal matrices.
\begin{equation}
    \mathcal{E}_d(\mathcal{O}) = (id\otimes P_\text{diag})(\mathcal{O}^{(\psi)}) = \frac{1}{d}\sum_{i,ab}\mathcal{O}^{ab}\psi^a_i\bar{\psi}^b_i\ket{i}\bra{i}
\end{equation}
We now calculate the correlation function of Poisson quantized operators $\lambda(\mathcal{E}_d(\mathcal{O}))$:
\begin{proposition}
    Consider the Poisson algebra $\mathbb{P}_{\mu\mathbb{E}\tr}L^\infty(\mathbb{C}^k, \ell_d^\infty(\mathbb{C}))$, we have the following correlation function (i.e., Poisson moment formula) of diagonal simple operators:
    \begin{equation}
        \varphi_{\mu\mathbb{E}\tr}(\lambda(\mathcal{E}_d(\mathcal{O}_1))\cdots\lambda(\mathcal{E}_d(\mathcal{O}_p))) = \frac{1}{d^p}\sum_{\pi\in \mathcal{S}_p}B_{|\pi|}(\mu d)\prod_{\gamma\in\pi}\tr(\overrightarrow{\prod}_{k\in\pi}\mathcal{O}_k)
    \end{equation}
    Moreover, the connected correlation function of diagonal simple operators is given by:
    \begin{equation}
        \frac{\mu d}{d^p}\sum_{\pi\in \mathcal{S}_p}\prod_{\gamma\in\pi}\tr(\overrightarrow{\prod}_{k\in\gamma}\mathcal{O}_k)
    \end{equation}
\end{proposition}
\begin{proof}
    The proof is a simple modification of the calculation for random matrices. 
    \begin{align}
        \begin{split}
            \varphi_{\mu\mathbb{E}\tr}&(\lambda(\mathcal{E}_d(\mathcal{O}_1))\cdots\lambda(\mathcal{E}_d(\mathcal{O}_p))) = \frac{1}{d^p}\sum_{\sigma\in\mathcal{P}_p}\prod_{A\in\sigma}\mu\mathbb{E}\tr(\overrightarrow{\prod}_{k\in A}(\sum_{i_k,a_kb_k}\mathcal{O}^{a_kb_k}\psi^{a_k}_{i_k}\bar{\psi}^{b_k}_{i_k}\ket{i_k}\bra{i_k}))
            \\
            &=\frac{1}{d^p}\sum_{\sigma\in\mathcal{P}_p}\prod_{A\in\sigma}\sum_{a_kb_k,i_k}(\prod_{k\in A}\mathcal{O}^{a_kb_k})\mu\mathbb{E}(\prod_{k\in A}\psi^{a_k}_{i_k}\bar{\psi}^{b_k}_{i_k})\tr(\prod_{k\in A}\ket{i_k}\bra{i_k})
            \\
            &=\frac{1}{d^p}\sum_{\sigma\in\mathcal{P}_p}\prod_{A\in\sigma}d\sum_{a_kb_k}(\prod_{k\in A}\mathcal{O}^{a_kb_k})\mu(\sum_{\pi_A\in\mathcal{S}_A}\prod_{k\in A}\delta_{a_kb_{\pi_A(k)}})
            \\&
            =\frac{1}{d^p}\sum_{\sigma\in\mathcal{P}_p}\prod_{A\in\sigma}\sum_{\pi_A\in\mathcal{S}_A}(\mu d)\prod_{\gamma\in\pi_A}\tr(\overrightarrow{\prod}_{k\in \gamma}\mathcal{O}_k)
            \\
            &=\frac{1}{d^p}\sum_{\pi\in\mathcal{S}_p}B_{|\pi|}(\mu d)\prod_{\gamma\in\pi}\tr(\overrightarrow{\prod}_{k\in\gamma}\mathcal{O}_k)
        \end{split}
    \end{align}
    where the last equation follows from the combinatorial lemma.

    To extract the connected correlation function, we simply focus on the contribution of the trivial partition $\{1,2,\cdots,p\}$. The corresponding summand is given by:
    \begin{equation}
        \frac{\mu d}{d^p}\sum_{\pi\in\mathcal{S}_p}\prod_{\gamma\in\pi}\tr(\overrightarrow{\prod}_{k\in\gamma}\mathcal{O}_k)\qedhere
    \end{equation}
\end{proof}
Recall that for the Poisson algebra of random simple operators, the large-$d$ limit of the correlation function is the correlation function of the Poisson algebra of non-random matrices. In particular, the large-$d$ connected correlation function is given by a single trace $\tr(\mathcal{O}_p\cdots\mathcal{O}_2\mathcal{O}_1)$. As a comparison, for the Poisson algebra of random diagonal matrices, the large-$d$ limit of the correlation function is given by a product of traces.
\begin{corollary}
    The leading contribution to the large-$d$ limit of the correlation function is given by:
    \begin{equation}
        \mu^p\prod_{1\leq k\leq  p}\tr(\mathcal{O}_k)
    \end{equation}
\end{corollary}
\begin{proof}
    When $d\rightarrow \infty$, the correlation function is dominated by the singleton partition: $\sigma_\text{singleton}:=\{\{1\},\{2\},\cdots,\{p\}\}$. Because each component of this partition is a singleton set, the permutation group associated with each component is trivial. Therefore, from the calculation in the Proposition, we have the following limit:
    \begin{align}
        \begin{split}
            \lim_{d\rightarrow\infty}\frac{1}{d^p}\sum_{\sigma\in\mathcal{P}_p}\prod_{A\in\sigma}\sum_{\pi_A\in\mathcal{S}_A}(\mu d)\prod_{\gamma\in\pi_A}\tr(\overrightarrow{\prod}_{k\in\gamma}\mathcal{O}_k) = \mu^p\prod_{1\leq k\leq d}\tr(\mathcal{O}_k)\qedhere
        \end{split}
    \end{align}
\end{proof}
Notice that in the limit $d\rightarrow\infty$, the finite-$d$ connected correlation function of diagonal simple operators vanishes. Only the totally disconnected contribution to the full correlation function survives. This is in complete contrast with the correlation function of general simple operators, where the connected contributions survive in the large-$d$ limit.

Because the diagonal simple operators mutually commute, we can directly calculate the moment generating function. Physically, this generating function should be understood as a partition function. The correlation function can be derived from the partition function.
\begin{proposition}
    The moment generating function of diagonal simple operators under the Poisson state is given by:
    \begin{equation}
        \varphi_{\mu\mathbb{E}\tr}(e^{u\lambda(\mathcal{E}_d(\mathcal{O}))}) =  \exp(\frac{\mu d}{\det(\mathbbm{1} - \frac{u}{d}T)} - \mu d)
    \end{equation}
    where $u < 0$ is a negative constant. 
\end{proposition}
\begin{proof}
    The proof is a simple calculation of the Gaussian integral.
    \begin{align}
        \begin{split}
            \varphi_{\mu\mathbb{E}\tr}(e^{u\lambda(\mathcal{E}_d(\mathcal{O}))}) &= \exp(\mu\mathbb{E}\tr(e^{u\mathcal{E}_d(\mathcal{O})} - \mathbbm{1}))
            \\
            &=\exp(\mu\mathbb{E}\tr(\sum_ie^{\frac{u}{d}\sum_{ab}\mathcal{O}^{ab}\psi^a_i\bar{\psi}^b_i}\ket{i}\bra{i} - \mathbbm{1}))
            \\
            &=\prod_i\exp(\mu\mathbb{E}(e^{\frac{u}{d}\sum_{ab}\mathcal{O}^{ab}\psi^a_i\bar{\psi}^b_i} - 1))
        \end{split}
    \end{align}
    The exponent is a Gaussian integral:
    \begin{equation}
        \mathbb{E}(e^{\frac{u}{d}\sum_{ab}\mathcal{O}^{ab}\psi^a_i\bar{\psi}^b_i}) = \int\big(\prod_{a}\frac{d\psi^a_id\bar{\psi}^a_i}{\pi}e^{-|\psi^a_i|^2}\big)e^{\frac{u}{d}\sum_{ab}\mathcal{O}^{ab}\psi^a_i\bar{\psi}^b_i} = \frac{1}{\det(\mathbbm{1} - \frac{u}{d}T)}
    \end{equation}
    where the determinant is taken over $k\times k$-matrices. Therefore, we have:
    \begin{equation}
        \varphi_{\mu\mathbb{E}\tr}(e^{u\lambda(\mathcal{E}_d(\mathcal{O}))}) = \prod_d\exp(\frac{\mu}{\det(\mathbbm{1} - \frac{u}{d}T)} - \mu) = \exp(\frac{\mu d}{\det(\mathbbm{1} - \frac{u}{d}T)} - \mu d)\qedhere
    \end{equation}
\end{proof}
From the perspective of 2D open/closed TQFT, it is not surprising that its correlation function can be reproduced by a commutative Poisson algebra. Recall that the crucial bordism in a 2D open/closed TQFT is the so-called open-to-closed bordism. Geometrically, this bordism folds and glues the two endpoints of an open boundary. In the canonical form of a 2D open/closed bordism \cite{Lauda-Pfeiffer}, each open boundary connects with the bulk manifold through an open-to-closed bordism, and the bulk manifold only has closed boundaries. 

In terms of a knowledgeable Frobenius algebra, this open-to-closed bordism is represented by a projection to a commutative subalgebra. Therefore, the correlation function of a 2D open/closed TQFT is calculated by first projecting the (possibly noncommutative) operators associated with open boundaries to a commutative subalgebra. From this perspective, it is not surprising that the 2D open/closed TQFT correlation function is modeled by the Poisson moment formula of a commutative Poisson algebra.
\subsection{Projection to central simple operators}\label{subsection:centralsimple}

In addition to the projection to diagonal simple operators, we can also go one step further and project to the central simple operators. This will yield a slightly different partition function and correlation functions. The central projection is given by the normalized trace:
\begin{equation}
    \tau_d:=(id\otimes\frac{\tr}{d})\circ\Psi_d: L^1_k(\mathbb{C})\rightarrow L^1(\mathbb{C}^k, L^1_d(\mathbb{C})) \rightarrow L^1(\mathbb{C}^k, \mathbb{C}\mathbbm{1})
\end{equation}
where the normalized trace $\frac{\tr}{d}$ maps $L^1_d(\mathbb{C})$ to its center which is isomorphic to $\mathbb{C}$. 
\begin{equation}
    \overline{\mathcal{O}^{(\psi)}}:=\tau_d(\mathcal{O}) = (id\otimes\frac{\tr}{d})(\mathcal{O}^{(\psi)}) = \frac{1}{d^2}(\sum_{i,ab}\mathcal{O}^{ab}\psi^a_i\bar{\psi}^b_i)\mathbbm{1}
\end{equation}
Notice that the channel $\tau_d$ preserves following two states:
\begin{equation}
    \mathbb{E}\otimes\tr(\overline{\mathcal{O}^{(\psi)}}) = \frac{1}{d^2}\sum_{i,ab}\mathcal{O}^{ab}\mathbb{E}(\psi^a_i\bar{\psi}^b_i)\tr(\mathbbm{1}) = \frac{1}{d}\sum_i\sum_{ab}\mathcal{O}^{ab}\delta_{ab} = \tr(\mathcal{O})
\end{equation}
We directly calculate the moment generating function.
\begin{proposition}
    Consider the Poisson algebra $\mathbb{P}_{\mu\mathbb{E}\tr}L^\infty(\mathbb{C}^k, \mathbb{C}\mathbbm{1})$, we have the moment generating function:
    \begin{equation}
        \varphi_{\mu\mathbb{E}\tr}(e^{u\lambda(\overline{\mathcal{O}^{(\psi)}})}) = \exp(\frac{\mu d}{\det(1 - \frac{u}{d^2}T)^d} - \mu d)\qedhere
    \end{equation}
    where $u < 0$ is a real constant. 
\end{proposition}
\begin{proof}
    This is essentially the same calculation as the diagonal simple operators.
    \begin{align}
        \begin{split}
            \varphi_{\mu\mathbb{E}\tr}(e^{u\lambda(\overline{\mathcal{O}^{(\psi)}})}) &= \exp(\mu\mathbb{E}\tr(e^{u\overline{\mathcal{O}^{(\psi)}}} - \mathbbm{1})) = \exp(\mu\mathbb{E}(e^{\frac{u}{d^2}\sum_{i,ab}\mathcal{O}^{ab}\psi^a_i\bar{\psi}^b_i} - 1)\tr(\mathbbm{1}))
            \\
            &=\exp(\mu d\mathbb{E}(e^{\frac{u}{d^2}\sum_{i,ab}\mathcal{O}^{ab}\psi^a_i\bar{\psi}^b_i} - 1))
        \end{split}
    \end{align}
    The exponent is a Gaussian integral:
    \begin{equation}
        \mathbb{E}(e^{\frac{u}{d^2}\sum_{i,a,b}\mathcal{O}^{ab}\psi^a_i\overline{\psi}^b_i}) = \int \bigg(\prod_{i,a}\frac{d\psi^a_i d\bar{\psi}^a_i}{\pi}e^{-|\psi^a_i|^2}\bigg)e^{\frac{u}{d^2}\sum_{i,a,b}\mathcal{O}^{ab}\psi^a_i\bar{\psi}^b_i}
    \end{equation}
    Note that this integral converges when $u < 0$. This Gaussian integral is a product of $d$ identical Gaussian integrals:
    \begin{equation}
        \mathbb{E}(e^{\frac{u}{d^2}\sum_{i,a,b}\mathcal{O}^{ab}\psi^a_i\bar{\psi}^b_i}) = \prod_{1\leq i\leq d}\int\bigg(\prod_a\frac{d\psi^a_i d\bar{\psi}^a_i}{\pi}e^{-|\psi^a_i|^2}\bigg)e^{\frac{u}{d^2}\sum_{ab}\mathcal{O}^{ab}\psi^a_i\bar{\psi}^b_i}
    \end{equation}
    And a simple calculation gives:
    \begin{equation}
        \int\bigg(\prod_a\frac{d\psi^a_i d\bar{\psi}^a_i}{\pi}e^{-|\psi^a_i|^2}\bigg)e^{\frac{u}{d^2}\sum_{ab}\mathcal{O}^{ab}\psi^a_i\bar{\psi}^b_i} = \frac{1}{\det(1 - \frac{u}{d^2}T)}
    \end{equation}
    where the determinant is taken over $k\times k$-matrices. Therefore, we have:
    \begin{equation}
        \varphi_{\mu\mathbb{E}\tr}(e^{u\lambda(\overline{\mathcal{O}^{(\psi)}})}) = \exp(\frac{\mu d}{\det(1 - \frac{u}{d^2}T)^d} - \mu d)\qedhere
    \end{equation}
\end{proof}
Compared with the partition function for diagonal simple operators, the only difference is in the determinant factor. 

We can derive the correlation function of central simple operators from this partition function. Alternatively, we can calculate the correlation function directly:
\begin{proposition}
    The correlation function of central simple operators is given by:
    \begin{equation}
        \varphi_{\mu\mathbb{E}\tr}(\lambda(\overline{\mathcal{O}^{(\psi)}_1})\cdots\lambda(\overline{\mathcal{O}_p^{(\psi)}}) = \frac{1}{d^{2p}}\sum_{\pi\in\mathcal{S}_p}B_{|\pi|}(\mu d)\prod_{\gamma\in\pi}d\tr(\overrightarrow{\prod}_{k\in\gamma}\mathcal{O}_k)
    \end{equation}
\end{proposition}
\begin{proof}
    \begin{align}
        \begin{split}
            \varphi_{\mu\mathbb{E}\tr}&(\lambda(\overline{\mathcal{O}^{(\psi)}_1})\cdots\lambda(\overline{\mathcal{O}^{(\psi)}_p})) = \frac{1}{d^{2p}}\sum_{\sigma\in\mathcal{P}_p}\prod_{A\in\sigma}\mu\mathbb{E}\tr(\prod_{k\in A}(\sum_{i_k,a_kb_k}\mathcal{O}^{a_kb_k}_k\psi^{a_k}_{i_k}\bar{\psi}^{b_k}_{i_k})\mathbbm{1})
            \\
            &=\frac{1}{d^{2p}}\sum_{\sigma\in\mathcal{P}_p}\prod_{A\in\sigma}\sum_{i_k, a_kb_k}(\prod_{k\in A}\mathcal{O}^{a_kb_k}_k)\mu\mathbb{E}(\prod_{k\in A}\psi^{a_k}_{i_k}\bar{\psi}^{b_k}_{i_k})\tr(\mathbbm{1})
            \\&
            =\frac{1}{d^{2p}}\sum_{\sigma\in\mathcal{P}_p}\prod_{A\in\sigma}\sum_{i_k,a_kb_k}(\prod_{k\in A}\mathcal{O}^{a_kb_k}_k)\mu(\sum_{\pi_A\in\mathcal{S}_A}\prod_{k\in A}\delta_{a_kb_{\pi_A(k)}}\delta_{i_k i_{\pi_A(k)}})\tr(\mathbbm{1})
            \\&
            =\frac{1}{d^{2p}}\sum_{\sigma\in\mathcal{P}_p}\prod_{A\in\sigma}\sum_{\pi_A\in\mathcal{S}_A}\mu\prod_{\gamma\in\pi_A}\tr(\overrightarrow{\prod}_{k\in\gamma}\mathcal{O}_k)(\sum_{i_k}\prod_{k\in A}\delta_{i_ki_{\pi_A(k)}})\tr(\mathbbm{1})
            \\&
            =\frac{1}{d^{2p}}\sum_{\sigma\in\mathcal{P}_p}\prod_{A\in\sigma}\sum_{\pi_A\in\mathcal{S}_A}\mu\big(\prod_{\gamma\in\pi_A}\tr(\overrightarrow{\prod}_{k\in\gamma}\mathcal{O}_k)\big)d^{1 + \text{cycles}(\pi_A)}
            \\&
            =\frac{1}{d^{2p}}\sum_{\sigma\in\mathcal{P}_p}\prod_{A\in\sigma}\sum_{\pi_A\in\mathcal{S}_A}(\mu d)\big(\prod_{\gamma\in\pi_A}d\tr(\overrightarrow{\prod}_{k\in\gamma}\mathcal{O}_k)\big)
            \\&
            =\frac{1}{d^{2p}}\sum_{\pi\in\mathcal{S}_p}B_{|\pi|}(\mu d)\prod_{\gamma\in\pi}d\tr(\overrightarrow{\prod}_{k\in\gamma}\mathcal{O}_k)
        \end{split}
    \end{align}
    where the last equation follows from the combinatorial lemma. 
\end{proof}
\begin{corollary}
    In the large-$d$ limit, the correlation function of central simple operators reduces to:
    \begin{equation}
        \mu^p\prod_{1\leq k\leq p}\tr(\mathcal{O}_k)
    \end{equation}
\end{corollary}
\begin{proof}
    For each partition $\sigma$ and each set of permutations $(\pi_A)_{A\in\sigma}$, the $d$-dependent factor is given by $d^{|\sigma| + \sum_{A\in\sigma}\text{cycles}(\pi_A)}$. Because $|\sigma| \leq p$ and $\sum_{A\in\sigma}\text{cycles}(\pi_A) \leq \sum_{A\in\sigma}|A| = p$. Both inequalities achieve a maximum when the partition $\sigma$ is the singleton partition and the set of partitions is trivial. In this case, we have:
    \begin{equation}
        \lim_{d\rightarrow\infty}\frac{1}{d^{2p}}\sum_{\pi\in\mathcal{S}_p}B_{|\pi|}(\mu d)\prod_{\gamma\in\pi}d\tr(\overrightarrow{\prod}_{k\in\gamma}\mathcal{O}_k) = \mu^p\prod_{1\leq k\leq p}\tr(\mathcal{O}_k)\qedhere
    \end{equation}
\end{proof}
Notice that this is the same large-$d$ limit as the correlation function of the diagonal simple operators.  

\subsection{Simplified JT gravity with End-of-the-World branes}\label{subsection:mathJTEOW}

In this section, we apply Poissonization to study simplified JT gravity with end-of-the-world branes \cite{Penington:2019kki}; the example discussed in Section \ref{sec:babyuniversesnonAbel}. In particular, we show that the correlation function of particular (reduced) density matrices calculated in \cite{Penington:2019kki} is the Poisson moment formula of random matrices. This demonstration is largely an exercise of matching notations. We will not discuss the physical meaning of the density matrices considered in \cite{Penington:2019kki} since this is well explained in the original paper and in the main body of the text.

We state the observation as the following corollary to Proposition \ref{proposition:workhorse}.
\begin{corollary}
    Consider a particular diagonal $d\times d$-matrix:
    \begin{equation}
        \mathcal{O}(\beta) := \sum_i 2^{1 - 2\mu}|\Gamma(\mu - \frac{1}{2} + i\sqrt{2E_i})|^2e^{-\beta E_i}\ket{i}\bra{i}
    \end{equation}
    where $\mu > 0$ is a constant parameter (tension of the brane)\footnote{This is not to be confused with the $\mu$ introduced earlier for the normalization of the weight.} and the set $\{E_i\}$ are the eigen-energies and $\ket{i}$ is the corresponding energy eigen-state \footnote{Here for simplicity of notation, we assume there is no degeneracy. Each energy $E_i$ corresponds to a distinct state $\ket{i}$. Our discussion can accommodate degeneracy with only changes in notations.}. Then, the Poisson moment formula for the Poisson quantized random operators $\lambda(\mathcal{O}(\beta)_{(\psi)})$ is given by (c.f. Equation \ref{equation:momentforJT}):
    \begin{equation}
        \varphi_{\mathbb{E}\tr}(\lambda(\mathcal{O}(\beta_1)_{(\psi)})\cdots\lambda(\mathcal{O}(\beta_p)_{\psi})) = \frac{1}{k^p}\sum_{\sigma\in \mathcal{P}_p}\prod_{A\in\sigma}\sum_{\pi_A\in S_A}k^{\text{cycles}(\pi_A^0\pi_A)}\prod_{\gamma\in \pi_A}\tr(\prod_{j\in \gamma}\mathcal{O}(\beta_j))
    \end{equation}
    where recall $\pi^0_A$ is the cyclic permutation as explained in the proof of Proposition F.1.
    
    This moment formula contains both connected and disconnected contributions. Since connected correlation function corresponds to the trivial partition $\{1,2,\cdots,p\}\in\mathcal{P}_p$, the connected correlation function is given by:
    \begin{equation}
        \mathbb{E}\tr(\mathcal{O}(\beta_1)_{(\psi)}\cdots\mathcal{O}(\beta_p)_{(\psi)}) =\frac{1}{k^p}\sum_{\pi\in S_p}\prod_{\gamma\in \pi}\tr(\prod_{j\in\gamma}\mathcal{O}(\beta_k))k^{\text{cycles}(\pi^0\pi)}
    \end{equation}
    where again $\pi_0$ is the cyclic permutation as explained in the proof of Proposition F.1.
\end{corollary}
We now match notation with the original text (see also Appendix D \cite{Penington:2019kki}).  Recall the key (reduced) density operator\footnote{The original paper assumes a pure state on the composite system of a black hole and its Hawking radiation and considers its reduction to the radiation subsystem.} considered in \cite{Penington:2019kki} is given by (using the original notation):
\begin{equation}\label{equation:JTdensity}
    \rho(\beta) := \frac{1}{k}\sum_{1\leq a,b\leq k}\bra{\psi_b(\beta)}\ket{\psi_a(\beta)}\ket{a}\bra{b} = \tr_B\bigg((\frac{1}{\sqrt{k}}\sum_a\ket{\psi_a(\beta)}_B\ket{a}_R)(\frac{1}{\sqrt{k}}\sum_j\bra{\psi_b(\beta)}_B\bra{b}_R)\bigg)
\end{equation}
where the partial trace is over the gravitational component $B$ \footnote{Following the original notation in \cite{Penington:2019kki}, $B$ stands for black hole.}. In terms of the original notation of \cite{Penington:2019kki}, the pure state $\frac{1}{\sqrt{k}}\sum_a\ket{\psi_a(\beta)}_B\ket{a}_R$ is denoted as $\ket{\Psi}$ (see also Section 2 \cite{Penington:2019kki}). 

Using the definition of $\ket{\psi_a(\beta)}$, we have (see also Equations D.2 and D.3 in \cite{Penington:2019kki}):
\begin{equation}
    \rho(\beta) = \frac{1}{k}\sum_{1\leq a,b\leq k}\sum_{1\leq i\leq d} 2^{1-2\mu}|\Gamma(\mu - \frac{1}{2} + i\sqrt{2E_i})|^2e^{-\beta E_i}\psi^a_i\bar{\psi}^b_i\ket{a}\bra{b}
\end{equation}
where $\psi_i^a$'s are standard complex normal random variables\footnote{In terms of the original notation, $\psi_i^a$ is denoted as $C_{i,a}$.}. The authors are interested in calculating the \textit{connected} correlation function\footnote{In the language of JT gravity with EOW branes, this connected correlation function calculates the p-point function of p open boundaries with lengths $\{\beta_m\}_{1\leq m\leq p}$. Each open boundary is bounded by EOW branes. For simplicity, we have not included possible contributions from closed boundaries.}:
\begin{equation}
    \mathbb{E}\tr(\rho(\beta_1)\cdots\rho(\beta_p)) = \frac{1}{k^p}\sum_{ab}\tr(\overrightarrow{\prod_{1\leq m\leq p}}\ket{a_m}\bra{b_m})\mathbb{E}(\prod_{1\leq m\leq p}\bra{\psi_{b_m}(\beta_m)}\ket{\psi_{a_m}(\beta_m)})
\end{equation}
The Equations D.2 and D.3 in \cite{Penington:2019kki} calculate the random expectation part of this correlation function, namely $\mathbb{E}(\prod_{1\leq m\leq p}\bra{\psi_{b_m}(\beta_m)}\ket{\psi_{a_m}(\beta_m)})$. Plug in the result derived in \cite{Penington:2019kki} to the equation above, we can express the connected correlation function as:
\begin{align}\label{equation:connectedJT}
    \begin{split}
        &\mathbb{E}\tr(\rho(\beta_1)\cdots\rho(\beta_p)) = \frac{1}{k^p}\sum_{a_mb_m}\tr(\overrightarrow{\prod_{1\leq m\leq p}}\ket{a_m}\bra{b_m})\\&\times\sum_{\pi\in S_p}\bigg[(\prod_m\delta_{a_mb_{\pi(m)}})\prod_{\gamma\in\pi}\tr(\prod_{l\in \gamma}\bigg[\sum_i 2^{1-2\mu}e^{-\beta_kE_i}|\Gamma(\mu - \frac{1}{2} + i\sqrt{2E_i})|^2\ket{i}\bra{i}\bigg])\bigg]
        \\&
        =\frac{1}{k^p}\sum_{\pi\in\mathcal{S}_p}\bigg[\sum_{a_mb_m}\prod_m \delta_{a_{\pi^0(m)}b_m}\prod_m\delta_{a_{\pi(m)}b_m}\bigg]\\&\times\bigg[\prod_{\gamma\in\pi}\tr(\prod_{l\in \gamma}\bigg[\sum_i 2^{1-2\mu}e^{-\beta_kE_i}|\Gamma(\mu - \frac{1}{2} + i\sqrt{2E_i})|^2\ket{i}\bra{i}\bigg])\bigg]
        \\&
        =\frac{1}{k^p}\sum_{\pi\in\mathcal{S}_p}\prod_{\gamma\in\pi}\tr(\prod_{l\in \gamma}\bigg[\sum_i 2^{1-2\mu}e^{-\beta_kE_i}|\Gamma(\mu - \frac{1}{2} + i\sqrt{2E_i})|^2\ket{i}\bra{i}\bigg])k^{\text{cycles}(\pi^0\pi)}
    \end{split}
\end{align}
In the first equation, we simply plug in the results derived in \cite{Penington:2019kki}. In the second equation, we used the fact that $\tr(\overrightarrow{\prod}_m\ket{a_m}\bra{b_m} = \delta_{a_{\pi_0(m)}b_m})$ where $\pi_0$ is the cyclic permutation $(12\cdots p)\in\mathcal{S}_p$.

In our framework, the density operator $\rho(\beta)$ is modeled by a Poisson quantized operator: $\lambda(\mathcal{O}(\beta)_{(\psi)})$\footnote{In fact, Equation D.10 of \cite{Penington:2019kki} precisely defines the CPTP map $\mathcal{O}(\beta)\mapsto \mathcal{O}(\beta)_{(\psi)}$. This is the process of "dressing" open boundary operators with endpoint Gaussian random variables. The new feature added by Poissonization is to consider the Poisson quantized operator $\lambda(\mathcal{O}(\beta)_{(\psi)})$ instead of $\mathcal{O}(\beta)_{(\psi)}$}. Observe that we have defined:
\begin{equation}
    \mathcal{O}(\beta_k) = \sum_i 2^{1-2\mu}e^{-\beta_kE_i}|\Gamma(\mu - \frac{1}{2} + i\sqrt{2E_i})|^2\ket{i}\bra{i}
\end{equation}
Therefore, in our notation, the connected correlation function derived in \cite{Penington:2019kki} can be expressed as:
\begin{equation}\label{equation:corrJTEOW}
    \mathbb{E}\tr(\mathcal{O}(\beta_1)_{(\psi)}\cdots\mathcal{O}(\beta_p)_{(\psi)}) =\frac{1}{k^p}\sum_{\pi\in\mathcal{S}_p}\prod_{\gamma\in\pi}\tr(\prod_{l\in\gamma}\mathcal{O}(\beta_l))k^{\text{cycles}(\pi^0\pi)}
\end{equation}
This is the same expression derived using Poissonization (c.f. Equation \ref{equation:connectedmomentforJT}). Since the general correlation function is a summation over products of connected correlation functions, the Poisson moment formula correctly reproduces the correlation functions of simplified JT gravity with EOW branes.

In the main text of \cite{Penington:2019kki}, the authors considered the planar limit of JT gravity with EOW branes (c.f. Section 2.5 of \cite{Penington:2019kki}). The planar limit is achieved by taking both $k$ and $d\sim e^{S_0}$ to infinity. In Poissonization, these two limits have distinct effects. When $k\rightarrow \infty$, the only surviving term in the connected correlation function corresponds to the cyclic permutation $(\pi^0)^{-1}$. The connected correlation function becomes:
\begin{align}
    &\lim_{k\rightarrow\infty}\mathbb{E}\tr(\mathcal{O}(\beta_1)_{(\psi)}\cdots\mathcal{O}(\beta_p)_{(\psi)}) = \tr\lb\prod_{1\leq k\leq p}\mathcal{O}(\beta_k)\rb \nn\\
    &= \sum_i \prod_{1\leq k\leq p}\lb 2^{1-2\mu}e^{-\beta_k E_i}|\Gamma(\mu - \frac{1}{2} + i\sqrt{2E_i})|^2\rb
\end{align}
We can rewrite this equation in terms of spectral density. Following \cite{Penington:2019kki},   we introduce a new variable $s > 0$ such that $\frac{s^2}{2} = E$. Then, in terms of a density function $\rho_\text{spec}(s)$ \footnote{In the original text \cite{Penington:2019kki}, the spectral density is simply denoted as $\rho(s)$. Here we added the subscript to emphasize that $\rho_\text{spec}(s)$ is the spectral density while $\rho(\beta)$ is the density matrix.}, we have:
\begin{equation}\label{equation:poissZn1}
    \lim_{k\rightarrow\infty}\mathbb{E}\tr(\mathcal{O}(\beta)_{(\psi)}^p) = \int_0^\infty ds\rho_\text{spec}(s)y_\beta(s)^p
\end{equation}
where $y_\beta(s) := 2^{1-2\mu}e^{-\beta s^2 /2}|\Gamma(\mu - \frac{1}{2} + is)|^2$ (c.f. Equation 2.32 \cite{Penington:2019kki}). Here $\rho_\text{spec}(s) = \sum_i\delta(s - \sqrt{2E_i})$.

Now taking the limit $d\sim e^{S_0}\rightarrow\infty$, all genus creation is suppressed. The spectral density reduces to the leading term:
\begin{equation}
    \rho_0(s) = \frac{s}{2\pi^2}\sinh(2\pi s)
\end{equation}
In this case, the connected correlation function becomes:
\begin{equation}\label{equation:poissZn}
    \lim_{k, d\rightarrow\infty}\mathbb{E}\tr(\mathcal{O}(\beta)_{(\psi)}^p) = \lim_{k,d\rightarrow\infty}\tr\lb\mathcal{O}(\beta)^p\rb = \int_0^\infty ds\rho_0(s)y_\beta(s)^p = Z_p
\end{equation}
It is easy to see that Equation \ref{equation:poissZn} is precisely the formula of the planar connected correlation function derived in \cite{Penington:2019kki} (c.f. Equation 2.32). The notation $Z_p$ follows from the original text \cite{Penington:2019kki}, and it denotes the connected correlation function with $p$ asymptotic boundaries and disk bulk topology.

\subsection{Open/Closed 2D TQFTs}\label{subsection:openclosedappendix}

In this section, we apply Poissonization to study open/closed 2D TQFT as defined by \cite{@Moore-Segal, banerjee2022comments, gardiner20212D}. In particular, we show that the correlation function of open/closed 2D TQFT is a Poisson moment formula. 

To compare our calculation with the results of \cite{banerjee2022comments}, recall the n-point correlation function (i.e., the bordism mapping n open boundaries to $\mathbb{C}$) derived in \cite{banerjee2022comments} is given by (see also Equation 4.13 of \cite{banerjee2022comments}):
\begin{equation}
    \bar{\mathcal{A}}^\text{open}(n,0):\mathcal{O}_1\cdots\mathcal{O}_n\mapsto \sum_{\sigma\in S_n}B_{r(\sigma)}(\lambda)e^\lambda\prod_{c\in\text{cycl}(\sigma)}\mu^{-1}\tr(\mathcal{O}_c)
\end{equation}
where $r(\sigma) = |\sigma|$ is the number of cycles in $\sigma$, $\text{cycl}(\sigma)$ is the set of irreducible cycles of $\sigma$, $\mathcal{O}_c = \overrightarrow{\prod}_{k\in c}\mathcal{O}_k$ is the ordered product of matrices with labels in the cycle $c$. Translating to our notation, this n-point correlation function can be written as:
\begin{equation}
    \overline{\mathcal{A}}^\text{open}(n,0)(\mathcal{O}_1,\cdots,\mathcal{O}_n) = \sum_{\pi\in S_n}B_{|\pi|}(\mu)e^\mu\prod_{\gamma\in\pi}\alpha^{-1}\tr(\overrightarrow{\prod}_{k\in\gamma}\mathcal{O}_k)
\end{equation}
When $\frac{1}{\alpha}$ is an integer, this correlation function is a Poisson moment formula:
\begin{equation}
    \frac{e^\mu}{\mu^{2p}}\varphi_{\mu\alpha\mathbb{E}\tr}(\lambda(\overline{\mathcal{O}^{(\psi)}_1})\cdots\lambda(\overline{\mathcal{O}^{(\psi)}_n})) = \overline{\mathcal{A}}^\text{open}(n,0)(\mathcal{O}_1,\cdots,\mathcal{O}_n)
\end{equation}
To check this, we simply plug in the Poisson moment formula for central simple operators and define $d = \frac{1}{\mu}$:
\begin{align}
    \begin{split}
        \frac{e^\mu}{\alpha^{2p}}\varphi_{\mu\alpha\mathbb{E}\tr}(\lambda(\overline{\mathcal{O}^{(\psi)}_1})\cdots\lambda(\overline{\mathcal{O}^{(\psi)}_n})) &= \frac{e^\mu}{\alpha^{2p}}\frac{1}{d^{2p}}\sum_{\pi\in\mathcal{S}_p}B_{|\pi|}(\mu\alpha d)\prod_{\gamma\in\pi}d\tr(\overrightarrow{\prod}_{k\in\gamma}\mathcal{O}_k)
        \\&=e^\mu\sum_\pi B_{|\pi|}(\mu)\prod_\gamma\alpha^{-1}\tr(\overrightarrow{\prod}_k \mathcal{O}_k)
    \end{split}
\end{align}
This reproduces the n-point correlation function of an open/closed 2D TQFT.

So far, we have considered the case where $\frac{1}{\alpha}$ is a positive integer. To consider the general case, we need to replace the standard normal random variable $\psi^a_i$ with a complex free field $\psi^a(x)$ with two-point function\footnote{Rigorously, the Dirac delta function cannot be used to construct Gaussian free fields. One has to use other kernel functions to approximate the Dirac delta function. We have omitted this limiting procedure here and focus on presenting the main calculation.}:
\begin{equation}
    \mathbb{E}(\psi^a(x)\bar{\psi}^b(y)) = \delta_{ab}\delta(x,y)
\end{equation}
With the Gaussian free fields, we consider a modified central simple operator:
\begin{equation}
    \widetilde{\overline{\mathcal{O}^{(\psi)}}} := \sum_{ab}\int_0^{1/\alpha}dx \mathcal{O}^{ab}\psi^a(x)\bar{\psi^b}(x)
\end{equation}
Note that in this modified central simple operator, we have taken $d = 1$. Then we can calculate the correlation function of the Poisson quantized modified central simple operators:
\begin{align}
    \begin{split}
        \varphi_{\mu\mathbb{E}}&(\lambda(\widetilde{\overline{\mathcal{O}^{(\psi)}_1}})\cdots\lambda(\widetilde{\overline{\mathcal{O}^{(\psi)}}_p}))=\sum_{\sigma\in\mathcal{P}_p}\prod_{A\in\sigma}\mu\mathbb{E}(\prod_{k\in A}(\sum_{a_kb_k}\int_0^{1/\alpha}dx_k \mathcal{O}^{a_kb_k}_k\psi^{a_k}(x_k)\bar{\psi}^{b_k}(x_k)))
        \\
        &=\sum_{\sigma\in\mathcal{P}_p}\prod_{A\in\sigma}\sum_{\{a_kb_k\}}(\prod_k\int_0^{1/\alpha}dx_k)(\prod_{k\in A}\mathcal{O}^{a_kb_k}_k)\mu\mathbb{E}(\prod_{k\in A}\psi^{a_k}(x_k)\bar{\psi}^{b_k}(x_k))
        \\&=\sum_{\sigma\in\mathcal{P}_p}\prod_{A\in\sigma}\sum_{\{a_kb_k\}}(\prod_k\int_0^{1/\alpha}dx_k)(\prod_k \mathcal{O}^{a_kb_k}_k)\mu(\sum_{\pi_A\in\mathcal{S}_A}\prod_{k\in A}\delta_{a_kb_{\pi_A(k)}}\delta(x_k,x_{\pi_A(k)}))
        \\
        &=\sum_{\sigma\in\mathcal{P}_p}\prod_{A\in\sigma}\sum_{\pi_A\in\mathcal{S}_A}\mu(\prod_{\gamma\in\pi_A}\tr(\overrightarrow{\prod}_{k\in A}\mathcal{O}_k))(\int_0^{1/\alpha}\cdots\int_0^{1/\alpha}dx_1dx_2\cdots dx_p\prod_k\delta(x_k,x_{\pi_A(k)}))
        \\&=\sum_{\sigma\in\mathcal{P}_p}\prod_{A\in\sigma}\sum_{\pi_A\in\mathcal{S}_A}\mu(\prod_{\gamma\in\pi_A}\tr(\overrightarrow{\prod}_{k\in A}\mathcal{O}_k))\alpha^{-\text{cycles}(\pi_A)}
        =\sum_{\sigma\in\mathcal{P}_p}\prod_{A\in\sigma}\mu(\prod_{\gamma\in\pi_A}\alpha^{-1}\tr(\overrightarrow{\prod}_{k\in \gamma}\mathcal{O}_k))
        \\&=\sum_{\pi\in\mathcal{S}_p}B_{|\pi|}(\mu)\prod_\gamma\alpha^{-1}\tr(\overrightarrow{\prod}_{k\in \gamma}\mathcal{O}_k)
    \end{split}
\end{align}
where the last equation follows from the combinatorial lemma. Therefore, the correlation function of open/closed 2D TQFT is a Poisson moment function:
\begin{equation}
    e^\mu\varphi_{\mu\mathbb{E}}(\lambda(\widetilde{\overline{\mathcal{O}_1^{(\psi)}}})\cdots\lambda(\widetilde{\overline{\mathcal{O}_n^{(\psi)}}})) = \sum_{\pi\in\mathcal{S}_n}e^\mu B_{|\pi|}(\mu)\prod_{\gamma}\alpha^{-1}\tr(\overrightarrow{\prod}_{k\in \gamma}\mathcal{O}_k)
\end{equation}
\subsection{Marolf-Maxfield's topological gravity with EOW branes}\label{subsection:MMEOWdetail}

In this section, we apply Poissonization to study the Marolf-Maxfield model of topological gravity with end-of-the-world brane \cite{MM}. In particular, we show that the correlation function of this toy model is the moment formula of Poisson quantized central simple operators.

Recall that in the Marolf-Maxfield model, the moment generating function is given by (see also Equation 3.34 \cite{MM}):
\begin{equation}\label{equation:MMkey}
    \langle\exp(u\widehat{Z} + \sum_{1\leq i,j \leq k}t_{ij}\widehat{(\psi_j, \psi_i)})\rangle = \exp(\lambda\frac{e^u}{\det(1 - t)})
\end{equation}
Here we have copied the original notations.

We use Poissonization of central simple operators with $d = 1$ to reproduce this moment generating function. Recall that when $d = 1$, we have (see also Appendix \ref{subsection:centralsimple}):
\begin{equation}
    \varphi_{\mu\mathbb{E}}(e^{\lambda(\overline{\mathcal{O}^{(\psi)}})}) = \exp(\frac{\mu}{\det(1 - \mathcal{O})}-\mu)
\end{equation}
In this case, the operator $\overline{\mathcal{O}^{(\psi)}}$ is given by: $\sum_{ab}\mathcal{O}^{ab}\psi^{ab}$. Thus, the moment generating function of the Marolf-Maxfield model is given by:
\begin{equation}
    e^\mu\varphi_{\mu\mathbb{E}}(e^{u\lambda(1) + \lambda(\overline{\mathcal{O}^{(\psi)}})}) = e^\mu\exp(\mu\mathbb{E}(e^ue^{\overline{\mathcal{O}^{(\psi)}}} - 1)) =  \exp(\mu\frac{e^u}{\det(1 - \mathcal{O})})
\end{equation}
where $u$ is a real parameter and $\lambda(1)$ is the Poisson quantization of identity operator. In the first equation, we have used the fact that $[\lambda(1), \lambda(\overline{\mathcal{O}^{(\psi)}})] = \lambda([1, \overline{\mathcal{O}^{(\psi)}}]) = 0$ and hence $e^{u\lambda(1) + \lambda(\overline{\mathcal{O}^{(\psi)}})} = e^{u\lambda(1)}e^{\lambda(\overline{\mathcal{O}^{\psi}})}$.

Expanding the moment generating function in terms of the eigenvectors of $\lambda(1)$ (the number operator), the moment generating function can be written as:
\begin{equation}
    e^\mu\sum_{n\geq 0}e^{un}\frac{e^{-\mu}\mu^n}{n!}\lb \frac{1}{\det(1 - \mathcal{O})} \rb^n = e^\mu\sum_{n\geq 0}e^{un}p_n(\mu)\left\langle\exp\lb \sum_{a,b}\mathcal{O}^{ab}\psi^{ab} \rb\right\rangle_{\lambda(1) = n}
\end{equation}
This is exactly Equation 3.37 in \cite{MM}.
\section{Universe field theory of JT gravity}\label{app:JT}

JT gravity is a simple 2D theory of gravity with the action
\begin{eqnarray}\label{JTaction2}
    I_{JT}=S_0\chi-\left(\frac{1}{2}\int_\mathcal{M}\sqrt{g}\phi(R+2)+\int_{\partial\mathcal{M}}\phi K\right)
\end{eqnarray}
where $\chi=2-2g-n$ is the Euler characteristic of a Riemann surface of genus $g$ with $n$ disks removed.

Consider the probe limit of a single asymptotic open boundary (parent universe) emitting and absorbing closed baby universes. We can choose boundary conditions that correspond to fixing the geodesic length of the closed universe to be $b$. To a time interval of length $T$ that does not scale with entropy, associate the Hilbert space $\mathcal{H}$ with $\psi(b)$ a wave-function on a single baby universe Hilbert space $\mathcal{H}=L^2(\mathbb{R}^+,bdb)$. A single absorption/emission event is an evolution in the Hilbert space $\mH_{open}\otimes\mathcal{H}$.
Since the emissions are identical, to describe multiple emissions we consider the Hilbert space $\mH_{open}\otimes \mathcal{F}_{sym}\lb\mathcal{H}\rb$ where $\mathcal{H}$ is the Fock space of closed universes:
\begin{eqnarray}
\mathcal{F}_{sym}\lb\mathcal{H}\rb\equiv \overline{\oplus_{n\geq 0}\otimes_\text{sym}^n\mathcal{H}}=   \overline{\mathbb{C}\ket{\Omega}\oplus \mathcal{H}\oplus \text{sym}(\mathcal{H}\otimes \mathcal{H})\otimes \cdots}
\end{eqnarray}
In the absence of matter fields, the Hilbert space of the open boundary $\mH_{open}$ furnishes a representation of the commutative algebra of functions of Hamiltonian $f(H)=\int dE \rho_0(E) f(E)$ where $\rho_0(E)=\frac{e^S}{2\pi^2}\sinh(2\pi\sqrt{2E})$ is the eigenvalue density. Another convenient basis is $l$-basis, where $l$ is the renormalized bulk geodesic length between the left and right boundaries in $\mH_{open}$. They are related to the bulk energy eigenbasis using
\begin{eqnarray}
    &&\psi_E(l)=\braket{l}{E}=\frac{u}{2}\int_{-\infty}^\infty dx\: e^{-u\cosh x}e^{-2is x}\nn\\
    &&s=\sqrt{2E},\qquad u=4e^{-l/2}\ .
\end{eqnarray}
The resolution of identity in the $l$-basis takes the form
\begin{eqnarray}
  &&e^{-S} \int e^{l}dl \braket{E'}{l}\braket{l}{E}=\frac{1}{\rho_0(E)}\delta(E-E')\\
&&    e^{-S}\int_{-\infty}^\infty e^l dl\: \ket{l}\bra{l}=1\ .
\end{eqnarray}
In the absence of topology fluctuations, the disk partition function captures the transition amplitude between thermofield doubles of different temperatures:
\begin{eqnarray}
    \braket{TFD_\beta}{TFD_{\beta'}}=\int dE \rho_0(E) e^{-(\beta+\beta')E/2}\ .
\end{eqnarray}
We can compute the disk amplitude for a transition to $\ket{l}_{open}$ after time $t$:
\begin{eqnarray}
    &&\bra{l}e^{-iHT/2}\ket{TFD_\beta}=\int_0^\infty dE \rho_0(E) e^{-(\beta/2+iT)E}\psi_E(l)\ .
\end{eqnarray}
    
The amplitude for an event that corresponds to the absorption of a single baby universe is 
$\ket{l,b}\equiv \ket{l}_{open}\otimes \ket{b}$ with $\ket{b}\in \mathcal{H}$:
\begin{eqnarray}
    &&\bra{l,b}e^{-iHT/2}\ket{TFD_\beta}=\int_0^\infty dE \rho_0(E) e^{-(\beta/2+iT)E}\psi_E(l) \lb\frac{\cos (b \sqrt{2E})}{\pi\rho_0(E)\sqrt{2E}} \rb
\end{eqnarray}
In the probe limit of a single open boundary, we can represent the emission of a baby universe with the action of the operator $\mO(b)$ below, which is diagonal in the energy eigenbasis of the parent universe \footnote{Note that the vector $\ket{E,b}=\mO(b)\ket{E}$ is not properly normalized. Its norm is given by the two-point function of $\mO(b)$, which diverges in the limit $E\to 0$. However, the operator $\mO(b)$ is still densely defined in the Hilbert space of the open boundary \cite{penington2023algebras}}. See Figure \ref{fig7} (a):
\begin{eqnarray}\label{Ob}
    &&\mO(b)=\int dE \rho_0(E) \lb\frac{\cos (b \sqrt{2E})}{\pi\rho_0(E)\sqrt{2E}} \rb\ket{E}\bra{E}
\end{eqnarray}
The amplitude for the emission of a baby universe of length $b$ in a thermofield double state is
\begin{eqnarray}
  Z_{trumpet}(\beta+iT,b)&&\equiv \bra{TFD_{\beta/2}}e^{-iH_{bulk}T/2}\ket{TFD_{\beta/2},b}\nn\\
  &&=\int_0^\infty dE\:\frac{\cos(b\sqrt{2E})}{\pi\sqrt{2E}}e^{-(\beta+iT)E}\ .
\end{eqnarray}
Since $\rho_0(E)\sim e^{S}$, the amplitude for the emission of a baby universe of circumference $b$ is exponentially suppressed. Therefore, we can ignore the creation of multiple baby universes in times of order $T=O(1)$. The operators $\mO(b)$ generate a commutative algebra of observables on the parent Hilbert space.

\begin{figure}[t]
    \centering
    \includegraphics[width=0.8\linewidth]{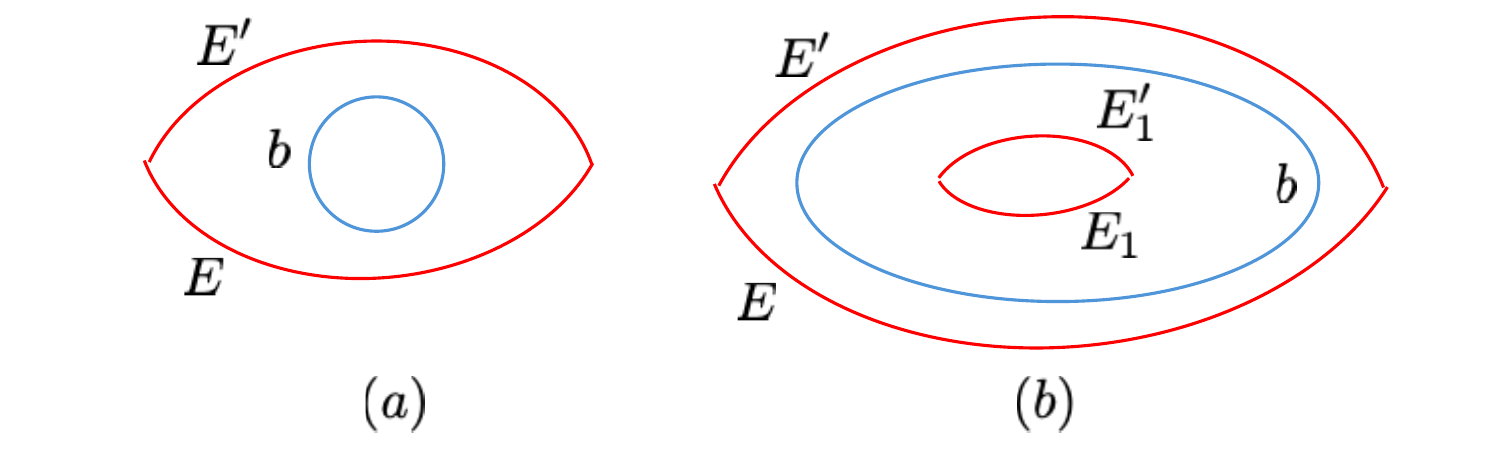}
    \caption{\small{(a) The operator $\mO(b)$ in JT gravity that creates an baby universe with length $b$. (b) The operator that creates an asymptotic boundary $Z(E_1,E'_1)$ can be decomposed as a second-order effect with emission and absorption of a baby universe \cite{saad2019late}.}}
    \label{fig7}
\end{figure}



 Going beyond the probe limit in JT, one can compute the amplitude for insertion of a closed asymptotic boundary as the operator $Z(E_1,E_1')$ in the one-boundary Hilbert space; see Figure \ref{fig7} (b):
 \begin{eqnarray}\label{twoptfunction}
     \bra{E'}Z(E'_1,E_1)\ket{E}&&=\int bdb\braket{E',b}{E}\braket{E_1',b}{E_1}\nn\\
     &&=\frac{\delta(E-E')\delta(E_1-E_1')}{\rho(E)\rho(E_1)}\frac{
     (E+E_1)
     }{4\pi^2\sqrt{EE_1}(E-E_1)^2}
 \end{eqnarray}
 In the limit $E\to E_1$, we find the level repulsion term $1/(E-E_1)^2$ often attributed to the ramp. More generally, acting on a single boundary with the operator $Z(E'_1,E_1)\cdots Z(E'_p,E_p)$ gives
\begin{eqnarray}\label{multiptJT}
    \bra{E'}Z(E'_1,E_1)\cdots Z(E'_p,E_p)\ket{E}=\frac{\delta(E-E')}{(2\pi)^{p}\rho(E)^p\sqrt{E}}\prod_{i=1}^p\frac{\delta(E_i-E'_i)(E+E_i)}{\rho(E_i)\sqrt{E_i}(E-E_i)^2}
\end{eqnarray}
It is convenient to define $z=\sqrt{E}$, and label $Z(E'_i,E_i)$ by a single variable $E_i$. Then, the factor in (\ref{twoptfunction}) can be written as
\begin{eqnarray}
    \frac{(E+E_1)}{(E-E_1)^2}=\frac{1}{2}\lb \frac{1}{(z-z_1)^2}+\frac{1}{(z+z_1)^2}\rb\ .
\end{eqnarray}
This makes it clear that (\ref{multiptJT}) corresponds to the correlators of a free $\mathbb{Z}_2$-twisted chiral bosons $\phi_+(z_1)$ and $\phi_-(z_2)$ \cite{post2022universe}. The universe field theory of JT (Kodaira-Spencer theory) also includes a cubic interaction term. The correlators of the interacting theory can be repackaged as the correlators of the free theory with an extra local operator insertion on the twisted vacuum\footnote{For more details, see Appendix C of a  \cite{post2022universe}.}.

\section{Poissonization and Late-Time Plateau in Spectral Form Factors}\label{app:plateau}
In this Appendix, we perform a simple calculation to schematically show how a late-time plateau can be captured in a Poisson algebra. First, we review the mathematical definition of spectral form factors and provide a simple generalization to the case where the underlying von Neumann algebra is semi-finite. Then, we apply our definition to the case of Poisson algebras and perform an explicit calculation inspired by pure JT gravity. We analyze the result and show that a late-time plateau appears in this calculation. In addition, we show that the global minimum of this spectral form factor is less than the late-time plateau, while the initial value of this spectral form factor is the global maximum, and the late-time plateau is less than this global maximum\footnote{Here minimum and maximum refer to the spectral form factor as a function of time}. These simple observations (together with the smoothness of the spectral form factor as a function of time) indicate the existence of a slope-ramp-plateau, which is believed to be characteristic of chaos.

Note that our discussion here differs from the usual discussion of the slope-ramp-plateau behavior. The slope-ramp-plateau behavior is usually studied in the context of ensemble theory, where a spectral form factor is calculated for each individual Hamiltonian, and the slope-ramp-plateau behavior appears only after taking the ensemble average of the entire family of spectral form factors. Here, our calculation is inspired by JT gravity, where the theory is effectively described by the double scaling limit of random matrices. The effective Hamiltonian is described in the form of a spectral density function. Instead of 'third quantizing' the random Hamiltonians before taking the ensemble average, alternatively, we consider 'third quantizing' the effective Hamiltonian using its spectral density and calculate a single spectral form factor. The randomness is effectively summarized in the spectral density in the double scaling limit.

\subsection{Spectral Form Factor in Semi-finite von Neumann Algebras}\label{subsubsection:spectralformfactorvna}
Given an $N\times N$ Hamiltonian (or generally any self-adjoint matrix) $H$ with eigenvalues $\lambda_1,...,\lambda_N$, the spectral form factor can be defined as\cite{sff-rm}:
\begin{equation}
    \textbf{SSF}(t) := \frac{1}{N^2}\sum_{1\leq i,j\leq N}e^{it(\lambda_i - \lambda_j)} = |\tau(e^{itH})|^2
\end{equation}
where $\tau := \frac{1}{N}\Tr$ is the normalized trace. In case of random matrices, we usually also consider the expectation value of the spectral form factor \cite{sff-rm}. In a typical chaotic system, the spectral form factor is expected to exhibit the so-called slope-ramp-plateau behavior. As time $t$ grows to $\infty$, the spectral form factor is expected to first decrease (the slope) to a global minimum, then increase (the ramp), and finally settle to a limiting value (the plateau). We will reproduce this behavior in a Poisson algebra.

The definition of a spectral form factor can be easily extended to more general semi-finite von Neumann algebras.
\begin{definition}
    Given a semi-finite von Neumann algebra $N$ (not to be confused with the size of matrices $N$) and the corresponding tracial state (or weight) $\tau$. Consider a Hamiltonian $H$ affiliated with $N$, the spectral form factor of $H$ is given by:
    \begin{equation}
        \textbf{SFF}_H(t) := |\tau(e^{itH})|^2\ .
    \end{equation}
\end{definition}

As an example of this definition, we consider the familiar case of random matrices in the large $N$ limit. As the size of the matrix grows to infinity, the discrete spectrum converges to a continuous spectral density. The commutative von Neumann algebra generated by the Hamiltonian has a normal weight (or state if normalized) defined by the spectral density $\tau(f(H)) := \int_0^\infty dE\rho(E)f(E)$ where $\rho(E)$ is the spectral density. Then the corresponding spectral form factor can be written as:
\begin{equation}
    \textbf{SSF}_H(t) = |\int_0^\infty dE\rho(E)e^{itH}|^2
\end{equation}
\subsection{Spectral Form Factor in Poisson Algebra}
Consider a semi-finite von Neumann algebra $N$ along with a tracial state (or weight) $\tau$. The corresponding Poisson algebra $\mathbb{P}_\tau N$ is also semi-finite, and the Poisson state $\varphi_\tau$ is a tracial state. Given a self-adjoint operator affiliated $H$ with $N$, the Poisson quantized operator $\lambda(H)$ is self-adjoint and affiliated with the Poisson algebra $\mathbb{P}_\tau N$. We can consider the spectral form factor:
\begin{equation}
    \textbf{SSF}_\text{Poiss}(t) = |\varphi_\tau(e^{itH})|^2
\end{equation}
As a concrete calculation, we consider the Poissonization of the commutative von Neumann algebra generated by the Hamiltonian in JT gravity. The corresponding weight is given by the regulated (approximate) spectral density $\tau(f(H)) = \int_0^\infty dE e^{-\beta E}\rho_0(E)f(E)$ where $\rho_0(E) = \frac{e^{S_0}}{4\pi^2}\sinh(2\pi\sqrt{E})$ and $\beta > 0$. The spectral density contains only leading order (genus zero) contribution. The normalization is necessary for technical reasons so that the normalized weight $\tau$ is finite (i.e. $\tau(1) < \infty$). We calculate the spectral form factor:
\begin{equation}
    \textbf{SSF}_\text{Poiss}(t) = |\varphi_\tau(e^{it\lambda(H)})|^2
\end{equation}
The calculation is straightforward:
\begin{align}
    \begin{split}
        \textbf{SSF}_\text{Poiss}(t) &= |\varphi_\tau(e^{it\lambda(H)})|^2 
        =|\exp(\tau(e^{itH} - 1))|^2
        \\&
        =|\exp(\int_0^\infty dE e^{-\beta E}\rho_0(E)(e^{itE} - 1))|^2
        \\&
        =\exp(\int_0^\infty dE e^{-\beta E}\rho_0(E)(2\cos(tE) - 2))
        \\&
        =\exp\lb 2\pi^{\frac{3}{2}}\lb \frac{e^{\frac{\pi^2\beta}{\beta^2 + t^2}}\cos(\frac{3}{2}\arctan(\frac{t}{\beta}) + \frac{\pi^2 t}{\beta^2 + t^2})}{(\beta^2 + t^2)^{\frac{3}{4}}} - \beta^{-\frac{3}{2}}e^{\frac{\pi^2}{\beta}} \rb \rb
    \end{split}
\end{align}
The time-dependence of this spectral form factor is explicit. We have the following long-time limit:
\begin{equation}
    \lim_{t\rightarrow\infty}\textbf{SSF}_\text{Poiss}(t) = \exp(-2\pi^{\frac{3}{2}}\beta^{-\frac{3}{2}}e^{\frac{\pi^2}{\beta}})
\end{equation}
At the other end, we have:
\begin{equation}
    \textbf{SSF}_\text{Poiss}(0) = 1 > \textbf{SSF}_\text{Poiss}(\infty)
\end{equation}
To see that a global minimum exists for this spectral form factor and that this global minimum is less than the limiting value, we analyze the time-dependent function:
\begin{equation}
    A(t) := \frac{e^{\frac{\pi^2\beta}{(\beta^2+t^2)}}}{(\beta^2 + t^2)^{\frac{3}{4}}}\cos(\frac{3}{2}\arctan(\frac{t}{\beta}) + \frac{\pi^2 t}{\beta^2 + t^2})
\end{equation}
The prefactor $\frac{e^{\frac{\pi^2\beta}{\beta^2 + t^2}}}{(\beta^2 + t^2)^{\frac{3}{4}}}$ is strictly positive and is monotonically decreasing, while the cosine factor is oscillating. A simple plot of this time-dependent function shows the slope-ramp-plateau behavior.
\begin{figure}[H]
    \centering
    \includegraphics[width = \linewidth]{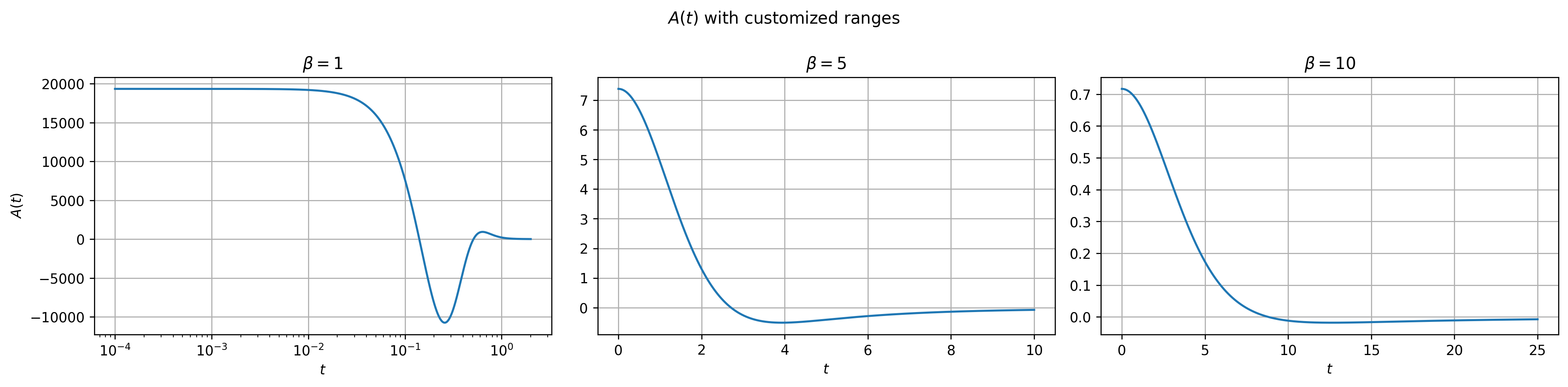}
    \caption{\footnotesize{The time-dependent factor for $\beta =1 , 5, 10$}}
\end{figure}
For $\beta = 1$, the plot is shown on the log scale to better demonstrate the details, while for $\beta = 5, 10$ the plots are shown on a normal scale. Although we plotted only for the time-dependent expononent function $A(t)$, the full spectral form factor is $\exp(2\pi^{\frac{3}{2}}(A(t) - \beta^{-\frac{3}{2}}e^{\frac{\pi^2}{\beta}}))$ and hence has the same qualitative behavior as $A(t)$. Note that for small $\beta$ (e.g. $\beta = 1$), the oscillating behavior of the cosine factor shows up for small $t$, creating a local maximum that is larger than the limiting plateau. But for large enough $\beta$ (e.g. $\beta = 5, 10$), the oscillating behavior of the cosine factor is washed out by the decay factor $\frac{e^{\frac{\pi^2\beta}{(\beta^2+t^2)}}}{(\beta^2 + t^2)^{\frac{3}{4}}}$ leaving only a global minimum with no local maximum that is above the limiting plateau. Because the spectral form factor is a constant times the exponential of $A(t)$, the same slope-ramp-plateau behavior is inherited by the spectral form factor.

\section{Set partitions and distributions}\label{app:setpartitions}

Throughout this work, the key combinatorial object is the set of partitions $\mathcal{P}_n$ on the finite set $[n] = \{1,2,\cdots,n\}$. In this subsection, we collect several useful facts about $\mathcal{P}_n$.

First a \textit{partition} $\sigma$ of $[n]$ is a grouping of the elements of $[n]$ into non-empty disjoint subsets:
\begin{equation}
    \sigma := \{A_1,A_2,\cdots,A_m: A_i \subset [n]\text{ , }A_i\cap A_j = \emptyset\}
\end{equation}
The cardinality of $\sigma$ is at least $1$ and at most $n$. For the unique partition with only $1$ element, we call it the trivial partition:
\begin{equation}
    \sigma_\text{trivial} = \{\{1,2,\cdots,n\}\}
\end{equation}
For the unique partition with $n$ elements, we call it the singleton partition:
\begin{equation}
    \sigma_\text{singleton} = \{\{1\},\cdots,\{n\}\}
\end{equation}
The number of partitions of $[n]$ with cardinality $m$ is \textit{the Stirling number of second kind} $S(n,m)$:
\begin{equation}
    S(n,m) = \frac{1}{m!}\sum_{p_1 + \cdots + p_m = n}\frac{n!}{p_1!\cdots p_m!}
\end{equation}
The total number of partitions of $[n]$ is \textit{the Bell number} $B_n$:
\begin{equation}
    B_n = \sum_{1\leq m\leq n}S(n,m)
\end{equation}
The generating function of the Stirling number of the second kind is \textit{the Touchard polynomial} $T_n(x)$:
\begin{equation}
    T_n(x) = \sum_{1\leq m\leq n}S(n,m)x^m
\end{equation}
From the definitions, the Bell number is simply the Touchard polynomial when $x = 1$:
\begin{equation}
    B_n = T_n(1)
\end{equation}
The moments of a Poisson distribution with intensity $\mu$ are given by the Touchard polynomials:
\begin{equation}
    \mathbb{E}(X^n) = \sum_{k\geq 0}k^n\frac{e^{-\mu}\mu^k}{k!}=  T_n(\mu)
\end{equation}
where $X$ is a Poisson random variable with intensity $\mu$. The additivity of Poisson intensities can be written in terms of a recursion relation between Touchard polynomials:
\begin{equation}
    T_n(\mu_1 + \mu_2) = \sum_{0\leq k\leq n}\binom{n}{k}T_k(\mu_1)T_{n - k}(\mu_2)
\end{equation}
where by convention $T_0(\mu) = 1$. And the moment generating function of a Poisson distribution can be written as:
\begin{equation}
    \exp(\mu(e^{\eta} - 1)) = \mathbb{E}(e^{\eta X}) = \sum_{n\geq 0}\frac{\eta^n}{n!}T_n(\mu) = \sum_{n\geq 0}\sum_{1\leq m\leq n}\frac{\eta^n}{n!}S(n,m)\mu^m
\end{equation}
where $X$ is a Poisson random variable with intensity $\mu$.

The Stirling numbers and the Touchard polynomials have shown up in various formulae involving normal ordering of number operators \cite{mansour2016commutation}. For example, we have:
\begin{equation}
    \widehat{N}^n = \sum_{1\leq m \leq n}S(n,m):\widehat{N}^m:
\end{equation}
where the number operators on the right hand side $:\widehat{N}^m:$ are normally ordered. We can regard the right-hand side as a normal-ordered operator:
\begin{equation}
    :T_n(\widehat{N}): = \sum_{1\leq m\leq n}S(n,m):\widehat{N}^m:
\end{equation}
The exponential of the number operator can be written as:
\begin{equation}
    \exp(\eta\widehat{N}) = \sum_{n\geq 0}\frac{\eta^n}{n!}\widehat{N}^n = \sum_{n\geq0}\sum_{1\leq m\leq n}\frac{\eta^n}{n!}S(n,m):\widehat{N}^m:   =:\exp((e^{\eta} - 1)\widehat{N}):
\end{equation}
\section{Universality of Poisson statistics}\label{app:universalityPoisson}

In this section, we show that the Poisson limit theorem is robust against perturbations. Specifically, we show that adding a mild correlation (thereby breaking the mutual independence condition) does not break the Poisson limit theorem. A mild correlation results in a scaled Poisson intensity. We also show that if the single event probability has a weak dependence on time, the Poisson limit theorem continues to hold. The resulting Poisson intensity would be time-dependent (i.e., a Poisson random process).
\subsection{Robustness against weak correlation}
Here, we show that the Poisson limit theorem is robust against mild mutual dependencies through a simplified model of correlation. This is by no means a complete proof, which requires a more detailed analysis, generalizing our example. This is beyond the scope of this work.

First, we present an alternative way of understanding the Poisson limit theorem based on a simple spin chain model. Consider a classical 1D spin chain with $M$ spins $\{s_i\}_{1\leq i\leq M}$. Each spin takes a value in $\{\pm 1\}$. Consider a single-body Hamiltonian: $H_0 := \sum_i s_i = \sum_i (2t_i - 1)$ where $t_i\in\{0,1\}$. If we take the inverse temperature to be:
\begin{equation}
    \beta_M := -\frac{1}{2}\log\frac{\frac{M_\text{up}}{M}}{1 - \frac{M_\text{up}}{M}} = -\frac{1}{2}\log\frac{M_\text{up}}{M - M_\text{up}}
\end{equation}
where $0 < M_\text{up} < M$ is an integer \textit{independent} of $M$. Notice that the inverse temperature is dependent on the length of the spin chain. As the size $M\rightarrow\infty$, $\beta_M\rightarrow\infty$, leading to the low temperature limit. In this case, we can calculate the expectation value:
\begin{align}
    \begin{split}
        \tr(\sum_i\frac{s_i + 1}{2}\frac{e^{-\beta_M H_0}}{Z_0}) &= \tr(\sum_i t_i \frac{e^{-2\beta_M \sum_i t_i}e^{M\beta_M}}{Z_0}) = \sum_i (1\times\frac{e^{-2\beta_M}}{Z_i} + 0\times\frac{1}{Z_i})
    \end{split}
\end{align}
where $Z_0$ is the partition function: $Z_0 = tr(e^{-\beta_M H_0}) = e^{M\beta_M}\tr(e^{-2\beta_M\sum_it_i})$ and $Z_i := \tr(e^{-2\beta_M t_i})$. Because $\frac{e^{-2\beta_M}}{Z_i} = \frac{\frac{M_\text{up}}{M - M_\text{up}}}{\frac{M_\text{up}}{M - M_\text{up}} + 1} = \frac{M_\text{up}}{M}$, it is clear that each $t_i$ is a Bernoulli random variable with success probabiliy $\frac{M_\text{up}}{M}$. The Poisson limit is achieved by considering the expectation of $\sum_it_i$ and taking the infinite chain limit $M\rightarrow\infty$. We reiterate that the inverse temperature goes to infinity with $M$ but the parameter $M_\text{up}$ is fixed:
\begin{align}
    \begin{split}
        \lim_{M\rightarrow\infty}\tr(\sum_it_i \frac{e^{-\beta_M H_0}}{Z_0}) &= \lim_{M\rightarrow\infty}\sum_{0\leq k\leq M}k\binom{M}{k}(\frac{M_\text{up}}{M})^k(1 - \frac{M_\text{up}}{M})^{M - k}
    \end{split}
\end{align}
Using this simple spin chain model, we now consider the effects of local correlations. For simplicity, we only consider the nearest-neighbor interaction term:
\begin{equation}
    H = H_0 - J_M\sum_i s_is_{i + 1}
\end{equation}
where $J_M  = O(\frac{1}{\log M}) > 0$. The interaction term gives more weight to neighbors with the same configuration, making such a configuration more likely.

Notice that this coupling constant $J_M$ depends on the system size. This is crucial to retain the Poisson limit, as we will see shortly. For large $M$, this coupling constant is vanishingly small, although the interaction term $J_M\sum_i s_is_{i + 1}$ may still be large for certain configurations. 

In this case, we can calculate:
\begin{equation}
    \tr(\sum_i t_i\frac{e^{-\beta_M H}}{Z})= \tr(\sum_i t_i\frac{e^{-(2\beta_M + 4\beta_M J)\sum_i t_i + 4 \beta_M J_M\sum_i t_i t_{i + 1} + (\beta_M M  + \beta_M M J_M)}}{Z})
\end{equation}
Define the modified partition functions:
\begin{equation}
    Z':=\tr(e^{-(2\beta_M + 4\beta_M J_M)\sum_i t_i + 4\beta_M J_M\sum_i t_it_{i + 1}})\text{ , }Z_\text{single} := \tr(e^{-(2\beta_M + 4\beta_M J_M)\sum_i t_i})
\end{equation}
For each fixed $M$, we consider the interaction term as a perturbation:
\begin{align}
    \begin{split}
        \tr(\sum_i t_i \frac{e^{-\beta_M H}}{Z}) &= \tr(\sum_i t_i \frac{e^{-(2\beta_M + 4\beta_M J_M)\sum_i t_i + 4\beta_M J_M\sum_i t_it_{i + 1}}}{Z'})
        \\&
        =\frac{Z_\text{single}}{Z'}\tr(\sum_i t_i\sum_{k\geq 0}\frac{(4\beta_M J_M)^k}{k!}(\sum_i t_it_{i + 1})^k\frac{e^{-(2\beta_M + 4\beta_M J_M)\sum_i t_i}}{Z_\text{single}})
    \end{split}
\end{align}
The leading order perturbation is given by:
\begin{align}
    \begin{split}
        (4\beta_M J_M)&\tr((\sum_i t_i)(\sum_j t_jt_{j + 1})\frac{e^{-(2\beta_M + 4\beta_M
        J_M)\sum_i t_i}}{Z_\text{single}}) 
        \\& = (4\beta_M J_M)\sum_j\tr((\sum_i t_i)t_jt_{j + 1}\frac{e^{-(2\beta_M + 4\beta_M J_M)\sum_it_i}}{Z_\text{single}})
    \end{split}
\end{align}
For each fixed $j$, the term is non-zero only when both $t_j$ and $t_{j + 1}$ are $1$. In addition, the distribution provided by $\frac{e^{-(2\beta_M + 4\beta_M J_M)\sum_it_i}}{Z_\text{single}}$ is a modification of the original multinomial distribution. Instead of success probability given by $\frac{e^{-2\beta_M}}{e^{-2\beta_M} + 1}$, the modified distribution has success probability given by $\frac{e^{-(2\beta_M + 4\beta_M J_M)}}{e^{-(2\beta_M + 4\beta_M J_M)} + 1}$. Denote the modified success probability as $\widetilde{p}$:
\begin{equation}
    \widetilde{p} = \frac{M^{1 + 2J_M}_\text{up}}{M^{1+ 2J_M}}
\end{equation}
For $J_M = O(\frac{1}{\log M})$, the modified probability is finite perturbation of the original probability $\frac{M_\text{up}}{M}$. In particular, $\lim_{M\rightarrow \infty}M\widetilde{p} < \infty$. This ensures a Poisson limit still exists.

Now we can calculate:
\begin{align}
    \begin{split}
        \tr((\sum_it_i)t_jt_{j + 1}\frac{e^{-(2\beta_M + 4\beta_M J_M)\sum_i t_i}}{Z_\text{single}}) &= \sum_{2\leq k\leq M}\widetilde{p}^2k\binom{M - 2}{k - 2}\widetilde{p}^{k - 2}(1 - \widetilde{p})^{M - k}
    \end{split}
\end{align}
Hence for each fixed $2\leq k\leq M$, the leading order perturbation leads to the following modification to the original multinomial distribution:
\begin{equation}
    \Delta\text{Prob}(\sum_i t_i = k) = M\frac{Z_\text{single}}{Z'}(4\beta_M J_M)\binom{M-2}{k-2}\widetilde{p}^k(1 - \widetilde{p})^{M - k}
\end{equation}
where the factor $M$ comes from the fact that there are $M$ interaction terms. Note $\binom{M - 2}{k-2} = \binom{M}{k}\frac{k(k -1)}{M(M - 1)}$. In addition, by definition we have $4\beta_M = O(\log M)$ and $Z' > Z_\text{single}$. Combining these considerations, we have:
\begin{equation}
    \lim_{M\rightarrow\infty}\Delta\text{Prob}(\sum_it_i = k) < \lim_{M\rightarrow\infty}C_k\frac{J_M\log M}{M}\binom{M}{k}\widetilde{p}^k(1-\widetilde{p})^{M-k}
\end{equation}
where $C_k$ is a positive constant depending only on $k$. We can calculate the limit on the right-hand side explicitly:
\begin{align}
    \begin{split}
        \lim_{M\rightarrow\infty}&\frac{J_M\log M}{M}\binom{M}{k}\widetilde{p}^k(1 - \widetilde{p})^{M - k} = \lim_{M\rightarrow\infty}\frac{J_M\log M}{M}\binom{M}{k}\frac{1}{M^k}(M\widetilde{p})^k(1-\widetilde{p})^{M - k}
        \\&=(\lim_{M\rightarrow\infty}\frac{J_M\log M}{M})(\lim_{M\rightarrow\infty}\binom{M}{k}\frac{1}{M^k})(\lim_{M\rightarrow\infty} (M\widetilde{p})^k)(\lim_{M\rightarrow\infty}(1 - \widetilde{p})^{M - k})
        \\&=0
    \end{split}
\end{align}
Therefore, the leading order modification to the limiting probability distribution is $0$. We can generalize this calculation to all orders of perturbations. For example, the second-order perturbation is bounded above by:
\begin{equation}
    (4\beta_M J_M)^2(\binom{M}{2}\binom{M-4}{k-4} + M\binom{M-2}{k-2})\widetilde{p}^k(1-\widetilde{p})^{M - k}
\end{equation}
where the first term comes from contributions $\sum_{j\neq j'}t_jt_{j +1}t_{j'}t_{j' + 1}$ and the second term comes from contributions $\sum_j(t_jt_{j+1})^2$. Because we have:
\begin{equation}
    \binom{M}{2}\binom{M-4}{k-4} + M\binom{M-2}{k-2} = \binom{M}{k}(O(\frac{M^2}{M^4}) + O(\frac{M}{M^2}))
\end{equation}
Hence, the second-order perturbation vanishes as $M\rightarrow\infty$.

Generally, the $k$-th order perturbation is bounded above by:
\begin{equation}
    (4\beta_M J_M)^k(\sum_{1\leq i\leq k}\binom{M}{i}\binom{M- 2i}{k - 2i})\widetilde{p}^k(1 - \widetilde{p})^{M - k}
\end{equation}
It follows from the same calculation that such a perturbation vanishes as $M\rightarrow\infty$.

Now, back to the $0$-th order term, this gives us a Poisson distribution with a modified intensity as $M\rightarrow\infty$:
\begin{align}
    \begin{split}
        \lim_{M\rightarrow\infty}\binom{M}{k}\widetilde{p}^k(1 - \widetilde{p})^{M - k} = \frac{e^{-\lim_{M\rightarrow\infty}M\widetilde{p}}(\lim_{M\rightarrow\infty}M\widetilde{p})^k}{k!}
    \end{split}
\end{align}
The limit $\lim_{M\rightarrow \infty}M\widetilde{p} = \lim_{M\rightarrow\infty}(4\beta_M J_M)M_\text{up}$. Compared with the original intensity $M_\text{up}$, the modification comes as a multiplicative constant $lim_{M\rightarrow\infty}4\beta_M J_M$. Since $\beta_M = O(\log M)$ and $J_M = O(\frac{1}{\log M})$, the product $\beta_M J_M = O(1)$ is a constant.
\subsection{Robustness against weak time-dependence - Poisson Random Process}
So far, we have focused on the Poisson limit theorem where the intensity is independent of time. The same Poisson limit holds if the intensity is dependent on time, and the limiting law would be a Poisson random process.

Recall that in the time-independent Bernoulli approximation, we assumed that the success probability of each Bernoulli event is proportional to the length of the discretized time interval. And the Poisson distribution arises from the limit:
\begin{equation}
    \lim_{M\rightarrow\infty}\binom{M}{k}(\lambda\frac{T}{M})^k(1 - \lambda\frac{T}{M})^{M-k} = \frac{e^{-\lambda T}(\lambda T)^k}{k!}
\end{equation}
where each time interval has length $\frac{T}{M}$ and the total length of time is $T$. Switching to the time-dependent case, the intensity parameter $\lambda$ now depends on time. We uniformly approximate the intensity function with an increasing sequence of step functions:
\begin{equation}
    \lim_{\alpha}\lambda_\alpha = \lambda\text{ , }\lambda_\alpha(t) = \sum_{i}\lambda^\alpha_i\chi_{I_i^\alpha}(t)
\end{equation}
where $I_i^\alpha\subset[0,T]$ is a finite time interval, $\chi_{I^\alpha_i}$ is the characteristic function of $I_i^\alpha$, and $\lambda_i^\alpha\in\mathbb{R}_+$ is a positive real number. Such an approximating sequence of simple functions always exists. 

For a given $\lambda_\alpha$, within each interval $I_i^\alpha$, we consider a Bernoulli approximation by dividing the interval $I_i^\alpha$ into equal-sized subintervals. When the number of divisions approaches infinity, the Poisson limit theorem shows that the limiting distribution is a Poisson distribution with intensity $\lambda_i^\alpha|I_i^\alpha|$ where $|I_i^\alpha|$ is the length of the subinterval $I_i^\alpha$. Each interval corresponds to a Poisson random variable $X_i^\alpha$. By the additivity of Poisson intensity, the distribution of the sum $\sum_i X_i^\alpha$ follows a Poisson distribution with intensity $\sum_i\lambda_i^\alpha |I_i^\alpha| = \int_0^T dt\lambda_\alpha(t)$. More precisely, we have:
\begin{equation}
    \text{Prob}(\sum_i X^\alpha_i = k) = \frac{e^{-\int_0^T dt\lambda_\alpha(t)}(\int_0^Tdt\lambda_\alpha(t))^k}{k!}
\end{equation}Taking the limit over $\alpha$, the monotone convergence theorem shows that:
\begin{equation}
    \lim_\alpha\frac{e^{-\int_0^T dt\lambda_\alpha(t)}(\int_0^Tdt\lambda_\alpha(t))^k}{k!} = \frac{e^{-\int_0^Tdt\lambda(t)}(\int_0^Tdt\lambda(t))^k}{k!}
\end{equation}
Hence, the limit is a Poisson random process with intensity given by $\lambda(t)$.

\section{Invitation to the mathematical theory of Poissonization}\label{app:math}

In this section, we introduce the mathematical theory of Poissonization. We focus on stating the results and proving only essential theorems.

Poissonization can be understood as a canonical quantization procedure (i.e. a functor) analogous to the second quantization. The input data is a physical description of a one-particle quantum system, and the output data is a physical description of a many-body quantum system. While the second quantization consumes the one-particle Hilbert space as the input, Poissonization consumes the following duplet of data:
\begin{itemize}
    \item An algebra of observables on the one-particle system;
    \item A reference state on the one-particle system\footnote{Here, a state is a normal faithful positive normalized functional. In general, we can drop the requirement of normalization and consider a general normal faithful \textit{weight}. For Poissonization, this weight can be finite or semi-finite \cite{T}.}
\end{itemize}
As an example, the duplet of data can be given by $\lb\mathbb{B}(\mathcal{H}), \rho_\beta\rb$ where $\mathcal{H}$ is the one-particle Hilbert space and $\rho_\beta$ is a Gibbs state associated with a Hamiltonian on the one-particle system.

The output of the second quantization is a Weyl algebra acting on a Fock space and a distinguished vacuum state. Analogously, Poissonization outputs an algebra of observables and a distinguished reference state. This algebra is called the \textit{Poisson algebra} and the state is called the \textit{Poisson state} \footnote{Both notions are canonically tied to the input data. Hence, more precisely, we should call the output the Poisson algebra and the Poisson state associated with the specific input data.}.

The second quantization also specifies a quantization map that takes in a vector in the one-particle Hilbert space and constructs a Weyl unitary. Analogously, Poissonization also defines such a quantization map (the \textit{Poisson quantization map}) which takes in an operator on the one-particle system and constructs an (affiliated\cite{T}) operator on the many-body system. The Poisson quantization map preserves commutation relations and admits a concrete description that forms the basis of many of the concrete calculations.

There exist several different but equivalent constructions of Poissonization. In Section \ref{sec:bipartitecoherent}, we have already seen how to define Poisson algebra and Poisson state of matrix algebra using \textit{one-sided number-preserving} operators. In this section, we will present two other constructions. All constructions are equivalent in the sense that they all construct an isomorphic Poisson algebra and Poisson state given the same input data. Mathematically, this equivalence is a non-trivial fact that requires a noncommutative generalization of Hamburger's moment theorem and an estimation of the moment growth of the Poisson moment formula (see also Appendix \ref{appendix:hamburger}) \cite{MAW}. 

In addition to the constructions, we will also briefly discuss the functoriality, modular theory, and other useful properties of Poissonization. For simplicity and concreteness, we will mostly restrict our attention to the Poissonization of matrix algebras. There are two major drawbacks of restricting to matrix algebras:
\begin{enumerate}
    \item Poissonization of matrix algebras always has a non-trivial center. This raises the question of whether Poisson algebra can be a factor. We will explain in detail how to construct the Poissonization of a type $I_\infty$ factor and show that the resulting Poisson algebra is the hyperfinite type $II_1$ factor.
    \item Many explicit formulae for Poissonization of matrix algebras rely on the fact that the input state (or weight) is finite. This raises the question of whether Poisson algebra can be constructed using unbounded weights. This is important in applications because it is frequently the case that the one-particle modular automorphism is given by some unbounded weight. Poissonization of type $I_\infty$ factor already provides a way to construct a Poisson algebra using unbounded weight (the tracial weight in this case). In this case, the construction proceeds via a finite-dimensional approximation. More general constructions are available. We comment on this generalization without going into the mathematical details. 
\end{enumerate}
To fix notation, we will denote the one-particle algebra of observables as $\mathcal{N}$ and the one-particle state (or weight) as $\omega$. The Poisson algebra will be denoted as $\mathbb{P}_\omega\mathcal{N}$ and the Poisson state $\varphi_\omega$.

\subsection{Poissonization as symmetric tensor algebra}

In this section, we present the first construction of Poissonization using the symmetric tensor algebra. We emphasize that in this section $\omega$ is finite (i.e. $\omega(1) < \infty$). As a caution, this construction does not work if $\omega$ is an infinite weight. In addition, we emphasize that to present a construction of Poissonization, we need to construct the Poisson quantization map, the Poisson state, and the Poisson algebra. These three objects are the essential ingredients of Poissonization. We will present some basic properties of these objects. And finally, we will derive the Poisson moment formula (i.e., the correlation functions of Poisson quantized operators).

Recall the (algebraic) symmetric tensor algebra of an associative algebra $\mathcal{N}$ is:
\begin{equation*}
    \mathcal{T}_\text{sym}(\mathcal{N}) = \oplus_{N\geq 0}\mathcal{N}^{\otimes_\text{sym} N}
\end{equation*}
where the tensor product is the symmetric tensor product. The Poisson quantization map is a linear map from $\mathcal{N}$ to the symmetric tensor algebra $\mathcal{T}_\text{sym}(\mathcal{N})$:
\begin{equation}
    \lambda:=\sum_{N\geq 0}\lambda_N:\mathcal{N}\rightarrow \mathcal{T}_\text{sym}(\mathcal{N})
\end{equation}
where the $N$-th degree Poisson quantization map is a linear map:
\begin{equation}
    \lambda_N:\mathcal{N}\rightarrow\mathcal{N}^{\otimes_\text{sym} N}: \mathcal{O}\mapsto \sum_{1\leq j\leq N}\lambda_N^{(j)}(\mathcal{O}):=\sum_{1\leq j\leq N}1\otimes\cdots\underbrace{\otimes \mathcal{O}\otimes}_\text{j-th position}\cdots\otimes 1
\end{equation}

The Poisson quantization map $\lambda$ is a Lie algebra homomorphism (i.e., preserves commutation relations):
\begin{equation}
    [\lambda(\mathcal{O}), \lambda(\widetilde{\mathcal{O}})] = \sum_{N\geq 0}[\lambda_N(\mathcal{O}), \lambda_N(\widetilde{\mathcal{O}})] = \sum_{N\geq 0}\sum_{1\leq j, k\leq N}[\lambda_N^{(j)}(\mathcal{O}), \lambda_N^{(k)}(\widetilde{\mathcal{O}})]
\end{equation}
When $j \neq k$, the commutator vanishes. Thus, we have:
\begin{equation}
    [\lambda(\mathcal{O}), \lambda(\widetilde{\mathcal{O}})] = \sum_{N\geq 0}\sum_{1\leq j\leq N}[\lambda_N^{(j)}(\mathcal{O}), \lambda_N^{(j)}(\widetilde{\mathcal{O}})] = \sum_{N\geq 0}\sum_{1\leq j\leq N}\lambda_N^{(j)}(\mathcal{O}\widetilde{\mathcal{O}} - \widetilde{\mathcal{O}}\mathcal{O}) = \lambda([\mathcal{O},\widetilde{\mathcal{O}}])
\end{equation}
Moreover, since each $N$-th degree quantization map preserves conjugation by definition (i.e. $\lambda_N(\mathcal{O}^\dagger) = \lambda_N(\mathcal{O})^\dagger$), the Poisson quantization map also preserves conjugation.

It is a fact that for each $N\geq 0$, the symmetric tensor product $\mathcal{N}^{\otimes_\text{sym} N}$ is the (ultra-)weak closure (under the tensor product weight $\omega^{\otimes N}$) of the following two sets of elements:
\begin{equation}
    \mathcal{N}^{\otimes_\text{sym} N} = \overline{\{\otimes^N \mathcal{O}: \mathcal{O}\in\mathcal{N}\}}^\text{w.o.t} = \overline{\{\Gamma_N(\mathcal{O}):\mathcal{O}\in\mathcal{N}\}}^\text{w.o.t}
\end{equation}
\begin{definition}
For a well-defined bounded weight $\omega$ on $\mN$, we define a normalized linear function on the symmetric tensor algebra called the {\it Poisson state of $\omega$}:
\begin{equation}
    \varphi_\omega := \sum_{N\geq 0}\frac{e^{-\omega(1)}}{N!}\omega^{\otimes N}
\end{equation}
where $\omega^{\otimes N}$ is a weight on $\mathcal{N}^{\otimes_\text{sym} N}$. Notice that for this definition to make sense, it is necessary that $\omega(1) < \infty$. 
\end{definition}
Using the GNS construction, we can take the (ultra-)weak closure of the algebraic symmetric tensor product under the Poisson state. The closure is defined as the \textit{Poisson algebra}.
\begin{equation}
    \mathbb{P}_\omega\mathcal{N} := \overline{\oplus_{N\geq 0}\mathcal{N}^{\otimes_\text{sym} N}}^{\varphi_\omega}
\end{equation}
Extension of $\varphi_\omega$ to the Poisson algebra is a normal faithful state. It is called the \textit{Poisson state}. The Poisson algebra is generated by the (integrated) Poisson quantized operators:
\begin{equation}
    \mathbb{P}_\omega\mathcal{N} = \overline{\{\Gamma(e^{i\mathcal{O}}):=e^{i\lambda(\mathcal{O})}:\mathcal{O}\in\mathcal{N}\text{ , }\mathcal{O}^\dagger = \mathcal{O}\}}^\text{w.o.t}
\end{equation}
Notice that because $\lambda$ preserves conjugation, $\lambda(\mathcal{O})$ is self-adjoint if $\mathcal{O}^\dagger = \mathcal{O}$. Hence, by Stone's theorem, $\lambda(\mathcal{O})$ generates a one-parameter unitary group \footnote{This unitary group is defined in the type I algebra $\mathbb{B}(\mathcal{H}_{\varphi_\omega})$, where $\mathcal{H}_{\varphi_\omega}$ is the GNS Hilbert space of the Poisson state.}. The integrated Poisson quantized operator can be represented explicitly:
\begin{equation}
    \Gamma(e^{i\mathcal{O}}) = e^{i\lambda(\mathcal{O})} = \sum_{N\geq 0}e^{i\lambda_N(\mathcal{O})} = \sum_{N\geq 0}\sum_{n\geq 0}\frac{i^n}{n!}\lambda_N(\mathcal{O})^n = \sum_{N \geq 0}(e^{i\mathcal{O}})^{\otimes N}
\end{equation}
From the integrated Poisson quantized operators, we can recover the noncommutative Poisson moment generating formula:
\begin{equation}
    \varphi_\omega(\Gamma(e^{it\mathcal{O}})) = \sum_{N\geq 0 }\frac{e^{-\omega(1)}}{N!}\omega^{\otimes N}((e^{it\mathcal{O}})^{\otimes N}) = \sum_{N\geq 0}\frac{e^{-\omega(1)\omega(e^{it\mathcal{O}})^N}}{N!} = \exp(\omega(e^{it\mathcal{O}} - 1))
\end{equation}
where $\mathcal{O}^\dagger = \mathcal{O}$ and $t \in\mathbb{R}$ is a parameter. Recall that for a classical Poisson random variable $X$ with intensity $\mu > 0$, the classical Poisson moment generating function is given by:
\begin{equation}
    \mathbb{E}(e^{itX}) = \sum_{N\geq 0}\frac{e^{-\mu}\mu^{N}e^{it N}}{N!}=\exp(\mu(e^{it} - 1))
\end{equation}

The domain of the Poisson state can be extended to include certain affiliated operators to the Poisson algebra. In particular, the following expressions are well-defined:
\begin{equation}
    \varphi_\omega\lb \lambda(x_1)...\lambda(x_n) \rb
\end{equation}
We will provide formulae to calculate these expressions shortly. These expressions are the noncommutative generalizations of the classical moments. Hence, we will call these expressions the \textit{Poisson moments}. When $\omega$ is finite (as is the case we focus on in this section), we can define Poisson moments for any product of Poisson quantized operators. When $\omega$ is not finite, the Poisson moments are only defined for products of Poisson quantized operators $\lambda(\mathcal{O})$'s where the operators 
$\mathcal{O}$'s are in the Tomita algebra of $(\mathcal{N},\omega)$ \cite{T}. 
\begin{theorem}
    The Poisson moment formula calculates the $n$-point correlation function of Poisson quantized operators $\{\lambda(\mathcal{O}_i)\}_{1\leq i\leq n}$ in the Poisson state $\varphi_\omega$. The formula is given by a summation over all possible partitions of the set $[n]:=\{1,\cdots,n\}$
    \begin{equation}
        \varphi_\omega(\lambda(\mathcal{O}_1)\cdots\lambda(\mathcal{O}_n)):=\sum_{\sigma\in\mathcal{P}_n}\prod_{A\in\sigma}\omega(\overrightarrow{\prod}_{i\in A}\mathcal{O}_i)
    \end{equation}
    where $\mathcal{P}_n$ is the set of partition on $[n]$, $\sigma$ is a particular partition, $A\in\sigma$ is a subset which is a part of the partition $\sigma$, and the final product is an ordered product because the elements $\{\mathcal{O}_i\}_{1\leq i\leq n}$ does not generally commute.
\end{theorem}
Before presenting the derivation of the Poisson moment formula, we first write down some examples of the Poisson moment formula.
\begin{enumerate}
    \item For the simplest case $n = 1$, the first moment is preserved:
    \begin{equation}
        \varphi_\omega(\lambda(\mathcal{O})) = \sum_{N\geq 0}\frac{e^{-\omega(1)}}{N!}\omega^{\otimes N}(\lambda_N(\mathcal{O})) = \sum_{N\geq 0}\frac{e^{-\omega(1)}}{N!}N\omega(1)^{N-1}\omega(\mathcal{O}) = \omega(\mathcal{O})
    \end{equation}
    \item For $n = 2$, the two-point Poisson moment formula contains two terms, because the set of partitions $\mathcal{P}(2)$ contains two partitions:
    \begin{equation}
        \mathcal{P}(2) = \{\{1\}\{2\}, \{1,2\}\}
    \end{equation}
    Therefore, we have:
    \begin{equation}
        \varphi_\omega(\lambda(\mathcal{O}_1)\lambda(\mathcal{O}_2)) = \omega(\mathcal{O}_1)\omega(\mathcal{O}_2) + \omega(\mathcal{O}_1\mathcal{O}_2)
    \end{equation}
    \item For $n =3$, the three-point Poisson moment formula contains 5 terms, because the set of partitions $\mathcal{P}(3)$ contains 5 partitions:
    \begin{align}
        \begin{split}
            \mathcal{P}(3)&=\{\{1\}\{2\}\{3\}, \{1,2\}\{3\}, \{1,3\}\{2\}, \{2,3\}\{1\}, \{1,2,3\}\}
        \end{split}        
    \end{align}
    Notice that the set partition is unordered. Therefore, we have:
    \begin{align}
        \begin{split}
            \varphi_\omega(\lambda(\mathcal{O}_1)\lambda(\mathcal{O}_2)\lambda(\mathcal{O}_3)) &= \omega(\mathcal{O}_1)\omega(\mathcal{O}_2)\omega(\mathcal{O}_3) + \omega(\mathcal{O}_1\mathcal{O}_2)\omega(\mathcal{O}_3) + \omega(\mathcal{O}_1\mathcal{O}_3)\omega(\mathcal{O}_2) \\&+ \omega(\mathcal{O}_2\mathcal{O}_3)\omega(\mathcal{O}_1) + \omega(\mathcal{O}_1\mathcal{O}_2\mathcal{O}_3)
        \end{split}
    \end{align}
\end{enumerate}
Now we are ready to present the derivation of the Poisson moment formula. The derivation is simply done by a brute-force calculation.
\begin{proof}
    First we calculate the product $\lambda(\mathcal{O}_1)\cdots\lambda(\mathcal{O}_n)$. Recall $\lambda(\mathcal{O}) = \sum_{N\geq 0}\lambda_N(\mathcal{O})$. To calculate this product, we only need to calculate the $N$-th degree component of this product:
    \begin{align}
        \begin{split}
            \lambda_m(\mathcal{O}_1)\cdots\lambda_m(\mathcal{O}_n) &= \sum_{1\leq i_1,\cdots,i_n\leq m}\bigotimes_{1\leq k\leq m}\overrightarrow{\prod_{j\in[n]:i_j = k}}\mathcal{O}_j
            \\& = \sum_{\vec{i}:[n]
            \rightarrow [m]}\bigotimes_{1\leq k\leq m}\overrightarrow{\prod_{j\in i^{-1}(k)}}\mathcal{O}_j
        \end{split}
    \end{align}
    where $\vec{i}(j) = i_j$ is a set map from $[n]$ to $[m]$. The product is an ordered product because the operators $\{\mathcal{O}_i\}_{1\leq i\leq n}$ do not generally commute. The set $i^{-1}(k)$ is the preimage of $\{k\}$ under the map $\vec{i}$. If this set is empty, the product is interpreted to be $1$. Using this formula, we can calculate the $m$-th degree component of the Poisson moment formula:
    \begin{equation}\label{equation:combinatorial_moment}
        \omega^{\otimes m}(\lambda_m(\mathcal{O}_1)\cdots\lambda_m(\mathcal{O}_n)) = \sum_{\vec{i}:[n]\rightarrow[m]}\prod_{1\leq k\leq m}\omega(\overrightarrow{\prod_{j\in i^{-1}(k)}}\mathcal{O}_j)
    \end{equation}
    For each set map $\vec{i}$, the collection of non-empty subset $\{i^{-1}(k)\neq\emptyset:k\in [m]\}$ forms a partition of $[n]$ \footnote{This is the induced partition induced by the map $\vec{i}$ and the singleton partition of $[m]$: $\{\{1\}\{2\}\cdots\{m\}\}$.}. Different set maps can induce the same partition of $[n]$. For any given partition $\sigma\in\mathcal{P}_n$, the number of set maps that induce $\sigma$ is given by:
    \begin{equation}
        |\sigma|!\binom{m}{|\sigma|}
    \end{equation}
    where $|\sigma|$ is the number of subsets in the partition $\sigma$ and $|\sigma| \leq m$ \footnote{This combinatorial factor arises by simply counting the number of ways to assign a distinct number between $1$ and $m$ to each subset $A\in \sigma$. Different assignments define different set maps.}. Using this combinatorial observation, we can rewrite the $m$-th degree component of the Poisson moment formula:
    \begin{equation}
        \omega^{\otimes m}(\lambda_m(\mathcal{O}_1)\cdots\lambda_m(\mathcal{O}_n)) = \sum_{\sigma\in\mathcal{P}_n, |\sigma| \leq m}|\sigma|!\binom{m}{|\sigma|}\omega(1)^{m - |\sigma|}\prod_{A\in\sigma}\omega(\overrightarrow{\prod}_{j\in A}\mathcal{O}_j)
    \end{equation}
    where the factor $\omega(1)^{m - |\sigma|}$ comes from the empty sets $\{k\in[m]: i^{-1}(k) = \emptyset\}$. These empty sets exist because $|\sigma|\leq m$.

    Now, applying the definition of Poisson state, we have:
    \begin{align}
        \begin{split}
            \varphi_\omega(\lambda(\mathcal{O}_1)\cdots\lambda(\mathcal{O}_n)) &= \sum_{m\geq 0}\frac{e^{-\omega(1)}}{m!}\sum_{\sigma\in\mathcal{P}_n, |\sigma|\leq m}|\sigma|!\binom{m}{|\sigma|}\omega(1)^{m - |\sigma|}\prod_{A\in\sigma}\omega(\overrightarrow{\prod_{j\in A}}\mathcal{O}_j)
            \\
            &=\sum_{\sigma\in\mathcal{P}_n}\sum_{m\geq |\sigma|}\frac{e^{-\omega(1)}\omega(1)^{m - |\sigma|}}{(m - |\sigma|)!}\prod_{A\in\sigma}\omega(\overrightarrow{\prod_{j\in A}}\mathcal{O}_j)
            \\
            &=\sum_{\sigma\in\mathcal{P}_n}\prod_{A\in\sigma}\omega(\overrightarrow{\prod_{j\in A}}\mathcal{O}_j)\qedhere
        \end{split}
    \end{align}
\end{proof}
Finally, we discuss the connected component of the Poisson moment formula, which is defined as the connected correlation function of Poisson quantized operators: 
\begin{equation}
    \varphi_\omega(\lambda(\mathcal{O}_1)\cdots\lambda(\mathcal{O}_n))_\text{conn} := (-i)^n\frac{\partial^n}{\partial s_1\cdots\partial s_n}|_{s_1,\cdots,s_n = 0}\log\varphi_\omega(e^{is_1\lambda(\mathcal{O}_1)}\cdots e^{is_n\lambda(\mathcal{O}_n)})
\end{equation}
To calculate the connected correlation function, we make a simple observation:
\begin{align}
    \begin{split}
        e^{is\lambda(\mathcal{O})}e^{it\lambda(\widetilde{\mathcal{O}})} &= \Gamma(e^{is\mathcal{O}})\Gamma(e^{it\widetilde{\mathcal{O}}}) = \sum_{N\geq 0}(e^{is\mathcal{O}})^{\otimes N} \sum_{M\geq 0}(e^{it\widetilde{\mathcal{O}}})^{\otimes M} \\&= \sum_{N\geq 0}(e^{is\mathcal{O}}e^{it\widetilde{\mathcal{O}}})^{\otimes N} = \Gamma(e^{it\mathcal{O}}e^{is\widetilde{\mathcal{O}}})
    \end{split}
\end{align}
where $\mathcal{O}^\dagger = \mathcal{O}\text{ , }\widetilde{\mathcal{O}}^\dagger = \widetilde{\mathcal{O}}$ and $s,t\in\mathbb{R}$ \footnote{Recall that in this section $\omega$ is finite. For $\omega$ infinite, this formula needs to be used with care because the domains of the Poisson quantization map and the integrated Poisson quantization map are no longer the entire input von Neumann algebra.}.
Using these observations, we can calculate the connected component of the Poisson moment formula:
\begin{align}
    \begin{split}
        \varphi_\omega(\lambda(\mathcal{O}_1)\cdots\lambda(\mathcal{O}_n))_\text{conn}&=(-i)^n\frac{\partial^n}{\partial s_1\cdots\partial s_n}|_{s_1,\cdots,s_n = 0}\log\varphi_\omega(e^{is_1\lambda(\mathcal{O}_1)}\cdots e^{is_n\lambda(\mathcal{O}_n)})
        \\
        &=(-i)^n\frac{\partial^n}{\partial s_1\cdots\partial s_n}|_{s_1,\cdots,s_n = 0}\log \varphi_\omega(\Gamma(e^{is_1\mathcal{O}_1}\cdots e^{is_n\mathcal{O}_n}))
        \\
        &=(-i)^n\frac{\partial^n}{\partial s_1\cdots\partial s_n}|_{s_1,\cdots,s_n = 0}\log\exp(\omega(e^{is_1\mathcal{O}_1}\cdots e^{is_n\mathcal{O}_n} - 1))
        \\&=(-i)^n\frac{\partial^n}{\partial s_1\cdots\partial s_n}|_{s_1,\cdots,s_n = 0}\omega(e^{is_1\mathcal{O}_1}\cdots e^{is_n\mathcal{O}_n} - 1)
        \\&=\omega(\mathcal{O}_1\cdots\mathcal{O}_n)
    \end{split}
\end{align}
This shows that \textbf{\textit{Poissonization preserves the connected correlation functions}}.

\subsection{Brief introduction to ultrafilter and ultraproduct}\label{appendix:ultra}

In this subsection, we introduce the notions of ultrafilter and ultraproduct to physicists. We will \textit{not} provide the rigorous mathematical definitions. Rather, we will focus on how to use these tools to study asymptotic properties of sequences of operators/operator algebras. 

For our purpose, an ultrafilter can be understood as a consistent set of rules to define limits for \textit{all} bounded sequences of complex numbers. It provides a generalization of the notion of convergence in standard analysis. This is particularly useful when we study \textit{asymptotic} properties of sequences of operators.

As a demonstration, we go back to the original proof by Murray and von Neumann \cite{MvN}that the hyperfinite type $II_1$ factor is not isomorphic to the free group factor (i.e., the large $N$ limit of Gaussian random matrices). The key technical tool in their proof is the notion of a nontrivial central sequence in the hyperfinite $II_1$ factor $\mathcal{R}$. In terms of infinite tensor product, $\mathcal{R}\cong\otimes_{n\geq 0}(M_2(\mathbb{C}), \frac{1}{2}\tr_2)$. A nontrivial central sequence in this case is a sequence of operators $(x_n)_{n\geq 0}$ such that:
\begin{enumerate}
    \item The sequence is uniformly bounded in norm: $\sup_n||x_n|| < \infty$;
    \item For any operator $y\in \mathcal{R}$, the sequence asymptotically commute with $y$:
    \begin{equation}
        [x_n,y]\rightarrow 0\text{ strongly}
    \end{equation}
    \item The following nontriviality conditions hold:
    \begin{equation}
        \forall n\geq 0\text{ , }\tau(x_n) = 0
    \end{equation}
    \begin{equation}
        \lim_{n\rightarrow\infty}\tau(x_n^\dagger x_n) \neq 0
    \end{equation}
    where $\tau$ is the canonical trace on $\mathcal{R}$.
\end{enumerate}
Murray and von Neumann showed that there exist nontrivial central sequences in the hyperfinite type $II_1$ factor, while the free group factors do not have nontrivial central sequences. This proof highlights the importance of asymptotic properties in operator algebras. 

It turns out that there is an alternative, more concise and exact formulation of these conditions. Since we want to study sequences of operators in $\mathcal{R}$, it is natural to consider the von Neumann algebra of sequences: $\ell_\infty(\mathbb{N},\mathcal{R}) := \{(x_n)_{n\geq 0}:x_n\in\mathcal{R}\text{ , }\sup_{n\geq 0}||x_n|| < \infty\}$. This algebra is the algebra of all sequences satisfying the first condition. Within $\ell_\infty(\mathbb{N},\mathcal{R})$, there exists a special kind of sequences - namely, the diagonal sequences: $\{(x_n)_{n\geq 0}:\forall n\text{ , }x_n = x\in\mathcal{R}\}$. This is a subalgebra in $\ell_\infty(\mathbb{N}, \mathcal{R})$ isomorphic to $\mathcal{R}$. Tentatively, an asymptotically central sequence can now be reformulated as an element in the relative commutant $\mathcal{R}'\cap \ell_\infty(\mathbb{N}, \mathcal{R})$.

One caveat is that there exist trivial examples of asymptotically central sequences. And we would like to get rid of these examples. For example, the following sequence is trivially asymptotically central:
\begin{equation}
    (\widetilde{x}_n:=\frac{1}{n}\mathbbm{1}^{\otimes n}\otimes\begin{bmatrix}
        0 & 1 \\ 1 & 0
    \end{bmatrix}\otimes\mathbbm{1}\otimes\cdots)_{n\geq 0}
\end{equation}
This is because for all $y\in\mathcal{R}$, we have:
\begin{eqnarray}
    \lim_{n\rightarrow\infty}|\tau(\widetilde{x}_ny)|\leq\lim_{n\rightarrow\infty}\tau(\widetilde{x}_n^\dagger\widetilde{x}_n)^{1/2}\tau(y^\dagger y)^{1/2} = \lim_{n\rightarrow\infty}\frac{\tau(y^\dagger y)^{1/2}}{n} = 0
\end{eqnarray}
And similarly $\lim_{n\rightarrow\infty}\tau(y\widetilde{x}_n) = 0$. From this and the fact that any vector $\ket{\xi}$ in the GNS representation can be written as $\ket{\xi} = z\ket{\tau^{1/2}}$ for some $z\in\mathcal{R}$\footnote{Here $\ket{\tau^{1/2}}$ is the vector implementing the canonical trace $\tau$ in the GNS Hilbert space.}, we have:
\begin{equation}
    \lim_{n\rightarrow\infty}\bra{\eta}[\widetilde{x}_n,y]\ket{\xi} = \lim_{n\rightarrow\infty}\tau(w^\dagger[\widetilde{x}_n,y]z) = \lim_{n\rightarrow\infty}\tau(\widetilde{x}_nyzw^\dagger) - \tau(yzw^\dagger\widetilde{x}_n) = 0
\end{equation}
where $\ket{\eta}  = w\ket{\tau^{1/2}}$ and $\ket{\xi} = z\ket{\tau^{1/2}}$. Thus $\lim_{n\rightarrow\infty}[\widetilde{x}_n,y]\ket{\xi} = 0$ for all $y\in\mathcal{R}$ and for all $\ket{\xi}$. 

The nontriviality condition gets rid of these trivial sequences. In terms of the algebra of bounded sequences $\ell_\infty(\mathbb{N}, \mathcal{R})$, we are tempted to get rid of these trivial sequences by quotienting out $\{(x_n)_{n\geq 0}:\lim_{n\rightarrow\infty}\tau(x^\dagger_nx_n) = 0\}$. This is the correct idea, but the key technical issue is that the standard limit $\lim_{n\rightarrow\infty}\tau(x^\dagger_nx_n)$ is not always defined. In principle, the sequence of numbers $(\tau(x^\dagger_nx_n))_{n\geq 0}$ is only a bounded sequence - it need not be convergent in the standard sense. \textit{This is where the notion of ultrafilter comes to the rescue}. As mentioned above, for us, an ultrafilter is simply a consistent set of rules that gives all bounded sequences a well-defined limit. In terms of notation, we denote this well-defined limit as: $\lim_{w}\tau(x^\dagger_nx_n)$ where $w$ is an ultrafilter\footnote{To be precise, $w$ is a nonprincipal ultrafilter: $w\in\beta\mathbb{N}\backslash\mathbb{N}$ where $\beta\mathbb{N}$ is the Stone-Cech compactification of the natural numbers.}. Using this limit, the following set is well-defined: $I_w := \{(x_n)_{n\geq 0}:\lim_w\tau(x^\dagger_nx_n) = 0\}$. It turns out that $I_w$ is a closed $*$-ideal and we can form the quotient von Neumann algebra:
\begin{equation}
    \mathcal{R}^w := \ell_\infty(\mathbb{N}, \mathcal{R}) / I_w
\end{equation}
\textit{This is the ultraproduct algebra of $\mathcal{R}$.} Now the asymptotic central sequences are simply elements in the relative commutant: $\mathcal{R}^w\cap\mathcal{R}'$. The fact that the hyperfinite type $II_1$ algebra has many nontrivial central sequences can be reformulated as a precise statement: $\mathcal{R}^w\cap \mathcal{R}' \neq\mathbb{C}$. The tools of ultrafilter and ultraproduct turn an asymptotic property (e.g., asymptotic centrality) into a property of the ultraproduct algebra $\mathcal{R}^w$. In this sense, one can regard the ultraproduct algebra as a way to encode asymptotic information. 

In the main text, we used a slightly more general version of the ultraproduct construction. For a sequence $(x_n)_{n\geq 0}$ in $\ell_\infty(\mathbb{N}, \mathcal{R})$, each operator $x_n$ takes value in the same algebra $\mathcal{R}$. Alternatively, we can allow for operators taking values in different algebras. For example, given a sequence of matrix algebras $\{M_n(\mathbb{C})\}$, we can form the von Neumann algebra of bounded sequences $\ell_\infty(\mathbb{N}, M_n(\mathbb{C})) := \{(x_n)_{n\geq 0}: x_n\in M_n(\mathbb{C})\}$. Using this, we can form the ultraproduct algebra $\prod^w M_n(\mathbb{C}) := \ell_\infty(\mathbb{N}, M_n(\mathbb{C})) / I_w$ where $I_w := \{(x_n)_{n\geq 0}\in \ell_\infty(\mathbb{N}, M_n(\mathbb{C})):\lim_w\tr_n(x_n^\dagger x_n) = 0\}$ is a closed $*$-ideal. Here $\tr_n$ is the canonical normalized trace on $M_n(\mathbb{C})$. This ultraproduct of matrix algebras is useful in capturing asymptotic properties of sequences of matrices. This makes the ultraproduct construction useful in studying large $N$ properties of matrix models. For example, the free group factor (i.e., the large $N$ limit of Gaussian random matrices) can be constructed from this ultraproduct construction. 

However, we emphasize that the ultraproduct algebra is typically \textit{not} hyperfinite and it can have drastically different properties than the original sequence of algebras. In addition, the ultraproduct algebra depends on the ultrafilter $w$ used in the construction. Often, to show that certain constructions or properties are independent of the ultrafilter, additional arguments would be needed. For our application to Poissonization, we will show that the Poisson algebra is independent of the ultrafilter used in the noncommutative Bernoulli approximation by using a noncommutative Hamburger moment theorem \cite{MAW}. 

\subsection{Poissonization as noncommutative Bernoulli approximation}\label{subsection:qpoiss}

In this section, we present an alternative construction of Poissonization from the perspective of noncommutative probability theory. As the name suggests, Poissonization can be understood as a noncommutative generalization of the classical Poisson random measure. The simplest case of a Poisson random measure is a Poisson random variable (where the underlying space is a singleton). A classical Poisson random variable can be constructed as the limit of a sequence of sums of Bernoulli random variables. This is the well-known \textit{law of rare events}.

To generalize this approximation to the noncommutative setting, two ingredients need to be modified:
\begin{enumerate}
    \item The classical Bernoulli random variables need to be generalized to the noncommutative setting. Poissonization replaces these random variables with Poisson quantized operators.
    \item The notion of limit needs to be modified so as to construct a well-defined Poisson algebra and Poisson state. Poissonization replaces the usual limit with the notion of an ultraproduct.
\end{enumerate}
Although the notion of an ultraproduct seems like a formidable mathematical tool, physicists have long been using this notion implicitly. More precisely, the algebra generated by master fields in the large $N$ limit of a matrix theory is an ultraproduct of matrix algebras \cite{statmech, singer, voiculescu1, VND, accardi1, Speicher:1993kt}. For example, for $SU(N)$-gauge theory on 2D Riemann surfaces, the algebra generated by the master fields in the large $N$ limit is a (non-hyperfinite) type $II_1$ factor, which is an ultraproduct of finite-dimensional matrix algebras \cite{singer}. 

Poissonization is \textit{not} free probability. But the same framework of noncommutative probability theory provides the right tools to construct Poissonization. The key observation of large $N$ theory is that the large $N$ limits of the correlation functions completely specify the theory. Mathematically, this observation can be made precise by a noncommutative Hamburger moment theorem \cite{MAW}. For Poissonization, the key is the Poisson moment formula. In this section, we use noncommutative Bernoulli approximation to construct a sequence of algebras, weights, and correlation functions whose limit is the Poisson moment formula. The limiting operators will be analogs of large $N$ master fields (these are the Poisson quantized operators), the limiting algebra will be the Poisson algebra, and the limiting weight will be the Poisson weight.

Certain parts of this section can be technical, but we have tried to make the material accessible. We first recall the classical Bernoulli approximation and the law of rare events. For completeness, we have included a discussion of the algebraic formulation of classical probability theory (i.e., the commutative von Neumann algebras) in Appendix \ref{app:math}. We then present the large $N$ limit as an ultraproduct, and apply our discussion to the classical Bernoulli approximations and construct the classical Poisson random variable as an ultraproduct of sums of Bernoulli random variables. Finally, we generalize the construction to the noncommutative case and observe that Poissonization can be understood as the large $N$ limit of particular sparse Bernoulli random matrices.

\subsubsection{Classical Bernoulli Approximation to Poisson Distribution}\label{subsubsection:bernoulli}

First, recall the classical Bernoulli approximation to the Poisson distribution. We already discussed this in Subsection \ref{subsec:rareevents} with a physics example; however, we repeat it here using mathematical terminology to motivate ultralimits. Suppose there exists a rare event that happens with probability $\mu\Delta T$ where $\Delta T$ is the length of the time interval. In other words, in each unit interval of time, this rare event happens with probability $\mu$. We would like to know what the probability is of $k$ such events happening during a unit time interval. One way to approach this problem is to coarse-grain the time interval into $M$ equi-length subintervals:
\begin{equation}
    [0,1] = [0,\frac{1}{M}]\sqcup[\frac{1}{M}, \frac{2}{M}]\sqcup\cdots\sqcup[\frac{M-1}{M},1]
\end{equation}
For large enough $M$, each time interval has negligible length. Therefore, the probability of a rare event happening more than once within a subinterval is dampened by a $\frac{1}{M}$. Hence, we can assume that the number of events within each subinterval is at most once. Then the total number of events within the unit time interval can be approximated by a summation:
\begin{equation}
    X_M = \sum_{1\leq i\leq M}X^{(i)}
\end{equation}
where $X^{(i)}$ is the number of events within the $i$-th subinterval. By construction $X^{(i)}$ is a Bernoulli random variable:
\begin{align}
    \begin{split}
        &X^{(i)} = 0 \text{ with probability }(1 - \frac{\mu}{M})
        \\
        &X^{(i)} = 1 \text{ with probability }\frac{\mu}{M}
    \end{split}
\end{align}
Then the probability distribution of the random variable $X_M$ is a convolution of $M$ Bernoulli random variables:
\begin{equation}
    X_M\sim (\frac{\mu}{M}\delta_1 + (1 - \frac{\mu}{M})\delta_0)^{* M}
\end{equation}
where $*M$ denotes the $M$-fold convolution of Bernoulli random variables. Concretely, we have:
\begin{equation*}
    \text{Prob}(X_M = k) = \binom{M}{k}(\frac{\mu}{M})^k(1- \frac{\mu}{M})^{M-k}
\end{equation*}
When the coarse graining becomes finer and finer (i.e. $M\rightarrow \infty$), we have the limiting distribution (the \textit{law of rare events}):
\begin{align}
    \begin{split}
       \lim_{M\rightarrow\infty}\text{Prob}(X_M = k) &= 
        \lim_{M\rightarrow\infty}\frac{(1 - \frac{\mu}{M})^M\mu^k}{k!}\frac{M!}{(M-k)!M^k(1 - \frac{mu}{M})^k} \\&= \frac{e^{-\mu}\mu^k}{k!}\lim_{M\rightarrow\infty}\prod_{0\leq i\leq k-1}\frac{M-i}{M-\mu} = \frac{e^{-\mu}\mu^k}{k!} 
    \end{split}
\end{align}
The resulting probability distribution is the classical Poisson distribution with intensity $\mu$. 
\subsubsection{Large $N$ Limit as ultraproduct}\label{subsubsection:large_N}
In this subsection, we introduce the mathematical tool of ultraproduct. Physicists have been using this tool implicitly when constructing master fields to study the large $N$ limits of matrix models \cite{Gopakumar-Gross}.

In matrix models, we are given a sequence of finite $N$ matrix algebras and finite $N$ correlation functions. The large $N$ limit of the correlation functions can often be calculated. From these limiting correlation functions, we would like to construct master fields and a state on the algebra generated by the master fields so that limiting correlation functions can be reproduced. 

For example, consider the non-interacting matrix model with the partition function:
\begin{equation}
    Z_N := \int_{\text{Herm}_N} \prod_{1\leq i\leq k}dX_i e^{-\frac{N}{2}\sum_{1\leq i\leq k}\tr(X_i^2)}
\end{equation}
where $X_i^\dagger = X_i$ are Hermitian \footnote{In other words, $X_i$'s are mutually independent Hermitian Gaussian random matrices.}. The finite $N$ correlation functions are:
\begin{equation}
    \langle \frac{\tr}{N}(f(X_1,\cdots,X_k))\rangle:=\frac{1}{Z_N}\int_{\text{Herm}_N} \prod_{1\leq i\leq k}dX_i e^{-\frac{N}{2}\sum_{1\leq i\leq k}\tr(X_i^2)}\frac{\tr}{N}(f(X_1,\cdots,X_k))
\end{equation}
where $\tr$ is the un-normalized trace on the $N$-by-$N$ matrices and $f(X_1,\cdots,X_k)$ is a noncommutative monomial (or simply a word) formed from products of the matrices $X_i$.

When $k = 1$ (one-matrix model), the large $N$ limit of the correlation functions can be calculated from the Wigner semi-circular law \cite{mehta2004random}. When $k > 1$, the key observation is the following factorization (free independence) \cite{voiculescu1}:
\begin{equation}
    \lim_{N\rightarrow\infty}\langle\frac{\tr}{N}(X_1^{\alpha_1} X_2^{\alpha_2}\cdots X_k^{\alpha_k})\rangle = \prod_{1\leq i\leq k}\lim_{N\rightarrow \infty}\langle\frac{\tr}{N}(X_i^{\alpha_i})\rangle
\end{equation}
where $\alpha_i \in \mathbb{N}$ and the product is ordered so that $X_i\neq X_j$. More general correlation functions can be calculated from this factorization rule. 

In this case, the master fields $\widehat{X_i}$ corresponding to $X_i$'s are given by semi-circular random variables that are mutually freely independent \cite{voiculescu1}. The algebra generated by these master fields is isomorphic to the free group factor with $k$ generators. And using the canonical trace on the free group factor, we have:
\begin{equation}
    \lim_{N\rightarrow\infty}\langle\frac{\tr}{N}f(X_1,\cdots,X_k)\rangle := \tau(f(\widehat{X}_1,\cdots,\widehat{X}_k))
\end{equation}
where $\widehat{X}_i$ is the master field corresponding to $X_i$ and $\tau$ is the canonical trace on the free group factor. Here $f$ is a noncommutative monomial.

Mathematically, the free group factor can be constructed using an ultraproduct of matrix algebras \cite{Speicher1992, Speicher:1993kt}. The key idea is to represent master fields as sequences of finite $N$ matrices and to construct the limiting algebra via a GNS-type construction. The limiting state is well-defined, and it relies on a well-defined limit of finite $N$ moment formulae. In the next section, we apply this procedure to recast the classical Bernoulli approximation as an ultraproduct.

\paragraph{Bernoulli approximation as ultraproduct}

In this section, we recast the classical Bernoulli approximation as a particular example of an ultraproduct. The commutative von Neumann algebra generated by the limiting Poisson random variable is the ultraproduct algebra. Recall that the Poisson random variable can be approximated by a sequence of sums of i.i.d. Bernoulli random variables. Based on the algebraic characterization of random variables, a Bernoulli distribution defines a state $\omega_p$ on the commutative von Neumann algebra $\ell_\infty^2$ where $\omega_p(\alpha_0\delta_0 + \alpha_1\delta_1) = \alpha_0(1-p) + \alpha_1 p$ ($\alpha_i$ are constants for $i = 0,1$). In addition, a sum of $M$ i.i.d. Bernoulli random variables $X_M = \sum_{1\leq i \leq M}X^{(i)}$ defines a tensor product state $\omega_p^{\otimes k}$ on the symmetric tensor algebra $\otimes_\text{sym}^M \ell_\infty^2$. Moreover, the distribution of $X_M$ is reproduced by the tensor product state and the spectral projections of the operator:
\begin{equation}
    \widetilde{\lambda}_M(\delta_1) = \sum_{1\leq i \leq M}1\otimes\cdots\underbrace{\otimes \delta_1\otimes}_\text{i-th position}\cdots\otimes 1
\end{equation}

The discussion above indicates that we should consider the following sequence of algebras and states:
\begin{equation}
    \{(\otimes_\text{sym}^M \ell_\infty^2, \omega_{p_M}^{\otimes M})\}_{M\geq [\lambda]+1}
\end{equation}
where for each $M$ the success probability $p_M = \frac{\mu}{M}$. The indices $M$ are required to be larger than $[\mu]+1$ so that $p_M$ is less than 1. 

As mentioned in our discussion of large $N$ limits, the ultraproduct construction represents the master field as a sequence of (uniformly bounded) operators. In our current application, we can be more explicit and consider the von Neumann algebra of (uniformly bounded) sequences:
\begin{equation}
    \ell_\infty(\mathbb{N}, \otimes_\text{sym}^M\ell_\infty^2) := \{(\mathcal{O}_M):\mathcal{O}_M\in\otimes_\text{sym}^M\ell_\infty^2\text{ for }M\geq [\mu]+1\text{, and }0\text{ otherwise}\}
\end{equation}
Heuristically, the limiting state $\omega$ is defined by the limit ``$\omega((\mathcal{O}_M)) := \lim_{M\rightarrow\infty}\omega_{p_M}^{\otimes M}(\mathcal{O}_M)$" where $\mathcal{O}_M\in\otimes_\text{sym}^M\ell_\infty^2$. The quotation mark emphasizes the fact that \textit{not all} bounded sequences have a well-defined limit in the standard sense \footnote{This is where an ultrafilter/ultraproduct comes to the rescue. The crucial fact is that given any (non-principal) ultrafilter, any bounded sequence has a well-defined limit. And for convergent sequences, this limit reduces to the standard notion of limit.}. Then, following the GNS construction, we quotient out the null ideal:
\begin{equation}
    \mathcal{I}_w := \{(\mathcal{O}_M)\in\ell_\infty(\mathbb{N},\otimes_\text{sym}^M\ell_\infty^2): \lim_{M\rightarrow w}\omega_{p_M}^{\otimes M}(\mathcal{O}_M^\dagger\mathcal{O}_M) = 0\}
\end{equation}
where $w$ is a non-principal ultrafilter (see also Appendix \ref{appendix:ultra}).
The resulting quotient algebra is \textit{the ultraproduct algebra}:
\begin{equation}
    \prod^w (\otimes_\text{sym}^M\ell_\infty^2) := \ell_\infty(\mathbb{N}, \otimes_\text{sym}^M\ell_\infty^2) / I_w
\end{equation}

We do not need the entire ultraproduct algebra to reproduce the Poisson random variable. We only need the sub-algebra generated by a particular sequence: 
\begin{equation}
    \widetilde{\lambda}(\delta_1) := [(\widetilde{\lambda}_M(\delta_1))]
\end{equation}
where the square bracket emphasizes that the operator $\widetilde{\lambda}(\delta_1)$ exists in the quotient algebra. Notice that the expectation of $\widetilde{\lambda}(\delta_1)$ can be readily calculated:
\begin{equation}
    \omega(\widetilde{\lambda}(\delta_1)) = \lim_{M}\omega_{p_M}^{\otimes M}(\widetilde{\lambda}_M(\delta_1)) = \lim_M M\omega_{p_M}(1)^{M-1}\omega_{p_M}(\delta_1) = \lim_M Mp_M = \mu
\end{equation}
To study the spectrum of $\widetilde{\lambda}(\delta_1)$, we decompose each component $\lambda_M(\delta_1)$:
\begin{equation}
    (\widetilde{\lambda}_M(\delta_1))_{M\geq [\mu]+1} = (\sum_{0\leq k\leq M}kp^M_k)_{M\geq [\mu]+1} 
\end{equation}
Hence for each $k\in\mathbb{N}$, we have a sequence of projections $(p^M_k)_{M\geq [\mu]+1}$ where for $M < k$ define $p^M_k = 0$. We can calculate the ultraproduct of this sequence of projections:
\begin{equation}
    \lim_{M}\omega_{p_M}^{\otimes M}(p^M_k) = \lim_{M}\binom{M}{k}(\frac{\mu}{M})^k(1 - \frac{\mu}{M})^{M-k} = \frac{e^{-\mu}\mu^k}{k!}
\end{equation}
where the ultraproduct in the second equation reduces to the usual limit because the sequence converges. This is nothing but the classical Bernoulli approximation formula. The sequence of projections $(p^M_k)_{M\geq[\mu]+1}$ gives a well-defined projection $\widetilde{\lambda}(p_k)$ in the ultraproduct algebra such that:
\begin{equation}
    \omega(\widetilde{\lambda}(p_k)) = \lim_{M}\omega_{p_M}^{\otimes M}(p_k^M) = \frac{e^{-\mu}\mu^k}{k!}
\end{equation}
Therefore, we have the spectral decomposition of $\widetilde{\lambda}(\delta_1)$:
\begin{equation}
    \widetilde{\lambda}(\delta_1) = \sum_{k\geq 0}k\widetilde{\lambda}(p_k)
\end{equation}
And the von Neumann subalgebra generated by $\widetilde{\lambda}(\delta_1)$ is the algebra of bounded sequences on $\mathbb{N}$, and the state $\omega$ restricted to this subalgebra reduces to the Poisson distribution. The isomorphism is given by
\begin{equation}
    \varphi:\ell_\infty(\mathbb{N})\rightarrow \langle\widetilde{\lambda}(\delta_1)\rangle:\sum_{k\geq 0}\alpha_k\delta_k\mapsto\sum_{k\geq 0}\alpha_k\widetilde{\lambda}(p_k)
\end{equation}
where $(\alpha_k)_{k\geq 0}$ is a bounded sequence of complex numbers. The limiting state $\omega$ reduces to:
\begin{equation*}
    \omega\circ\varphi(\sum_{k\geq 0}\alpha_k\delta_k) = \sum_{k\geq 0}\alpha_k \omega(\widetilde{\lambda}(p_k)) = \sum_{k\geq 0}\alpha_k \frac{e^{-\lambda}\lambda^k}{k!}
\end{equation*}
This concludes the ultraproduct modeling of Bernoulli approximation.
\subsubsection{Noncommutative Bernoulli approximation to Poissonization}\label{subsubsection:general}
The classical Bernoulli approximation is the simplest example of Poissonization. The resulting Poisson algebra is a commutative von Neumann algebra generated by a single Poisson quantized operator $\widetilde{\lambda}(\delta_1)$, and the Poisson state is the Poisson distribution. Rephrasing this construction in a slightly more concise and abstract way, the input data in the Bernoulli approximation is the trivial von Neumann algebra $\mathbb{C}$ and the input weight is the linear functional of multiplication by a positive constant $\mu > 0$:
\begin{equation}
    \omega_\mu:\mathbb{C}\rightarrow\mathbb{C}:z\mapsto \mu z
\end{equation}
In the ultraproduct construction (as opposed to the symmetric tensor algebra construction), the Poisson quantization map is given by a sequence of linear maps:
\begin{equation}
    \widetilde{\lambda}_M:\mathbb{C}\rightarrow \otimes_\text{sym}^M\ell_\infty^2:\delta_1\mapsto \sum_{1\leq i\leq M}1\otimes\cdots\underbrace{\otimes \delta_1\otimes}_\text{i-th position}\cdots\otimes 1
\end{equation}
where we have regarded $\mathbb{C}$ as the von Neumann algebra of measurable functions on the one-point space $\{1\}$ and $\delta_1$ generates this trivial algebra. The full Poisson quantization map is defined to be:
\begin{equation}
    \widetilde{\lambda}:\mathbb{C}\rightarrow \ell_\infty(\mathbb{N}, \otimes_\text{sym}^M\ell_\infty^2)/I_w: \delta_1\mapsto [(\widetilde{\lambda}_M(\delta_1))]
\end{equation}
where again the square bracket emphasizes that the operator exists in the quotient algebra.

The Poisson algebra is the subalgebra generated by $\widetilde{\lambda}(\delta_1)$ in the ultraproduct algebra and the Poisson state is the limiting state $\varphi_{\omega_\mu}:=\lim_M \omega_{p_M}^{\otimes M}$ restricted to this subalgebra.

The generalization to more general von Neumann algebras is straightforward. Here we restrict ourselves to consider tracial von Neumann algebras as inputs \footnote{The reason is that the ultraproduct definition needs to be modified for non-tracial weights. This involves introducing additional (but standard) machineries.}. Instead of $(\mathbb{C}, \omega_\mu)$, the input data is $(\mathcal{N},\tau)$ where $\tau$ is a tracial finite weight (i.e.$\tau(1)<\infty$). The Poisson quantization map again is a sequence of linear maps:
\begin{eqnarray}\label{equation:qpoissquant}
    \widetilde{\lambda}_M: \mathcal{N}\rightarrow \otimes_\text{sym}^M(\mathcal{N}\otimes\ell_\infty^2):\mathcal{O}\mapsto &&\sum_{1\leq i\leq M}1\otimes\cdots\underbrace{\otimes(\mathcal{O}\otimes\delta_1)\otimes}_\text{i-th position}\cdots\otimes 1 \nn\\
    &&= \sum_{1\leq i\leq M}1\otimes\cdots\underbrace{\otimes\begin{bmatrix}
        \mathcal{O}&0\\0&0
    \end{bmatrix}\otimes}_\text{i-th position}\cdots\otimes 1
\end{eqnarray}
where $M \geq [\tau(1)] + 1$, and the symmetric tensor product is taken over the tensor algebra $\mathcal{N}\otimes\ell_\infty^2\cong\mathcal{N}\oplus\mathcal{N}$. In terms of the matrix notation, $\mathcal{N}\otimes\ell_\infty^2$ is the set of diagonal $2\times2$ matrices where each matrix entry is an operator in $\mathcal{N}$. The $M$-th state on the algebra $\mathcal{N}\otimes\ell_\infty^2$ is given by the tensor product of the state $\tau_N:=\frac{\tau}{\tau(1)}\otimes\begin{bmatrix}
    \frac{\tau(1)}{N}&0\\0&1-\frac{\tau(1)}{N}
\end{bmatrix}$. We consider the ultraproduct 
\begin{equation}
    \prod^w \big(\otimes_\text{sym}^M(\mathcal{N}\otimes\ell_\infty^2), \tau_M^{\otimes M}\big)
\end{equation}
where $w$ is a non-principal ultrafilter. The Poisson quantization map is given by:
\begin{equation}
    \widetilde{\lambda}:\mathcal{N}\rightarrow \prod^w\big(\otimes_\text{sym}^M(\mathcal{N}\otimes\ell_\infty^2), \tau_M^{\otimes M}\big):\mathcal{O}\mapsto [(\widetilde{\lambda}_M(\mathcal{O}))]
\end{equation}
where again the square bracket emphasizes that the operator exists in the quotient. The Poisson algebra is the subalgebra generated by:
\begin{equation}
    \mathbb{P}_\tau\mathcal{N}:=\langle\widetilde{\lambda}(\mathcal{O}):\mathcal{O}^\dagger =\mathcal{O}\in\mathcal{N}\rangle\subset\prod^w(\otimes_\text{sym}^M(\mathcal{N}\otimes\ell_\infty^2), \tau_M^{\otimes M})
\end{equation}
The Poisson state $\varphi_\tau$ is the restriction of the limiting state $\lim_{M\rightarrow w}\tau_M^{\otimes M}$ to the Poisson algebra \footnote{Technically, the Poisson quantized operator $\widetilde{\lambda}(\mathcal{O})$ is generally \textit{not} bounded operators in the Poisson algebra but \textit{affiliated operators} with the Poisson algebra. Affiliated operators commute with the commutant of the Poisson algebra. But they are not necessarily bounded. For example, in the case of Bernoulli approximation, $\widetilde{\lambda}(\delta_1)$ has spectrum $\mathbb{N}$ and is not bounded. Nevertheless, by Stone's theorem, essentially self-adjoint affiliated operators generate a one-parameter family of unitary operators that exist inside the von Neumann algebra. This is what it means for the Poisson quantized operators to generate the Poisson algebra.}.

By definition of $\widetilde{\lambda}_M$, it is not difficult to see that $\widetilde{\lambda}_M$ preserves commutation relation:
\begin{equation}
    \widetilde{\lambda}_M([\mathcal{O},\widetilde{\mathcal{O}}]) = [\widetilde{\lambda}_M(\mathcal{O}), \widetilde{\lambda}_M(\widetilde{\mathcal{O}})]
\end{equation}
Hence, the Poisson quantization map also preserves commutation relations. 

In addition, we can readily calculate the expectation of the Poisson quantized operator under the Poisson state:
\begin{equation}
    \varphi_\tau(\widetilde{\lambda}(\mathcal{O})) = \lim_M \tau_M^{\otimes M}(\widetilde{\lambda}_M(\mathcal{O})) = \lim_M M\tau_M(1)^{M-1}\tau_M(\mathcal{O}\otimes\delta_1) = \lim_M M\frac{\tau(\mathcal{O})}{\tau(1)}\frac{\tau(1)}{M} = \tau(\mathcal{O})
\end{equation}
where we have used the fact that $\tau_M(1) = (\frac{\tau(1)}{M} + (1 - \frac{\tau(1)}{M}))\frac{\tau(1)}{\tau(1)} = 1$.

We can also derive the Poisson moment formula. Notice that by the same argument of the proof of Equation \ref{equation:combinatorial_moment}, we have:
\begin{equation}
    \tau_M^{\otimes M}(\widetilde{\lambda}_M(\mathcal{O}_1)\cdots\widetilde{\lambda}_M(\mathcal{O}_n)) = \sum_{\sigma\in\mathcal{P}_n, |\sigma| \leq M}|\sigma|!\binom{M}{|\sigma|}\prod_{A\in\sigma}\tau_M(\overrightarrow{\prod}_{j\in A}(\mathcal{O}_j\otimes\delta_1))
\end{equation}
Because $\delta_1$ is a projection, we have $\overrightarrow{\prod}_{j\in A}(\mathcal{O}_j\otimes\delta_1) = (\overrightarrow{\prod}_{j\in A}\mathcal{O}_j)\otimes\delta_1$. Hence we have $\tau_M(\overrightarrow{\prod}_{j\in A}(\mathcal{O}_j\otimes\delta_1)) = \frac{\tau(1)}{M}\frac{\tau(\overrightarrow{\prod}_{j\in A}x_j)}{\tau(1)}$. And hence, we have:
\begin{equation}
    \tau_M^{\otimes M}(\widetilde{\lambda}_M(\mathcal{O}_1)\cdots\widetilde{\lambda}_M(\mathcal{O}_n)) = \sum_{\sigma\in\mathcal{P}_n, |\sigma| \leq M}|\sigma|!\binom{M}{|\sigma|}\frac{1}{M^{|\sigma|}}\prod_{A\in\sigma}\tau(\overrightarrow{\prod}_{j\in A}\mathcal{O}_j)
\end{equation}
where $|\sigma|$ is the number of subsets in the partition $\sigma$. Observe that the standard limit exists for the leading combinatorial factor:
\begin{equation}
    \lim_{M\rightarrow\infty}|\sigma|!\binom{M}{|\sigma|}\frac{1}{M^{|\sigma|}} = \lim_{M\rightarrow\infty}\frac{M(M-1)\cdots(M-|\sigma|+1)}{M^{|\sigma|}} =1 
\end{equation}
Taking the limit, we have the Poisson moment formula:
\begin{equation}
    \varphi_\tau(\widetilde{\lambda}(\mathcal{O}_1)\cdots\widetilde{\lambda}(\mathcal{O}_n)) = \lim_{M}\tau_M^{\otimes M}(\widetilde{\lambda}_M(\mathcal{O}_1)\cdots\widetilde{\lambda}_M(\mathcal{O}_n)) = \sum_{\sigma\in\mathcal{P}_n}\prod_{A\in\sigma}\tau(\overrightarrow{\prod}_{j\in A}\mathcal{O}_j)
\end{equation}

We end this section with two comments:
\begin{enumerate}
    \item As emphasized, an ultraproduct of von Neumann algebras is typically not hyperfinite. However, the Poisson algebra as a proper subalgebra of an ultraproduct is hyperfinite if the input von Neumann algebra is hyperfinite. This is not easy to see in the ultraproduct construction of Poissonization, but it is easy to see in the symmetric tensor product construction of Poissonization. The missing link is to show that both constructions result in the same Poisson algebra and the same Poisson state. This is achieved by a moment-growth estimation\footnote{For physicists, such moment-growth estimations have shown up in the study of quantum chaos and specifically in the study of Universal Operator Growth Hypothesis\cite{PCASA}.} and a noncommutative Hamburger moment theorem\cite{MAW}(see also Appendix \ref{appendix:hamburger}).
    \item We have chosen to present Bernoulli approximation using the algebraic formulation of random variables to emphasize the mathematical structure of this construction and to demonstrate that the noncommutative generalization is natural and straightforward. Now that we are familiar with the noncommutative Bernoulli approximation, we can restate the construction in the more traditional language of random matrices. We devote the final part of this section to this reformulation.
\end{enumerate}
In the following, we restrict our attention to the case where the input von Neumann algebra is a matrix algebra $M_k(\mathbb{C})$ and the input state is the normalized trace $\frac{\tr}{k}$. First we redefine the map $\widetilde{\lambda}_M$ as:
\begin{equation}
    \widetilde{\lambda}_M(\mathcal{O}) = \sum_{1\leq i\leq M}1\otimes\cdots\underbrace{\otimes(\mathcal{O}X^{(i)}_M)\otimes}_\text{i-th position}\cdots\otimes 1
\end{equation}
where $\mathcal{O}\in M_k(\mathbb{C})$ is a matrix and $X_M^{(i)}$ is a classical Bernoulli random variable with success probability $\frac{1}{M}$. Further, $X_M^{(i)}$ and $X_M^{(i')}$ are mutually independent. Since $\mathcal{O}$ is a matrix, $\widetilde{\lambda}_M(\mathcal{O})$ is a sum of $M$ Bernoulli random matrices, where each summand $1\otimes\cdots\underbrace{\otimes(\mathcal{O}X^{(i)}_M)\otimes}_\text{i-th position}\cdots\otimes 1$ is non-zero with probability $\frac{\tau(1)}{M}$. Notice that because of the tensor product structure, each summand $1\otimes\cdots\underbrace{\otimes(\mathcal{O}X^{(i)}_M)\otimes}_\text{i-th position}\cdots\otimes 1$ has at least half of its entries equal to $0$. Therefore, $\widetilde{\lambda}_M(\mathcal{O})$ is a sum of sparse random matrices. 

The limiting state is the familiar state in matrix models - namely, the expectation value of the normalized trace:
\begin{equation}
    \lim_M \mathbb{E}_M\frac{\tr}{k^M}(\widetilde{\lambda}_M(\mathcal{O}_1)\cdots\widetilde{\lambda}_M(\mathcal{O}_n))
\end{equation}
where $k^M$ is the dimension of the random matrix $\widetilde{\lambda}_M(\mathcal{O})$. Here, the expectation value is taken over all independent Bernoulli random variables. Notice that the normalized trace can be decomposed into a tensor product of normalized traces on $k\times k$-matrices:
\begin{equation}
    \frac{\tr}{k^M} = (\frac{\tr}{k})^{\otimes M}
\end{equation}
where $\tr$ is the unnormalized trace and we have used the fact that $M_{k^M}(\mathbb{C}) = \otimes^M M_k(\mathbb{C})$. Using these observations, we can calculate the finite $N$ moment formula. Notice that the combinatorics in the calculation is the same combinatorics as before (see also \ref{equation:combinatorial_moment}):
\begin{align}
    \begin{split}
        \mathbb{E}_M\frac{\tr}{k^M}(\widetilde{\lambda}(\mathcal{O}_1)\cdots\widetilde{\lambda}(\mathcal{O}_n)) &= \sum_{\vec{i}:[n]\rightarrow[M]}\prod_{1\leq \alpha\leq M}\mathbb{E}\frac{\tr}{k}(\overrightarrow{\prod}_{j\in i^{-1}(\alpha)}(\mathcal{O}_jX_M^{(\alpha)}))
        \\
        &=\sum_{\vec{i}:[n]\rightarrow[M]}\prod_{1\leq \alpha\leq M}\mathbb{E}\big((X_M^{(\alpha)})^{|i^{-1}(\alpha)|}\big)\frac{\tr}{k}(\overrightarrow{\prod}_{j\in i^{-1}(\alpha)}\mathcal{O}_j)
    \end{split}
\end{align}
Because powers of a Bernoulli random variable with value in $\{0,1\}$ is the same random variable, we have:
\begin{equation}
    \mathbb{E}\big((X_M^{(\alpha)})^{i^{-1}(\alpha)}\big) = \mathbb{E}(X_M^{(\alpha)}) = \frac{1}{M}
\end{equation}
Hence, the finite $M$ moment formula can be written as:
\begin{align}
    \begin{split}
        \mathbb{E}_M\frac{\tr}{k^M}(\widetilde{\lambda}(x_1)\cdots\widetilde{\lambda}(x_n)) &= \sum_{\vec{i}:[n]\rightarrow[M]}\prod_{1\leq\alpha\leq M}\frac{1}{M}\frac{\tr}{k}(\overrightarrow{\prod}_{j\in i^{-1}(\alpha)}\mathcal{O}_j)
        \\&=\sum_{\sigma\in\mathcal{P}_n, |\sigma|\leq M}|\sigma|!\binom{M}{|\sigma|}\frac{1}{M^{|\sigma|}}\prod_{S\in \sigma}\frac{\tr}{k}(\overrightarrow{\prod}_{j\in S}\mathcal{O}_j)
    \end{split}
\end{align}
Therefore, in the large $M$ limit, we have the Poisson moment formula:
\begin{equation}
    \lim_{M\rightarrow\infty}\mathbb{E}_M\frac{\tr}{k^M}(\widetilde{\lambda}(\mathcal{O}_1)\cdots\widetilde{\lambda}(\mathcal{O}_n)) = \sum_{\sigma\in\mathcal{P}_n}\prod_{S\in\sigma}\frac{\tr}{k}(\overrightarrow{\prod}_{j\in S}\mathcal{O}_j)
\end{equation}
In short, \textbf{\textit{Poissonization of matrix algebra is the large $M$ limit of particular sparse Bernoulli random matrices}}. 

\subsection{Properties of Poissonization}

In this section, we mention three key properties of Poissonization:
\begin{enumerate}
    \item The GNS Hilbert space of the Poisson state is a symmetric Fock space. The Poisson algebra represented on this symmetric Fock space is in its standard form. In particular, the commutant of the Poisson algebra is another Poisson algebra.
    \item Poissonization lifts unital CP maps on the one-particle algebra of observables to unital CP maps on the Poisson algebra.
    \item Modular automorphism associated with the input weight is lifted to a modular automorphism of the Poisson state.
\end{enumerate}

\subsubsection{Standard representation of Poisson algebra}
\begin{theorem}
    The GNS Hilbert space of the Poisson algebra $\mathbb{P}_\omega \mathcal{N}$ under the Poisson state $\varphi_\omega$ is isomorphic to the symmetric Fock space $\mathcal{F}_\text{sym}(L_2(\mathcal{N},\omega))$ where $L_2(\mathcal{N},\omega)$ is the GNS Hilbert space of $N$ under the (finite) weight $\omega$.
\end{theorem}
We construct a basis for the GNS Hilbert space explicitly using an inductive procedure\footnote{This inductive construction is similar to the classical moment-to-cumulant transformation.}. Throughout the construction, all operators $\mathcal{O}$ are assumed to have zero singleton expectation value: $\omega(\mathcal{O})= 0$ \footnote{When $\omega$ is finite, this is not restrictive since we can always subtract the expectation value from the operator, i.e. $\mathcal{O}\mapsto \mathcal{O} - \omega(\mathcal{O})\mathbbm{1}$.}. Consider the following recursive formula:
\begin{equation}\label{equation:recursion}
    \lambda_\emptyset(\mathcal{O}_1,...,\mathcal{O}_n) := \lambda(\mathcal{O}_1)\lambda_\emptyset(\mathcal{O}_2,...,\mathcal{O}_n) - \sum_{2\leq i\leq n}\lambda_\emptyset(\mathcal{O}_2,...,\mathcal{O}_1\mathcal{O}_i,...,\mathcal{O}_n)
\end{equation}
For the base case $n = 1$, we define:
\begin{equation}
    \lambda_\emptyset(\mathcal{O}) := \lambda(\mathcal{O})
\end{equation}
Using the definition of $\lambda(\mathcal{O}) = \sum_{N\geq 0}\sum_{1\leq j\leq N}1\otimes...\underbrace{\otimes \mathcal{O}\otimes}_\text{the j-th position}...\otimes 1$, the following formula for $\lambda_\emptyset(\mathcal{O}_1,...,\mathcal{O}_n)$ holds:
\begin{equation}\label{equation:generalbasis}
    \lambda_\emptyset(\mathcal{O}_1,...,\mathcal{O}_n):= \sum_{N\geq n}\lambda_{\emptyset,N}(\mathcal{O}_1,...\mathcal{O}_n)  = \sum_{N\geq n}\mathcal{O}_1\otimes_\text{sym}...\otimes_\text{sym}\mathcal{O}_n\otimes_\text{sym} 1^{\otimes_\text{sym}(N - n)}
\end{equation}
where $\otimes_\text{sym}$ is the symmetric tensor product and $\lambda_{\emptyset,N}(\mathcal{O}_1,...,\mathcal{O}_n)$ denotes the $N$-th degree component of $\lambda_\emptyset(\mathcal{O}_1,...,\mathcal{O}_n)$. Notice that there is no contribution from degrees less than $n$. To be more concrete, we give some examples of $\lambda_\emptyset(\mathcal{O}_1,...\mathcal{O}_n)$ when $n$ is small:
\begin{enumerate}
    \item When $n = 2$, the first non-zero term is:
    \begin{equation}
        \lambda_{\emptyset, 2}(\mathcal{O}_1,\mathcal{O}_2) = \mathcal{O}_1\otimes_\text{sym}\mathcal{O}_2 = \mathcal{O}_1\otimes\mathcal{O}_2 + \mathcal{O}_2 \otimes\mathcal{O}_1
    \end{equation}
    \begin{align}
        \begin{split}
            \lambda_{\emptyset, 3}(\mathcal{O}_1,\mathcal{O}_2) = \mathcal{O}_1\otimes_\text{sym}\mathcal{O}_2\otimes_\text{sym} 1 &= \mathcal{O}_1\otimes\mathcal{O}_2 \otimes 1+ \mathcal{O}_2\otimes\mathcal{O}_1\otimes 1 
            \\& + \mathcal{O}_1\otimes 1\otimes\mathcal{O}_2 + \mathcal{O}_2 \otimes1\otimes\mathcal{O}_1
            \\& +1\otimes\mathcal{O}_1\otimes\mathcal{O}_2 + 1\otimes\mathcal{O}_2\otimes\mathcal{O}_1
        \end{split}
    \end{align}
    In general, we need to sum over all possible ways of assigning 2 operators to $N$ slots where each operator is assigned a different slots. For example when $N = 3$, there are $\binom{3}{2}\times 2! = 6$ distinct assignments:
    \begin{equation}
        (1,2)\mapsto(1,2), (2,1), (1,3), (3,1), (2,3), (3,2)
    \end{equation}Denote the set of such set maps as $\mathcal{P}_{2,N}:=\{\vec{j}:[2]\rightarrow[N]\text{s.t.}j_1 \neq j_2\}$, then we can write:
    \begin{equation}
        \lambda_{\emptyset, N}(x_1,x_2) = \sum_{\vec{j}\in\mathcal{P}_{2,N}}1\otimes...\underbrace{\otimes x_1\otimes}_{j_1 = k\text{ ,k-th position}}...\underbrace{\otimes x_2\otimes}_{j_2 = \ell\text{ ,$\ell$-th position}}...\otimes 1
    \end{equation}
    \item When $n = 3$, the first non-zero term is:
    \begin{equation}
        \lambda_{\emptyset, 3}(x_1,x_2,x_3) = \sum_{\pi\in\mathcal{S}_3}x_{\pi(1)}\otimes x_{\pi(2)}\otimes x_{\pi(3)}
    \end{equation}
    where $\mathcal{S}_3$ is the permutation group of $3$ elements. For general terms, an analogous formula as $\lambda_{\emptyset, N}(x_1,x_2)$ holds. Denote the set of all possible ways of assigning 3 operators to $N$ slots as $\mathcal{P}_{3,N}$. Then we have:
    \begin{align}
        \begin{split}
            \lambda_{\emptyset, N}&(x_1,x_2,x_3) \\&=  \sum_{\vec{j}\in\mathcal{P}_{3,N}}1\otimes...\underbrace{\otimes x_1\otimes}_{j_1 = k\text{ ,k-th position}}...\underbrace{\otimes x_2\otimes}_{j_2 = \ell\text{ ,$\ell$-th position}}...\underbrace{\otimes x_3\otimes}_{j_3 = m\text{ ,m-th position}}...\otimes 1
        \end{split}
    \end{align}
\end{enumerate}
The general formula of $\lambda_\emptyset(\mathcal{O}_1,...,\mathcal{O}_n)$ follows from induction and recursive formula (c.f. Equation \ref{equation:recursion}). We sketch the argument below. By inductive hypothesis, $\lambda_\emptyset(\mathcal{O}_2,...,\mathcal{O}_n)$ only contains terms where each operator of $\mathcal{O}_2,...,\mathcal{O}_n$ is placed at different tensorial positions. No two operators are assigned the same position. The inductive formula (c.f. Equation \ref{equation:recursion}) involves a product $\lambda(\mathcal{O}_1)\lambda_\emptyset(\mathcal{O}_2,...,\mathcal{O}_n)$. This product contains two types of terms:
\begin{enumerate}
    \item The newly added operator $\mathcal{O}_1$ is assigned to a previously empty position \footnote{Empty positions are occupied by the identity.}.
    \item The newly added operator $\mathcal{O}_1$ is assigned to a previously occupied position. Suppose that position is occupied with the operator $\mathcal{O}_i$ ($2\leq i\leq n$), then the addition of $\mathcal{O}_1$ modifies the occupant operator to $\mathcal{O}_1\mathcal{O}_i$.  
\end{enumerate}
Then by general symmetric argument, we can see that the second type of terms sum up to exactly $\sum_{2\leq i\leq n}\lambda_\emptyset(\mathcal{O}_2,...,\mathcal{O}_1\mathcal{O}_i,...,\mathcal{O}_n)$. This is canceled by the second term in the recursive formula. The remaining expression is precisely Equation \ref{equation:generalbasis}.

We can calculate the inner product $\bra{\varphi_\omega^{1/2}}\lambda_\emptyset(\widetilde{\mathcal{O}}_1,..,\widetilde{\mathcal{O}}_m), \lambda_\emptyset(\mathcal{O}_1,\cdots,\mathcal{O}_n)\ket{\varphi_\omega^{1/2}}$. Without loss of generality, we assume $m \geq n$:
\begin{align}
    \begin{split}
        \bra{\varphi_\omega^{1/2}}&\lambda_\emptyset(\widetilde{\mathcal{O}}_1,..,\widetilde{\mathcal{O}}_m), \lambda_\emptyset(\mathcal{O}_1,\cdots,\mathcal{O}_n)\ket{\varphi_\omega^{1/2}} \\&= \sum_{N\geq m}\frac{e^{-\omega(1)}}{N!}\omega^{\otimes N}\bigg(\lb\widetilde{\mathcal{O}}^\dagger_1\otimes_\text{sym}\cdots\otimes_\text{sym}\widetilde{\mathcal{O}}^\dagger_m\otimes_\text{sym} 1^{\otimes_\text{sym}(N-m)}\rb\\&\cdot\lb\mathcal{O}_1\otimes_\text{sym}\cdots\otimes_\text{sym}\mathcal{O}_n\otimes_\text{sym} 1^{\otimes_\text{sym}(N-n)}\rb\bigg)
        \\&=\sum_{N\geq m}\frac{e^{-\omega(1)}}{N!}\delta_{nm}(n!\binom{N}{n})\omega(1)^{N-n}\sum_{\pi\in S_n}\prod_{1\leq i\leq n}\omega(\widetilde{\mathcal{O}}^\dagger_i\mathcal{O}_{\pi(i)})
        \\&=\delta_{nm}\sum_{\pi\in S_n}\prod_{1\leq i\leq n}\omega(\widetilde{\mathcal{O}}^\dagger_i\mathcal{O}_{\pi(i)})
    \end{split}
\end{align}
where we have used the following two facts:
\begin{enumerate}
    \item Because the singleton expectation vanishes, the non-zero contributions to $\omega^{\otimes N}((\widetilde{\mathcal{O}}^\dagger_1\otimes_\text{sym}\cdots\otimes_\text{sym}\widetilde{\mathcal{O}}^\dagger_m\otimes_\text{sym}1^{\otimes_\text{sym}(N-m)})(\mathcal{O}_1\otimes_\text{sym}\cdots\otimes_\text{sym}\mathcal{O}_n\otimes_\text{sym} 1^{\otimes_\text{sym}(N-n)}))$ come from pairing $\widetilde{\mathcal{O}}^\dagger_j$ with $\mathcal{O}_i$. Therefore, this expectation value is non-zero if and only if $n = m$
    \item The factor $n!$ in the third line of the derivation comes from the symmetric tensor product. For example $\omega^{\otimes 2}((\widetilde{\mathcal{O}}^\dagger_1\otimes_\text{sym} \widetilde{\mathcal{O}}^\dagger_2)(\mathcal{O}_1\otimes_\text{sym}\mathcal{O}_2)) = \omega^{\otimes 2}((\widetilde{\mathcal{O}}^\dagger_1\otimes \widetilde{\mathcal{O}}^\dagger_2 + \widetilde{\mathcal{O}}^\dagger_2 \otimes \widetilde{\mathcal{O}}^\dagger_1)(\mathcal{O}_1\otimes \mathcal{O}_2 + \mathcal{O}_2\otimes \mathcal{O}_1)) = 2\omega(\widetilde{\mathcal{O}}^\dagger_1\mathcal{O}_1)\omega(\widetilde{\mathcal{O}}^\dagger_2\mathcal{O}_2) + 2\omega(\widetilde{\mathcal{O}}^\dagger_1\mathcal{O}_2)\omega(\widetilde{\mathcal{O}}^\dagger_2\mathcal{O}_1)$. And the binomial factor $\binom{N}{n}$ comes from choosing $n$ positions out of $N$ to place the nontrivial operators. 
\end{enumerate}
This calculation shows that the vectors $\{\lambda_\emptyset(\mathcal{O}_1,...,\mathcal{O}_n)\ket{\varphi_\omega^{1/2}}\}$ form a basis for the symmetric Fock space $\mathcal{F}_\text{sym}\lb L_2(\mathcal{N},\omega)\rb$ \footnote{Notice that this basis of $\lambda_\emptyset$-vectors is \textit{not} a finite-dimensional modification of the familiar symmetric basis in the sense that the difference between $\lambda_\emptyset(\mathcal{O}_1,\cdots,\mathcal{O}_n)\ket{\varphi_\omega^{1/2}}$ and $\mathcal{O}_1\otimes_\text{sym}\cdots\otimes_\text{sym}\mathcal{O}_n\ket{\varphi_\omega^{1/2}}$ is a infinite sum of mutually orthogonal vectors:
\begin{eqnarray}
&&\lambda_\emptyset(\mathcal{O}_1,\cdots,\mathcal{O}_n)\ket{\varphi_\omega^{1/2}} - \mathcal{O}_1\otimes_\text{sym}\cdots\otimes_\text{sym}\mathcal{O}_n\ket{\varphi_\omega^{1/2}} \nn\\
&&= \sum_{N\geq n+1}\mathcal{O}_1\otimes_\text{sym}\cdots\otimes_\text{sym} \mathcal{O}_n\otimes_\text{sym} 1^{\otimes_\text{sym}(N-n)}\ket{\varphi_\omega^{1/2}} 
\end{eqnarray}
From this perspective, the $\lambda_\emptyset$-basis cannot be obtained by a small (i.e., finite-rank or compact) perturbation of the usual symmetric basis.}.

\subsubsection{Poissonization and unital CP maps}

We move on to discuss in what sense Poissonization preserves unital CP maps. If the unital CP map further preserves the input weight (finite and faithful):
\begin{equation}
    \omega_M\circ \Psi = \omega_N
\end{equation}
Then there exists a unital CP map between the two corresponding Poisson algebras:
\begin{equation}
    \Gamma(\Psi):=\sum_{N\geq 0}\Psi^{\otimes N}:\mathbb{P}_{\omega_N}\mathcal{N}\rightarrow\mathbb{P}_{\omega_M}\mathcal{M}
\end{equation}
Concretely, this lifted unital CP map acts on integrated Poisson quantized operators in the obvious way:
\begin{equation}
    \Gamma(\Psi)\lb\Gamma(e^{i\mathcal{O}})\rb = \Gamma\lb\Psi(e^{i\mathcal{O}})\rb
\end{equation}
It also acts on the Poisson quantized operators in the obvious way:
\begin{equation}
    \Gamma(\Psi)\lb\lambda(\mathcal{O})\rb = \lambda\lb\Psi(\mathcal{O})\rb
\end{equation}
The lifted unital CP map preserves the Poisson state:
\begin{equation}
    \varphi_{\omega_M}\circ\Gamma(\Psi) = \varphi_{\omega_N}
\end{equation}

A special example of a unital CP map is a normal $*$-homomorphism. Thus if we have a normal $*$-homomorphism $\pi:\mathcal{N}\rightarrow\mathcal{M}$ that preserves the weight: $\omega_M\circ\pi = \omega_N$, then we have a lifted normal $*$-homomorphism: $\Gamma(\pi):\mathbb{P}_{\omega_N}\mathcal{N}\rightarrow\mathbb{P}_{\omega_M}\mathcal{M}$ such that the Poisson state is preserved: $\varphi_{\omega_M}\circ\Gamma(\pi) = \varphi_{\omega_N}$.

\subsubsection{Modular automorphism of Poisson state}

Finally, we explain the relation between the modular automorphism of the Poisson state and the modular automorphism of the input weight. The Poisson quantization map intertwines these two modular automorphisms in the sense that:
\begin{equation}
    \sigma_t^{\varphi_\omega}(\lambda(\mathcal{O})) = \lambda(\sigma_t^\omega(\mathcal{O}))
\end{equation}
where $\sigma_t^{\varphi_\omega}$ is the modular automorphism of the Poisson state and $\sigma_t^\omega$ is the modular automorphism of the input weight.

In addition, the Poisson quantization map also intertwines the modular conjugations:
\begin{equation}
    J_{\varphi_\omega}(\lambda(\mathcal{O})) = \lambda(J_\omega(\mathcal{O}))
\end{equation}
where $J_{\varphi_\omega}$ is the modular conjugation of the Poisson state and $J_\omega$ is the modular conjugation of the input weight.

This intertwining property allows us to calculate the spectrum of the modular operator $\Delta_{\varphi_\omega}$ from the spectrum of the modular operator $\Delta_\omega$. This is useful to determine the type of the Poisson algebra. We will briefly mention some results along this line in the next section. 

\subsection{General statement of Poissonization}

In this final subsection, we collect the general theorem of Poissonization. The discussions in this work have focused on specific examples of Poissonization. It fails to emphasize the functoriality of Poissonization. As mentioned before, Poissonization is extremely general. Just like the second quantization, which can be applied to any one-particle Hilbert space, Poissonization can be applied to any one-particle operator algebra in the following sense:
\begin{theorem}
    For any von Neumann algebra $\mathcal{N}$ and any normal faithful semifinite weight $\omega$, the Poisson algebra $\mathbb{P}_\omega\mathcal{N}$ and the Poisson state $\varphi_\omega$ are well-defined.
\end{theorem}
\subsection{Noncommutative Hamburger moment theorem}\label{appendix:hamburger}

Throughout this work, we have presented several constructions of Poissonization. Yet we have not shown these constructions are equivalent. Here we quote the result that allows us to prove the equivalence of different constructions of Poissonization. This result is based on a noncommutative generalization of the classical Hamburger's moment theorem\cite{MAW}. This generalization is done in the framework of noncommutative probability theory\cite{quantStoch}. One of the key ingredients of noncommutative probability theory is the correlation function. In a sense, noncommutative probability theory aims to construct and study the properties of operator algebras using combinatorial structures in correlation functions. To demonstrate the use of noncommutative probability theory, we will revisit the well-known second quantization and reformulate its construction using tools from noncommutative probability theory. 

\paragraph{Hamburger moment problem and its reformulation in noncommutative probability theory:} This classical moment problem considers the reverse-engineering question whether a sequence of real numbers can \textit{uniquely} specify a probability measure on $\mathbb{R}$. Suppose two such probability measures have already been found $d\nu_1(x)$ and $d\nu_2(x)$, these two measures are the same if the \textit{moment growth condition} is satisfies:
\begin{equation}
    |\int_{-\infty}^\infty x^nd\nu_1(x)| = |\int_{-\infty}^\infty x^nd\nu_2(x)| = |\mu_n| \leq Ce^nn^n
\end{equation}
where $C>0$ is an absolute constant.

To introduce the noncommutative generalization, we first rethink the essential input to the classical Hamburger moment problem. The problem statement specifies a sequence of real numbers $\{\mu_1,\mu_2,\cdots\}$. One suspects these real numbers might be moments of some probability measure $d\nu(x)$. To make this suspicion mathematically precise, we need two ingredients. First, we consider the abstract algebra generated by two symbols $1$ and $x$. In fact, this algebra is nothing but the familiar algebra of polynomials: $\mathbb{C}[x]$. Note that there exists an adjoint operation on polynomials:
\begin{equation}
    (\sum_{0\leq n\leq N}\alpha_nx^n)^\dagger = \sum_{0\leq n\leq N}\overline{\alpha_n}x^n
\end{equation}
This means that $\mathbb{C}[x]$ is a unital $*$-associative algebra \footnote{Unital means there exists a multiplicative unit in the algebra. Associativity means that the multiplication is associative.}. In addition, the sequence of real numbers $\{\mu_1,\mu_2,\cdots\}$ defines a linear functional:
\begin{equation}
    \phi:\mathbb{C}[x]\rightarrow\mathbb{C}:\sum_{0\leq n\leq N}\alpha_nx^n\mapsto \sum_{0\leq n\leq N}\alpha_n\mu_n
\end{equation}
where $\mu_0$ is defined to be 1. Hence, this linear functional is normalized:
\begin{equation}
    \phi(1) = \mu_0 = 1
\end{equation}

The existence part of the Hamburger moment theorem states that there exists a probability measure on $\mathbb{R}$ whose moments are given by $\{\mu_0,\mu_1,\cdots\}$ if and only if the Hankel operator is positive definite:
\begin{equation}
    \begin{bmatrix}
        \mu_0&\mu_1&\mu_2&\cdots&\mu_N&\cdots\\
        \mu_1&\mu_2&\mu_3&\cdots&\mu_{N+1}&\cdots\\
        \mu_2&\mu_3&\mu_4&\cdots&\mu_{N+2}&\cdots\\
        \cdots\\
        \mu_N&\mu_{N+1}&\mu_{N+2}&\cdots\\
        \cdots
    \end{bmatrix}\geq0
\end{equation}
Equivalently, for any finite sequence of complex parameters $(\alpha_n)_{0\leq n\leq N}$, we have:
\begin{equation}
    \begin{bmatrix}
        \overline{\alpha_0}&\overline{\alpha_1}&\cdots&\overline{\alpha_N}&0&\cdots
    \end{bmatrix}\begin{bmatrix}
        \mu_0&\mu_1&\mu_2&\cdots&\mu_N&\cdots\\
        \mu_1&\mu_2&\mu_3&\cdots&\mu_{N+1}&\cdots\\
        \mu_2&\mu_3&\mu_4&\cdots&\mu_{N+2}&\cdots\\
        \cdots\\
        \mu_N&\mu_{N+1}&\mu_{N+2}&\cdots\\
        \cdots
    \end{bmatrix}
    \begin{bmatrix}
        \alpha_0\\\alpha_1\\\alpha_2\\\cdots\\\alpha_N\\0\\\cdots
    \end{bmatrix} \geq 0
\end{equation}
In other words:
\begin{equation}
    \sum_{0\leq n,m\leq N}\overline{\alpha_n}\alpha_m\mu_{n+m} \geq 0
\end{equation}
In terms of the linear functional $\phi$, we have:
\begin{equation}
    \phi((\sum_{0\leq n\leq N}\alpha_nx^n)^\dagger(\sum_{0\leq n\leq N}\alpha_nx^n)) = \sum_{0\leq n,m\leq N}\overline{\alpha_n}\alpha_m\phi(x^{n+m}) = \sum_{0\leq n,m\leq N}\overline{\alpha_n}\alpha_m\mu_{n+m} \geq 0
\end{equation}
In addition, since the measure is a probability measure, we have a normalization condition:
\begin{equation}
    \int_{-\infty}^\infty d\nu(x) = 1
\end{equation}
where $d\nu(x)$ is the probability measure whose moments are $\{\mu_0,\mu_1,\cdots\}$. This is exactly the normalization of the functional $\phi$. Therefore, the existence of a probability measure on $\mathbb{R}$ whose moments are $\{\mu_0,\mu_1,\cdots\}$ is equivalent to the linear functional $\phi$ being positive and unital:
\begin{equation}
    \phi(f(x)^\dagger f(x))\geq 0\text{ for any }f(x)\in\mathbb{C}[x]
\end{equation}
\begin{equation}
    \phi(1) = 1
\end{equation}
Given such a linear functional, the GNS construction gives a representation of the polynomial algebra on the Hilbert space $L^2(\mathbb{R},d\nu(x))$:
\begin{equation}
    \pi:\mathbb{C}[x]\rightarrow \mathcal{L}(L^2(\mathbb{R},d\nu(x))):f(x)\mapsto M_{f(x)}
\end{equation}
where $\mathcal{L}(L^2(\mathbb{R},d\nu(x)))$ is the space of linear operators densely defined on $L^2(\mathbb{R},d\nu(x))$, and $M_{p(x)}$ is the multiplication operator:
\begin{equation}
    (M_{f(x)}g)(x) = f(x)g(x)
\end{equation}
Notice that this is indeed the GNS construction because the inner product is given by the positive functional $\phi$:
\begin{align}
    \begin{split}
        ||\pi(f(x))1||^2 &= ||M_{f(x)}1||^2 = ||f(x)||^2 =\int_{-\infty}^\infty |f(x)|^2D\nu(x) = \sum_{0\leq n\leq N}\overline{\alpha_n}\alpha_m\int_{-\infty}^\infty x^{n+m}d\nu(x)
        \\
        &=\sum_{0\leq n,m\leq N}\overline{\alpha_n}\alpha_m\mu_{n+m} = \phi(f(x)^\dagger f(x)) < \infty
    \end{split}
\end{align}
where $f(x)=\sum_{0\leq n\leq N}\alpha_nx^n$. In particular, all polynomials are elements of the Hilbert space $L^2(\mathbb{R},d\nu(x))$. The representation $\pi$ is unital:
\begin{equation}
    \pi(1) = id
\end{equation}
And it preserves the adjoint:
\begin{equation}
    \pi(f(x)^\dagger) = M_{f(x)}^\dagger
\end{equation}
Moreover, although the multiplication operator $M_{f(x)}$ may be unbounded in general, all polynomials are in its domain, because all moments of the probability distribution are finite:
\begin{equation}
    ||M_{f(x)}g(x)||^2 = \int_{-\infty}^\infty |f(x)g(x)|^2D\nu(x) = \phi(|f(x)g(x)|^2) = \sum_{0\leq n,m\leq N}\overline{\beta_n}\beta_m\mu_{n+m} < \infty
\end{equation}
where $f(x)g(x) = \sum_{0\leq n\leq N}\beta_nx^n$. 

The pair $(\mathbb{C}[x],\phi)$ is a basic example of a noncommutative probability space. The representation $\pi$ is a basic example of the GNS representation of a noncommutative probability space. More generally, a noncommutative probability space is a pair of data:
\begin{enumerate}
    \item A unital $*$-associative algebra $\mathcal{A}$
    \item And a normalized positive linear functional $\phi:\mathcal{A}\rightarrow\mathbb{C}$ such that $\phi(x^\dagger x)\geq 0$ and $\phi(1) = 1$
\end{enumerate}

Previously we construct the GNS representation of $(\mathbb{C}[x],\phi)$ on the Hilbert space $L^2(\mathbb{R},d\nu)$. However, the choice of Hilbert space is not unique. Different choices result in different models of the same noncommutative probability space $(\mathbb{C}[x],\phi)$. For example, one may find another probability measure $d\nu'(x)$ on $\mathbb{R}$ whose moments are $\{\mu_0,\mu_1,\cdots\}$. Then the GNS representation of $(\mathbb{C}[x],\phi)$ can act on the Hilbert space $L^2(\mathbb{R},d\nu'(x))$. Within each GNS representation, the noncommutative probability space generates a von Neumann algebra. This is the algebra generated by (the spectral projections of) the basis elements:
\begin{equation}
    \overline{\pi_\nu(\mathbb{C}[x])} := \langle\pi_\nu(x), \pi_\nu(1)\rangle\subset\mathbb{B}(L^2(\mathbb{R},d\nu))
\end{equation}
where $\{1,x\}$ are the generators that generate the polynomial algebra and $\pi_\nu$ is the GNS representation of $(\mathbb{C}[x],\phi)$. The probability measure $\nu$ induces a state on this von Neumann algebra:
\begin{equation}
    \phi_\nu\bigg(\pi_\nu\big(f(x)\big)\bigg) := \int_{-\infty}^\infty \bigg(\pi_\nu\big(f(x)\big)1\bigg) d\nu(x) = \int_{-\infty}^\infty M_{f(x)}1d\nu(x) = \int_{-\infty}^\infty f(x)d\nu(x)
\end{equation}

The uniqueness of the probability measure can now be translated into the uniqueness of the von Neumann algebra generated by the noncommutative probability space $(\mathbb{C}[x],\phi)$. The probability measure is unique if and only if there exists an isomorphism of von Neumann algebras:
\begin{equation}
    \psi:\overline{\pi_\nu(\mathbb{C}[x])}\rightarrow\overline{\pi_{\nu'}(\mathbb{C}[x])}
\end{equation}
such that the map preserves the state: $\phi_{\nu'}\circ\psi = \phi_\nu$. And by the Hamburger moment theorem, the measure is unique if and only if the moment growth condition is satisfied:
\begin{equation}
    \phi(x^n) = \int_{-\infty}^\infty x^nd\nu(x) \leq Ce^nn^n
\end{equation}
The noncommutative probability space $(\mathbb{C}[x],\phi)$ is generated by a single non-trivial element $x$. In general, a noncommutative probability space $(\mathcal{A},\phi)$ can be generated by a set $S:=\{s_i\in\mathcal{A}\}_{i\in I}$. And a noncommutative moment is defined to be an expression of the form:
\begin{equation}
    \phi(s_1s_2\cdots s_n)
\end{equation}

We are finally ready to state the noncommutative generalization of the Hamburger moment theorem \cite{MAW}:
\begin{theorem}
    Let $(\mathcal{A},\phi)$ be a noncommutative probability space. Let $\pi_1,\pi_2$ be two representations of $(\mathcal{A},\phi)$. Then there exists a state-preserving von Neumann algebra isomorphism:
    \begin{equation}
        \psi:\overline{\pi_1(\mathcal{A})}\rightarrow\overline{\pi_2(\mathcal{A})}
    \end{equation}
    if and only if there exists a length function on a generating set $S$ of $\mathcal{A}$:
    \begin{equation}
        \ell:S\rightarrow(0,\infty)
    \end{equation}
    such that the noncommutative \textit{moment growth} condition is satisfied:
    \begin{equation*}
        |\phi(s_1\cdots s_n)| \leq Cn^n\prod_{1\leq i\leq n}\ell(s_i)
    \end{equation*}
\end{theorem}
In the case of $(\mathbb{C}[x],\phi)$, we can simply take the length function to be the constant function with value being the natural exponent:
\begin{equation}
    \ell(x) := e
\end{equation}

\paragraph{Equivalence of different constructions of Poissonization:} We can now immediately apply this theorem to show the equivalence of various constructions of Poissonization. The noncommutative probability space is the \textit{free} algebra generated by the \textit{abstract} symbols:
\begin{equation}
    \mathcal{A}_\text{Poiss}:=\langle1,\lambda(\mathcal{O}):x\in\mathcal{N}_\text{Tomita}\rangle
\end{equation}
where $\mathcal{N}_\text{Tomita}$ is the Tomita subalgebra of $\mathcal{N}$ under the weight $\omega$. When $\omega(1) < \infty$, $\mathcal{N}_\text{Tomita}$ can be taken as the entire algebra $\mathcal{N}$.

$\mathcal{A}_\text{Poiss}$ is a unital $*$-associative algebra where elements are words:
\begin{equation}
    \lambda(\mathcal{O}_1)\lambda(\mathcal{O}_2)\cdots\lambda(\mathcal{O}_n)
\end{equation}
The adjoint operation is given by:
\begin{equation}
    \lambda(\mathcal{O})^\dagger := \lambda(\mathcal{O}^\dagger)
\end{equation}
The state $\phi_\text{Poiss}$ is defined explicitly by the formula:
\begin{equation}
    \phi_\text{Poiss}(\lambda(\mathcal{O}_1)\cdots\lambda(\mathcal{O}_n)) = \sum_{\sigma\in\mathcal{P}_n}\prod_{A\in\sigma}\omega(\overrightarrow{\prod}_{i\in A}x_i)
\end{equation}
where $\mathcal{P}_n$ is the set of partitions on $\{1,..,n\}$ and $A\in\sigma$ is an element in the partition $\sigma$. 

On the von Neumann algebra $\mathcal{N}$ (or the Tomita algebra $\mathcal{N}_\text{Tomita}$ in case $\omega$ is infinite), we define the semi-norm:
\begin{equation}
    |||x||| := \max\{||x||, \omega(x^\dagger x)^{1/2}, \omega(xx^\dagger)^{1/2}\}
\end{equation}
And the length function is a rescaling of the semi-norm:
\begin{equation}
    \ell(x):=e|||x|||
\end{equation}
Then the \textit{Poisson moment formula} satisfies the growth condition:
\begin{align}
    \begin{split}
        |\phi_\text{Poiss}(\lambda(\mathcal{O}_1)\cdots\lambda(\mathcal{O}_n))| &\leq \sum_{\sigma\in\mathcal{P}_n}\prod_{A\in\sigma}|\omega(\overrightarrow{\prod}_{i\in A}x_i)| \leq \sum_{\sigma\in\mathcal{P}_n}\prod_{A\in\sigma}\prod_{i\in A}|||x_i|||
        \\
        &\leq |\mathcal{P}_n| \prod_{1\leq i\leq n}|||x_i||| \leq Ce^nn^n\prod_{1\leq i\leq n}|||x_i||| = Cn^n\prod_{1\leq i\leq n}\ell(x_i)
    \end{split}
\end{align}
where we have used the Cauchy-Schwartz estimate:
\begin{align}
    \begin{split}
        |\omega(xy_1\cdots y_nz)| &\leq \omega(xx^\dagger)^{1/2}\omega(z^\dagger y^\dagger y z)^{1/2}\leq \omega(xx^\dagger)^{1/2}||y_1..y_n||\omega(z^\dagger z)^{1/2} \\
        &\leq \omega(xx^\dagger)^{1/2}\prod_{1\leq i\leq n}||y_i||\omega(z^\dagger z)^{1/2}
        \leq |||x|||(\prod_{1\leq i\leq n}|||y_i|||)|||z|||
    \end{split}
\end{align}
In addition, we have used the fact that the cardinality of the set of partitions $\mathcal{P}_n$ is given by the classical Bell number $B_n$, and the Bell number is bounded above by:
\begin{equation}
    B_n\leq Ce^nn^n
\end{equation}
where $C>0$ is an absolute constant.

Therefore, from the noncommutative Hamburger moment theorem, all formulations of Poissonization are equivalent because all formulations take the same input and produce the same Poisson moment formula, which satisfies the growth condition.

{\it The key idea of this proof is that the Poisson algebra can be fully characterized by the combinatorial structure of the Poisson moment formula.} The abstract noncommutative probability space $(\mathcal{A}_\text{Poiss}, \phi_\text{Poiss})$ serves to formalize the combinatorial structure and to add in the necessary analytical control over the moment formula. To determine an analytical object (in this case, the Poisson algebra) from combinatorial information is one of the key ideas in noncommutative probability theory. Another example of noncommutative probability theory is free probability, where in physics language, many of the properties are fixed by the structures of large $N$ correlation functions. Yet another example is the second quantization. As a demonstration, we elaborate on the noncommutative probability perspective on the second quantization below.

    \bibliographystyle{unsrt}
    \bibliography{abstractPoiss}
 \end{document}